\documentclass[10pt,b5paper,openany]{book} 
\usepackage[b5paper, portrait,left=33mm,right=25mm,top=25mm,bottom=25mm]{geometry}

\usepackage[T1]{fontenc}
\usepackage{lmodern}
\usepackage{tabularx}
\usepackage{graphicx}
\usepackage[figurename=Figure,labelfont=bf,labelsep=period]{caption}
\usepackage{subcaption}
\usepackage{amsmath}
\usepackage{physics}
\usepackage{tikz-cd}
\usepackage{amsthm}
\usepackage{mathtools}
\usepackage{newtxtext,newtxmath}
\usepackage[colorlinks=true,citecolor=red,linkcolor=black]{hyperref}
\usepackage{enumitem}
\usepackage{setspace}
\usepackage{array}
\usepackage[square,numbers]{natbib}
\usepackage{multicol}
\usepackage{multirow}
\usepackage{xtab}
\usepackage{longtable}
\usepackage[titletoc]{appendix}
\usepackage{pdfpages}

\newtheorem{theorem} {Theorem}[section]

\newtheorem{definition}[theorem] {Definition}

\newtheorem{corollary}[theorem] {Corollary}
\newtheorem{conjecture}{Conjecture}
\newtheorem{proposition} [theorem]{Proposition}

\newcommand*\had{\vcenter{\hbox{\includegraphics[scale=0.4]{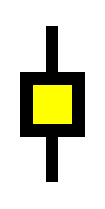}}}}
\newcommand*\white{\vcenter{\hbox{\includegraphics[scale=0.15]{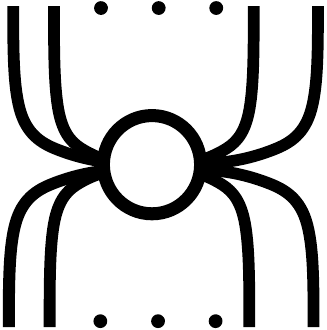}}}}
\newcommand*\blue{\vcenter{\hbox{\includegraphics[scale=0.15]{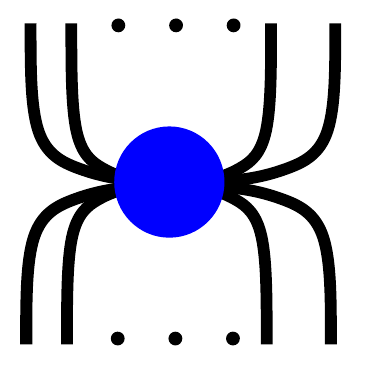}}}}
\newcommand*\blueb{\vcenter{\hbox{\includegraphics[scale=0.15]{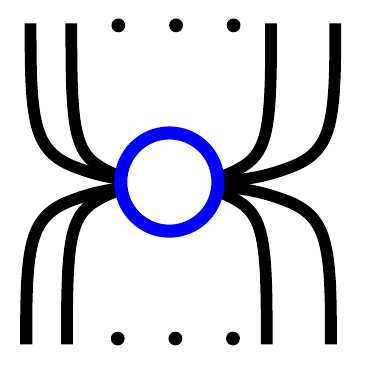}}}}
\newcommand*\black{\vcenter{\hbox{\includegraphics[scale=0.15]{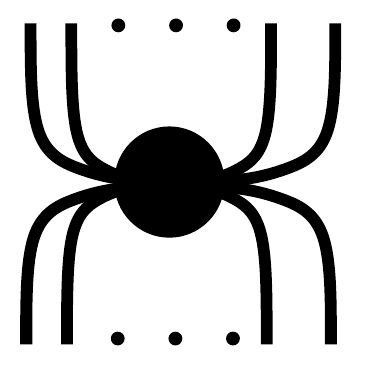}}}}
\newcommand*\blackb{\vcenter{\hbox{\includegraphics[scale=0.15]{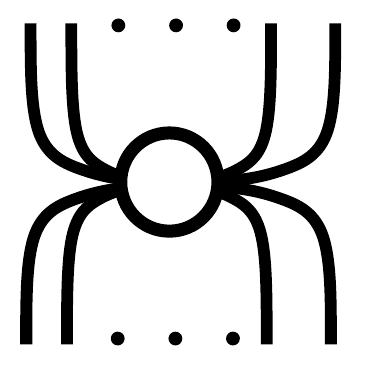}}}}

\newcommand*\cwgr{\vcenter{\hbox{\includegraphics[scale=0.15]{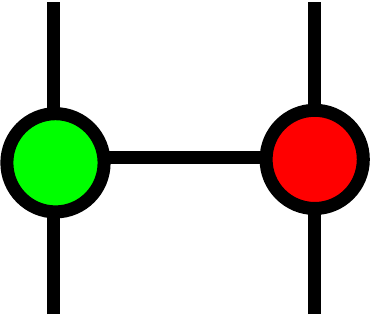}}}}
\newcommand*\cwgg{\vcenter{\hbox{\includegraphics[scale=0.15]{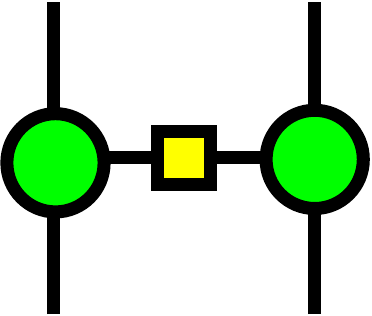}}}}
\newcommand*\cwrr{\vcenter{\hbox{\includegraphics[scale=0.15]{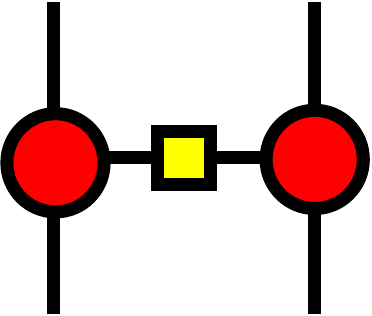}}}}
\newcommand*\zspider{\vcenter{\hbox{\includegraphics[scale=0.15]{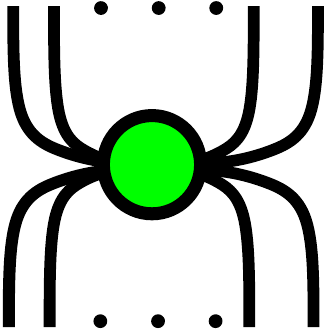}}}}
\newcommand*\xspider{\vcenter{\hbox{\includegraphics[scale=0.15]{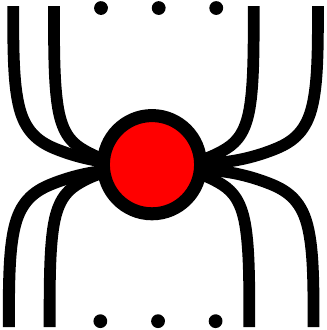}}}}

\newcommand*\wunit{\vcenter{\hbox{\includegraphics[scale=0.15]{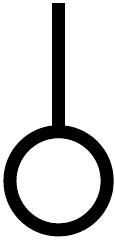}}}}
\newcommand*\bunit{\vcenter{\hbox{\includegraphics[scale=0.15]{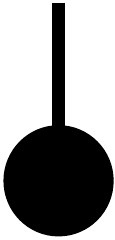}}}}
\newcommand*\sqrtwo{\vcenter{\hbox{\includegraphics[scale=0.15]{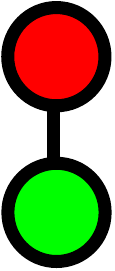}}}}
\newcommand*\starscalar{\vcenter{\hbox{\includegraphics[scale=0.15]{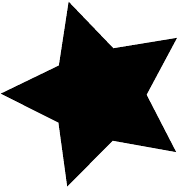}}}}

\newcommand*\purple{\vcenter{\hbox{\includegraphics[scale=0.15]{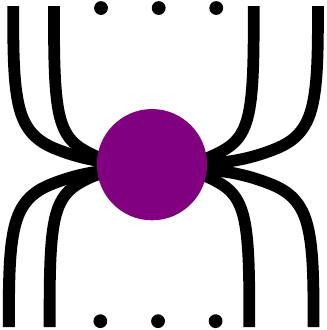}}}}

\newcommand*\orange{\vcenter{\hbox{\includegraphics[scale=0.15]{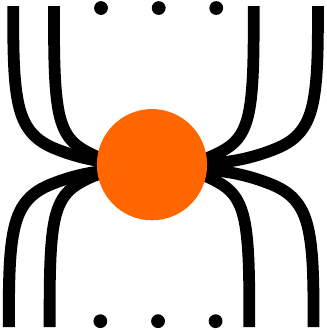}}}}

\newcommand*\pink{\vcenter{\hbox{\includegraphics[scale=0.15]{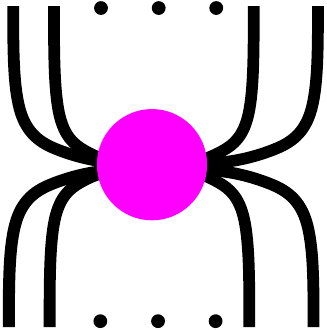}}}}
\newcommand*\brown{\vcenter{\hbox{\includegraphics[scale=0.15]{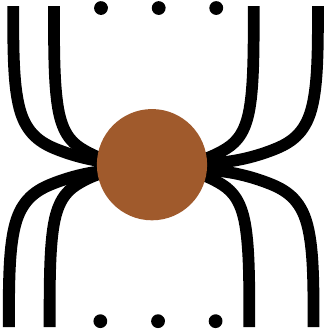}}}}

\newcommand*\cyan{\vcenter{\hbox{\includegraphics[scale=0.15]{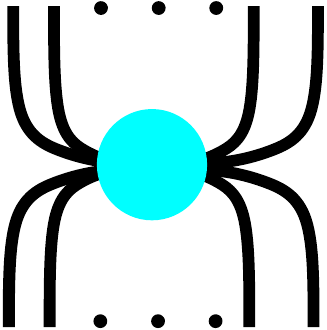}}}}

\newcommand*\pinkb{\vcenter{\hbox{\includegraphics[scale=0.15]{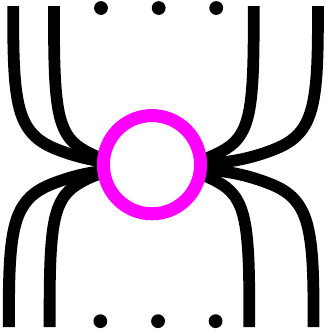}}}}

\newcommand*\cwbkp{\vcenter{\hbox{\includegraphics[scale=0.15]{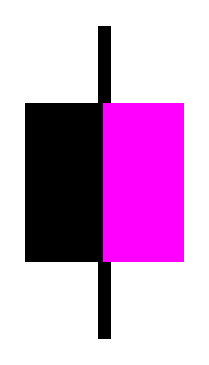}}}}
\newcommand*\egonebaseone{\vcenter{\hbox{\includegraphics[scale=0.2]{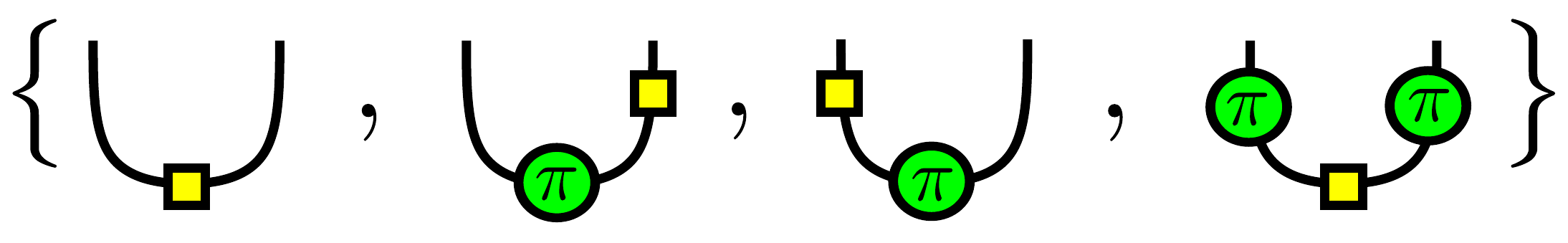}}}}
\newcommand*\egonebasetwo{\vcenter{\hbox{\includegraphics[scale=0.2]{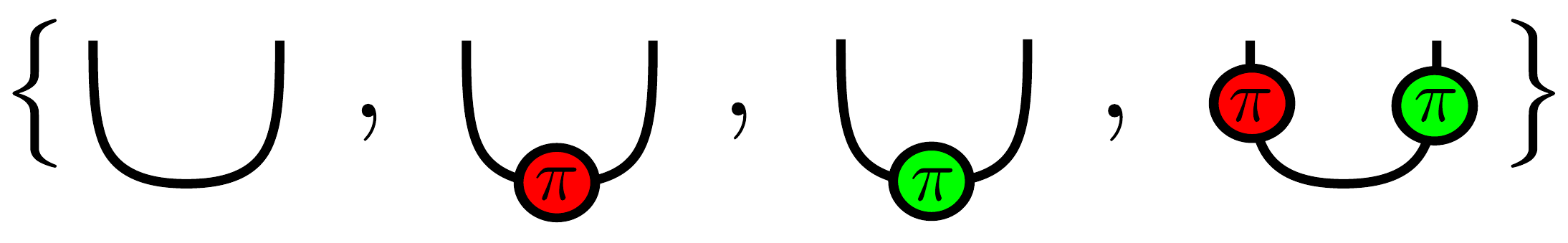}}}}

\begin{document}

\frontmatter
% You will need to make the title all-caps
%\title{Role of Compositionality in Constructing Complementary Classical Structures Within Qubit Systems}

%\author{Siti Aqilah Muhamad Rasat}
%\author[2]{Author Two}
%\author[3]{Author Three}

%\affil[1]{Institute for Mathematical Research (INSPEM), Universiti Putra Malaysia, 43400 Serdang, Selangor, Malaysia. Email: aqilahrasat@gmail.com}
%\affil[2]{Second affiliation address}
%\affil[2]{Third affiliation address}

%\maketitle

%\input{title}

\begin{titlepage}
\begin{center}
    \vspace*{1cm}
    \includegraphics[]{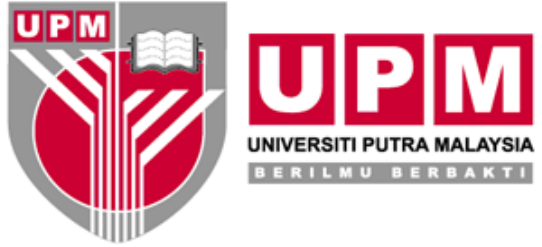}
    \vspace{1cm}
    \\
    \textbf{THE ROLE OF COMPOSITIONALITY IN CONSTRUCTING COMPLEMENTARY CLASSICAL STRUCTURES WITHIN QUBIT SYSTEMS}
    \vfill
    \textbf{By}
    \vspace{0.5cm}
    \\
    \textbf{SITI AQILAH BINTI MUHAMAD RASAT}
    \vfill
    \textbf{Thesis Submitted to the School of Graduate Studies, Universiti Putra Malaysia,}\\
    \textbf{in Fulfilment of the Requirements for the Degree of Master of Science}\\
    % date of viva
\end{center}
\end{titlepage}

\newpage
\noindent All material contained within the thesis, including without limitation text, logos, icons, photographs and all other artwork, is copyright material of Universiti Putra Malaysia unless otherwise stated. Use may be made of any material contained within the thesis for non-commercial purposes from the copyright holder. Commercial use of material may only be made with the express, prior, written permission of Universiti Putra Malaysia.\\

\noindent Copyright $\copyright$ Universiti Putra Malaysia

\newpage

\newpage

\addcontentsline{toc}{chapter}{Abstract}

\begin{center}
Abstract of thesis presented to the Senate of Universiti Putra Malaysia in fulfilment of the requirement for the degree of Master of Science.\\

\vspace{1cm}

\textbf{THE ROLE OF COMPOSITIONALITY IN CONSTRUCTING COMPLEMENTARY CLASSICAL STRUCTURES WITHIN QUBIT SYSTEMS}\\
\vspace{1cm}
By\\
\vspace{1cm}
\textbf{SITI AQILAH BINTI MUHAMAD RASAT}\\

%MONTH YEAR VIVA
\end{center}

\vspace{1cm}

\begin{longtable}{lcl}
\textbf{Chair} & \textbf{:} & \textbf{Prof. Madya Dr. Hishamuddin Bin Zainuddin}\\
\textbf{Faculty} & \textbf{:} & \textbf{Institute For Mathematical Research}
\end{longtable}

\vspace{1cm}

\noindent Observables in a quantum system, represented by a Hilbert space, are given by the orthogonal bases of the aforementioned Hilbert space. Categorical Quantum Mechanics provides further abstraction of such observables, allowing for a diagrammatic representation of measurements that extends to quantum processes. Our research studies this abstraction of observables, which has been dubbed as classical structures, in a subtheory of quantum mechanics which focuses on qubit systems (or 2-dimensional quantum system and its composites). We have constructed a procedure that takes the complementary classical structures of a single qubit system and compose them separably via the Kronecker product or 'entangle' them via Bell states to obtain complementary classical structures in $n$-qubit systems. In this present work, we apply our procedure to two qubit and three qubit systems as examples. Then, using rewriting rules of ZX-calculus and tools in graph theory, we searched for maximal complete sets of mutually complementary classical structures (the categorical counterpart of mutually unbiased bases)  among our constructed composite classical structures. For two qubits, we found 13 maximal complete sets of mutually complementary classical structures, and for three qubits, we found 32,448 maximal complete sets of mutually complementary classical structures.

\newpage

\addcontentsline{toc}{chapter}{Abstrak}

\begin{center}
Abstrak tesis yang dikemukakan kepada Senat Universiti Putra Malaysia sebagai memenuhi keperluan untuk ijazah Sarjana Sains.\\

\vspace{1cm}
\textbf{PERANAN PENGGUBAHAN DALAM MEMBINA STRUKTUR KLASIK PELENGKAP DALAM SISTEM QUBIT}\\
\vspace{1cm}
Oleh\\
\vspace{1cm}
\textbf{SITI AQILAH BINTI MUHAMAD RASAT}\\

%MONTH YEAR VIVA
\end{center}

\vspace{1cm}

\begin{longtable}{lcl}
\textbf{Pengerusi} & \textbf{:} & \textbf{Prof. Madya Dr. Hishamuddin Bin Zainuddin}\\
\textbf{Fakulti} & \textbf{:} & \textbf{Institut Penyelidikan Matematik}
\end{longtable}

\vspace{1cm}

\noindent Pembolehcerap dalam sistem kuantum, diwakili oleh ruang Hilbert, diberi oleh asas ortogon ruang Hilbert. Mekanik Kuantum Berkategori menyediakan selanjutnya keabstrakan pembolehcerap tersebut, membenarkan perwakilan berdiagram pengukuran yang merangkumi proses kuantum. Penyelidikan kami mengkaji keabstrakan pembolehcerap yang dijuluki sebagai struktur klasik, dalam subteori mekanik kuantum yang memfokuskan sistem qubit (atau sistem kuantum dua dimensi dan kompositnya). Kami telah membangunkan satu prosedur bagi struktur klasik pelengkap satu sistem qubit tunggal dan menghurainya secara terpisah melalui hasil darab Kronecker atau mengusutkannya melalui keadaan Bell bagi memperolehi struktur klasik pelengkap dalam sistem n-qubit. Dalam kajian ini, kami menggunakan prosedur tersebut untuk sistem dua dan tiga qubit sebagai contoh. Kemudian, dengan petua penulisan semula kalkulus ZX dan alatan teori graf, kami menggelintar bagi set lengkap maksimum struktur klasik yang saling melengkapi (struktur berkategori setara bagi asas saling saksama) antara struktur klasik komposit yang terbina. Bagi dua qubit, kami menjumpai 13 set lengkap maksimum struktur klasik yang saling melengkapi, dan bagi 3 qubit, kami menjumpai 32,448 set lengkap maksimum struktur klasik yang saling melengkapi.

\tableofcontents
\setlength{\parskip}{2em}

\addcontentsline{toc}{chapter}{List of Figures}
\listoffigures

\addcontentsline{toc}{chapter}{List of Tables}
\listoftables

%\input{abbrev}
% make a note that set of complementary classical structures means set of mutually complementary classical structures, unless otherwise stated

\addcontentsline{toc}{chapter}{List of Abbreviations}
\chapter*{List of Abbreviations}

\begin{longtable}{p{2cm}p{8.5cm}}
SLOCC & Stochastic Local Operations and Classical Communications\\
\\
TQFT & Topological Quantum Field Theory\\
\\
CQM & Categorical Quantum Mechanics\\
\\
CPM & Completely Positive Map(s)\\
\\
LHS & Left Hand Side\\
\\
RHS & Right Hand Side\\
\\
CS & Classical Structure(s)\\
\\
SC & Separably Composed\\
\\
NS & Non-Separable\\
\\
SCCS & Separably Composed Classical Structure(s)\\
\\
NSCS & Non-Separable Classical Structure(s)\\
\\
CD & Complementarity Diagram(s)
\end{longtable}

\mainmatter

\chapter{Introduction}
% add figures now!

Categorical quantum mechanics (CQM) reconstructs quantum processes as diagrams, providing an intuitive way of performing computations that significantly simplifies complex equation-based calculations in quantum mechanics. For example, quantum teleportation can be described using the following picture in CQM:
\begin{equation*}
    \includegraphics[scale=0.3]{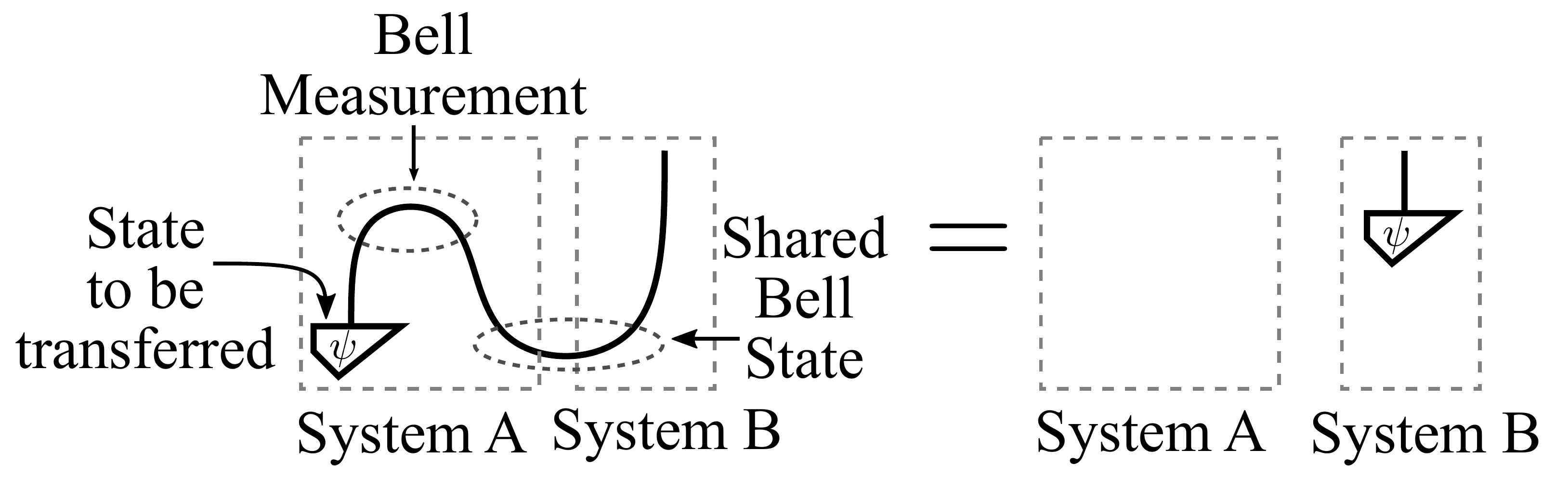}
\end{equation*}
where the boxes labelled `System $A$' and `System $B$' represent spatially separated systems. 

The same picture could be translated into a quantum circuit. However, the diagrams of CQM not only can be translated into mathematical terms, they form a category, specifically of the monoidal type, and if one wishes to express these diagrams as morphisms within a category in a more traditional manner, they can do so via a functor (details can be found in \cite{AbramskyCQM2009,Selinger2007a}). CQM has also evolved since its inception to provide a more general setting for processes which allows for a more detailed picture of quantum processes. This can be found in Section 4.2 of \cite{Coecke2016a}. 

In this chapter, we review some basic concepts of the so called process theories and structures within process theories that are integral to our work. The definitions are taken from various articles on CQM and process theories \cite{AbramskyCQM2009,Coecke2006, Coecke2011b, Coecke2011e, Coecke2015CQM}. Then we provide a brief outline of this thesis: our objectives, the procedure used to obtain our objective, and our results.

\section{Process Theories}\label{sec:process-intro} 

A process theory contains two main ingredients: a collection of systems (e.g. natural numbers, grocery items, information), classified by their types, and a collection of processes (e.g. arithmetic, food preparation, algorithm) which transform those systems. For each process, the type of system that it can transform and the type of system it produces must be specified. Furthermore, there is a means of composing processes of which the composition itself is a process in the theory. 
\begin{figure}[h]
\centering
    \includegraphics[scale=0.3]{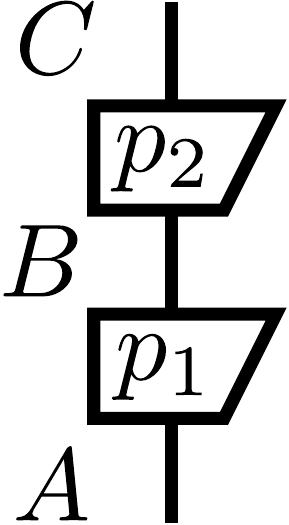}
    \caption{Sequential composition of processes.}
    \label{fig:compose}
\end{figure}

In this thesis, we adopt the convention of reading diagrams from bottom to top. So, in Fig. \ref{fig:compose}, we have two processes which are composed sequentially, that is, the process labelled $p_1$ is followed by the process labelled $p_2$. Notice that the two processes are joined together by the wire between them. Here, we need a compatibility condition. $p_1$ transforms some system of type $A$ into some system of type $B$. If $p_2$ does not transform a system of type $B$, it should not be able to transform the system produced by $p_1$ and the diagram above is meaningless. Therefore, when composing two processes, $p_1$ and $p_2$, sequentially, the type of system produced by $p_1$ must match the type of system which $p_2$ transforms.

In summary, 
\begin{definition}\cite{Coecke2011e, Coecke2015CQM}
A process theory consists of:
\begin{enumerate}
\item A collection $T$ of system-types, denoted by wires:
\begin{equation*}
    \includegraphics[scale=0.3]{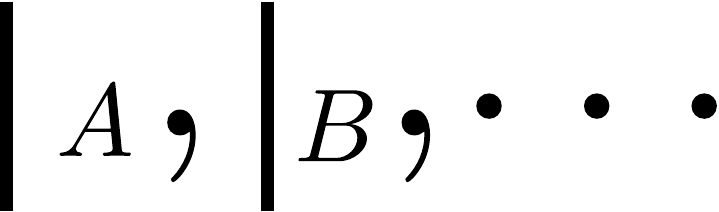}
\end{equation*}
\item A collection $P$ of processes, denoted by boxes, where the type of system transformed by a process and the type of system it produces belong to $T$:
\begin{equation*}
    \includegraphics[scale=0.3]{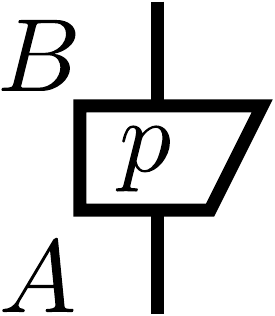}
\end{equation*}
\item Each system-type has a unique identity process satisfying the following equation for any process $p\in P$ and pair of system-types $A,B\in T$:
\begin{equation*}
    \includegraphics[scale=0.3]{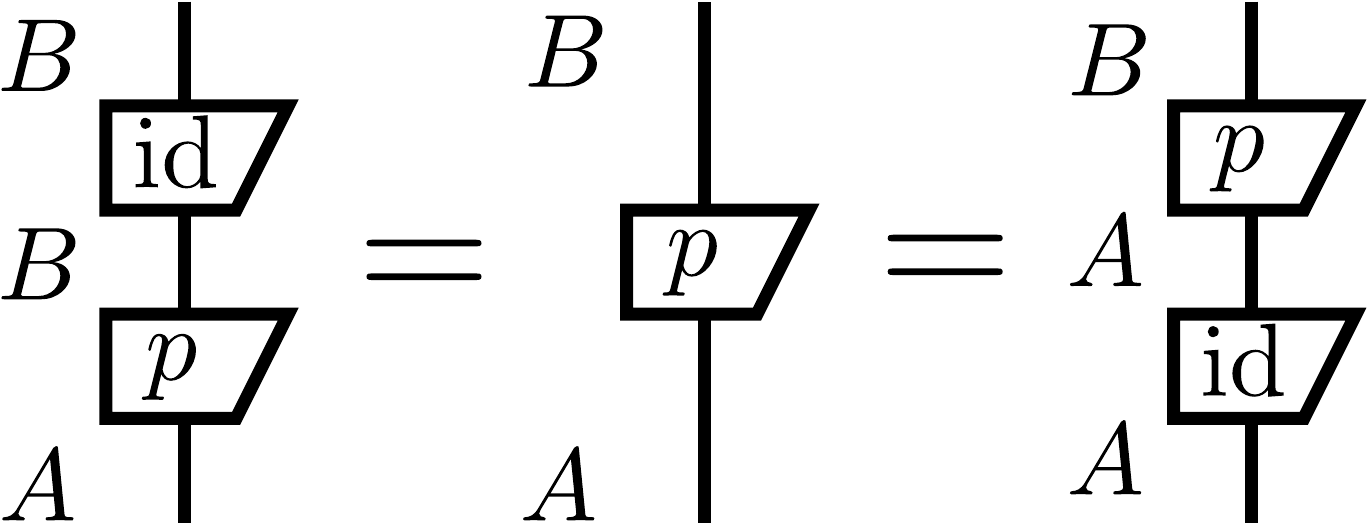}
\end{equation*}
\item A means of composing processes which forms a diagram that is also a process in $P$, that is, $P$ is closed under (sequential) composition (see Fig. \ref{fig:compose}). 
\end{enumerate}
\end{definition}

The identity process can be considered as the `do nothing' process, i.e. composing it to another process, say $p$ results in $p$. So we can define the identity process for a system-type $A$ as follows:
\begin{equation}
    \includegraphics[scale=0.3]{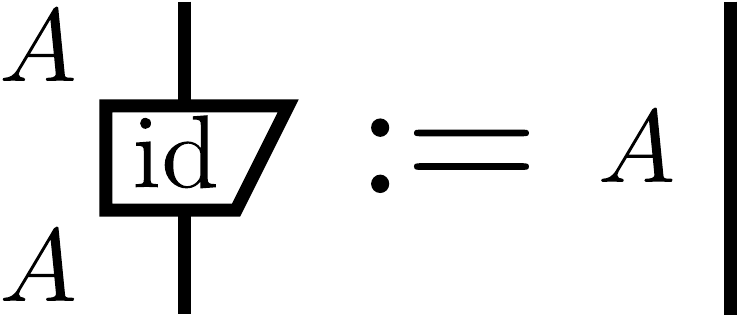}
\end{equation}

A process theory is a category (definition given below); where the system-types are the objects, the processes are the morphisms, the composition of morphisms is defined by the sequential composition between processes. 

\begin{definition} \cite{Awodey2010}
A category consists of the following data:
\begin{itemize}
    \item Objects, usually denoted by uppercase letters: $A$, $B$, $C$, ...
    \item Morphisms, usually denoted by lowercase letters: $f$, $g$, $h$, ...
    \item For each morphism $f$, there are given objects,
    \begin{center}dom$(f)$, cod$(f)$\end{center} 
    called the domain and codomain of $f$. We write:
    $$f:A\rightarrow B$$
    to indicate that $A=\text{dom}(f)$ and $B=\text{cod}(f)$.
    \item Given morphisms $f:A\rightarrow B$ and $g:B\rightarrow C$, that is, with
    $$\text{cod}(f)=\text{dom}(g)$$
    there is a given morphism
    $$g\circ f:A\rightarrow C$$
    \item For each object $A$, there is given an morphism
    $$1_A:A\rightarrow A$$
    called the identity morphism of $A$.
\end{itemize}
These data are require to satisfy the following laws:
\begin{itemize}
    \item Associativity:
    $$h\circ(g\circ f)=(h\circ g)\circ f$$
    for all $f:A\rightarrow B$, $g:B\rightarrow C$, $h:C\rightarrow D$.
    \item Unit:
    $$f\circ 1_A=f=1_B\circ f$$
    for all $f:A\rightarrow B$.
\end{itemize}
\end{definition}

A great interest in physics is the study of multipartite systems. We can describe a multipartite system by composing its subsystems but this type of composition must be different than the sequential composition we described above. Since we have set the directional convention to be along the vertical axis, we can compose two systems by placing them side by side so we may distinguish it from the sequential composition. We call this type of composition as parallel composition. Then it follows that we can transform the resulting composite system using two distinct processes, each of which acts on a different subsystem as in Fig. \ref{fig:parallel-comp}.
\begin{figure}[h]
\centering
    \includegraphics[scale=0.3]{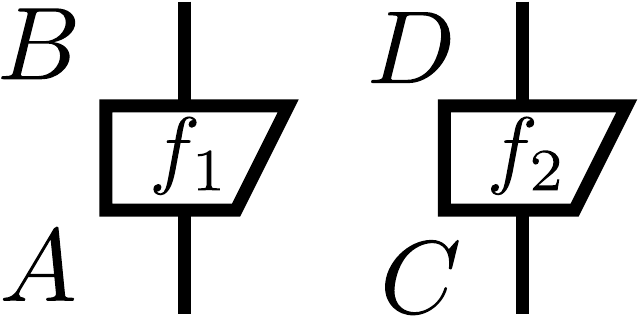}
    \caption{Parallel composition of two processes.}
     \label{fig:parallel-comp}
\end{figure}

Now we have two ways of composing processes: sequentially and in parallel. However, unlike sequential composition, there is no compatibility condition which needs to be satisfied in order to compose two process in a parallel manner. That is, with parallel composition, the processes remain independent. 

As a shorthand, we adopt the symbols $\otimes$ and $\circ$ for parallel and sequential compositions, respectively. That is:
\begin{equation*}
    \includegraphics[scale=0.3]{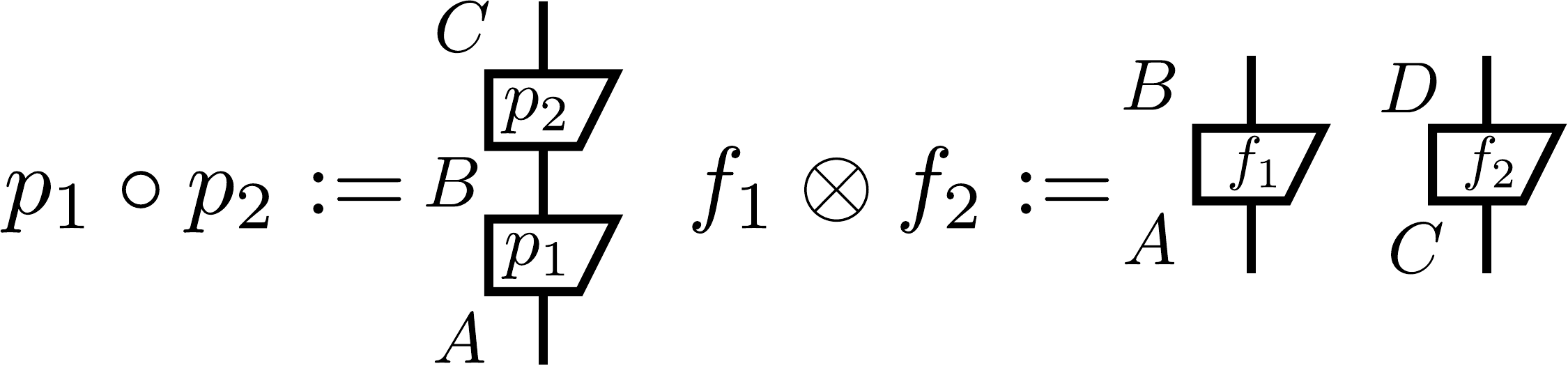}
\end{equation*}

Sequential composition has an order built into it, but to express different orderings for parallel composition, we need a process called `swap' which satisfies the equation below:
\begin{equation*}
    \includegraphics[scale=0.3]{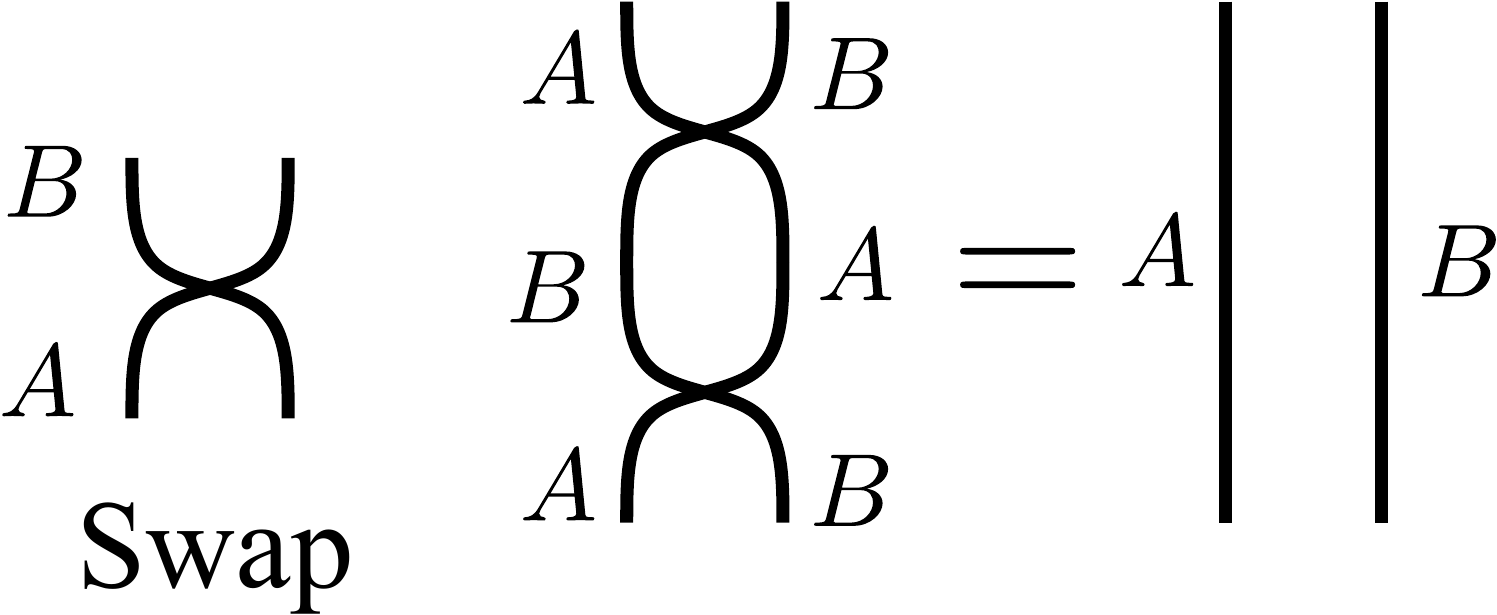}
\end{equation*}
% swap and swap identity
Sequentially composing swaps allow us to permute the ordering of a composite system. For example, we can transform $A\otimes B\otimes C\otimes D$ to one of its permutations in the following way:
\begin{equation*}
    \includegraphics[scale=0.3]{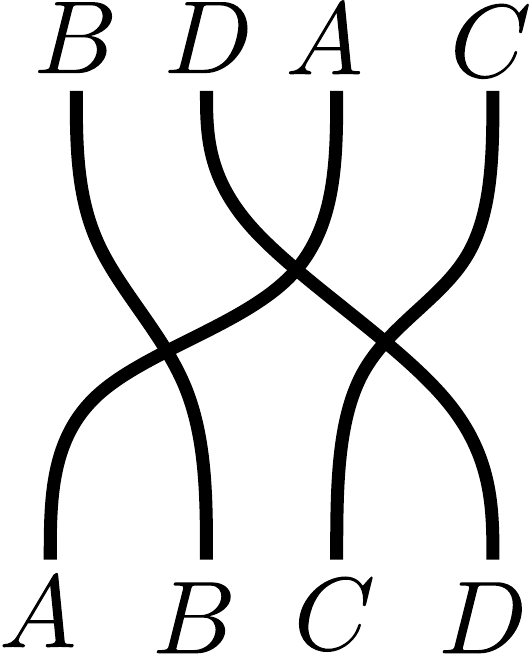}
\end{equation*}

Notice that the associativity of both compositions are built into their diagrams: 
\begin{equation*}
    \includegraphics[scale=0.3]{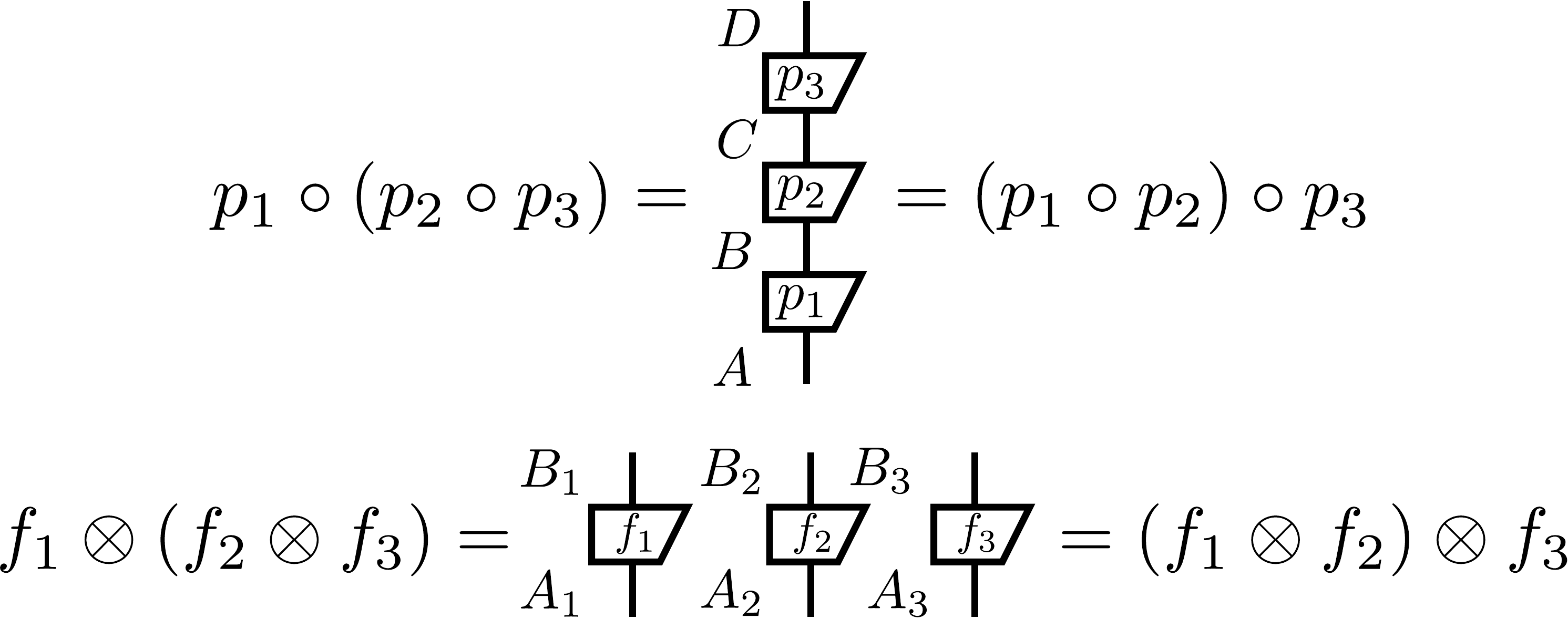}
\end{equation*}

We already have the identity process for each system-type with respect to sequential composition, i.e. the do nothing process. For parallel composition, there is no identity with respect to process. Instead, there is a special system, given the symbol $I$, which is the so called identity for system-types:
\begin{equation*}
    \includegraphics[scale=0.3]{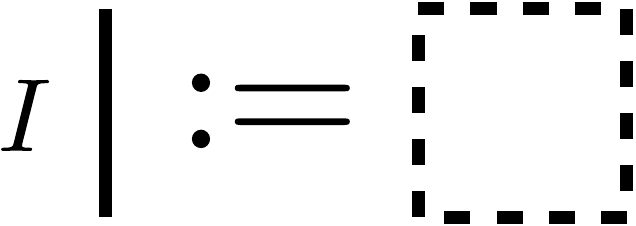}
\end{equation*}
% empty system
We shall forgo the dashed squares in forthcoming diagrams. The reason we included them to the diagrams above is to highlight the emptiness of $I$. So then, $A$, $A\otimes I$ and $I\otimes A$ are all equivalent to each other:
\begin{equation*}
    \includegraphics[scale=0.3]{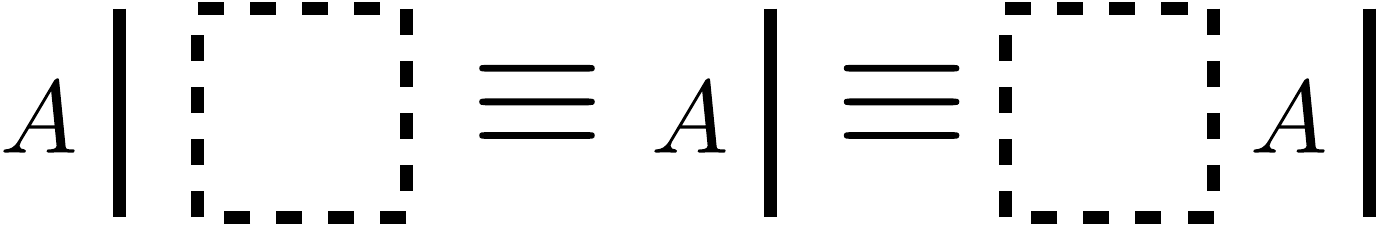}
\end{equation*}
$I$ is called the `empty system', which is represented by the empty diagram. Therefore, we are able to include those processes with no input and/or no output within our collection of processes. We call a process with no input as `state', a process with no output as `effect', and a process with neither input nor output as `scalar'. As the reader might have noticed, these terms are borrowed from physics. In fact, within a quantum process theory, these special processes match their names \cite{AbramskyCQM2009}.
\begin{equation*}
    \includegraphics[scale=0.3]{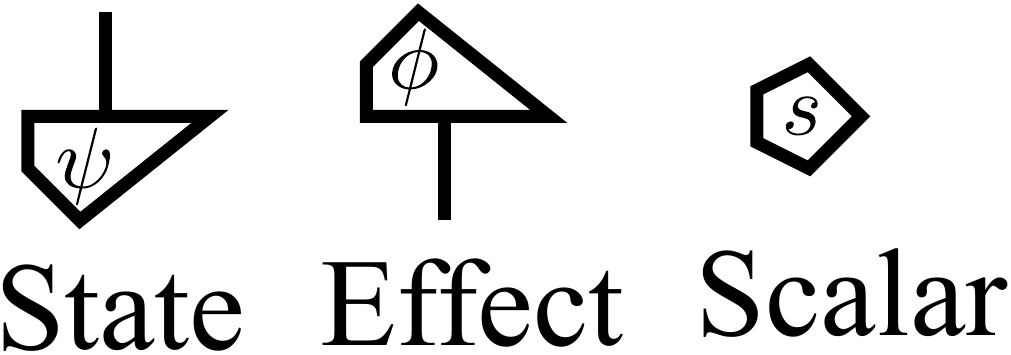}
\end{equation*}
% state, effect, scalar
When $\circ$-composing a state and effect, we obtain the Born rule, that is, the following diagram is \textit{the probability amplitude of obtaining the effect $\phi$  when measuring the state $\psi$}:
\begin{equation*}
    \includegraphics[scale=0.3]{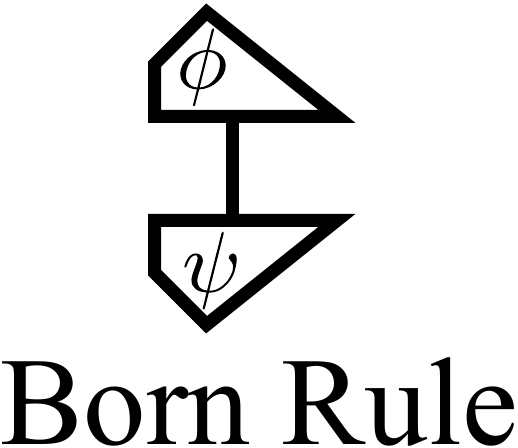}
\end{equation*}

\begin{definition} \cite{Coecke2015CQM}
A circuit diagram is a diagram constructed from processes, including identities and swaps, by $\circ$- and $\otimes$-compositions.  
\end{definition}

A process theory equipped with a $\otimes$-composition, where every diagram is a circuit diagram, is a symmetric monoidal category, with $\otimes$-composition as the equipped bifunctor and the empty system as the unit object. This correspondence between a process theory and a category is in fact an equivalence \cite{Joyal1991}.

\begin{definition}\cite{Selinger2007a}
A symmetric monoidal category is a category $\mathbf{C}$ equipped with a binary operation $\otimes:\mathbf{C}\times \mathbf{C}\rightarrow \mathbf{C}$ which is functorial, and there exist natural isomorphisms $\alpha,\sigma,\lambda,\rho$ for objects $A$, $B$, $C$ together with a distinguished object $I$:
$$\alpha_{A,B,C}:A\otimes(B\otimes C)\cong(A\otimes B)\otimes C,$$
$$\sigma_{A,B}:A\otimes B\cong B\otimes A,$$
$$\rho_C:I\otimes C\cong C \cong C\otimes I:\lambda_C$$
$\alpha,\sigma,\lambda,\rho$ satisfy certain coherence conditions. 
\end{definition}

A key difference between $\circ$ and $\otimes$ is the dependency between the processes which they compose:
\begin{equation*}
    \includegraphics[scale=0.3]{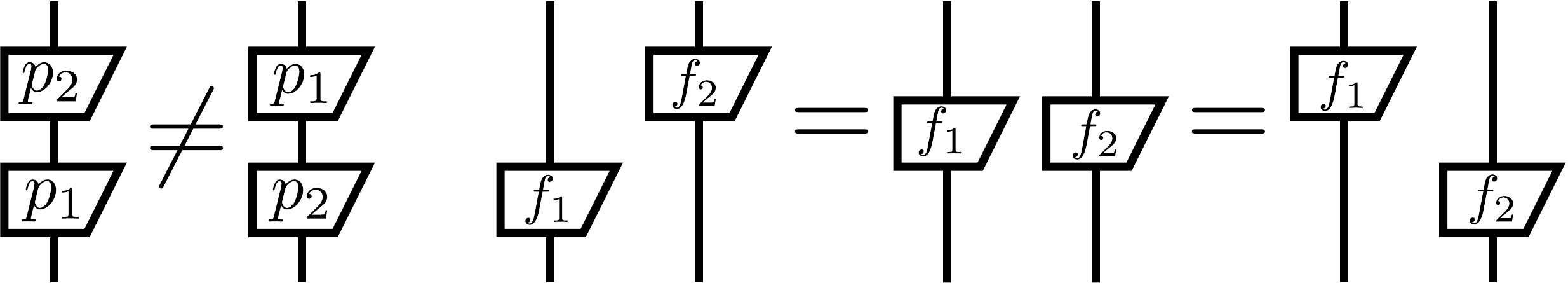}
\end{equation*}
If we ignore any directional convention, the $\circ$-composition of two processes can be viewed as a correlation, giving us an option to ignore the causal component that is inherent in input-output reading of processes. However, on the other hand, imposing a directional convention on a process theory gives context to its diagrams whereby we know the starting point when reading a diagram. We can obtain both these features if we impose what we shall call a `compact structure' on a process theory.

\begin{definition}\label{def:compact} \cite{AbramskyCQM2009}
A system-type $A$ is compact if $A\otimes A$ has a state -- called cup, $\cup$ -- and an effect -- called cap, $\cap$ -- which satisfy the following equation:
\begin{equation}\label{eq:pre-yank}
    \includegraphics[scale=0.3]{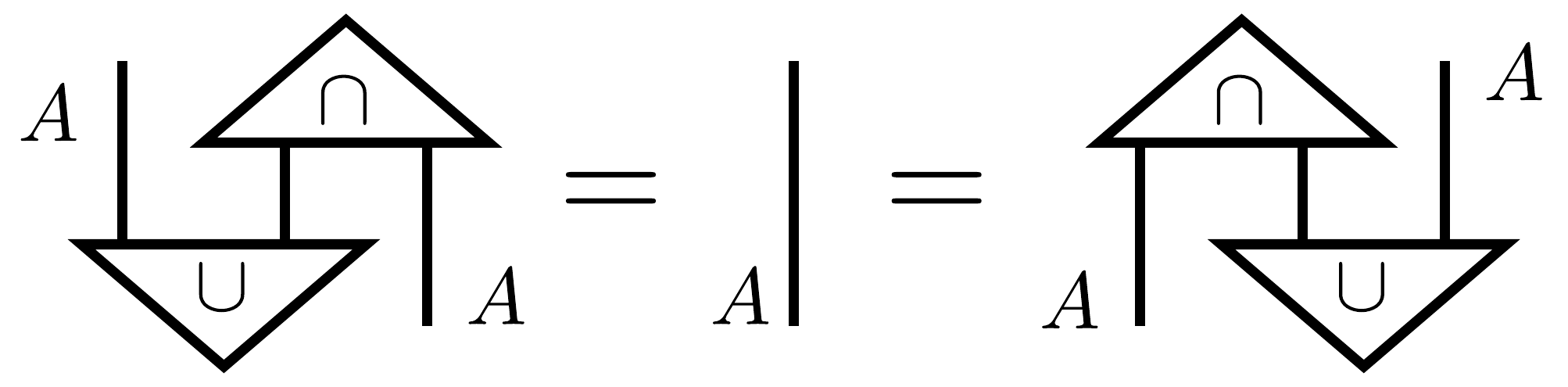}
\end{equation}
We call a process theory where every system-type is compact as a compact process theory. A diagram in a compact process theory is called a string diagram. 
\end{definition}

Since cups and caps are so special, we use the following diagrams for them:
\begin{equation*}
    \includegraphics[scale=0.3]{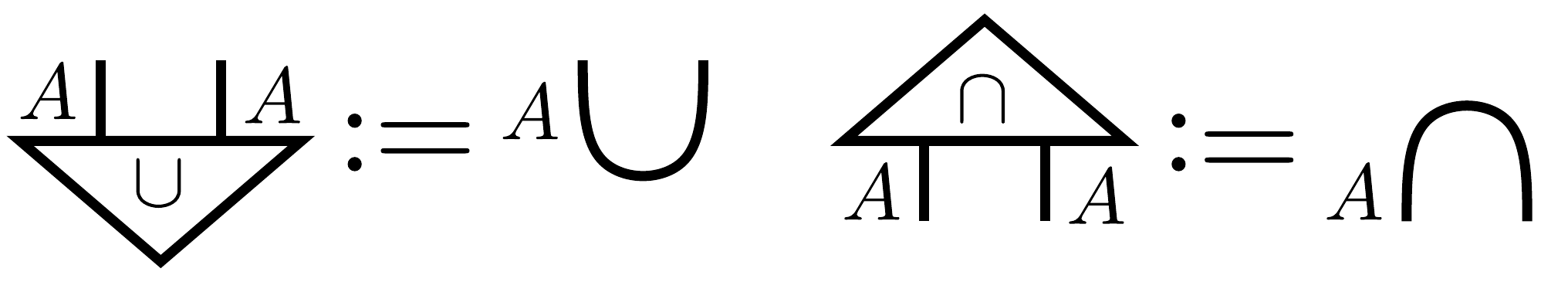}
\end{equation*}
Then Eq. \ref{eq:pre-yank} becomes:
\begin{equation}\label{eq:yank}
    \includegraphics[scale=0.3]{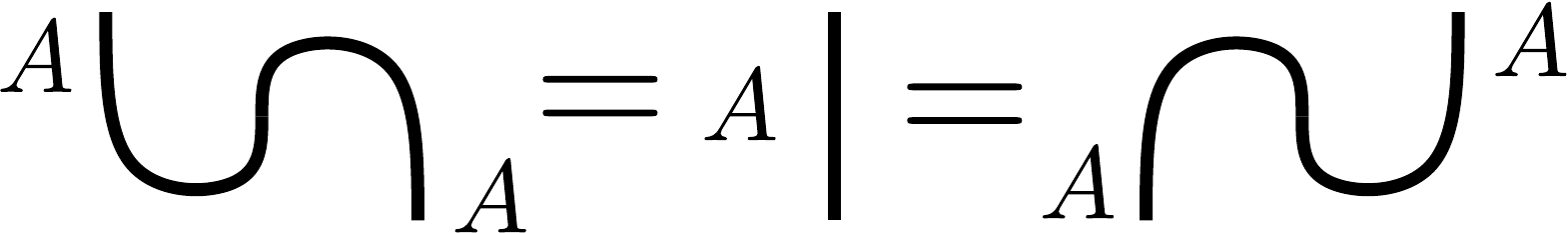}
\end{equation}% yanking eqn

The categorical counterpart of a compact process theory is a compact closed category; where the cup and cap of each system-type are its unit and counit, respectively, and the dual object of the system-type is itself. Note that an object in a compact closed category is not necessarily self-dual in general. However, we shall only encounter objects of a compact closed category which are self-dual in this thesis. As such, whenever we mention a compact closed category, we mean a compact closed category where every object is self-dual.

In a compact process theory, there is a bijective correspondence between states and processes:
\begin{equation*}
    \includegraphics[scale=0.3]{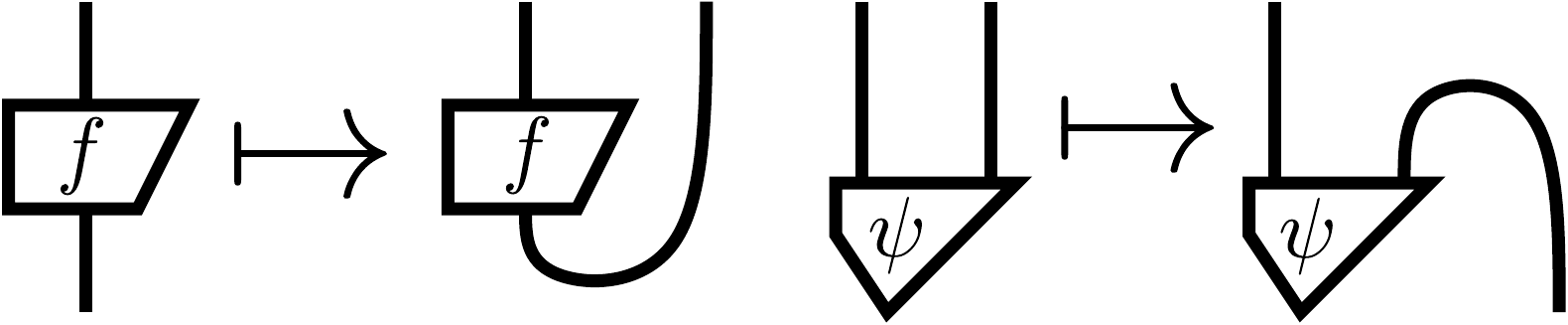}
\end{equation*}
% process state duality
The result above is called the `process-state duality'. In category theory, this equivalence is captured by an isomorphism between two functors which are the categorical counterparts of the mappings in the previous diagrams \cite{Selinger2007a}. 

\begin{definition}\cite{Selinger2007a}
A compact closed category is a symmetric monoidal category where each object $A$ is assigned a dual object $A^*$, together with a unit morphism $\eta_A:I\rightarrow A^*\otimes A$ and a counit morphism $\epsilon_A:A\otimes A^*\rightarrow I$, such that $\lambda_A^{-1}\circ(\epsilon_A\otimes 1_A)\circ\alpha_{A,A^*,A}^{-1}\circ(1_A\otimes\eta_A)\circ\rho_A=1_A$ and $\rho_{A^*}^{-1}\circ(1_{A^*}\otimes\epsilon_A)\circ\alpha_{A^*,A,A^*}\circ(\eta_A\otimes 1_{A^*})\circ\lambda_A=1_{A^*}$.
\end{definition}

We can also obtain the process-effect duality from the cup and cap but later, we shall introduce the dagger structure which renders these dualities as equivalent, that is, if we have  process-state duality, then we must have process-effect duality, and vice versa.

Within a compact process theory, any correlation that a pair of systems might have is generated by either the cup or the cap. That is, when two systems $A$ and $B$ are correlated via a state $\chi$, there must be a process $c$ such that:
\begin{equation}
    \includegraphics[scale=0.3]{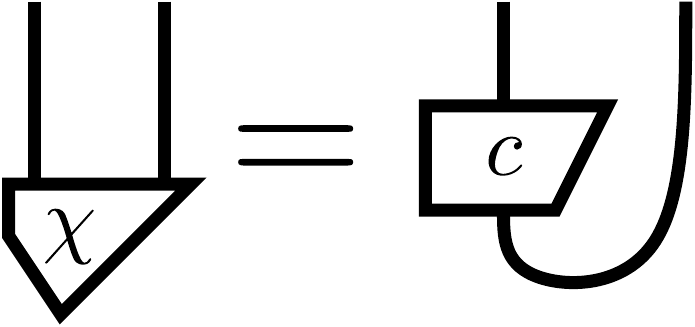}
\end{equation}
In quantum theory, the process-state duality makes sure that there is only a single class of maximally entangled states (with respect to SLOCC) for a bipartite system \cite{Coecke2010}. Another phenomena which can be expressed using the compact structure is quantum teleportation (see Fig. \ref{fig:teleport1}).

\begin{figure}[!ht]
\centering
\includegraphics[scale=0.3]{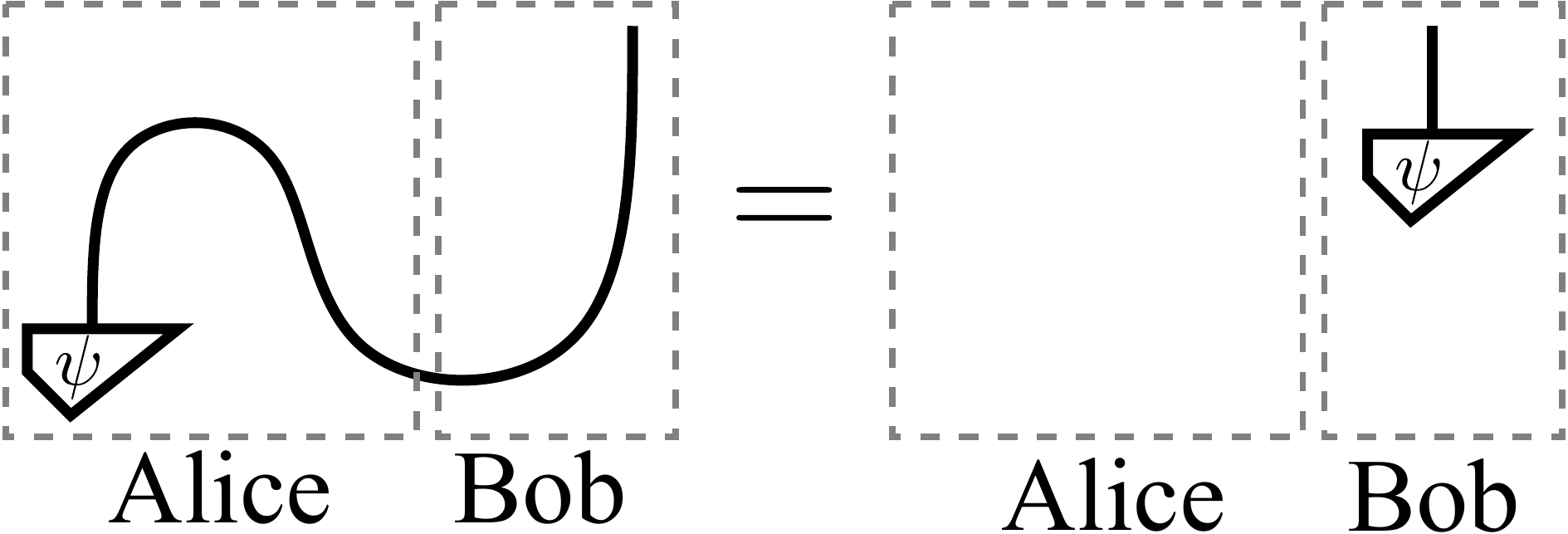}
\caption{Quantum Teleportation}    
\label{fig:teleport1}
\end{figure}
% diagram for quantum teleportation (with bob and alice)

The diagram in Fig. \ref{fig:teleport1} is not yet a complete description of quantum teleportation since the measurement performed by Alice is probabilistic and hence, the state which reaches Bob is one of $U_i\circ \psi$ where $U_i$ is a collection of unitary operators (or processes) which represent the possible outcomes of Alice's measurement. 

To describe unitary processes, we need a structure which describes adjoints. This is captured by the dagger structure \cite{Selby2017Adjoint}. In writing, we use the symbol $\dagger$ to denote the dagger structure. When applying $\dagger$ to a process, we obtain its adjoint, which is also a process. 

Within a process theory with a dagger structure, when applying $\dagger$ to a process twice, we obtain the process again. Therefore, it makes sense to graphically represent the adjoint of a diagram (and consequently, processes) as its reflection on the horizontal axis: 
\begin{equation*}
    \includegraphics[scale=0.3]{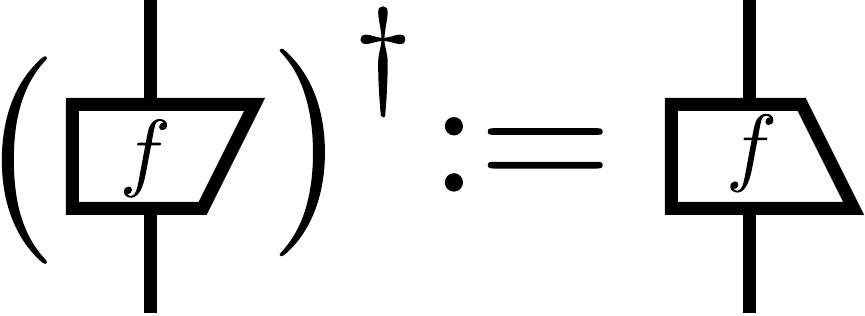}
\end{equation*}
The categorical counterpart of a process theory with a $\dagger$ structure is called a \textbf{dagger category} \cite{AbramskyCQM2009, Selinger2007a}. 

\begin{definition}\cite{Selinger2007a}
A dagger category is a category $\mathbf{C}$ together with an involutive, identity-on-objects, contravariant functor $\dagger:\mathbf{C}\rightarrow\mathbf{C}$.
\end{definition}

A compact structure and a dagger is compatible when:
\begin{equation*}
    \includegraphics[scale=0.3]{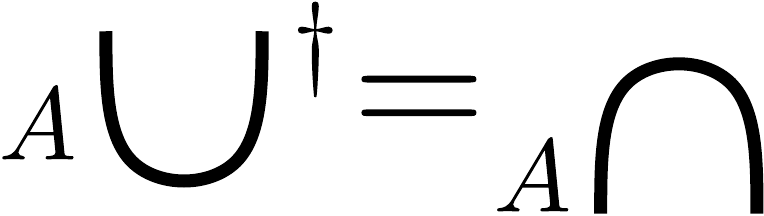}
\end{equation*}

\begin{definition}
A process is unitary if it satisfies the following equation:
\begin{equation}\label{eq:unitary}
    \includegraphics[scale=0.3]{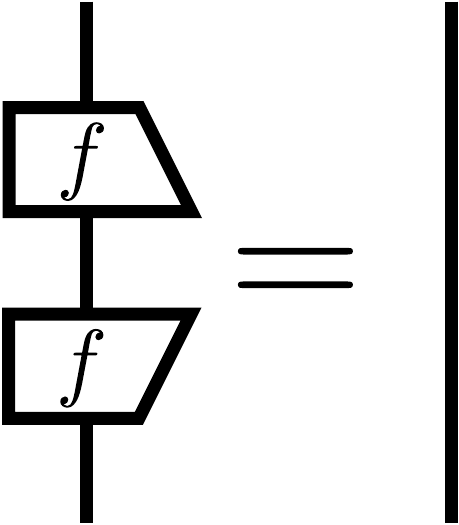}
\end{equation}
% unitary diagram
\end{definition}

Therefore, in the quantum teleportation diagram from Fig. \ref{fig:teleport1}, Bob may obtain the state that Alice sent by performing the adjoint of $U_i$ on the resulting state:
\begin{equation*}
    \includegraphics[scale=0.3]{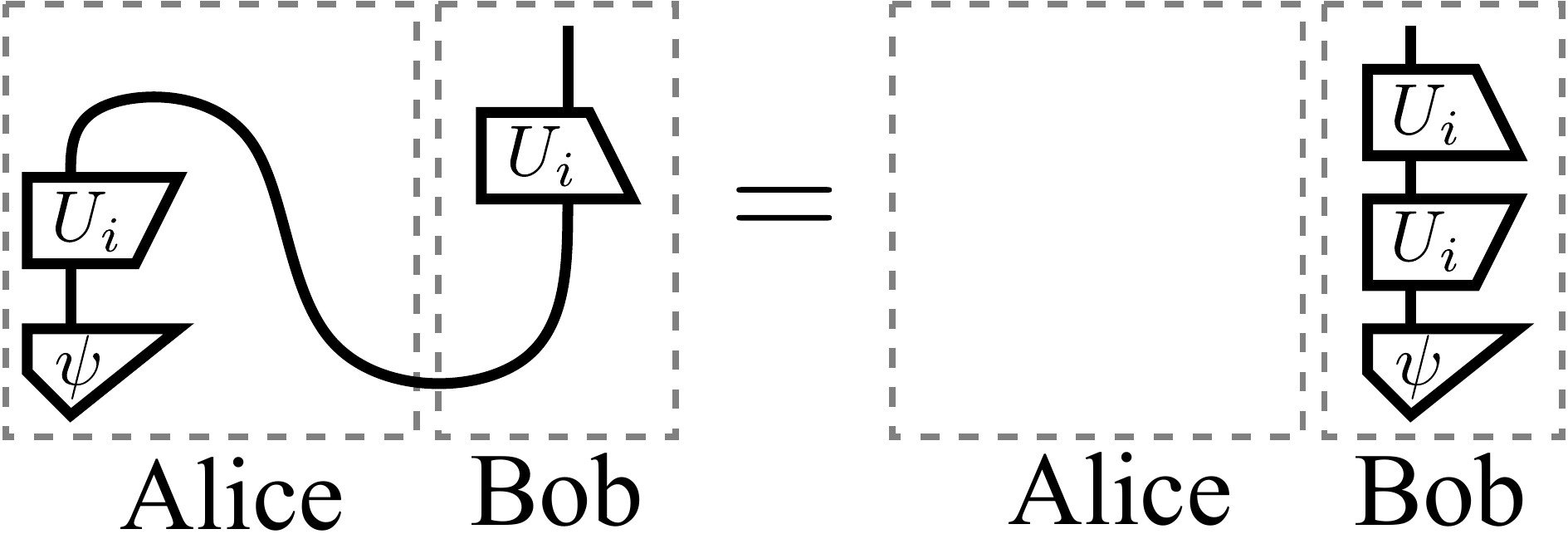}
\end{equation*}
% teleportation diagram 
But how would Bob know which $U_i^\dagger$ to perform? This can only be done by telling Alice the outcome of his measurement through a classical channel since another quantum communication would result in another probabilistic measurement. In the next section, we provide a brief review on how to present `classicality' in the diagrammatic language of process theories. 

We can also introduce diagrammatic counterparts of familiar concepts from the compact and dagger structures.

\begin{definition}\label{def:transpose}\cite{AbramskyCQM2009}
 The transpose of a process $f$ is defined as follows:
 \begin{equation*}
     \includegraphics[scale=0.3]{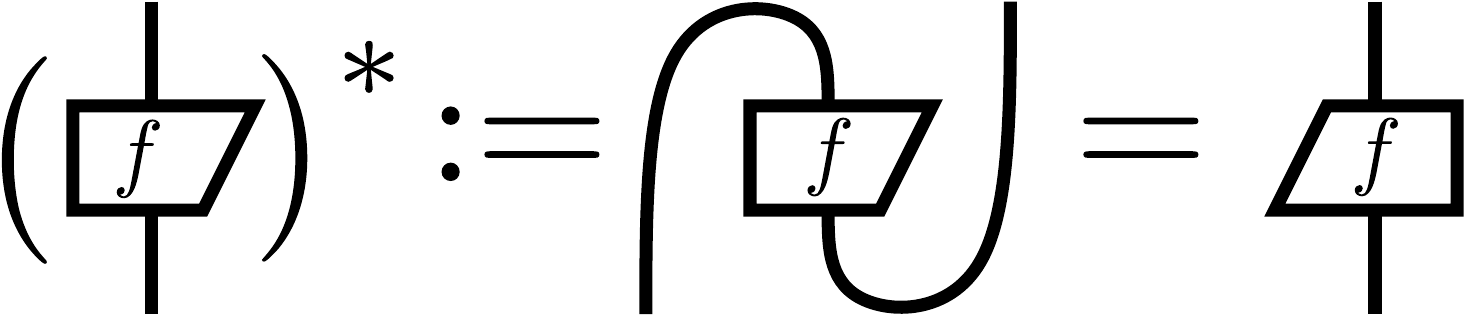}
 \end{equation*}
\end{definition}

\begin{definition}
The conjugate of a process $f$ is defined as follows:
\begin{equation*}
    \includegraphics[scale=0.3]{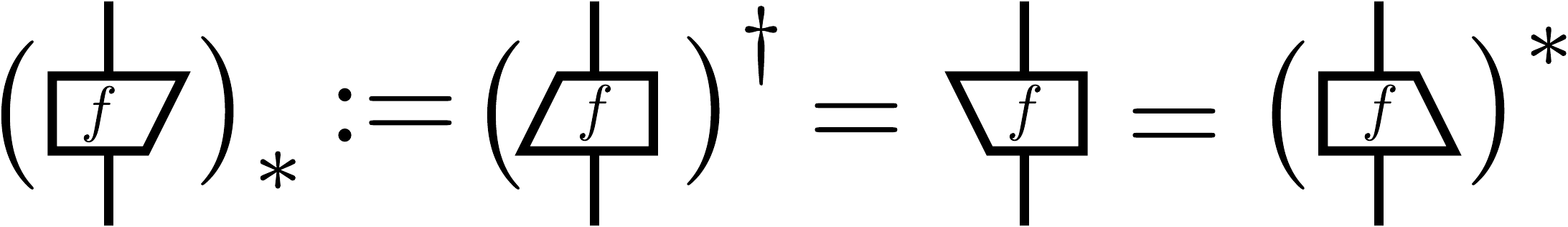}
\end{equation*}
\end{definition}

\section{Measurements as Diagrammatic Algebras}\label{sec:intro-measalg}

In quantum theory, there are several no-go theorems which distinguish between quantum data-types and classical data-types. Two of them are the no-cloning and the no-deleting theorems which state that quantum data-types cannot be copied or deleted as can be done to classical data-types. In a process theory, these two features of quantum data-types are the basis for describing classicality.

The mathematical framework for quantum theory relies on linear algebra: a state of a system is represented by the elements of a Hilbert space and the processes performed on states are represented by linear maps between Hilbert spaces. Due to linearity, quantum states cannot be copied uniformly. That is, given the possible states of a quantum system, only a portion of those states can be copied by the same cloning process. These states turn out to be orthogonal to each other. In fact, the vectors representing these states form an orthogonal basis of the Hilbert space which contains the system's possible states. The same applies to the deletion of quantum states.  

The category of finite-dimensional complex Hilbert spaces (objects) and the linear maps between them (morphisms) is a dagger compact category. Therefore, quantum processes can be depicted as diagrams of a compact process theory equipped with a dagger structure. Let $\mathcal{B}$ be an orthogonal basis of a finite-dimensional complex Hilbert space $H$. There is a process $\triangledown: H\rightarrow H\otimes H$, which copies the elements of $\mathcal{B}$, and a process $\top:H\rightarrow\mathbb{C}$, which deletes the elements of $\mathcal{B}$. Suppose $\mathcal{B}=\{\ket{j}\}_{j=1}^N$ where $N$ is dimension of $H$. Then:
\begin{eqnarray}
\triangledown(\ket{j})=\ket{j}\ket{j}\\
\top(\ket{j})=1
\end{eqnarray}

$H$ together with $\triangledown,\top,\triangledown^\dagger,\top^\dagger$ form a classical structure. In a process theory, we represent copying and deleting along with their adjoints as the following diagrams:
\begin{equation*}
    \includegraphics[scale=0.2]{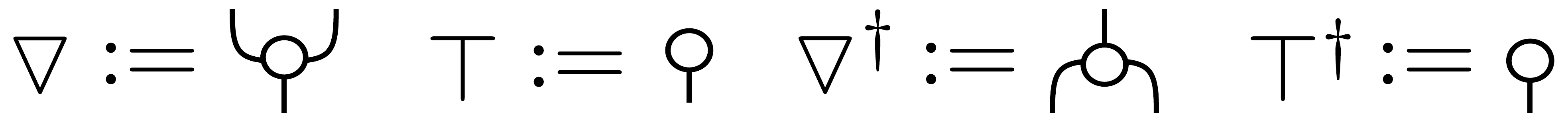}
\end{equation*}
\citet{Coecke2013New} showed that any orthogonal basis of a finite-dimensional complex Hilbert space correspond to a classical structure of the Hilbert space. Furthermore, this correspondence is bijective. 

We only deal with finite-dimensional complex Hilbert spaces in this work, so henceforth, when we mention Hilbert spaces, we are referring to finite-dimensional complex Hilbert spaces. 

The following is an alternative definition of a classical structure which takes the copying and deleting processes as particular spiders.

\begin{definition}\label{def:CS}
For a system-type $A$, a classical structure of $A$ consists of spiders with $m$ inputs and $n$ outputs, where $m,n\in\mathbb{N}\cup 0$:
\begin{equation*}
    \includegraphics[scale=0.2]{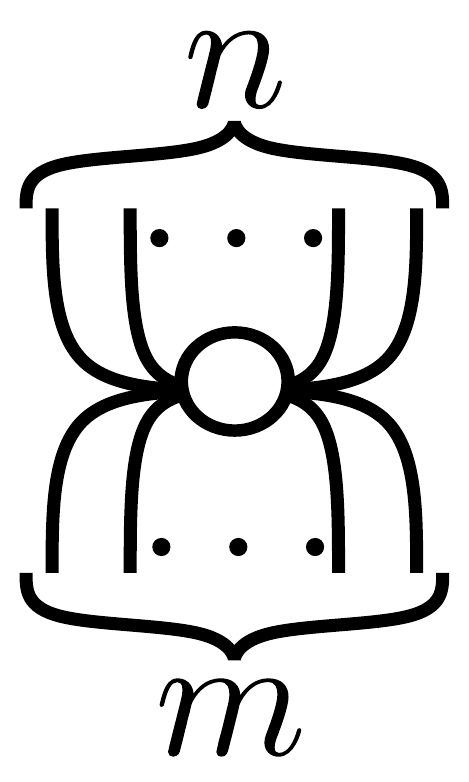}
\end{equation*}
such that $m_1,n_1$- and $m_2,n_2$-spiders fuse together when $k$ outputs of one are joined to the $k$ inputs of the other:
\begin{equation}\label{eq:spider-fuse}
    \includegraphics[scale=0.2]{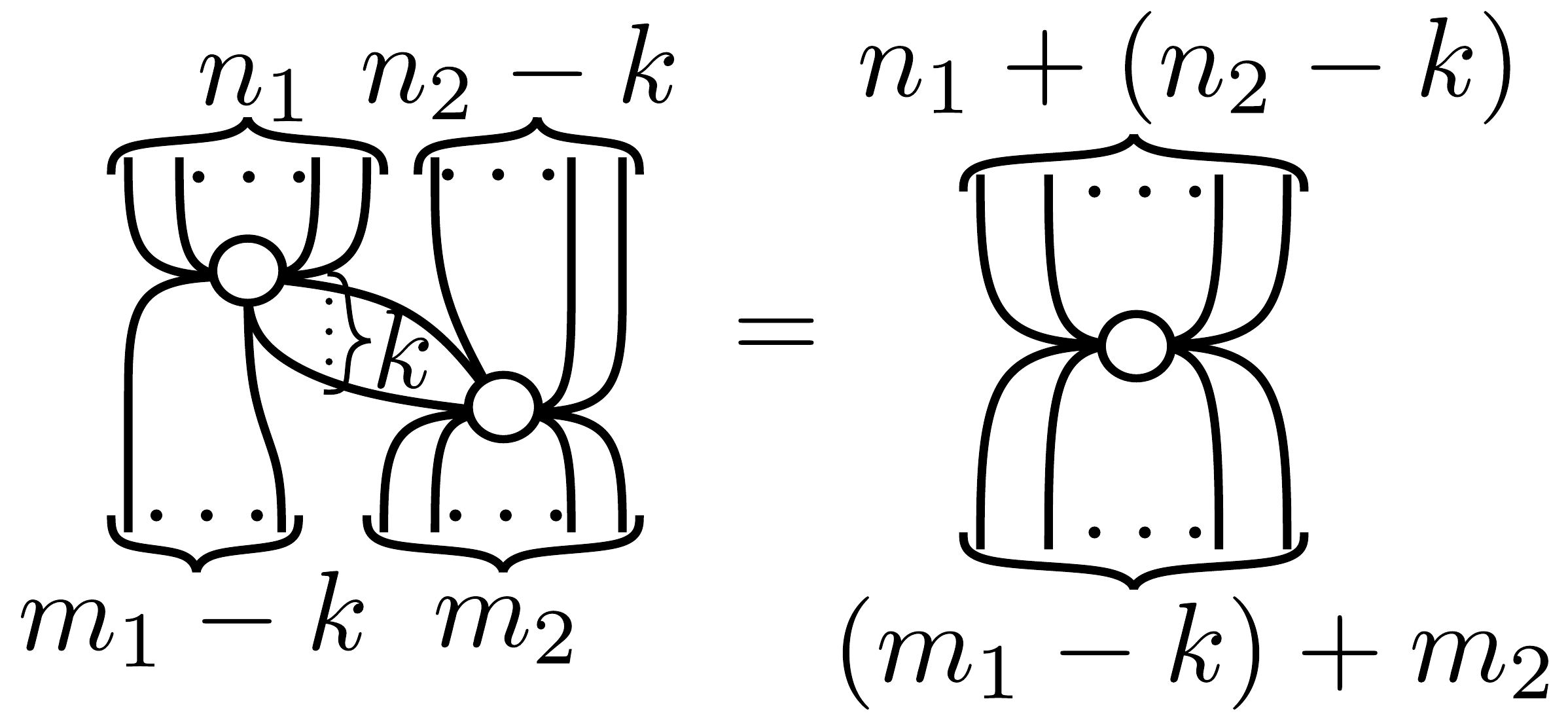}
\end{equation}
and a spider with permuted inputs/outputs is the same as the unpermuted one:
\begin{equation}
    \includegraphics[scale=0.2]{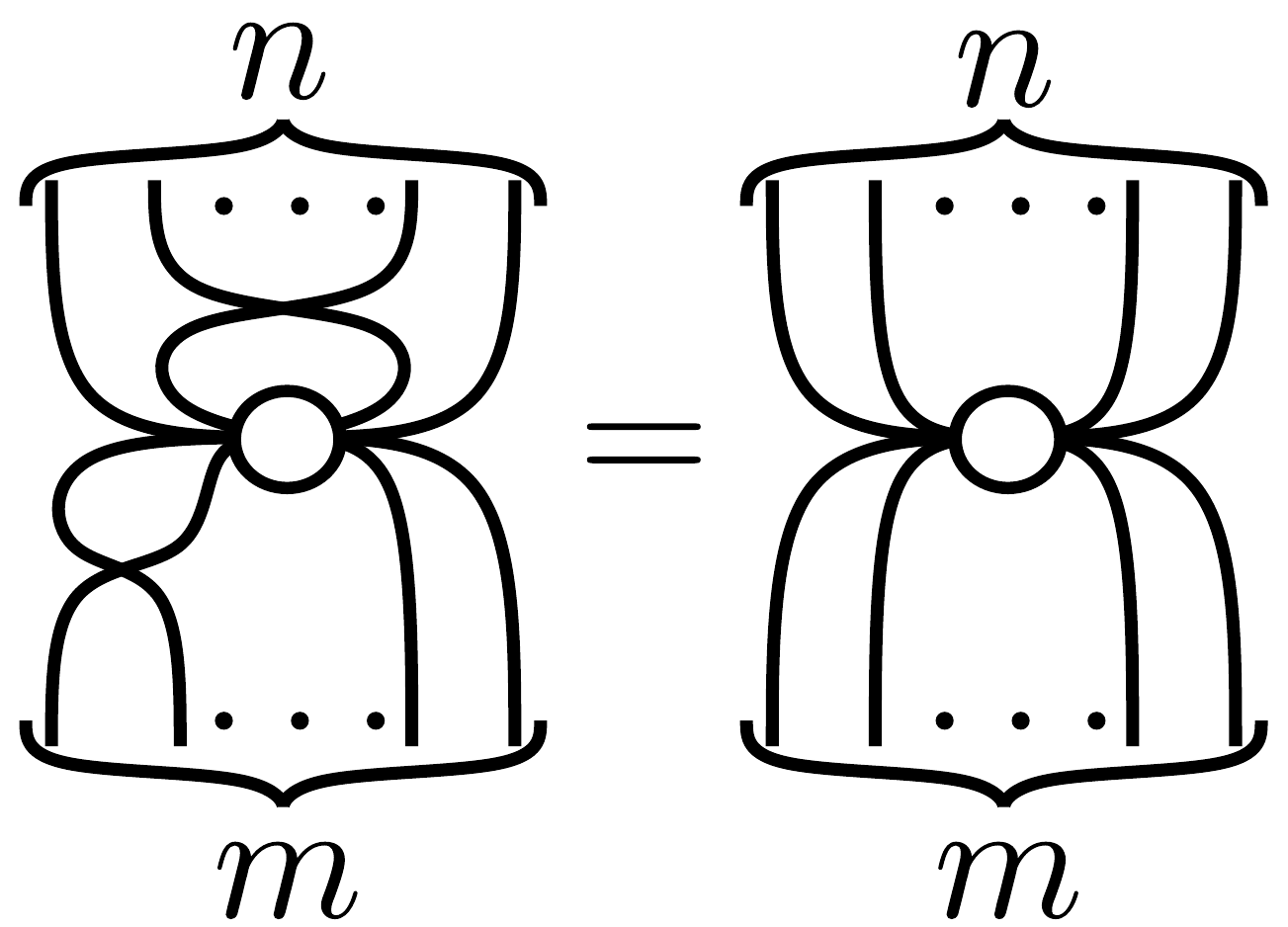}
\end{equation}
\end{definition}

A system with a classical structure is automatically compact. That is, the cup is the 2,0-spider and the cap is the 0,2-spider. 
\begin{equation}
    \includegraphics[scale=0.2]{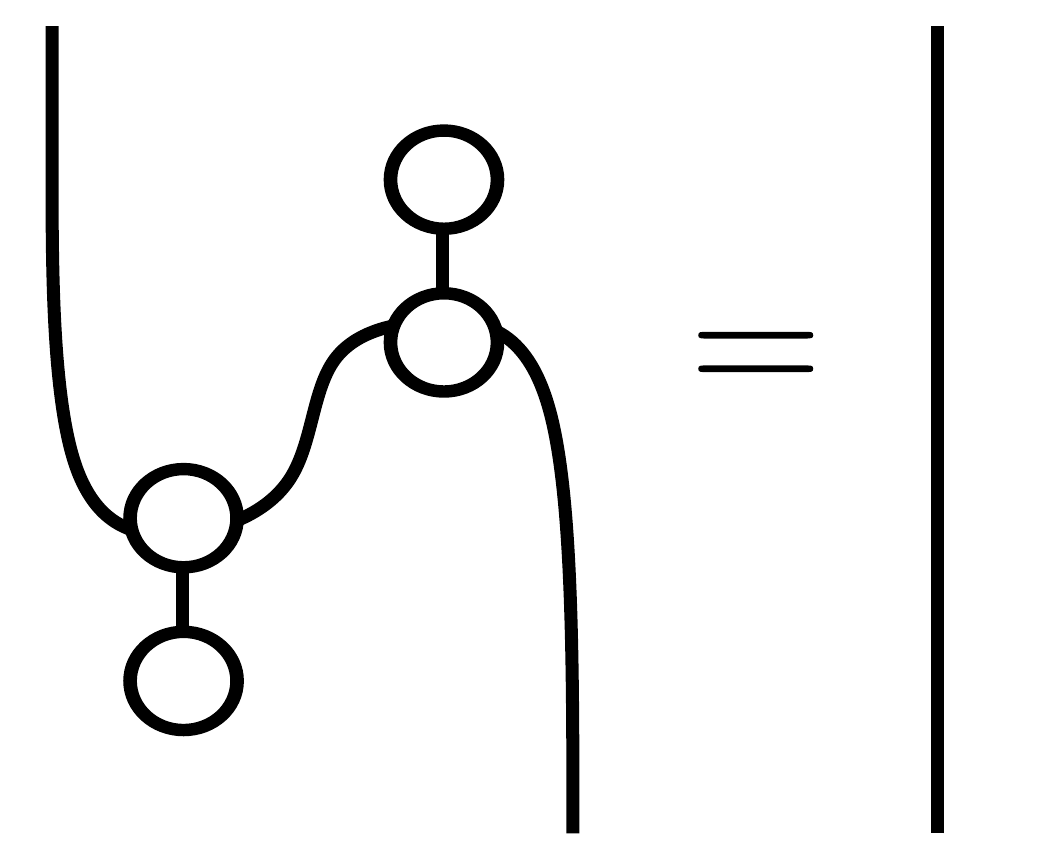}
\end{equation}

Notice that spiders in a classical structure include those without an input or output, i.e. when $m$ or $n$ is equal to 0. When a spider is fused with these types of spiders, the inputs/outputs joined are terminated. For example:
\begin{equation}
    \includegraphics[scale=0.2]{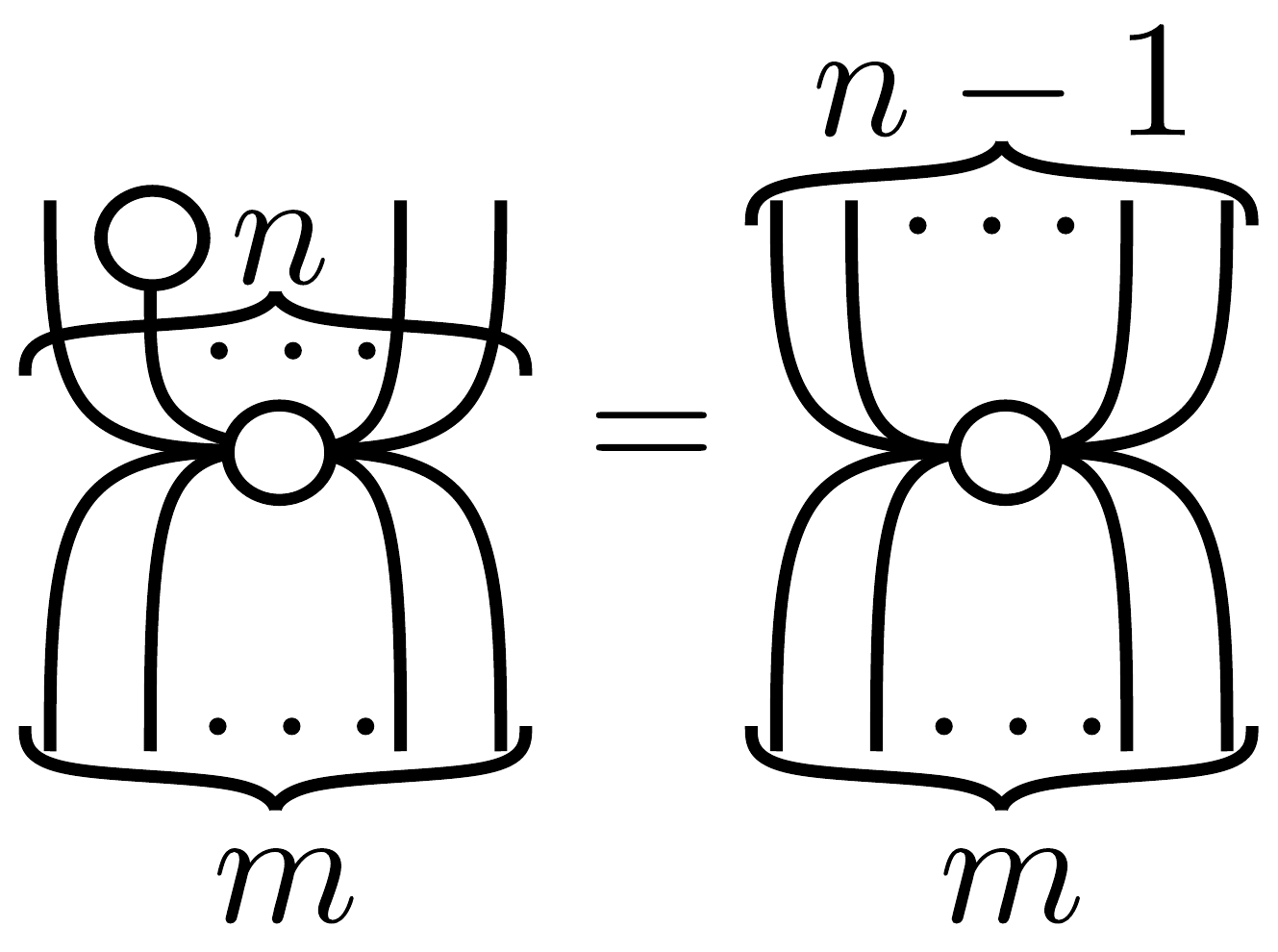}
\end{equation}

The categorical treatment of spiders can be found in \cite{Lack2004} and a graphical proof showing the equivalence between the definition of classical structures as an algebra and the one consisting of spiders can be found in \cite{Coecke2008POVMsSums}. 

Incidentally, when a measurement is performed on a pure quantum state, the outcomes are represented by an orthogonal basis. In particular, when a measurement with $N$ outcomes is performed on a state $\ket{\psi}=\sum_{j=1}^N\psi_j\ket{j}$, $\ket{\psi}$ will be transformed to one of the states in $\{\ket{j}\}_{j=1}^N$, which forms an orthogonal basis, and the probability of obtaining the $j$-th outcome is $|\psi_j|^2$. In fact, the 2,1-spider takes on the role of measurement in the `single-line means classical' and `double-line means quantum' interpretation of processes \cite{Coecke2016a}. The 1,2-spider, otherwise known as the adjoint of the 2,1-spider, then becomes the process which encodes classical data into a quantum system. 

So when classical data is encoded into a quantum system and then measured, we obtain the identity process:
\begin{equation}\label{eq:decoherence}
    \includegraphics[scale=0.2]{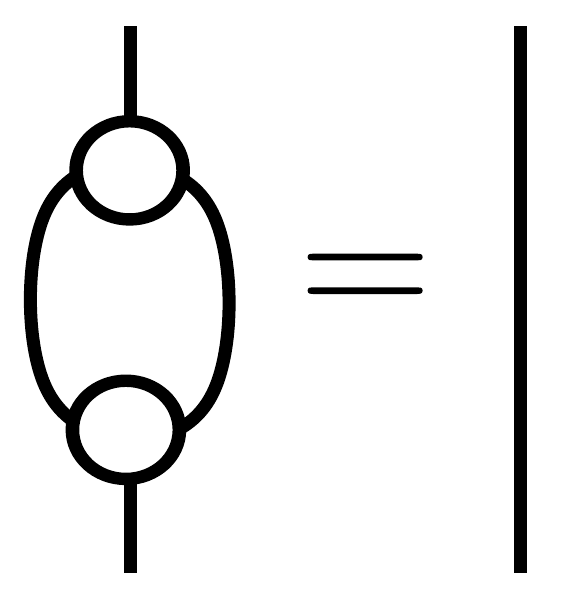}
\end{equation}

However, a system-type can have more than a single classical structure. So, suppose we have two classical structures. If we encode a quantum system with classical data using the 1,2-spider of one of the classical structures and then measure it with the 2,1-spider of the other classical structure, what will happen to the previous diagram?  

The two extremes of the measurement of a state in a basis $\mathcal{B}$ are
(1) accurate result when the measurement is described by $\mathcal{B}$, and (2) Completely random result when the measurement is described by a basis complementary to $\mathcal{B}$. (1) is described by the Eq. \ref{eq:decoherence}, but (2) is a different story. One would expect a disconnection between the input and output of a similar procedure to Eq. \ref{eq:decoherence} for (2) as the measurement gives us no information on the system. However, to obtain this disconnection, the diagram on LHS of Eq. \ref{eq:decoherence} needs a slight modification:

\begin{definition}\label{def:complement}
Two classical structures, distinguished by the colours of their nodes are complementary if:
\begin{equation}
    \includegraphics[scale=0.2]{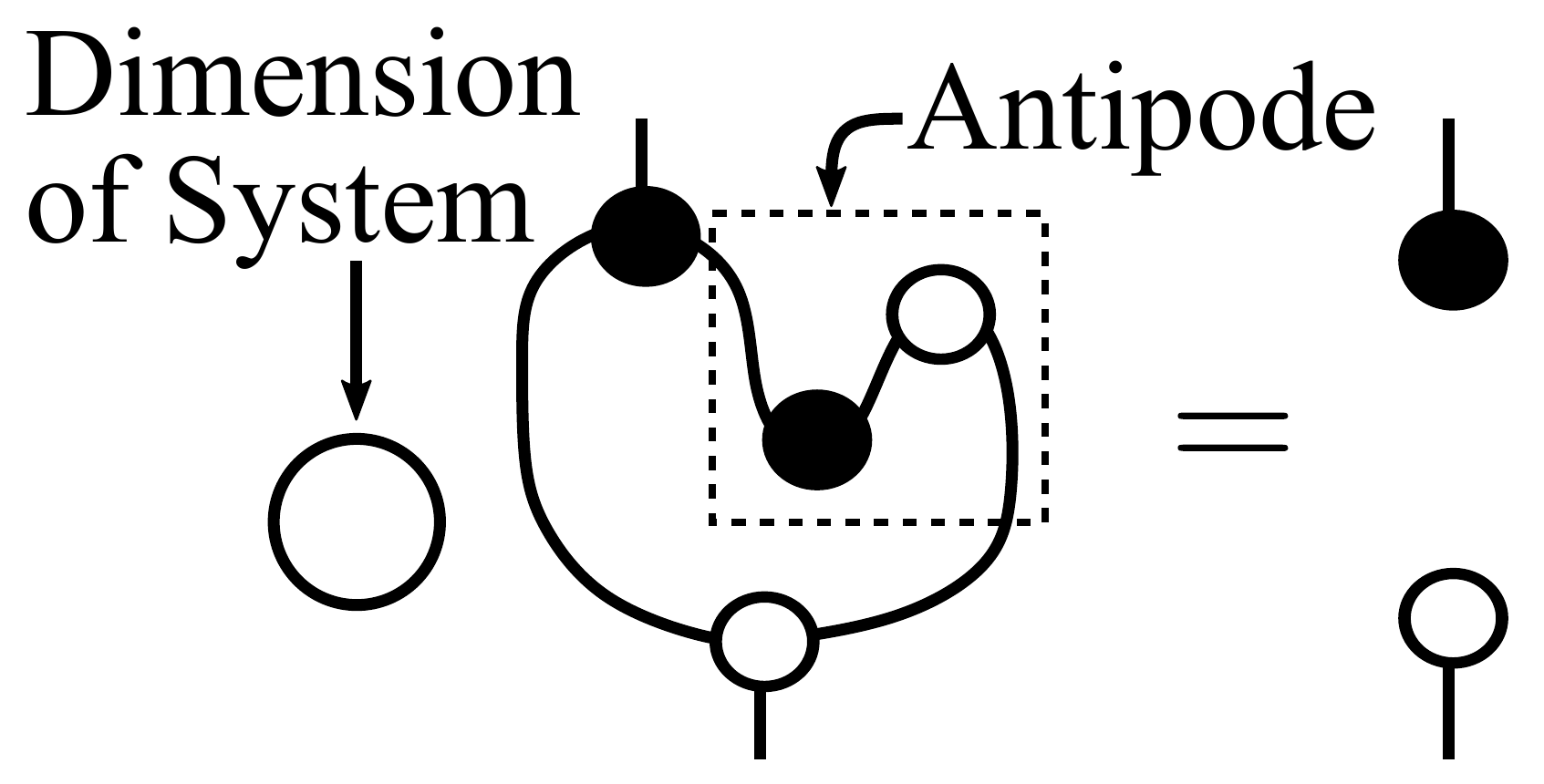}
\end{equation}
\end{definition}

\section{Thesis Outline}

This thesis is inspired by \citet{Romero2005}, where the authors took the complete\footnote{In this thesis, a set of mutually unbiased bases is complete if no other basis can be added into the set to form a larger set of mutually unbiased bases. The same goes for a set of mutually complementary classical structures where it is complete if no other classical structure can be added to the set to form a larger set of mutually complementary classical structure.} set of 9 mutually unbiased bases for $N$ qubits from \cite{Wootters1989}, say $\mathsf{MUB}_0$, and constructed new complete sets of 9 mutually unbiased bases through the following procedure:
\begin{enumerate}
\item Each of the basis in $\mathsf{MUB}_0$ is described by a set of $2^N-1$ mutually commuting operators which are the Kronecker products of Pauli operators;
\item These constituent Pauli operators are `rotated' by multiplying them with local unitary operators;
\item Entanglement between two qubits are generated or eliminated by applying the controlled-$Z$ operator to them. 
\end{enumerate}

This procedure is able to produce all the entanglement configurations available for complete sets of 9 mutually unbiased bases. For example, the entanglement configurations for three qubits are (2,3,4), (1,6,2), (0,9,0), and (3,0,6), where the first coordinate is the number of separable bases, the second coordinate is the number of biseparable bases, and the third coordinate is the number of non-separable bases in the set of MUBs. However, the separability of each basis is not immediately obvious. One would need to perform some straightforward but tedious computations to find out the separability of the basis represented by each set of mutually commuting operators. 

In this section, we showed how a classical structure is an abstraction of an orthogonal basis. Conceptually, complementarity between classical structures and unbiasedness between orthogonal bases coincide since they both represent an incompatible measurement of an observable which gives a completely random result. Section 8 of reference \cite{Coecke2011b} shows the relationship between complementarity and unbiasedness. So, replicating the results of \cite{Romero2005} should be possible within the framework of process theories. This is the subject of our current interest as a successful translation of those results in process theories shall provide a more explicit presentation of the entanglement structure in members of a complete set of MUBs on $N$ qubits.    

In Chapter \ref{sec:composeCS}, we shall devise two methods for composing two classical structures on a single qubit so that we can obtain a classical structure on two qubits. One method composes two classical structures on a single qubit and forms a classical structure on two qubits where the constituent classical structures remain separable. The other method composes two classical structures on a single qubit via one of three processes which are entangled on two qubits and forms a non-separable classical structure on two qubits. We shall call these bipartite processes `connecting wires'. We then extend the two methods so we could construct classical structures on $N$ qubits where $N>2$. 

% The aforementioned `connecting wires' uses the idea that we may generate entanglement between two qubits by applying a controlled gate on the qubits. In this work, we use the process theory versions of the controlled-$X$ and controlled-$Z$ gates to generate entanglement. We also utilize another property of the two controlled gates, i.e. each is its own inverse.  

The aforementioned constituent classical structures are classical structures on a single qubit that are abstractions of the eigenbases of the Pauli $X$, Pauli $Y$ and Pauli $Z$ operators. We provide a review of of the classical structures which correspond to the Pauli $X$ and $Z$ operators which form a diagrammatic calculus called ZX-calculus in Chapter \ref{sec:diagram-algebra}. Using ZX-calculus, we construct the classical structure which corresponds to the Pauli $Y$ operator. Composing these classical structures, either separably or via the connecting wires, we shall obtain a collection of classical structures on multiple qubits. 

In Chapter \ref{sec:results}, we identified the precise condition for a pair of classical structures, obtained via the procedure devised in Chapter \ref{sec:composeCS}, to be complementary. Then we found maximal complete sets of mutually complementary classical structures for the cases of two and three qubits. This is done using tools in graph theory where complete sets of mutually complementary classical structures are complete subgraphs of the graph consisting of classical structures as its vertices and there is an edge between two classical structures if they are complementary. To identify these edges, we found a generating set of pairs of classical structures. Each member of this set represents a class of pairs of classical structures with equivalent complementarity diagrams. This allows us to significantly reduce the number of pairs of classical structures to be checked for complementarity.

\section{Objectives}

The objectives of this study are twofold:
\begin{enumerate}
    \item First, we construct classical structures on $N$ qubits ($N\geq 2)$ by identifying its constituents as complementary classical structures on a single qubit, and two methods of composing these constituents, i.e. separably and non-separably, which make explicit the entanglement structure of the composite classical structure.
    \item Then, we study how the entanglement structure of the composite classical structures that we have constructed relates to their complementarity. In this thesis, we present our findings through examples in two qubits and three qubits. In particular, we searched for maximal complete sets of mutually complementary classical structures among composite classical structures on two and three qubits.  
\end{enumerate}

\chapter{Literature Review}

Mac Lane and Eilenberg introduced category theory while working on problems in algebraic topology, a field which applies the tools of algebra to topological spaces. The paper which introduced the concept of categories, {\it General Theory of Natural Equivalences} \cite{Eilenberg1945}, was not well received as it was said to lack any content, and category theory was later dubbed "general abstract nonsense." At the time, Mac Lane and Eilenberg did not seriously consider expanding their study on categories, and instead thought of it more as a convenient language for mathematicians in diverse fields to understand each other \cite{MacLane1988ConceptsPerspective}. However, it is out of that simple objective that a new mathematics, which in the following decades became an active area of research, was born.

\section{Applied Category Theory via Diagrams}\label{sec:LR-ACT}

In the decades since the publication of \citeauthor{Eilenberg1945}'s seminal work, category theory has become an important tool for studying the foundation of mathematics \cite{Lambek1988IntroductionLogic}, but more than that, it provides a way for various fields, in and outside of mathematics, to communicate with each other \cite{Baez2010PhysicsStone}. For example, in computer science, we can form a category by taking data types as objects and the programs which process the data as morphisms. We can find similar structures between the aforementioned category with the category whose objects are propositions, and morphisms are proofs that lead one proposition (the assumption) to another (the conclusion), providing a bridge between computation and logic. Both form monoidal categories, which are categories equipped with a bifunctor that allows for the parallel composition of objects and morphisms.  

% relativity and qm 

Another application of category theory is topological quantum field theory (TQFT), a field of study which attempts to reconcile between Einstein's theory of relativity and quantum mechanics by exploring similar structures in the space occupied by the two theories \cite{Atiyah1988, Kock2004FrobeniusTheories}. It does this by associating compact oriented manifolds to Hilbert spaces. This association can be represented by a monoidal functor between \textbf{nCob}, a category with compact oriented $(n-1)$-dimensional manifolds as objects and oriented $n$-dimensional cobordisms as morphisms, and \textbf{FHilb}, a category with complex Hilbert spaces as objects and linear operators between the Hilbert spaces as morphisms \cite{Baez2009}.

% feynman, penrose, joyal street
TQFT is not the sole branch of physics that was discovered to contain applications of category theory. The famed Feynman diagrams were found to have categorical roots and were formalized in reference \cite{Joyal1991}. \citeauthor{Joyal1991} then showed the connection between the Yang-Baxter equation and braiding in Knot Theory \cite{Joyal1993BraidedCategories} using what they called braided tensor categories. Before \citet{Joyal1991}, physicists have used diagrams such as those invented by Feynman and Penrose to assist in complicated computations. By formalizing these diagrams to be objects and morphisms in a category, specifically a monoidal category, Joyal and Street provided a way for physicists to reason using diagrams in a mathematically rigorous way. More recently, \citet{Selinger2011} provided a systematic survey of the various graphical languages and their corresponding monoidal categories. 

\section{Categorical Quantum Mechanics} \label{sec: LR-CQM}

Categorical quantum mechanics (henceforth, abbreviated to CQM) is a field of study which not only aims to reconstruct the formalism of quantum mechanics with categories as its backbone, but it also seeks to formalize a graphical description that provides a high level computational method for quantum theory while seeking for new insights into its logical structure \cite{AbramskyCQM2009}. In this section, we shall provide a review of CQM from its inception in 2004 to recent progress that is relevant to the present thesis. 

In their paper {\it A Categorical Semantics of Quantum Protocols} \cite{Abramsky2004}, Abramsky and Coecke proposed what they called a strongly closed compact category with biproducts to be the mathematical framework for quantum mechanics, abstracting from \textbf{FHilb}. Throughout the years, the term strongly closed compact category has evolved into dagger compact category, emphasizing the functor equipped with the category that is normally referred to as a dagger. In the same paper \cite{Abramsky2004}, it was shown that well-known quantum protocols --- teleportation, logic-gate teleportation, and entanglement swapping --- are compatible with a quantum mechanical description in a dagger compact category. Abramsky and Coecke further showed that dagger compact categories are able to accommodate abstract notions of scalars, adjoints, and Born Rule. 

As CQM evolved into what it is today, Abramsky, Coecke, and their colleague, Aleks Kissinger made sure to provide updated reviews of the topic \cite{AbramskyCQM2009,Coecke2006,Coecke2015,Coecke2015CQM,Coecke2016a,Coecke2017PicturingProcesses}, each with a different perspective and provides new insights on CQM. To familiarize with the categories behind CQM, reference \cite{Selinger2007a} by Peter Selinger provides a comprehensive look into them. In the same work, the correspondence between the diagrams in CQM and dagger compact categories was also shown. For a more current and comprehensive paper on the underlying categories of CQM, the reader may refer to \cite{Coecke2014CategoriesChannels}.   

\subsection{Additive vs Multiplicative}\label{sec: LR-addvsmult}

CQM shifts the question about the logical structure of quantum mechanics away from the study of lattices, which lacks a satisfactory treatment of composite systems, and brought forth the tensor product as a primitive. Coecke emphasized this multiplicative structure in \cite{Coecke2007} by showing that the bulk of the required linear structures is purely multiplicative and arises from the tensor (or monoidal) product that is equipped with a dagger compact category. As CQM aims to reduce any unnecessary baggage that comes with the Hilbert space formalism in order to build a more efficient formalism, it is important to examine the roles of additive and multiplicative structures in quantum mechanics.

As a consequence of references \cite{Coecke2007} and \cite{Selinger2007a}, Coecke provided an axiomatic description of mixed states \cite{Coecke2008a}. A quantum mixed state is described by a convex combination of linear operators. The following is an equational description of a mixed quantum state:
\begin{equation*}
    \rho = \sum_{j=1}^n p_j \ket{\psi_j}\bra{\psi_j}
\end{equation*}
where each $p_j$ is the probability that $\rho$ is in the pure state described by the vector $\ket{\psi_j}$. Using Selinger's abstract completely positive maps, Coecke makes the sum in the description above implicit, allowing for a less cumbersome graphical representation of mixed states. 

This suggests that CQM may render sums as redundant. However, due to the probabilistic nature of quantum mechanics which relies heavily on the use of sums, \citeauthor{AbramskyCQM2009} included biproducts in their categorical description of quantum mechanics. This is unfortunate as their 2-dimensional string diagrams do not allow for another type of composition. 

We shall discuss this issue further in the next section, but eventually, a description that does not rely on sums was found \cite{Coecke2008c, Coecke2008POVMsSums,Selinger2007a,Coecke2008a}. This led to the construction of a category of observable algebras and superoperators \cite{Coecke2014CategoriesChannels}, which was shown to embed the category of completely positive maps with biproducts \cite{Heunen2014}.

\subsection{Abstract Bases}\label{sec: LR-bases}

 In reference \cite{Coecke2008c}, it was found that sums are an implicit implementation of the capability to copy and delete quantum and classical data types. The key to this procedure is the no-cloning and no-deleting theorems in quantum mechanics. These theorems tell us that quantum and classical data types can be distinguished through their respective capabilities to be copied and deleted; that is, classical data types can be copied and deleted and quantum data types cannot. 

For a system represented by an object $A$ in a symmetric monoidal category, we can represent the capability to copy information about the system by the existence of a morphism of the form $A\rightarrow A\otimes A$, called comultiplication, and the capability to delete information by the existence of a morphism of the form $A\rightarrow I$, called counit, where $I$ is the monoidal unit of the category and $\otimes$ is the tensor product equipped with the category. These morphisms must satisfy the axioms of a cocommutative comonoid (the dual of a commutative monoid), and a classical structure can be described as an object equipped with a comultiplication and a counit morphisms. 

In a dagger symmetric monoidal category, we automatically obtain a monoid from a comonoid as the equipped dagger functor dualizes the morphisms in the category. To obtain abstract bases, we need a comonoid and the induced monoid to satisfy the Frobenius law, resulting in a dagger Frobenius algebra. It was shown by \citet{Coecke2013New} that there is a bijective correspondence between orthogonal bases of finite complex Hilbert spaces with special dagger commutative Frobenius algebras in \textbf{FHilb}. \citet{Coecke2008POVMsSums} also provided an abstraction of POVMs, and showed Naimark's theorem using a purely graphical proof technique. \citet{Coecke2010a} then extended their work on abstract quantum measurements to provide a categorical description of classical operations which consequently, provided a resource sensitive account of quantum-classical interactions. 

\subsection{Observables}\label{sec: LR-observables}

In the Hilbert space formalism of quantum mechanics, an observable can be described by an eigenbasis of a Hermitian operator, and so we can use the abstract bases in reference \cite{Coecke2013New} to describe observables in the setting of a symmetric monoidal category. Unlike observables in classical physics which always admit sharp values at the same time, not all quantum observables are compatible in such a way. 

This kind of incompatibility is referred to as complementarity in quantum mechanics. \citet{Coecke2008InteractingObservables} showed that any pair of complementary observables forms a bialgebra. They then extended their idea on complementarity in a longer paper with a similar title \cite{Coecke2011b} which showed that two complementary observables not only form a bialgebra, but they also form a Hopf algebra. Other developments on complementarity in CQM are detailed in \cite{Coecke2012} which introduced a stronger notion of complementarity, and \cite{Heunen2012a} which relates notion of complementarity in von Neumann algebras, Hilbert spaces, and orthomodular lattices. 

Abstract observables also allow for a comparison between local and non-local theories. In reference \cite{Coecke2011c}, a category for stabilizer qubit theory --- which is non-local --- and a category for Spekkens' toy theory -- which is local --- were constructed. The two categories, denoted as \textbf{Stab} and \textbf{Spek}, were shown to be similar except for one key aspect: their phase groups. That is, the phase group for \textbf{Stab} is the cyclic group of order 2, and the phase group for \textbf{Spek} is the Klein group. Reference \citet{Coecke2011c} further showed the relationship between phase groups and GHZ state correlations, introducing a key property for classifying between local and non-local behaviours.    

Mutually unbiased qubit theories led to the invention of ZX-calculus, a graphical calculus based on the Pauli $Z$ and $X$ spin observables which is intended to be used for quantum computation. In fact, ZX-calculus has been implemented in Quantomatic, an open source software which uses string diagrams to construct equational proofs \cite{Dixon2009,Kissinger2015,Quantomatic}. ZX-calculus was shown to be complete for stabilizer quantum mechanics\cite{Backens2013}, and similar graphical calculus was also developed for Spekkens' toy bit theory \cite{Backens2016}. In an article by \citet{Duncan2012}, ZX-calculus was applied to measurement-based quantum computation. More recently, a revised version of ZX-calculus was considered in the setting qutrits \cite{Wang2018qutritZX}.

Another important consequence of abstract observables is the ability to depict classical-quantum interactions using the graphical language of CQM. As far as we can tell, \citet{Coecke2008c} provided the first graphical description of classical-quantum interactions. \citet{Coecke2012EnvironmentMechanics} then refined the notion of classicality by proposing that the distinguishing property between classical and quantum data types as their abilities to be broadcast. In the article, \citeauthor{Coecke2012EnvironmentMechanics} unified a notion of environment with classical structures to define abstract classical channels, quantum measurements, and classical control. Furthermore, they adjoined the notion of complementarity we discussed earlier to derive some classically controlled quantum protocols. \citet{Coecke2016a} provided an updated discussion on classical-quantum interactions in CQM which could aide comprehension of the articles mentioned as it focuses more on the diagrammatic language of CQM than its categorical counterpart.     

The construction of abstract observables highlights the internal algebraic structure of observables. That is, the multiplication morphism of an abstract observable can be treated as a binary operation where the unit morphism is its identity element. It turns out that we can also reveal the algebraic structure of maximally entangled states through this construction. \citet{Coecke2010} showed that the GHZ-state and W-state (maximally entangled states for tripartite states) induce commutative Frobenius algebras. However, the commutative Frobenius algebra induced by the GHZ-state is special, and the one induced by the W-state is anti-special. This classification of maximally entangled tripartite states yields a compositional graphical model for expressing general multipartite states.  

\subsection{Foundational Issues} \label{sec: LR-foundation}

In reference \cite{Coecke2013}, \citeauthor{Coecke2013} proposed an alternative framework to quantum logic. They recast the order theoretic structure that quantum logic is built upon so that it comes with a primitive composition operation, making the construction completely compositional. When interpreted in \textbf{FHilb}, this construction yields the projection lattices of arbitrary finite-dimensional $C^*$-algebras. This hints at the viability of a root concept in CQM, i.e compositionality, as an alternative to quantum logic's propositional structure. \citet{Harding2009} has also showed the connection between the notion of a preparation (of a system), as defined by \citet{Abramsky2004}, to the notion of orthoalgebras in quantum logic.

Outside of quantum logic, there are also works investigating the connection between CQM and other areas of study which aim to answer the foundational issues of quantum mechanics. \citet{Barnum2013} investigated the connection between CQM and the so called convex operational models, which generalize the probabilistic structure of quantum mechanics using techniques from measure theory and functional analysis. In the article, the category for convex operational models, which is symmetric monoidal, was constructed, along with compact and dagger structures in such a category. \citet{Chiribella2014DistinguishabilityTheories} also studied CQM from the context of generalized probabilistic theories. This is part of Chiribella's research to construct a framework for operational-probabilistic theories \cite{Chiribella2011InformationalTheory, Chiribella2010}, which uses diagrams similar to those in CQM. Chiribella's work on causality has also been applied to causal structures in CQM by \citet{coecke2013causalProcesses}.  

Furthermore, the graphical language of CQM is general enough that one of its founders, Bob Coecke, has dubbed it generalized process theories \cite{Coecke2015CQM,Coecke2016a}, and it became the framework for the construction of theories that can be applied beyond quantum mechanics. Examples include resource theories \cite{Coecke2016} and natural language processing \cite{Piedeleu2015OpenProcessing}. 

\section{Contribution to the Literature}

Since its introduction in 2004, CQM has become an active field of research, so much so that its graphical approach has been applied to areas outside of quantum mechanics. Along with natural language processing, the techniques of CQM has been adapted for network theory \cite{Baez2012} --- providing pathways to applications in chemistry and biology ---  and important tools in mathematical modelling and signal flow graphs, such as \cite{Bonchi2017InteractingAlgebras, Fong2016} and control theory \cite{Baez2015CategoriesControl}. 

A benefit of CQM's graphical approach, otherwise known as process theories, is its ability to make explicit difficult mathematical concepts, rewriting them as intuitive diagrams. One example is quantum teleportation, which we touched upon in Section \ref{sec:process-intro}. Another example is linear algebra where a matrix --- usually presented as number arrays --- are replaced by flow diagrams, which makes explicit the linear map equivalent to it \cite{Sobocinski201510.Algebra}. That is, a vector can be inputted into the matrix, like a mathematical machine, and inside the machine, the entries of the vector are multiplied by numbers that are the entries of the matrix, resulting in an output that is a vector.  

In this work, we would like to depict the entanglement structure of mutually unbiased bases in qubit systems that appeals to our intuition in ways similar to \cite{Coecke2017PicturingProcesses} and \cite{Sobocinski201510.Algebra}. We do this by utilizing the compositionality of process theories to construct classical structures on multiple qubits from classical structures on a single qubit and compose them via parallel composition (see Section \ref{sec:process-intro}) or bipartite processes that we call `connecting wires'. Then we determine the complementarity between each pair of the resulting classical structures by checking whether or not they satisfy Definition \ref{def:complement}.  

There have been studies related to mutually unbiased bases within the setting of CQM such as the work by \citet{Musto2016} and that by \cite{Evans2009MUB}. However, as far as we are aware, our method of composing classical structures has not appeared in any existing literature and we have not found any literature which utilizes the compositionality of process theories to study entanglement structure on mutually unbiased bases in the same way that we do in the present work.

\chapter{Complementary Classical Structures}\label{sec:diagram-algebra}

Applications of MUBs can be found in quantum key distribution, quantum state determination, detection of quantum entanglement, and various other areas of study in quantum mechanics. Therefore, for any model of quantum mechanics, it should be of utmost interest to have a description of MUBs. Indeed, we gave this description in Definition \ref{def:complement} within a process-theoretic framework.

While techniques for finding a pair of MUBS in Hilbert spaces are available \cite{DurtMUB2010}, searching for \textit{maximal complete} sets of MUBs is much trickier. In fact, searching for a maximal complete set of MUBs for a Hilbert space with non prime power dimension is such a notoriously difficult problem that even for the smallest dimension of six, experts only have a strong suspicion that there are at most three mutually unbiased bases.
 
As far as we are aware, CQM has not provided anything new towards solving this problem, nor do we claim that our procedure will lead to a solution, but we do believe that the compositionality of quantum processes --- in particular, how it depicts correlation and separability --- will provide a better understanding towards the advantages %(or pitfalls) 
of composing complementary classical structures on systems to obtain complementary structures for a larger system. % we shall also show how this can be done more efficiently through diagrams. 

\section{Strong Complementarity}

Before we proceed to our main goal, we introduce strong complementarity. As the name suggests, it is a stronger version of complementarity. In particular, it implies complementarity, but the converse is not true in general.

\begin{definition}\label{def:strong-complementarity}
\cite{Coecke2012}
Two classical structures with spiders $\black$ and $\blackb$, respectively, are strongly complementary if they satisfy the following equations:

\begin{center}
\begin{minipage}{0.48\textwidth}
\begin{equation}\label{eq:strong-complement-copy1}
    \includegraphics[scale=0.2]{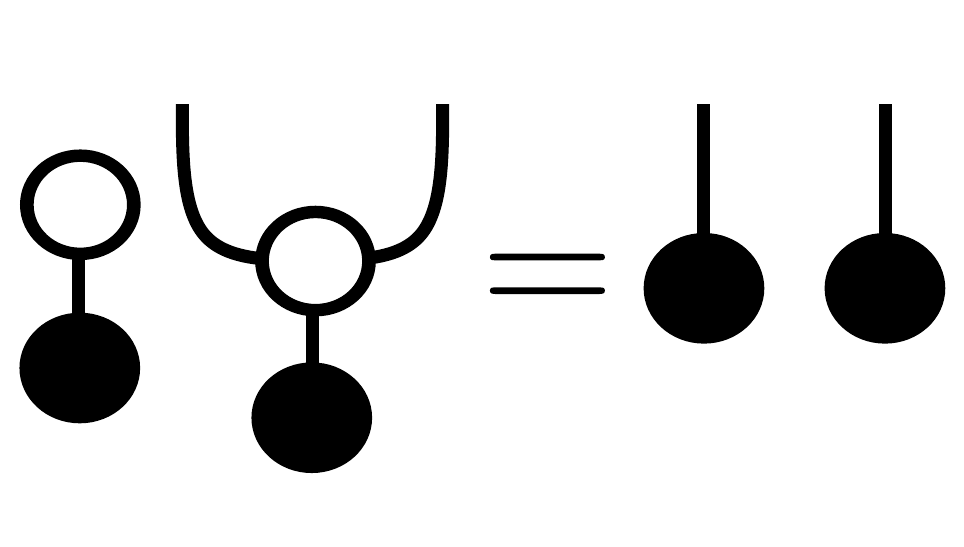}
\end{equation}
\end{minipage} 
\begin{minipage}{0.48\textwidth}
\begin{equation}\label{eq:strong-complement-copy-2}
    \includegraphics[scale=0.2]{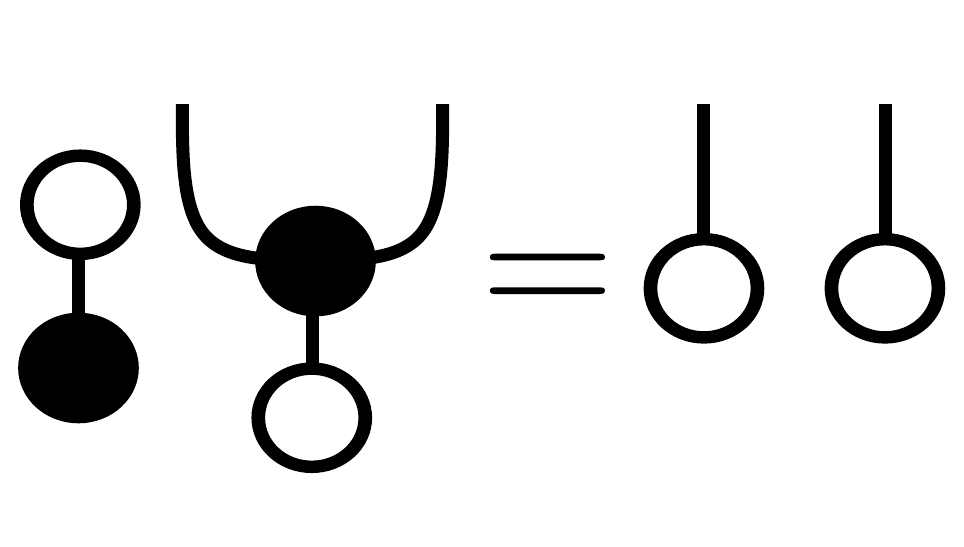}
\end{equation}
\end{minipage} 
\\
\begin{minipage}{0.48\textwidth}
\begin{equation}\label{eq:strong-complement-scalar}
    \includegraphics[scale=0.2]{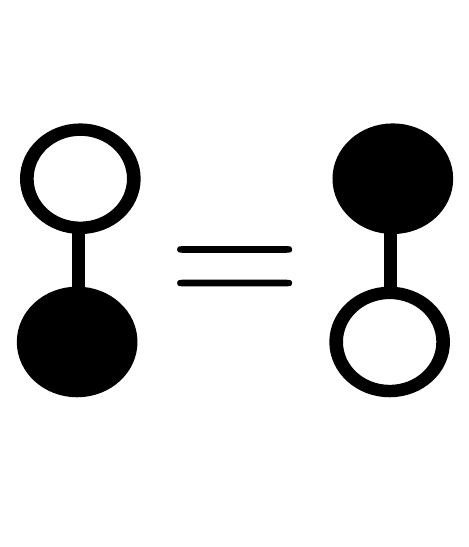}
\end{equation}
\end{minipage}
\begin{minipage}{0.48\textwidth} 
\begin{equation}\label{eq:strong-complement-bialg}
    \includegraphics[scale=0.2]{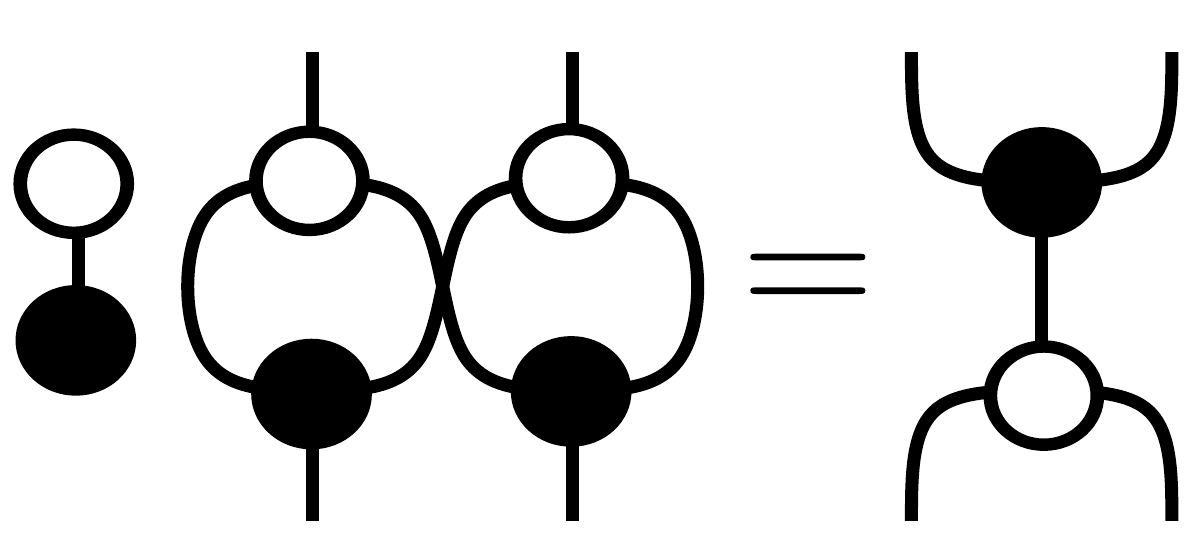}
\end{equation}
\end{minipage}
\end{center}
\end{definition}

\noindent Note that the equations in Definition \ref{def:strong-complementarity} are symmetric, i.e. the colours of the nodes can be swapped. 

For a Hilbert space with dimension $\geq 2$, it was shown that the largest set of pairwise strongly complementary classical structures cannot contain more than two members \cite{Coecke2012}. 

Two strongly complementary structures --- albeit, with some additions --- forms a graphical calculus for qubits called ZX-calculus. ZX-calculus is shown to be sound and universal for pure qubit quantum mechanics and was shown to be complete for stabilizer quantum mechanics \cite{Backens2016}. 

In the next section, we briefly recount the generators and rewrite rules of ZX-calculus and show how ZX-calculus relates to the Pauli operators on a single qubit.

\section{ZX-calculus}

Before we proceed, we shall summarize our notations for the eigenbases of the Pauli operators. The eigenbasis for the Pauli $Z$ operator is denoted by $\{\ket{0},\ket{1}\}$. This is our chosen standard basis. We denote the eigenbases of the Pauli $X$ and $Y$ operators respectively as $\{\ket{0_X},\ket{1_X}\}$ and $\{\ket{0_Y},\ket{1_Y}\}$. Members of both bases can be rewritten with respect to the standard basis:
\begin{eqnarray}
\ket{j_X}=\frac{1}{\sqrt{2}}(\ket{0}+(-1)^j\ket{1})\\
\ket{k_Y}=\frac{1}{\sqrt{2}}(\ket{0}+(-1)^k i\ket{1})
\end{eqnarray}

ZX-calculus consists of two classical structures, commonly represented by $\zspider$ and $\xspider$, which are strongly complementary. We refer to the classical structure consisting of $\zspider$ as $\mathcal{Z}$ and the classical structure consisting of $\xspider$ as $\mathcal{X}$. Spiders of $\mathcal{Z}$ and $\mathcal{X}$ can have phases $-\pi <\theta\leq\pi$, and a plain spider is a spider with phase 0. In addition to the spiders of $\mathcal{Z}$ and $\mathcal{X}$, ZX-calculus also consists of a scalar called \textit{star} and an operation called \textit{Hadamard}, denoted by the following diagrams:
\begin{equation*}
    \includegraphics[scale=0.2]{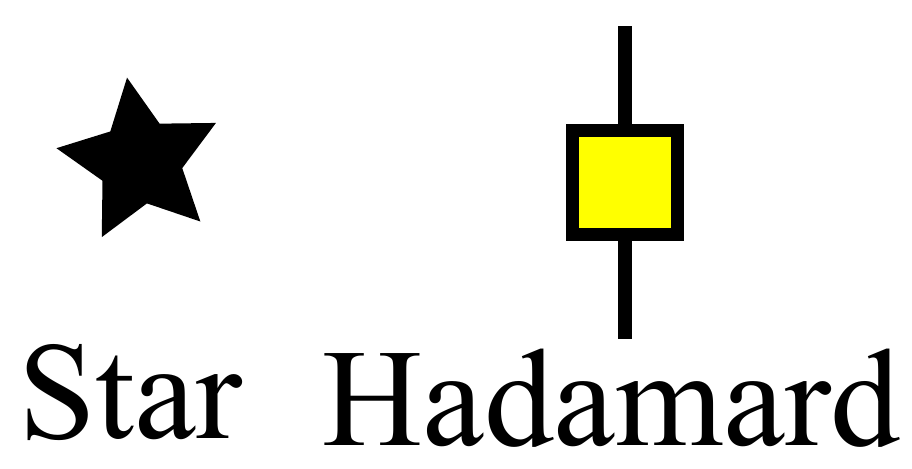}
\end{equation*}
Each generator of ZX-calculus has an interpretation in pure qubit quantum mechanics:
\begin{eqnarray*}
\includegraphics[scale=0.2]{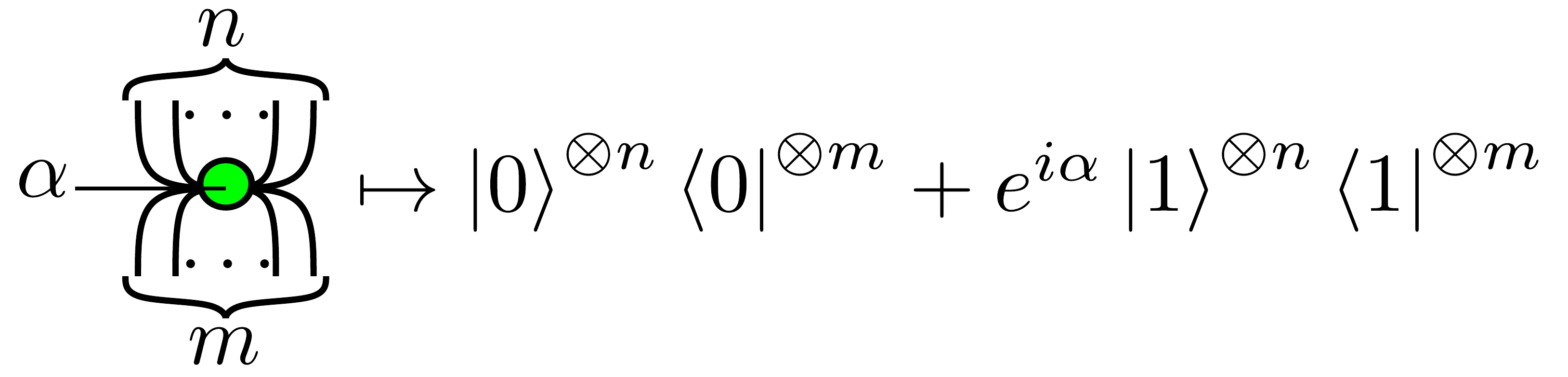}\\
\includegraphics[scale=0.2]{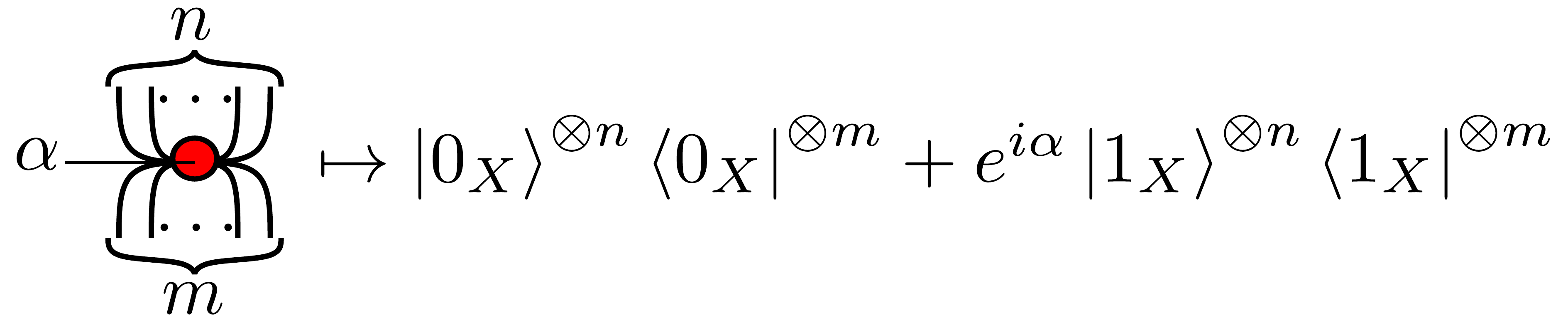}\\
\includegraphics[scale=0.2]{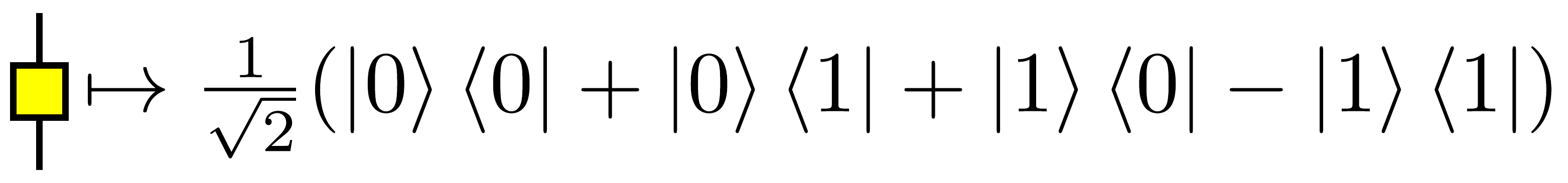}\\
\includegraphics[scale=0.2]{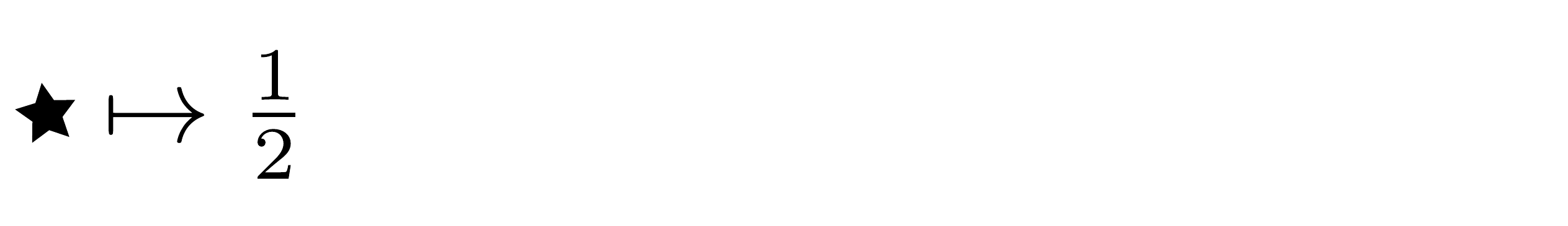}
\end{eqnarray*}
where $\alpha\in[-\pi,\pi]$.

The interpretation above can be represented formally as a monoidal functor. Details can be found in \cite{Backens2016}. 

Recall from Section \ref{sec:process-intro} where sequential composition of processes is $\circ$ and parallel composition of processes is $\otimes$. A diagram in ZX-calculus is a composition, either in sequence or in parallel, of the generators above. So to interpret a diagram in ZX-calculus as operations on qubit systems, one shall need to apply either $\circ$ or $\otimes$ on the interpretations of the generators above. Below is an example of an interpretation of a diagram in ZX-calculus to an operation on qubit systems:
\begin{equation*}
    \includegraphics[scale=0.2]{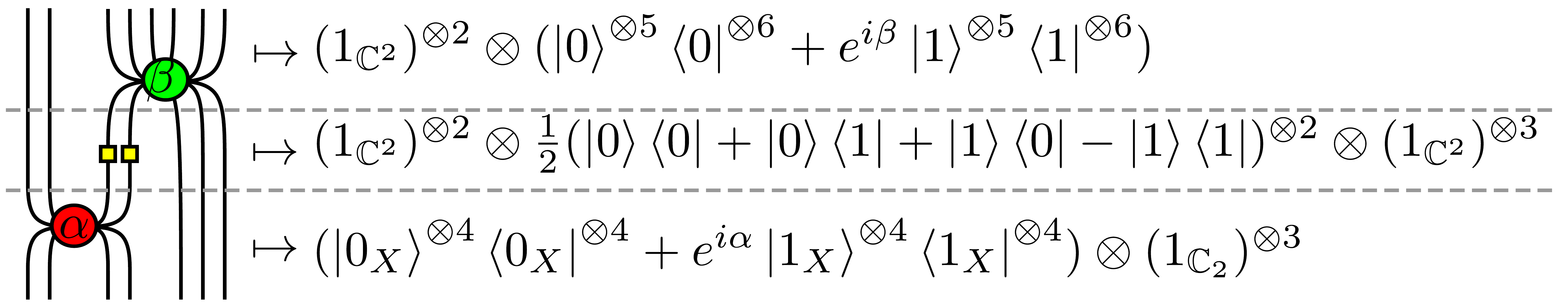}
\end{equation*}
where the diagram is partitioned into three subdiagrams with their interpretations on the right. The complete interpretation of the diagram is then $[(\ket{0_X}^{\otimes 4}\bra{0_X}^{\otimes 4}+e^{i\alpha}\ket{1_X}^{\otimes 4}\bra{1_X}^{\otimes 4})\otimes (1_{\mathbb{C}_2})^{\otimes 3}]\circ[(1_{\mathbb{C}^2})^{\otimes 2}\otimes\frac{1}{2}(\ket{0}\bra{0}+\ket{0}\bra{1}+\ket{1}\bra{0}-\ket{1}\bra{1})^{\otimes 2}\otimes (1_{\mathbb{C}^2})^{\otimes 3}]\circ[(1_{\mathbb{C}^2})^{\otimes 2}\otimes(\ket{0}^{\otimes 5}\bra{0}^{\otimes 6}+e^{i\beta}\ket{1}^{\otimes 5}\bra{1}^{\otimes 6})]$. While it is certainly straightforward to switch between diagrams and their equational counterparts, we opt not to provide interpretations of all diagrams in ZX-calculus since the resulting interpretation of even a simple diagram is space-consuming. 

The generators of ZX-calculus satisfy certain axioms called rewrite rules. These rewrite rules are given as Eqs. \ref{eq:zx1-fuse-green}-\ref{eq:zx16-zero}. The following rewrite rules are necessary to obtain a complete description of stabilizer quantum mechanics \cite{Backens2016}:

\begin{center}
\begin{minipage}{0.48\textwidth}
 The green fuse rule:
\begin{equation}\label{eq:zx1-fuse-green}
    \includegraphics[scale=0.2]{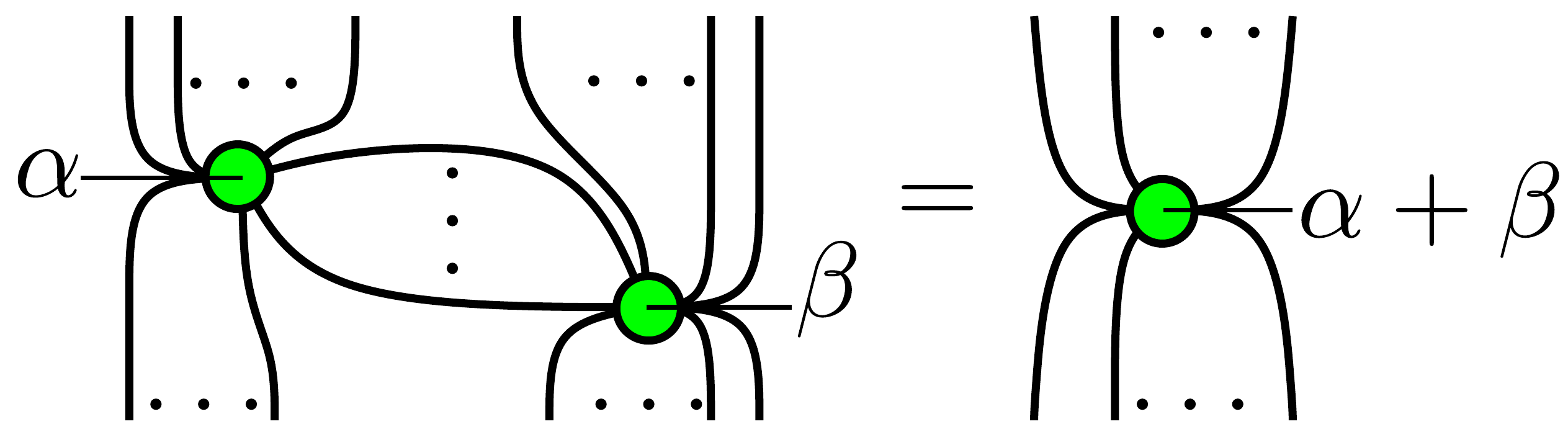}
\end{equation} 
 \end{minipage}
\begin{minipage}{0.48\textwidth}
The green loop rule:
\begin{equation}\label{eq:zx3-loop-green}
    \includegraphics[scale=0.2]{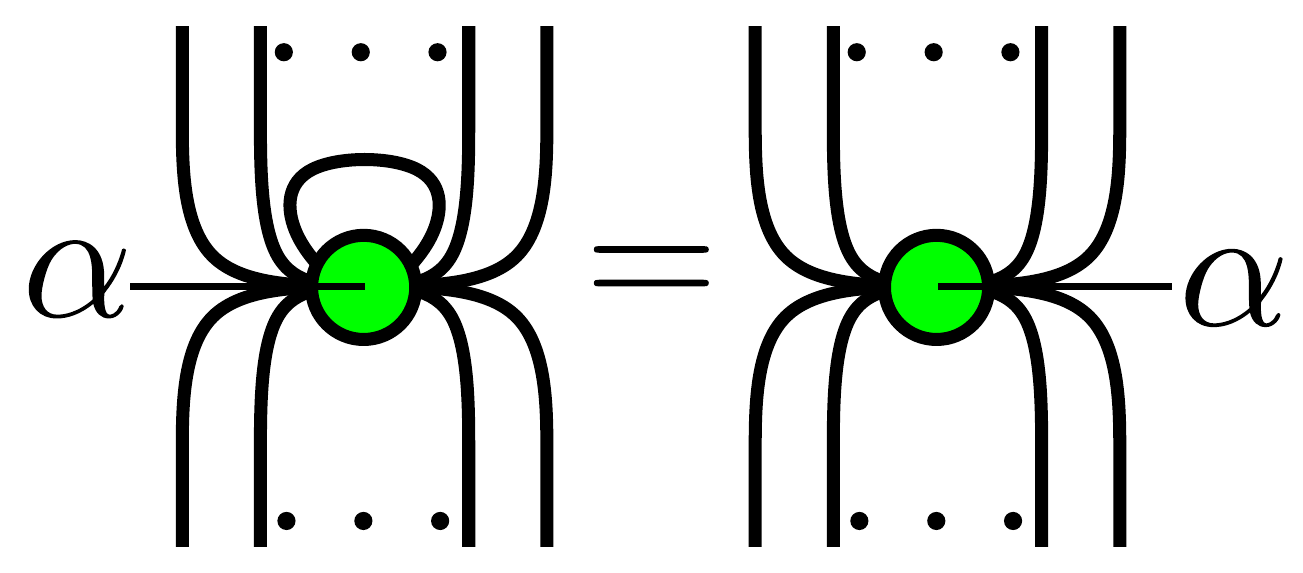}
\end{equation}
\end{minipage}
\end{center}

\begin{center}
\begin{minipage}{0.48\textwidth}
The red fuse rule:
\begin{equation}\label{eq:zx2-fuse-red}
    \includegraphics[scale=0.2]{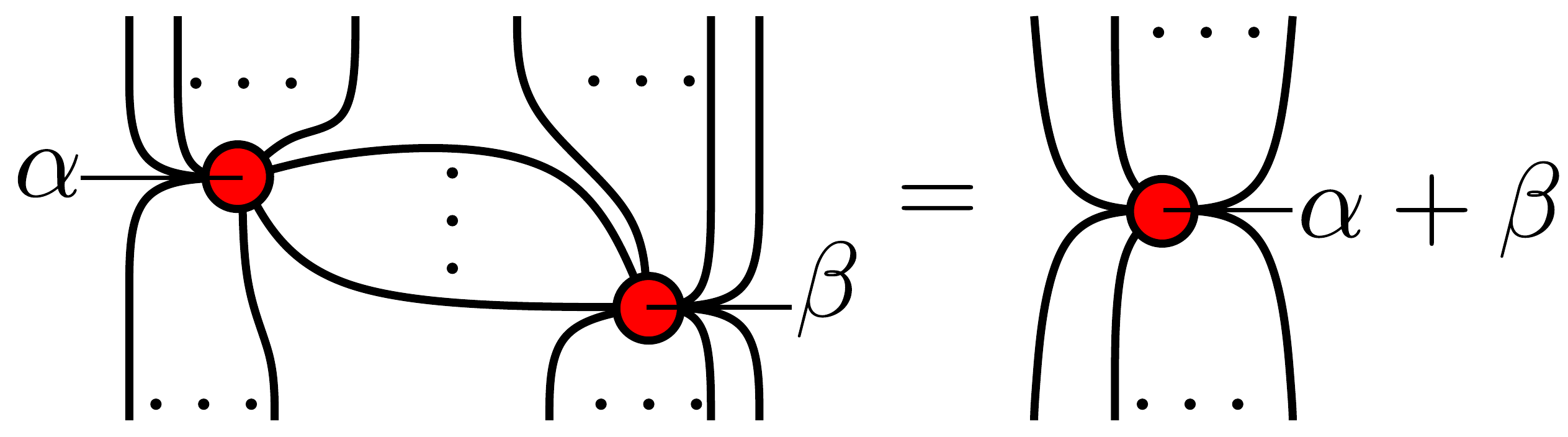}
\end{equation}
\end{minipage}
\begin{minipage}{0.48\textwidth}
The red loop rule:
\begin{equation}\label{eq:zx4-loop-red}
    \includegraphics[scale=0.2]{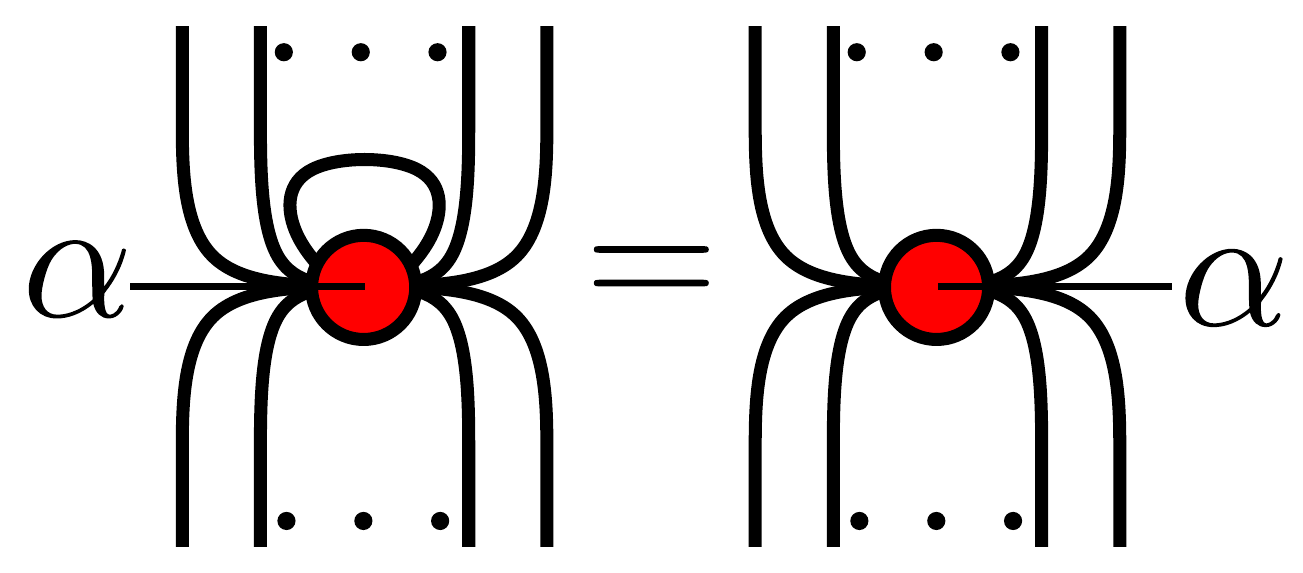}
\end{equation}
\end{minipage}
\end{center}

\begin{center}
\begin{minipage}{0.48\textwidth}
The green cup rule:
\begin{equation}\label{eq:zx5-cup-green}
    \includegraphics[scale=0.2]{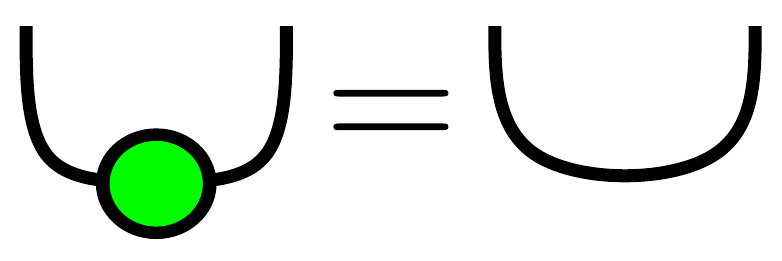}
\end{equation} 
\end{minipage}
\begin{minipage}{0.48\textwidth}
 The red cup rule:
\begin{equation}\label{eq:zx6-cup-red}
    \includegraphics[scale=0.2]{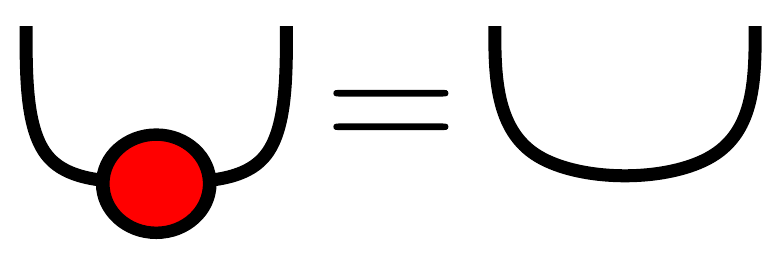}
\end{equation}
\end{minipage}
\end{center}

\begin{center}
\begin{minipage}{0.48\textwidth}
The green $\pi$-copy rule: 
\begin{equation}\label{eq:zx7-copy-green-pi}
    \includegraphics[scale=0.2]{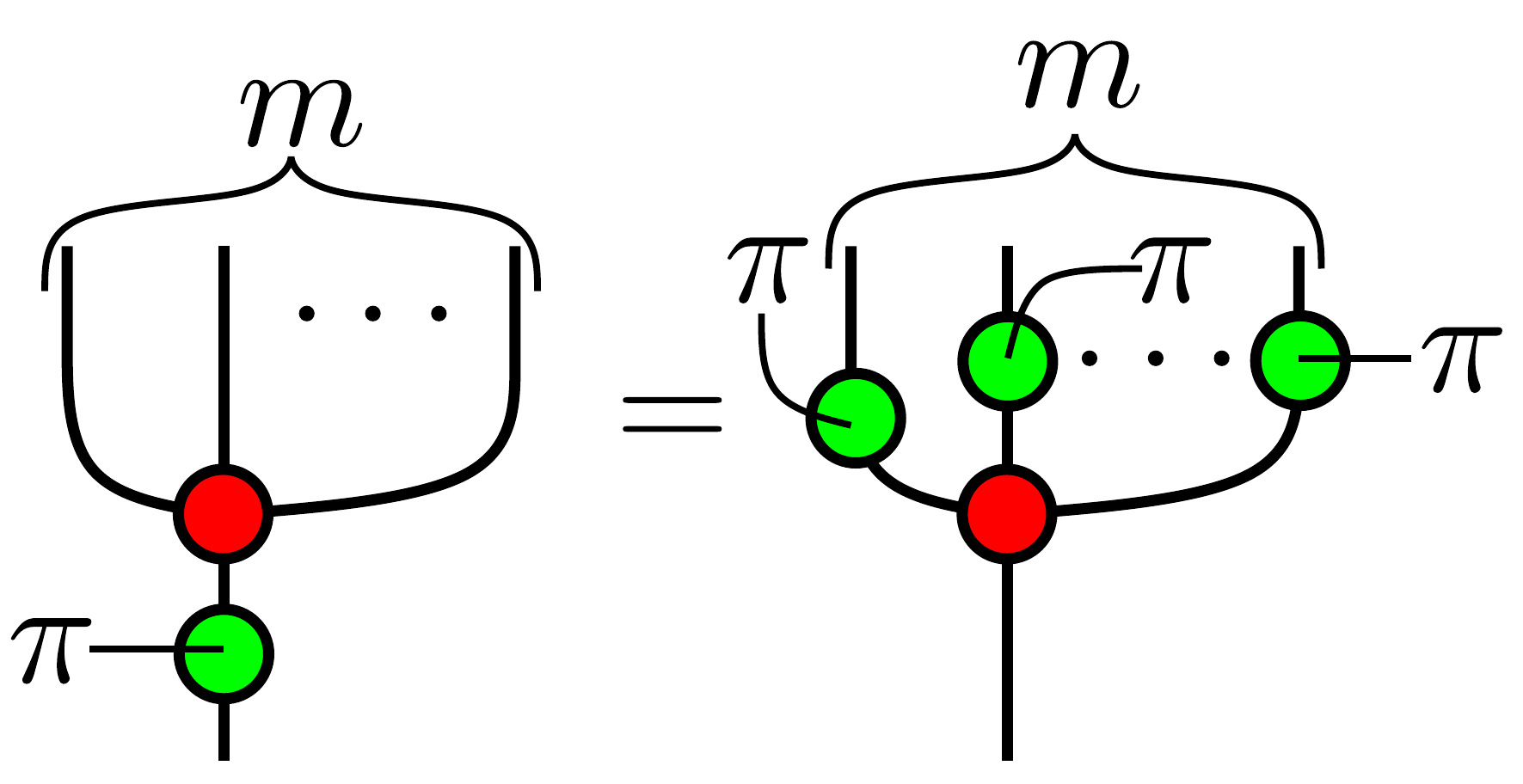}
\end{equation}
\end{minipage}
\begin{minipage}{0.48\textwidth}
The red $\pi$-copy rule:
\begin{equation}\label{eq:zx8-copy-red-pi}
    \includegraphics[scale=0.2]{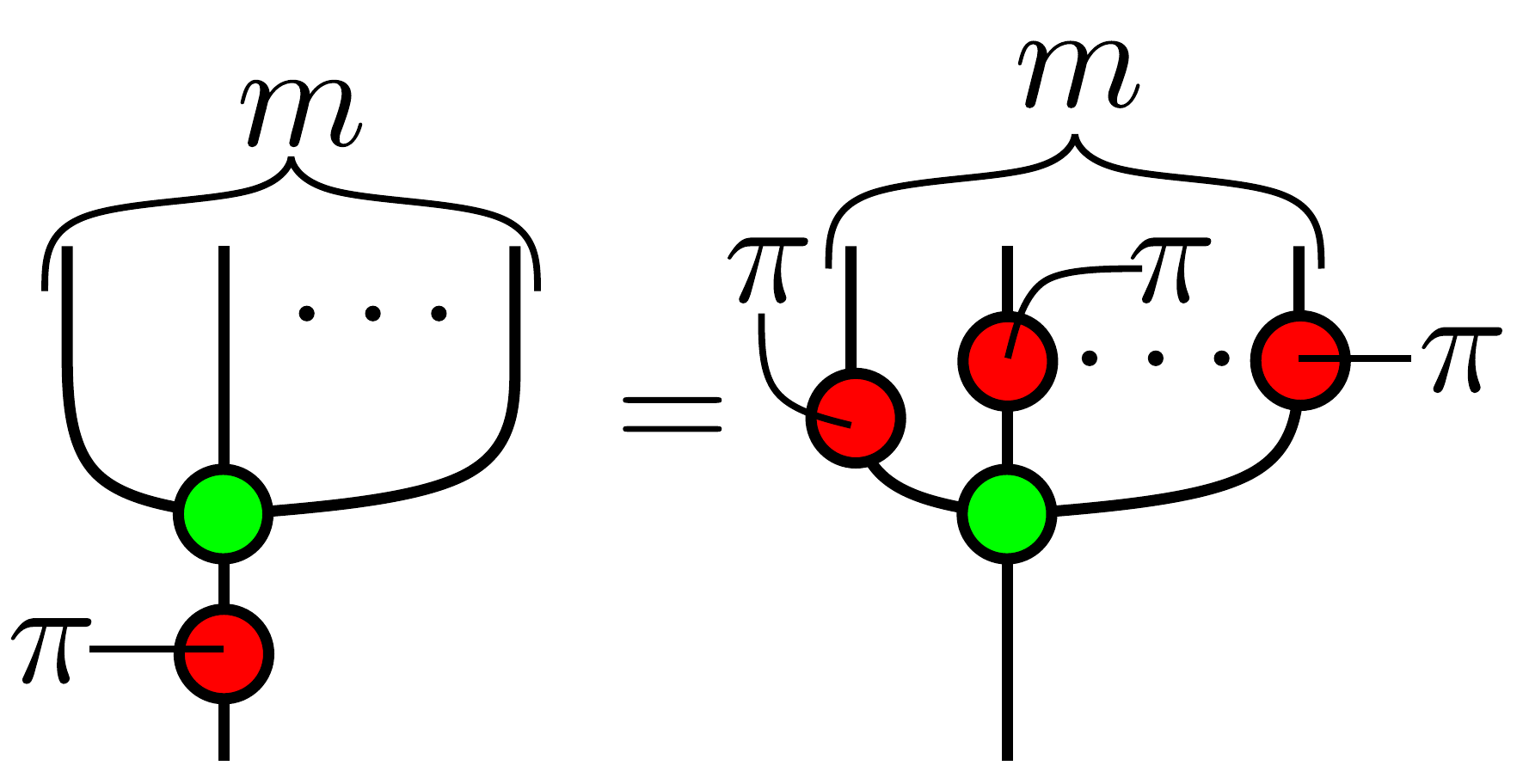}
\end{equation}
\end{minipage}
\end{center}

\begin{center}
\begin{minipage}{0.48\textwidth}
The green copy rule: 
\begin{equation}\label{eq:zx9-copy-green}
\includegraphics[scale=0.2]{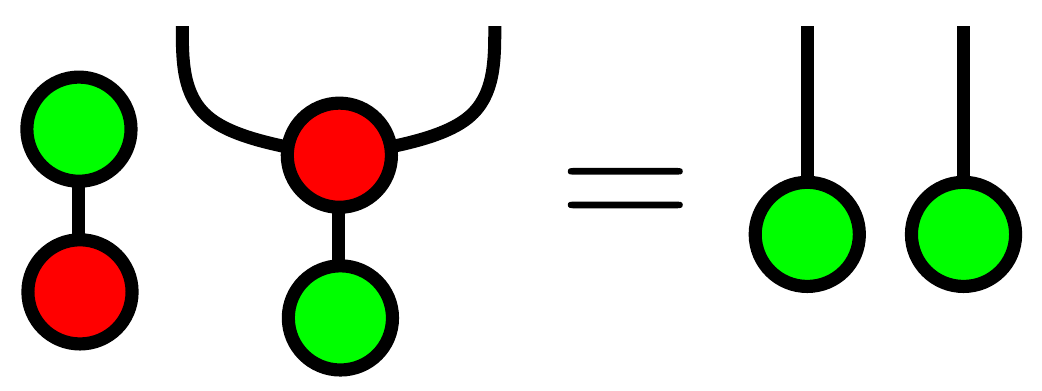}    
\end{equation}
\end{minipage}
\begin{minipage}{0.48\textwidth}
 The red copy rule:
\begin{equation}\label{eq:zx10-copy-red}
    \includegraphics[scale=0.2]{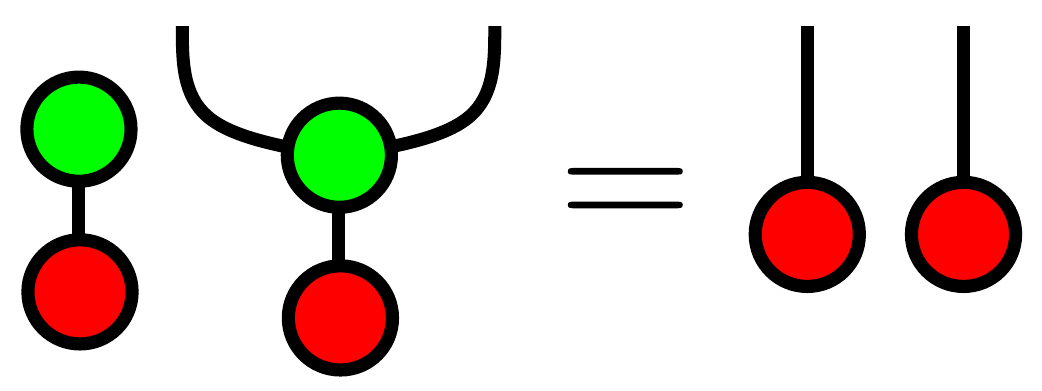}
\end{equation}
\end{minipage}
\end{center}

\begin{center}
\begin{minipage}{0.48\textwidth}
The bialgebra rule:
\begin{equation}\label{eq:zx11-bialgebra}
    \includegraphics[scale=0.2]{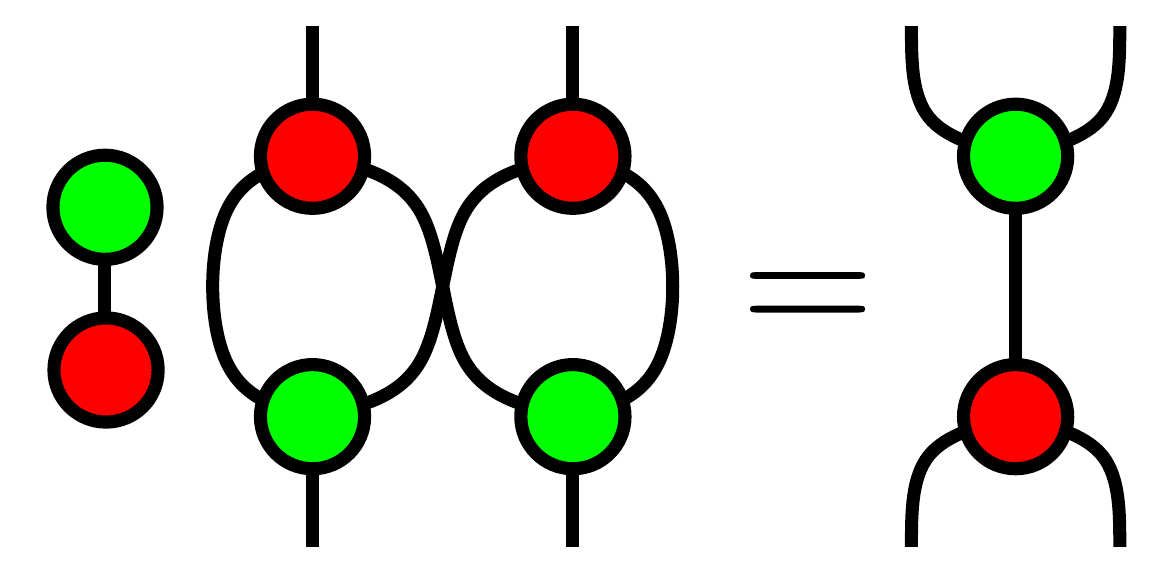}
\end{equation}
\end{minipage}
\begin{minipage}{0.48\textwidth}
The colour change rule:
\begin{equation}\label{eq:zx12-had-transform}
    \includegraphics[scale=0.2]{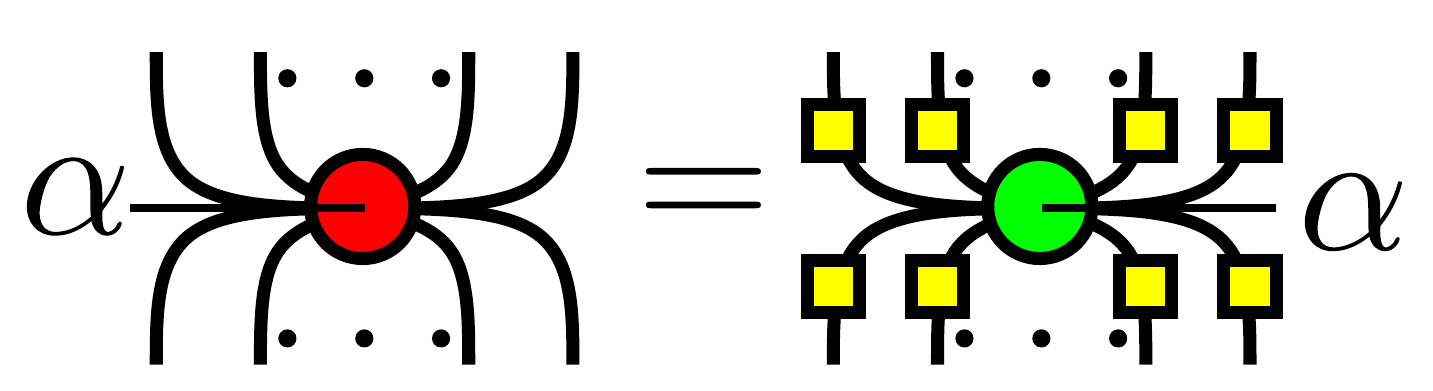}
\end{equation}
\end{minipage}
\end{center}

\begin{center}
\begin{minipage}{0.48\textwidth}
The Euler decomposition rule:
\begin{equation}\label{eq:zx13-had-euler}
    \includegraphics[scale=0.2]{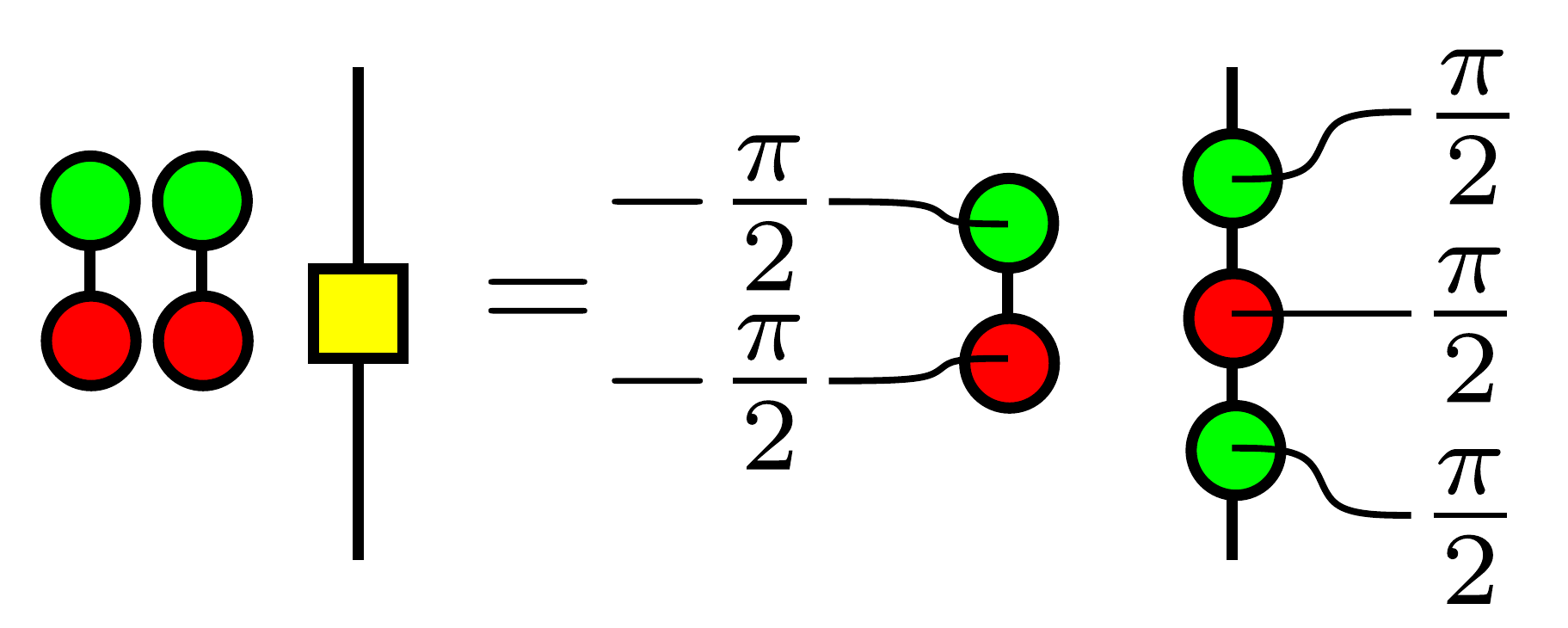}
\end{equation}
\end{minipage}
\begin{minipage}{0.48\textwidth}
The zero rule:
\begin{equation}\label{eq:zx14-pi-disconnect}
    \includegraphics[scale=0.2]{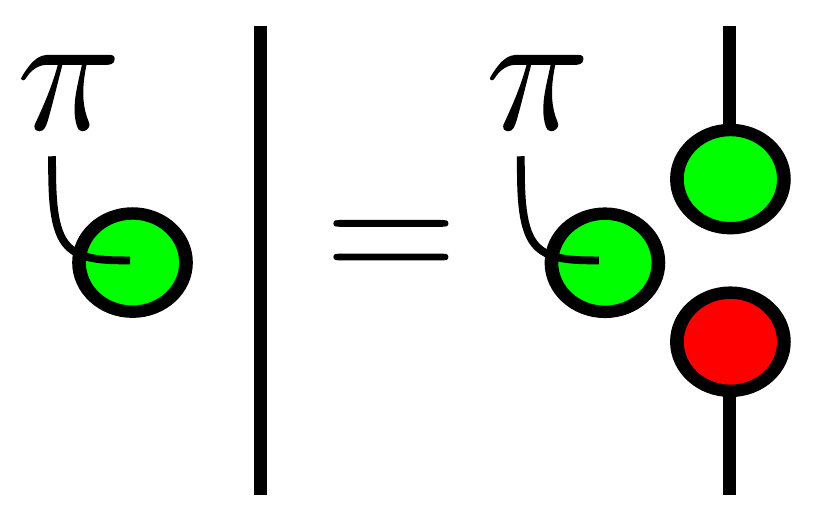}
\end{equation}
\end{minipage}
\end{center}

\begin{center}
\begin{minipage}{0.48\textwidth}
The zero rule:
\begin{equation}\label{eq:zx15-star}
    \includegraphics[scale=0.2]{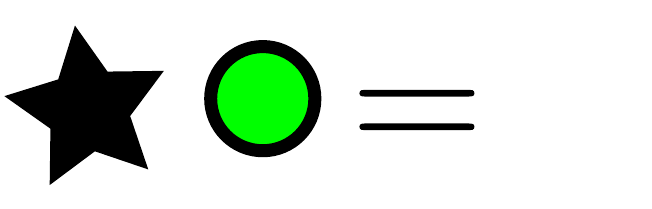}
\end{equation}
\end{minipage} 
\begin{minipage}{0.48\textwidth}
The zero scalar rule:
\begin{equation}\label{eq:zx16-zero}
    \includegraphics[scale=0.2]{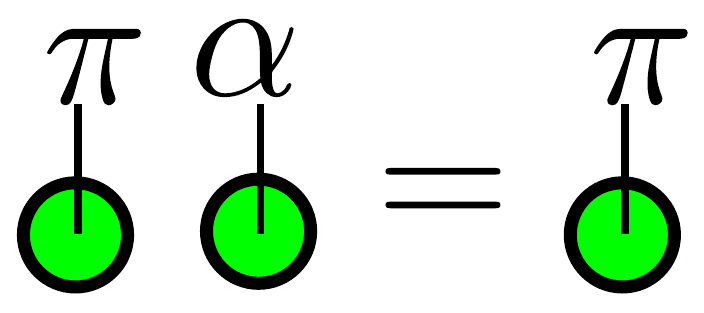}
\end{equation}
\end{minipage}
\end{center}
% zx-calculus rewrite rules including derived rules

From Eqs. \ref{eq:zx1-fuse-green}-\ref{eq:zx16-zero}, we can derive the following equations:

\begin{center}
\begin{minipage}{0.48\textwidth}
The Hopf rule:
\begin{equation}\label{eq:dzx1-hopf}
    \includegraphics[scale=0.2]{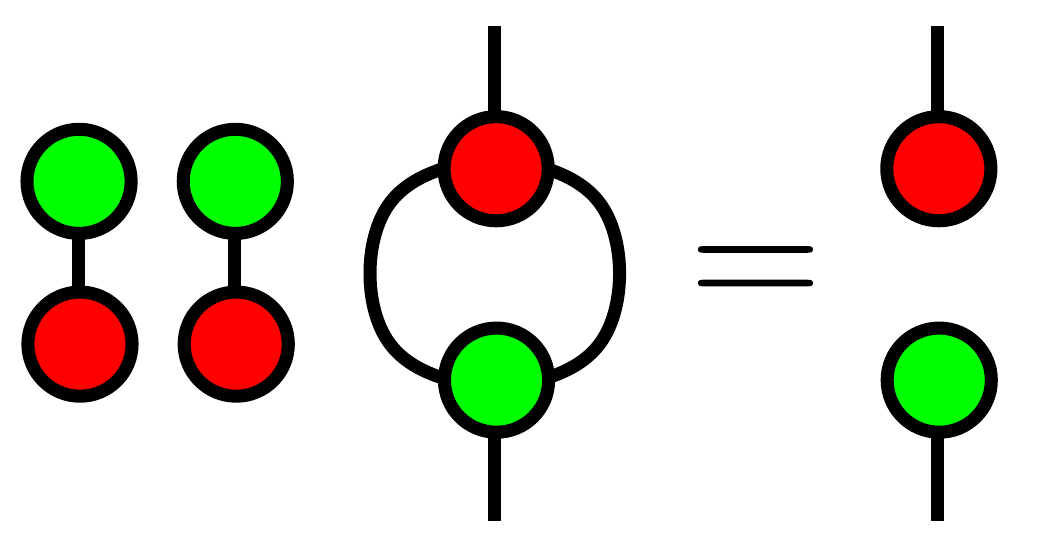}
\end{equation}
\end{minipage}
\begin{minipage}{0.48\textwidth}
The yanking rule:  
\begin{equation}\label{eq:dzx2-yank}
\includegraphics[scale=0.2]{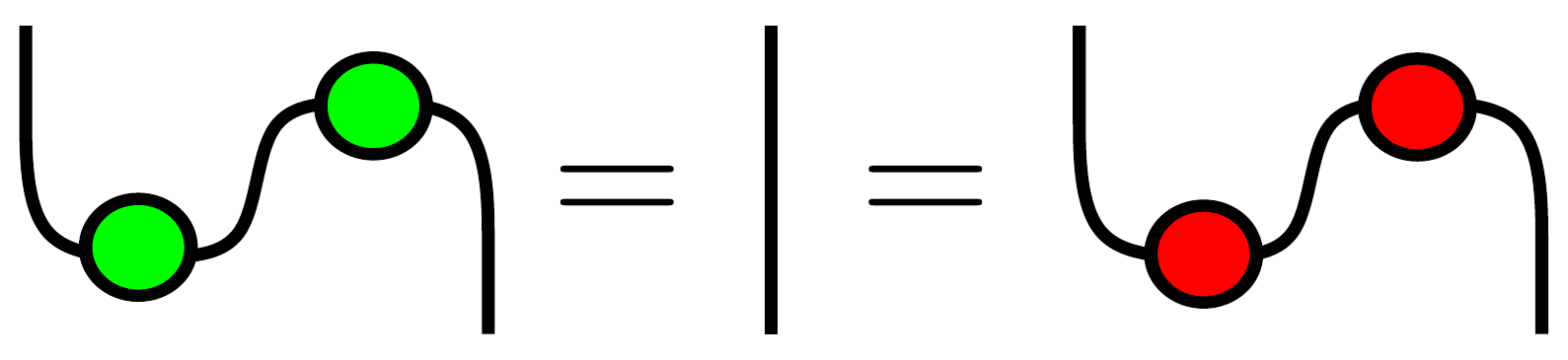} 
\end{equation} 
\end{minipage}
\end{center}

\begin{center}
\begin{minipage}{0.48\textwidth}
The identity rule:
\begin{equation}\label{eq:dzx3-identity}
    \includegraphics[scale=0.2]{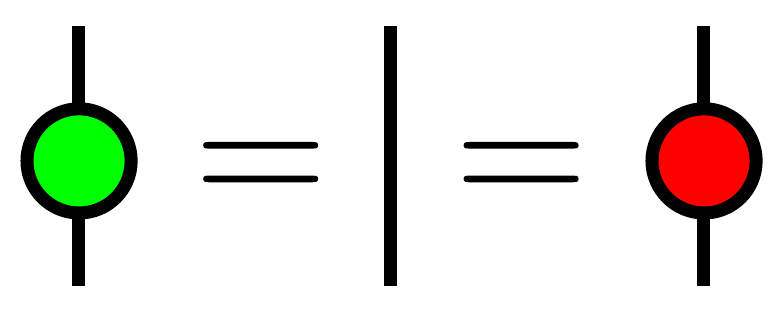}
\end{equation}
\end{minipage}
\begin{minipage}{0.48\textwidth}
The star-inverse rule:
\begin{equation}\label{eq:dzx4-star-inverse}
    \includegraphics[scale=0.2]{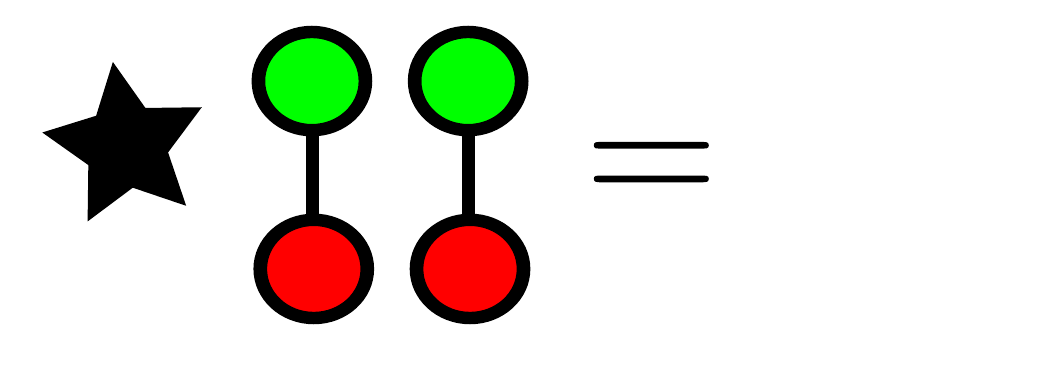}
\end{equation} 
\end{minipage}
\end{center}

\begin{center}
\begin{minipage}{0.48\textwidth}
The Hadamard-unitary rule:
\begin{equation}\label{eq:dzx5-had-unitary}
    \includegraphics[scale=0.2]{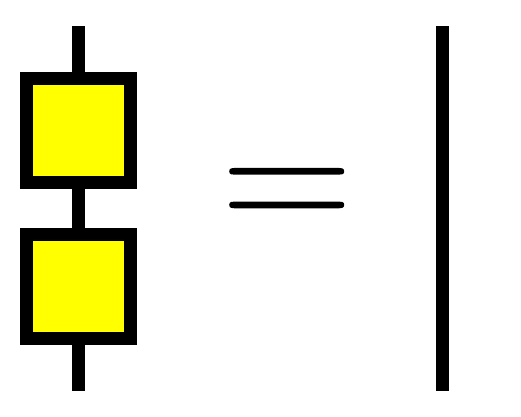}
\end{equation}
\end{minipage}
\begin{minipage}{0.48\textwidth}
\end{minipage}
\end{center}

The proofs for Eqs. \ref{eq:dzx1-hopf}-\ref{eq:dzx5-had-unitary} can be found in reference \cite{Backens2016}. We highlight these derived rewrite rules since we shall use them frequently in proving our results. 

\begin{proposition}
If a spider has a loop with a Hadamard, it is equal to the same spider sans the loop and its phase is added by $\pi$, up to scalars: 
\begin{equation}\label{eq:dzx6-loop-green-had}
    \includegraphics[scale=0.2]{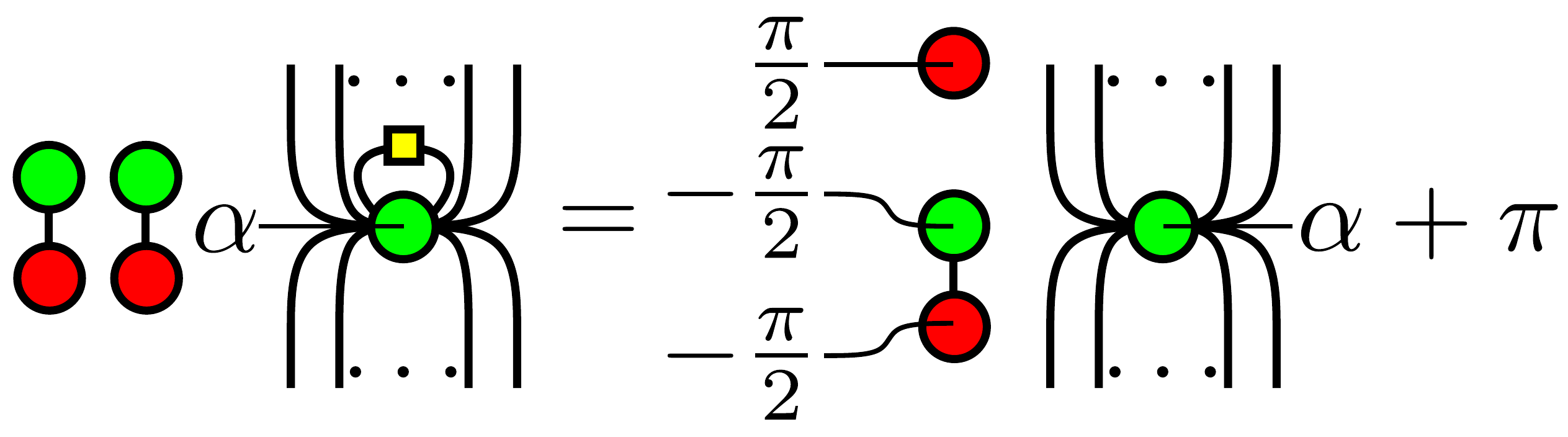}
\end{equation}
\begin{equation}\label{eq:dzx7-loop-red-had}
    \includegraphics[scale=0.2]{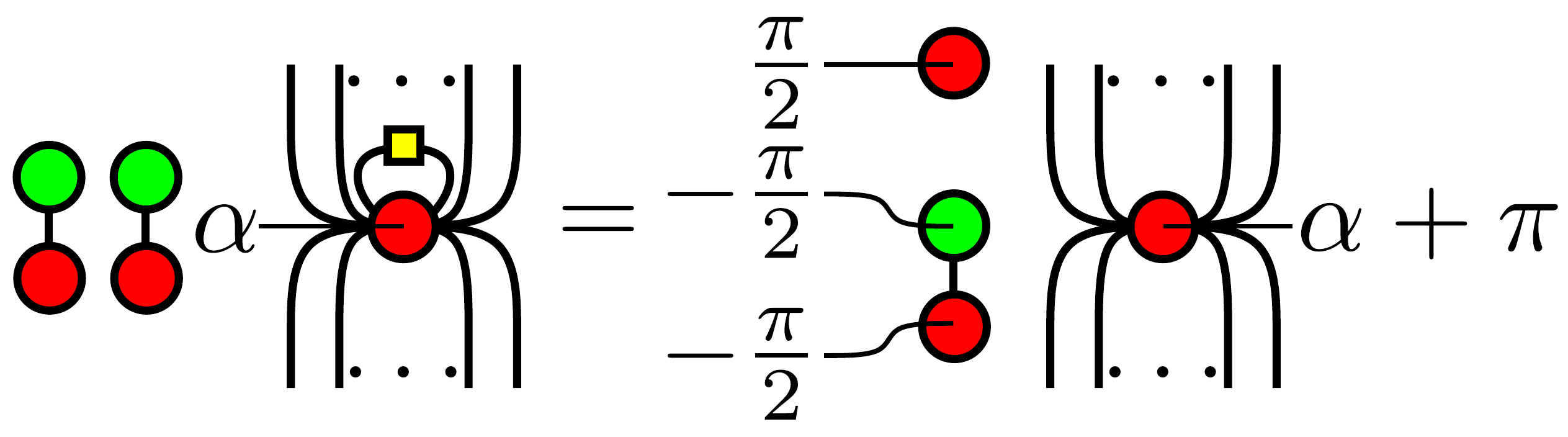}
\end{equation}

\begin{proof}
\begin{eqnarray*}
\includegraphics[scale=0.18]{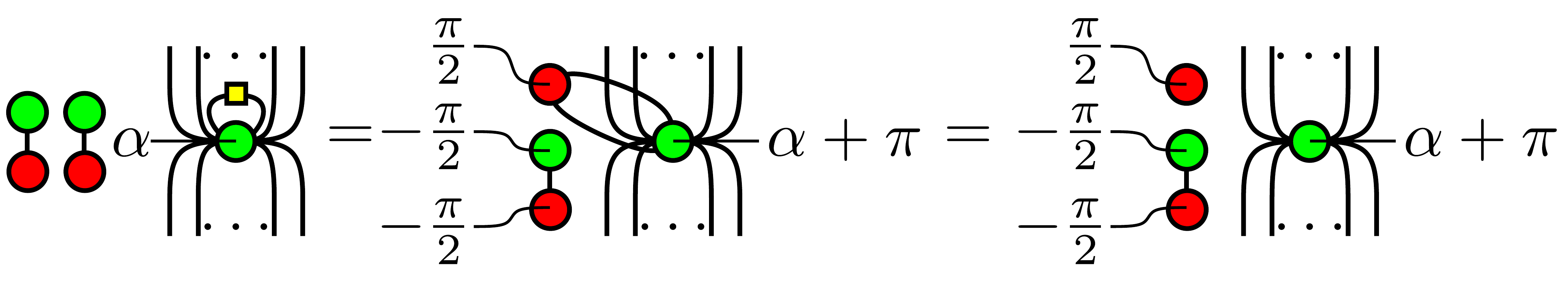}\\
\includegraphics[scale=0.18]{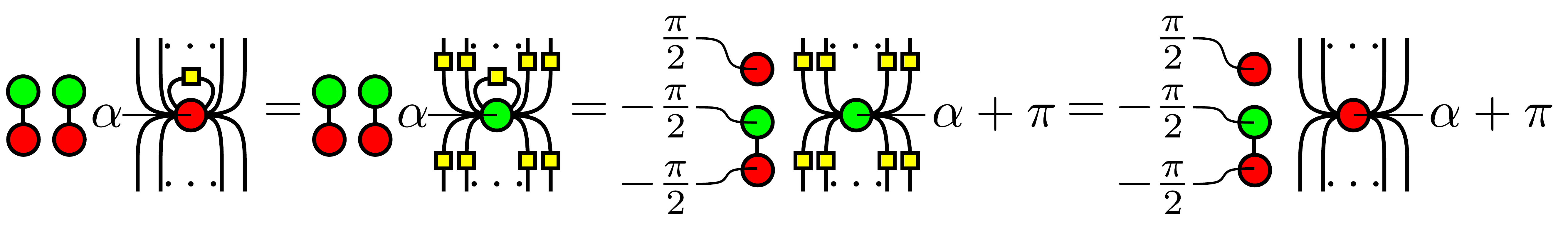}
\end{eqnarray*}
\end{proof}
\end{proposition}
\begin{proposition}\label{prop:hopf-had}
Eq. \ref{eq:dzx1-hopf} can be re-expressed as the following equations:
\begin{center}
\begin{minipage}{0.48\textwidth}
\begin{equation}
    \includegraphics[scale=0.2]{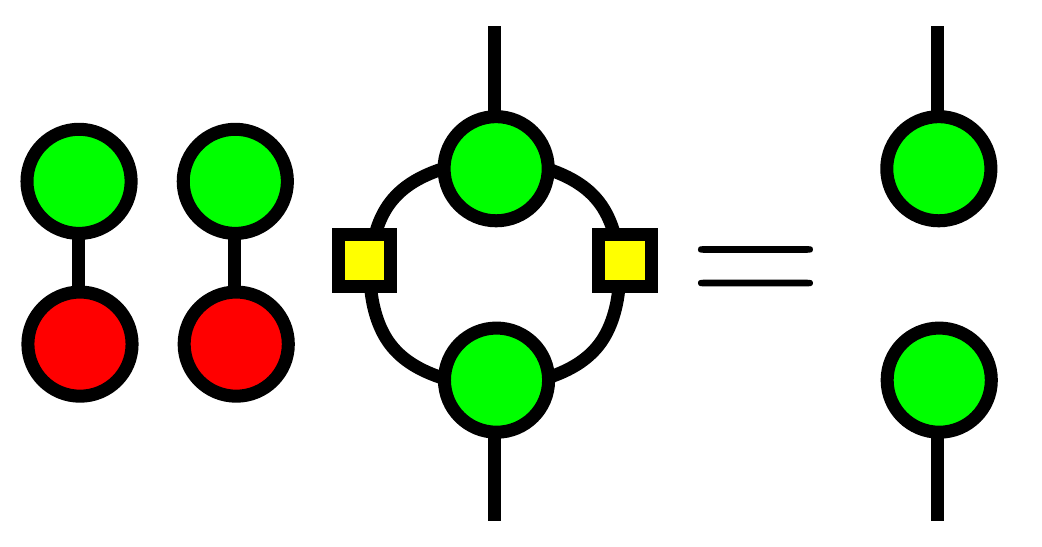}
\end{equation}
\end{minipage}
\begin{minipage}{0.48\textwidth}
\begin{equation}
    \includegraphics[scale=0.2]{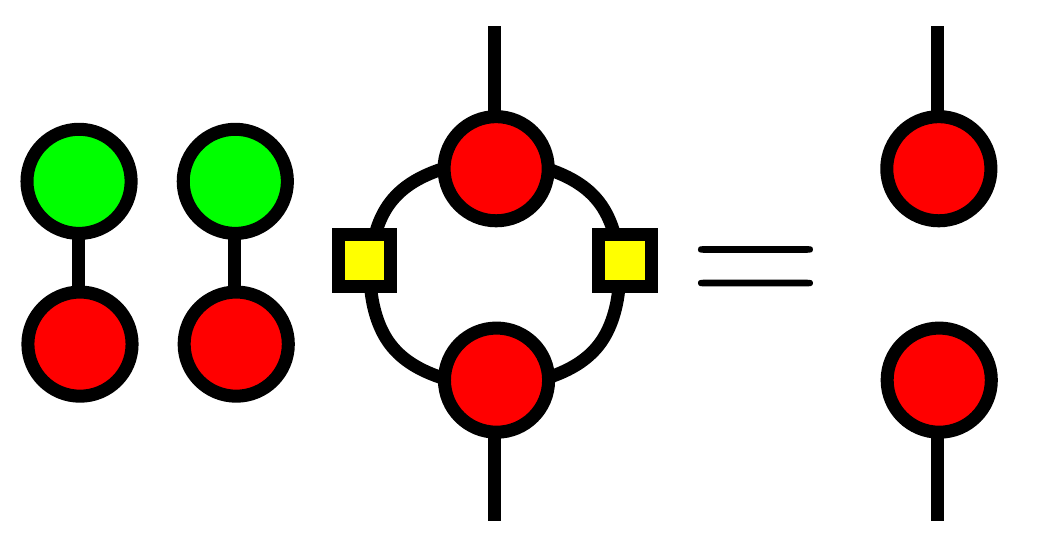}
\end{equation}
\end{minipage}
\end{center}

\begin{proof}
\begin{eqnarray}
    \includegraphics[scale=0.2]{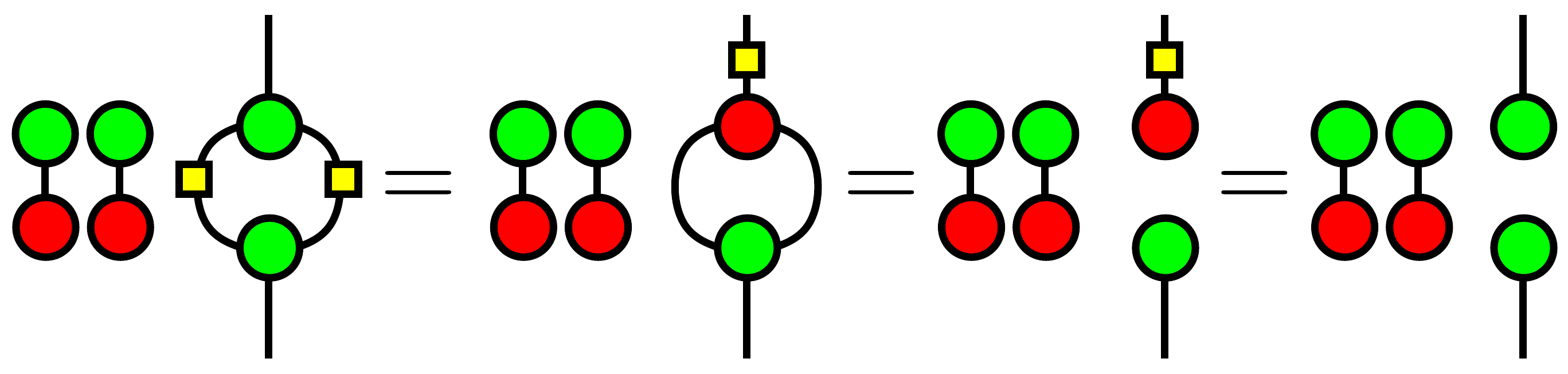}\\
    \includegraphics[scale=0.2]{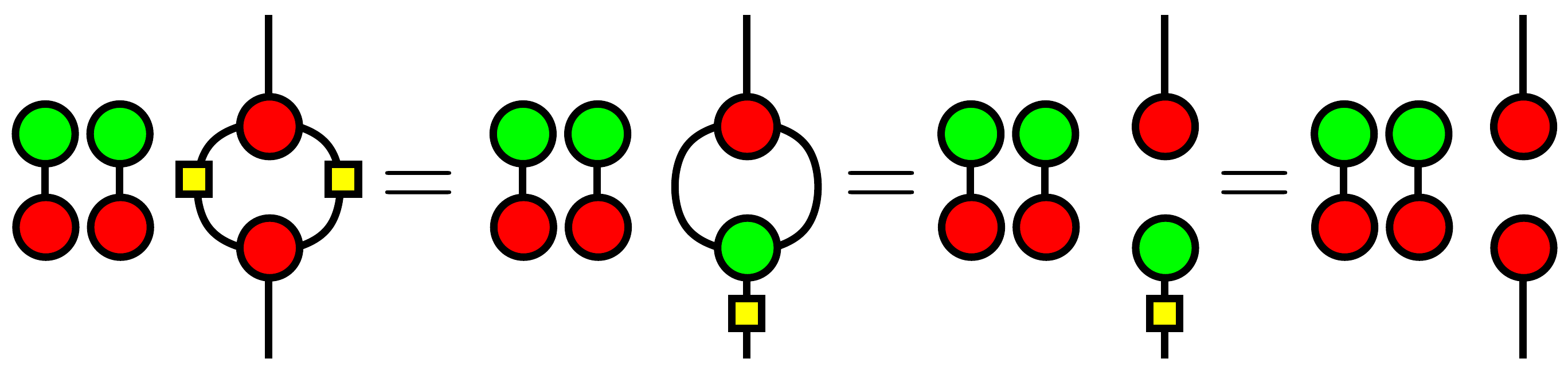}
\end{eqnarray}
\end{proof}
\end{proposition}

For the sake of simplicity, when rewriting diagrams using Eqs. \ref{eq:zx1-fuse-green}-\ref{eq:dzx7-loop-red-had}, we ignore scalars. That is, when two diagrams are equivalent up to scalars, we treat them as equal to each other. 

It is no accident that the classical structures in the ZX-calculus are labelled $\mathcal{Z}$ and $\mathcal{X}$. The corresponding bases for the classical structures are, respectively, the eigenbases of the Pauli $Z$ and Pauli $X$ operators. Furthermore, we can obtain the Pauli $Z$ and $X$ operators from the 1,1-spiders with phase $\pi$ of $\mathcal{Z}$ and $\mathcal{X}$, respectively. 

If we omit the star, the Hadamard, and the rewrite rules which involve phase $\pi$ from the ZX-calculus, we could be describing any pair of strongly complementary classical structures. Thus, a (modified) ZX-calculus can be applied to qudit systems. The only caveat is the fact that the ZX-calculus as described above is not complete for quantum mechanics \cite{Backens2016,Jeandel2019CompletenessZX-Calculus}, but this should not prevent us from using it to study fragments of quantum mechanics as long as we are cautious about the Euler decompositions that we use in our computations. In our case, we do not utilize any Euler decompositions of unitaries other than $\had$ so the incomplete version of ZX-calculus that we use does not affect our computations. 

\section{Constructing \texorpdfstring{$\mathcal{Y}$}{TEXT}}

As in the approach by \citet{Romero2005}, we need the three representatives of the complete set of 3 MUBs of a single qubit system. We already have $\mathcal{Z}$ and $\mathcal{X}$ --- which correspond to the Pauli $Z$ and $X$ operators, respectively  --- but how do we represent the classical structure corresponding to the Pauli $Y$ operator?  

ZX-calculus is universal for pure qubit quantum mechanics \cite{Backens2016}. That is, for any state or operator on qubits, there exists a diagram in the ZX-calculus which represents it. So, it is possible to construct the classical structure corresponding to the Pauli $Y$ operator, using only the generators of the ZX-calculus. We denote this classical structure by $\mathcal{Y}$. By doing this, we do not have to define additional rewrite rules since the ones in the ZX-calculus already suffice.

A $m,n$-spider of the $\mathcal{Y}$ structure is the following linear map, where $j^k\in\{0,1\}$:
\begin{align*}
\ket{(j^1)_Y}\otimes\ket{(j^2)_Y}\otimes\cdots\otimes\ket{(j^m)_Y} &\mapsto \ket{(j^1)_Y}^{\otimes n}\\
\text{if }\ket{(j^k)_Y}=\ket{(j^{k'})_Y}& \text{ for all } k,k'\in\{1,...,m\}\\ \ket{(j^1)_Y}\otimes\ket{(j^2)_Y}\otimes\cdots\otimes\ket{(j^m)_Y} & \mapsto 0\\
\text{if }\ket{(j^k)_Y}\not=\ket{(j^{k'})_Y} & \text{ for some } k,k'\in\{1,...,m\}
\end{align*}

Rewriting the linear map above in terms of $\ket{0}$ and $\ket{1}$ will not simplify it. So, instead of doing that, we exploit the symmetry provided by $\dagger$ (see Section \ref{sec:process-intro}). 

First, we consider the $m,1$-spider of $\mathcal{Y}$. Let $c_k$ be the number of $k\in\{0,1\}$ in $\ket{j^1}\otimes\ket{j^2}\otimes\cdots\otimes\ket{j^m}$ where $j^l\in\{0,1\}$ for each $l\in\{1,2,...m\}$. The $m,1$-spider of $\mathcal{Y}$ can be written as:
\begin{align*}
\ket{j^1}\otimes\ket{j^2}\otimes\cdots\otimes\ket{j^m}& \mapsto (-i)^{c_1}\ket{0} \\
 \text{ if } & c_1 \text{ is even or }0\\
\ket{j^1}\otimes\ket{j^2}\otimes\cdots\otimes\ket{j^m}& \mapsto (-i)^{c_1-1}\ket{1}\\
\text{ if } & c_1 \text{ is odd}
\end{align*}
We obtained the above linear map from the result that the 2,1-spider of a classical structure --- as the binary operation --- forms a group with the unbiased states of the classical structure (see Section 7.4 of \cite{Coecke2011b}). 

Compare this to the $m,1$-spider of the $\mathcal{X}$:
\begin{align*}
\ket{j^1}\otimes\ket{j^2}\otimes\cdots\otimes\ket{j^m}& \mapsto \ket{0} \\
 \text{ if } & c_1 \text{ is even or }0\\
\ket{j^1}\otimes\ket{j^2}\otimes\cdots\otimes\ket{j^m}& \mapsto \ket{1}\\
\text{ if } & c_1 \text{ is odd}
\end{align*}

The difference between the $m,1$-spiders of $\mathcal{Y}$ and $\mathcal{X}$ is the scalar factor on RHS of the mappings. That is, the $m,1$-spider of $\mathcal{Y}$ will take a tensor product of states in the standard basis and produce a state in the standard basis with some scalar factor 1 or $-1$. In contrast, the $m,1$-spider of $\mathcal{X}$ will take a tensor product of states in the standard basis and produce a state in the standard basis with scalar factor 1.  

This provides us with a hint about the spiders of the $\mathcal{Y}$, i.e they can be constructed from spiders of $\mathcal{X}$. Now, the value of $(-i)^{c_1}$ depends on the number of $\ket{1}$ on LHS of the mapping above. We obtain 1 if $c_1$ or $c_1-1$ is a multiple of four, and -1 if $c_1$ or $c_1-1$ is not divisible by 4. So, to obtain the $m,1$-spider of the $\mathcal{Y}$, we need to compose the $m,1$-spider of $\mathcal{X}$ with some process which `checks' the number of $\ket{1}$ inputted to the spider. Furthermore, such a process should only change the scalar part of the state. We can do this using the controlled-$Z$ operator, which, in ZX-calculus, is denoted by:
\begin{equation*}
    \includegraphics[scale=0.2]{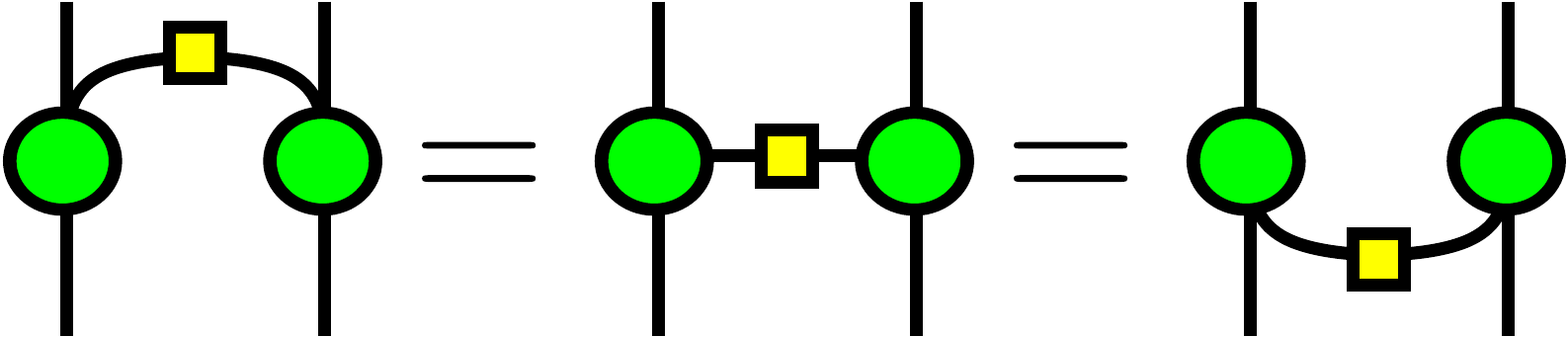}
\end{equation*}
The derivation of the diagram for the controlled-$Z$ operator can be found in reference \cite{Coecke2011b}. 

The controlled-$Z$ operator is an operator on two qubits, and it can be viewed as a process which checks whether one of the two qubits inputted into it is $\ket{0}$ or not. 

We argue that the following process composed of controlled-$Z$ operators checks the number $\ket{1}$ in a product state on $m$ qubits consisting of $\ket{0}$ and $\ket{1}$, which we shall denote as $\mathsf{CZ}_m$:
\begin{equation*}
    \includegraphics[scale=0.2]{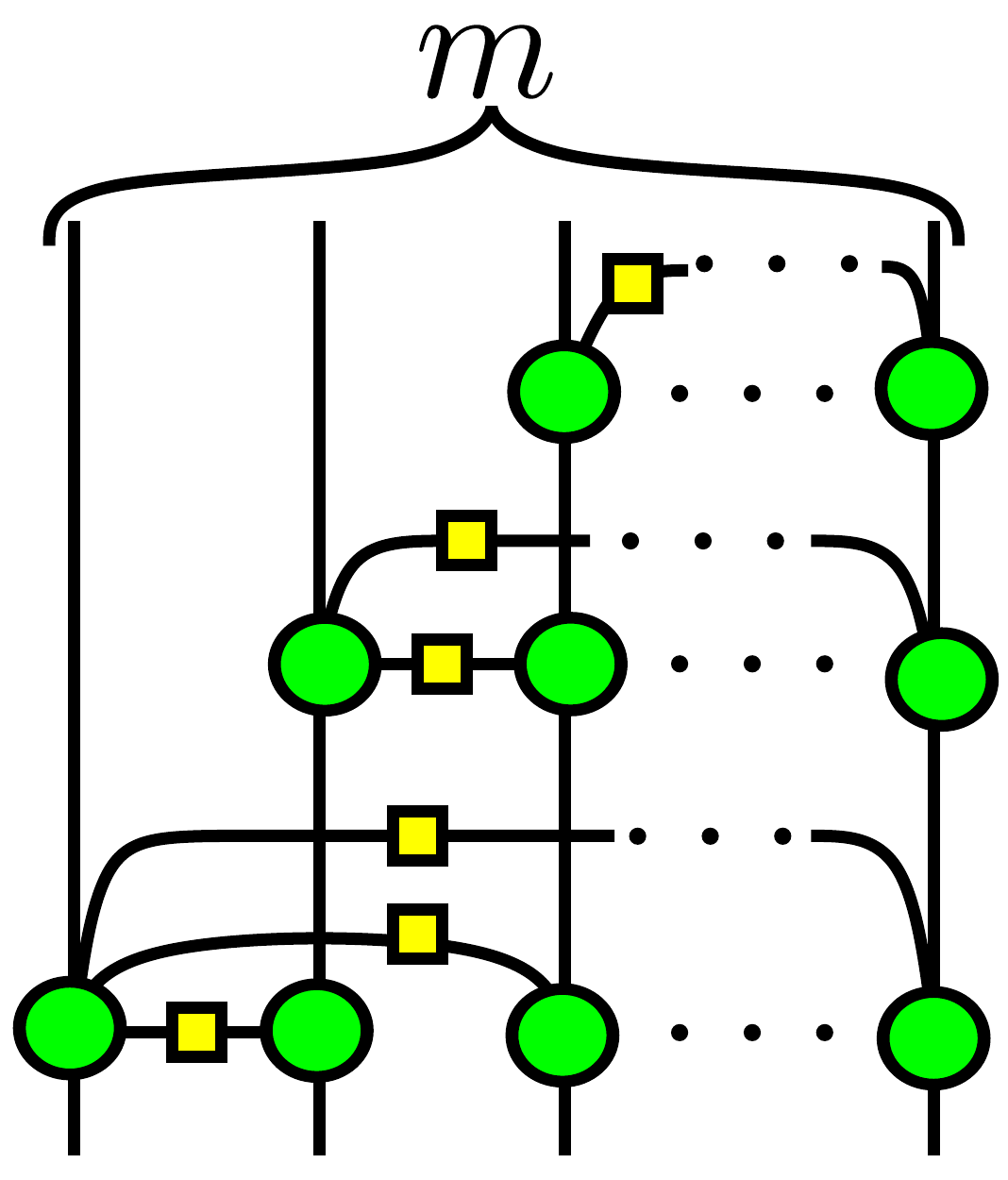}
\end{equation*}
% the horrible diagram
% an example
An example of the previous diagram for 4 qubits:
\begin{equation*}
    \includegraphics[scale=0.2]{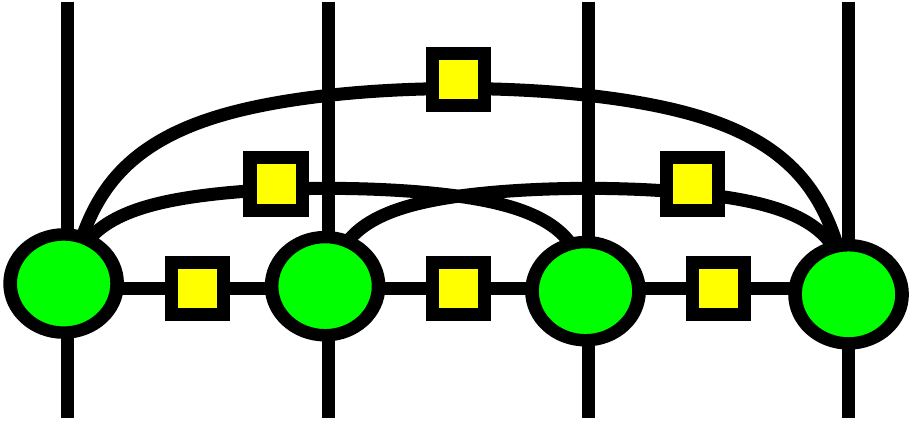}
\end{equation*}

We can break $\mathsf{CZ}_m$ down into levels. At the $k$-th level (see Fig. \ref{fig:level-k-CZ}), there are $m-k$ controlled-Z operators acting on the $k$-th qubit and the qubits to its right. If the $k$-th qubit is $\ket{1}$ and the number of $\ket{1}$ to the right of the $k$-th qubit is odd, -1 should be multiplied to the scalar factor of the output state. Otherwise, the sign of the scalar factor remains the same. 

\begin{figure}[!ht]
\begin{center}
 \includegraphics[scale=0.2]{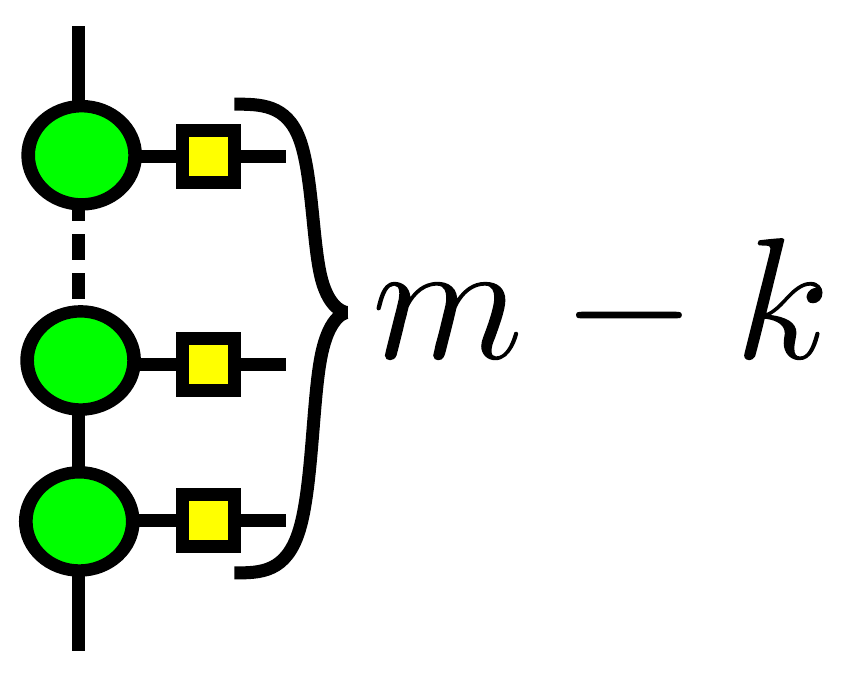}
    \caption{Level $k$ of $\mathsf{CZ}_m$}
    \label{fig:level-k-CZ}
    \end{center}
\end{figure}

The sign of the scalar factor changes $\frac{c_1-1}{2}$ times if $c_1$ is odd, and the sign of the scalar factor changes $\frac{c_1}{2}$ times if $c_1$ is even. Starting from +1, the resulting sign is then -1 if it changes an odd number of times, and it is +1 if it changes an even number of times. Thus, for $\ket{j^1j^2...j^m}$, where $j^k\in\{0,1\}$, $\mathsf{CZ}_m$ transforms $\ket{j^1j^2...j^m}$ into $-\ket{j^1j^2...j^m}$ if $c_1$ or $c_1-1$ is divisible by 4. Otherwise, $\mathsf{CZ}_m$ does nothing to it. Therefore, if we compose $\mathsf{CZ}_m$ to the $m,1$-spider of the $\mathcal{X}$, we obtain the $m,1$-spider of $\mathcal{Y}$.

The $1,n$-spider of a classical structure is the adjoint of its $n,1$-spider. So, the diagram in Fig. \ref{fig:y-spider} should be the $m,n$-spider of $\mathcal{Y}$.

\begin{figure}[!ht]
    \centering
    \includegraphics[scale=0.2]{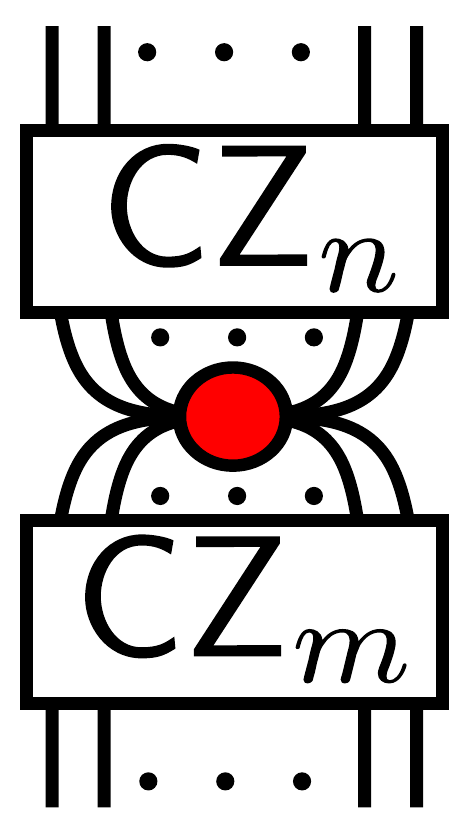}
    \caption{Spider of $\mathcal{Y}$}
    \label{fig:y-spider}
\end{figure}

There is an alternative way to express the spiders of $\mathcal{Y}$. In this thesis, we have opted to use the spider in Fig. \ref{fig:y-spider} since it does not involve spiders with phases other than 0 or $\pi$. The following theorem provides this alternative spider:

\begin{theorem}\label{thm:y-equivalence}
The diagram in Fig. \ref{fig:y-spider} has the following equivalent form:
\begin{equation}\label{eq:y-equivalence}
    \includegraphics[scale=0.2]{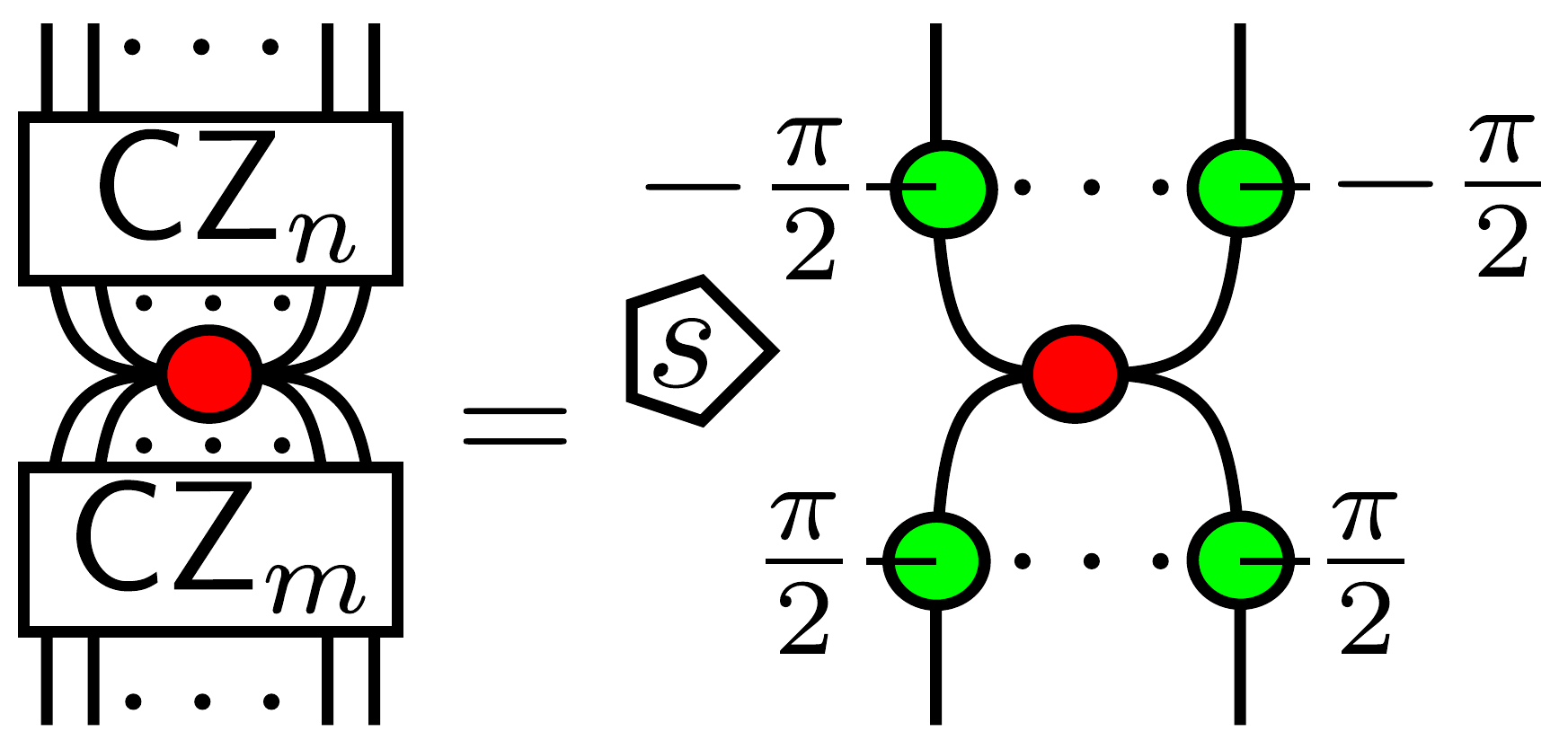}
\end{equation}
for some scalar $s$.
\end{theorem}

For the following results, we define $\theta=\frac{\pi}{2}$.

We need to prove the following propositions before we can prove Theorem \ref{thm:y-equivalence}. 
\begin{proposition}\label{prop:y-equivalence-1}
The following equation gives Eq. \ref{eq:y-equivalence} for the case of $m=2$ and $n=1$:
\begin{equation}
  \includegraphics[scale=0.2]{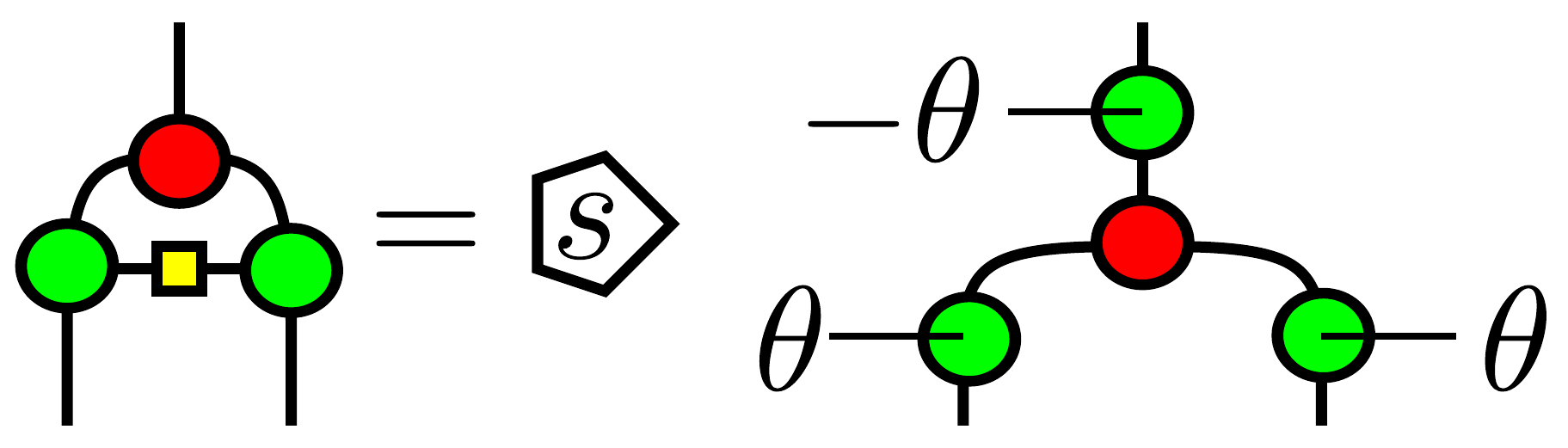}
\end{equation}
for some scalar $s$.
\begin{proof}
\begin{equation*}
    \includegraphics[scale=0.2]{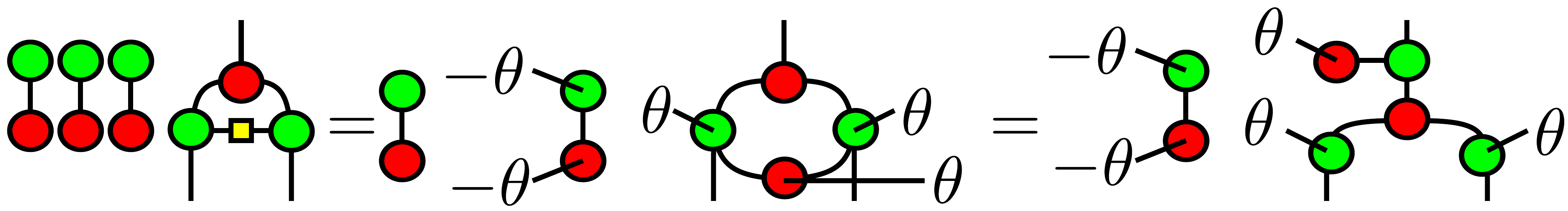}
\end{equation*}
Above, the first equality is obtained from Eq. \ref{eq:zx13-had-euler}, and the final equality is obtained from Eq. \ref{eq:zx11-bialgebra}.
\begin{equation*}
    \includegraphics[scale=0.2]{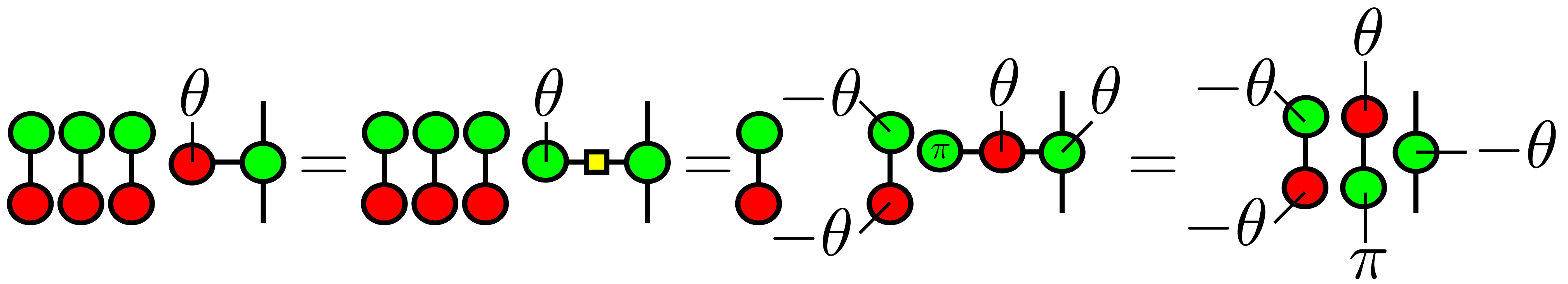}
\end{equation*}
Above, the second equality is obtained from Eq. \ref{eq:zx13-had-euler}, and the final equality is obtained from Eqs. \ref{eq:zx9-copy-green} and \ref{eq:zx7-copy-green-pi}.
Thus, we have:
\begin{equation*}
    \includegraphics[scale=0.2]{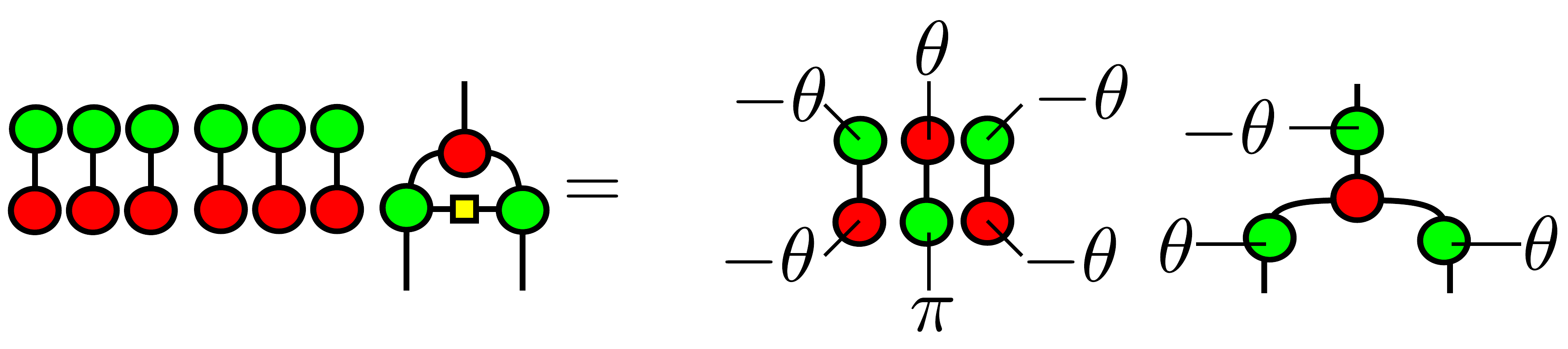}
\end{equation*}
By Eq. \ref{eq:dzx4-star-inverse},
\begin{equation}
    \includegraphics[scale=0.2]{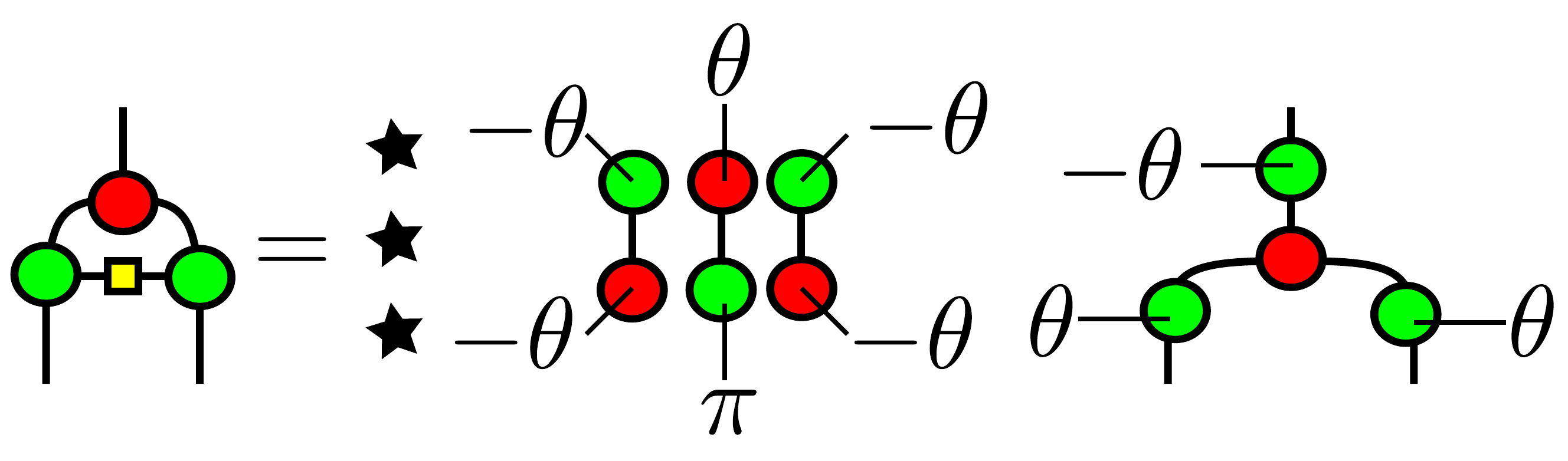}
\end{equation}
\end{proof}
\end{proposition}

\begin{proposition}\label{prop:y-equivalence-2}
Let us temporarily denote the diagram in Fig. \ref{fig:y-spider} as $S_{m,n}$. We can obtain $S_{m+1,1}$ by composing $S_{2,1}$ to any of the input legs of $S_{m,1}$: 
\begin{equation}
    \includegraphics[scale=0.2]{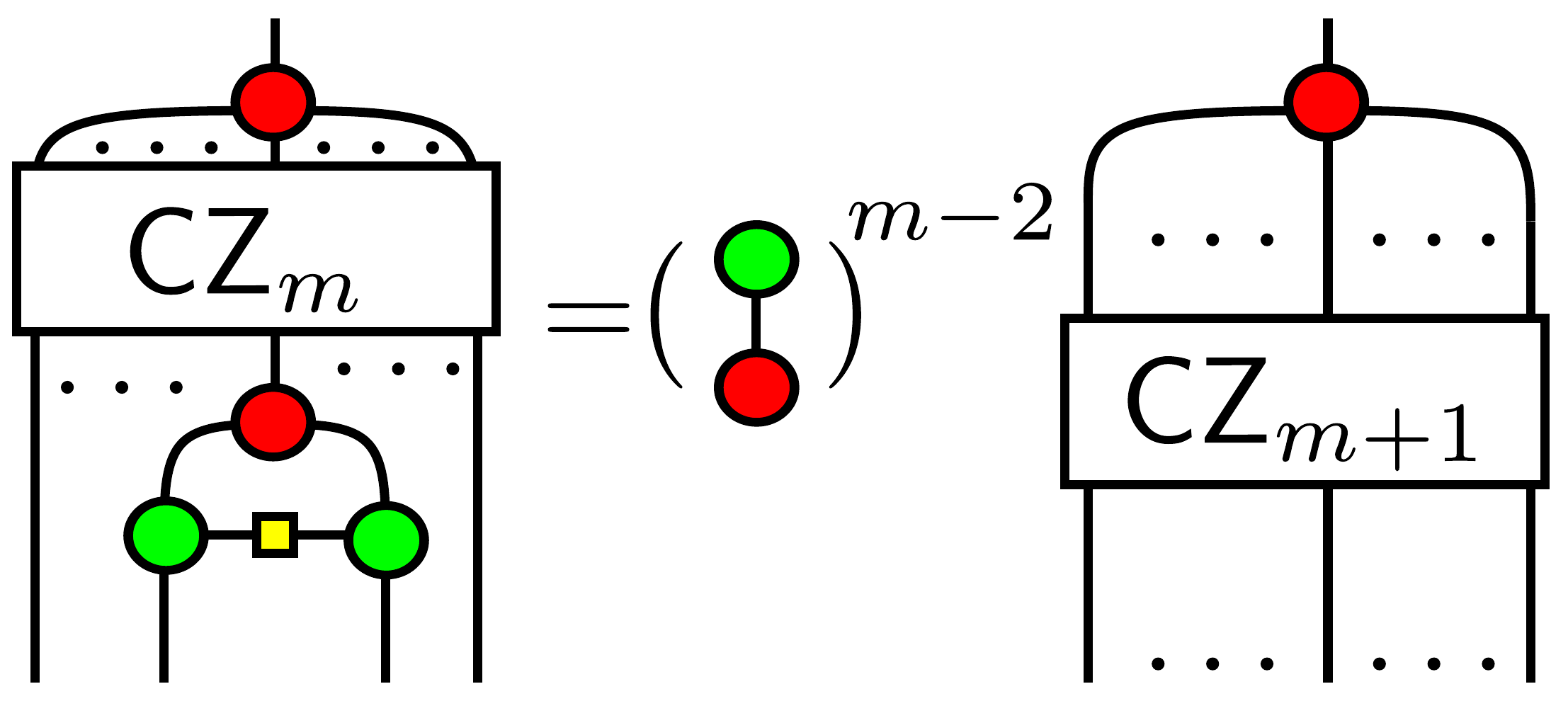}
\end{equation}
\begin{proof}
First, we prove for the case where $m=2$.
\begin{equation*}
    \includegraphics[scale=0.2]{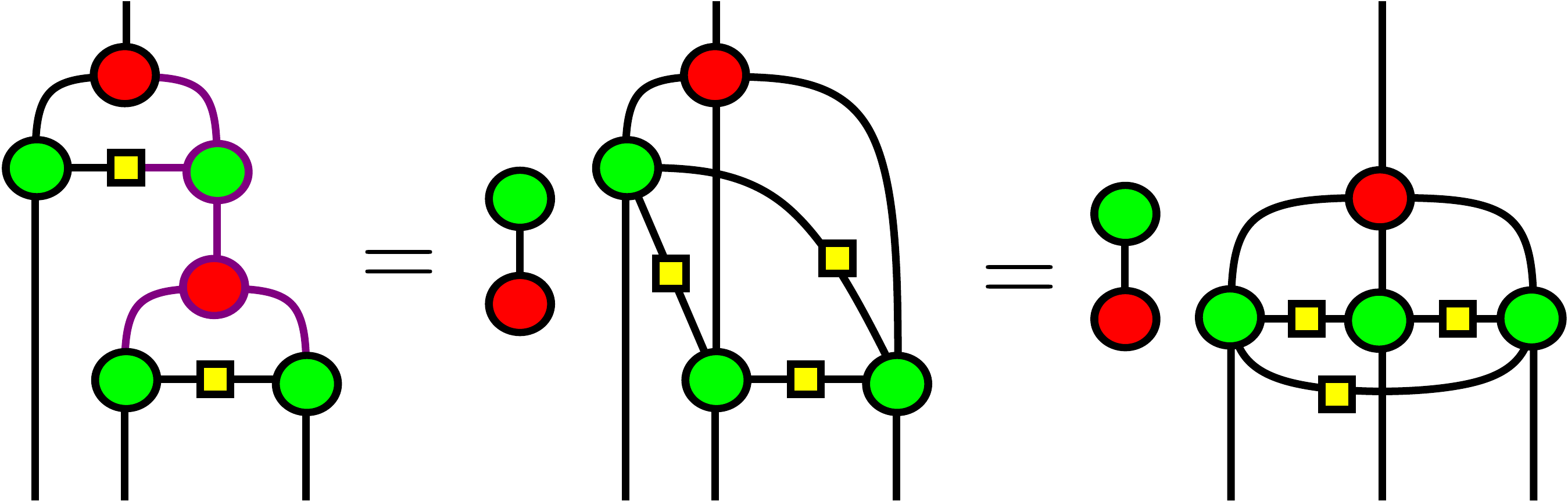}
\end{equation*}
Then, we show for the case of arbitrary $m\in\mathbb{N}$ where $m>2$:
\begin{equation*}
    \includegraphics[scale=0.2]{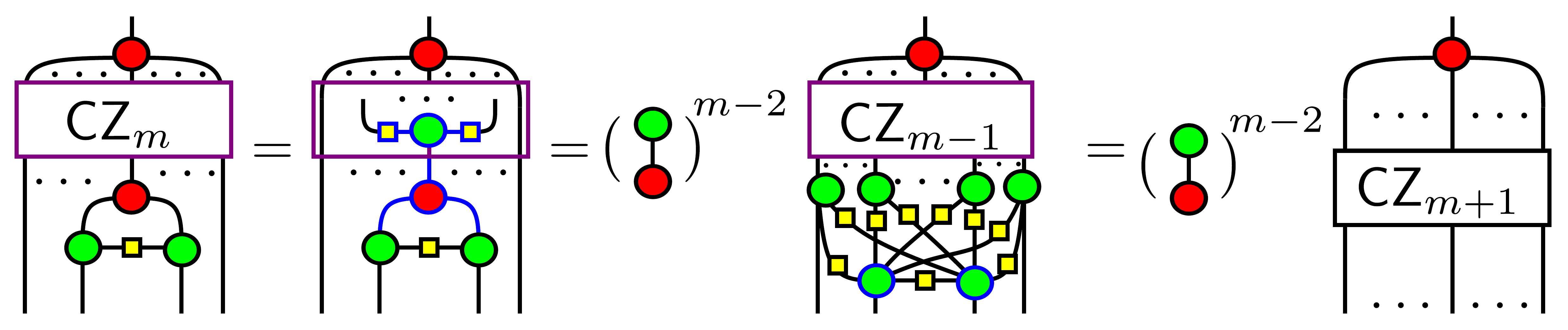}
\end{equation*}
In the second diagram, the purple box is $\mathsf{CZ}_m$ and the diagram inside it is a sub-diagram of $\mathsf{CZ}_m$. To obtain the second equality, Eq. \ref{eq:zx11-bialgebra} is repeatedly applied to the part of the second diagram that is outlined blue. The green node in $\mathsf{CZ}_m$ of the second diagram is then rewritten into $m-1$ red node and these nodes fuse to the other green nodes in $\mathsf{CZ}_m$ of the second diagram, resulting in $\mathsf{CZ}_{m-1}$ in the third diagram. The green nodes which are outlined blue in the third diagram are joined to each green node in $\mathsf{CZ}_{m-1}$ once via $\had$. The two green nodes are also joined to each other by $\had$. Thus, we have the final diagram. 
\end{proof}
\end{proposition}

Before proceeding, we define:

\begin{equation}\label{eq:green-phase}
    \includegraphics[scale=0.2]{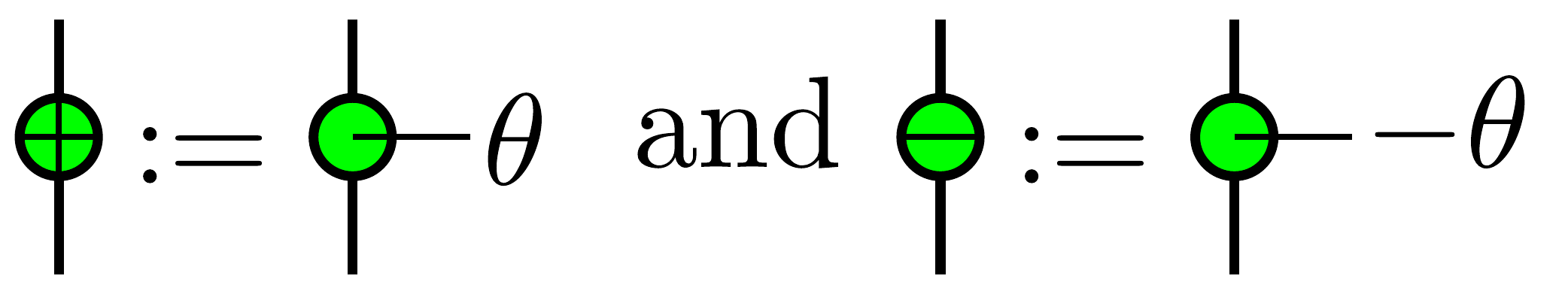}
\end{equation}

\begin{proof}(Theorem \ref{thm:y-equivalence})

\noindent According to Proposition \ref{prop:y-equivalence-1}:
\begin{equation*}
    \includegraphics[scale=0.2]{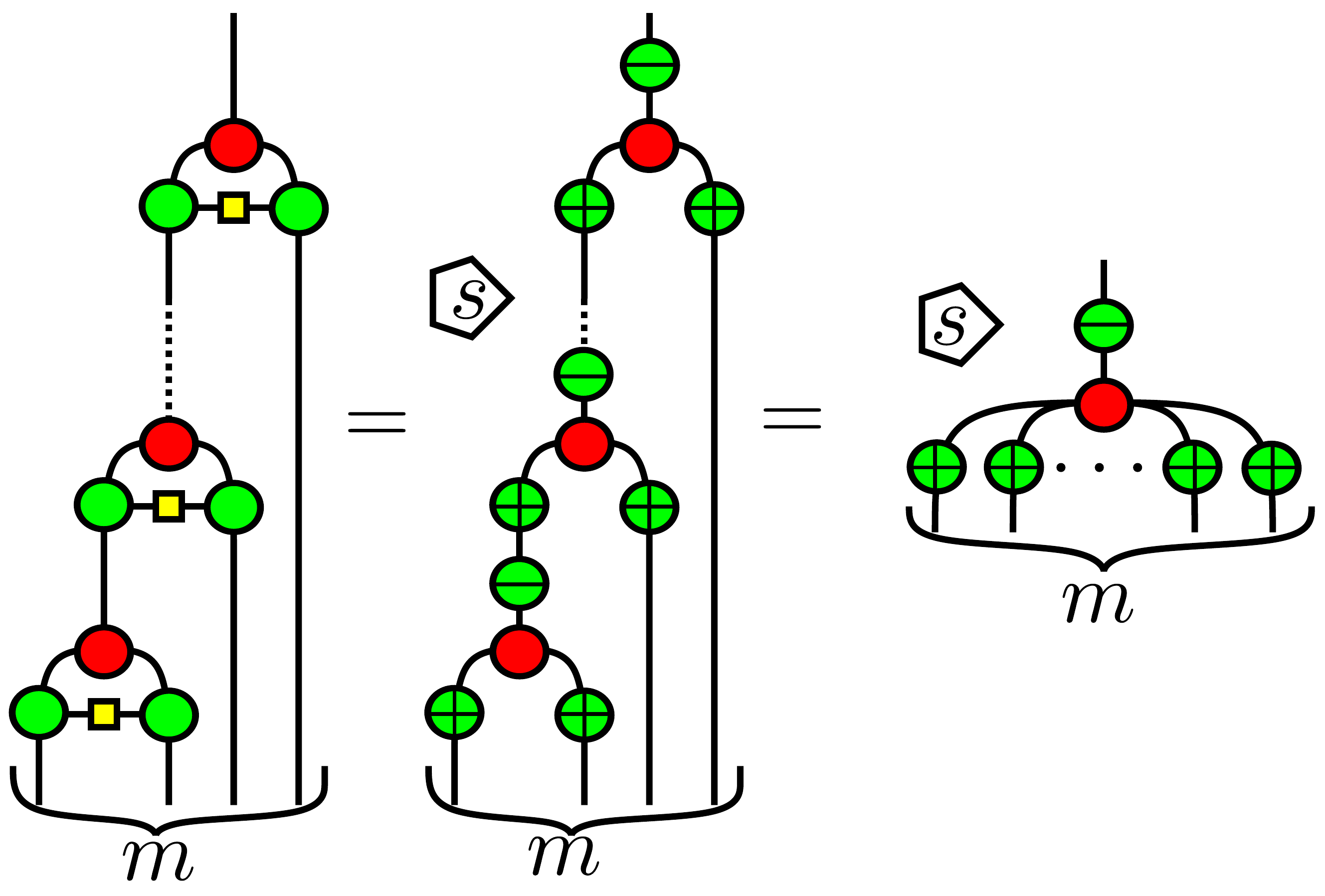}
\end{equation*}
for some scalar $s$.

However, applying Proposition \ref{prop:y-equivalence-2} to LHS of equation above, we have:
\begin{equation}\label{eq:y-equiv-1}
    \includegraphics[scale=0.2]{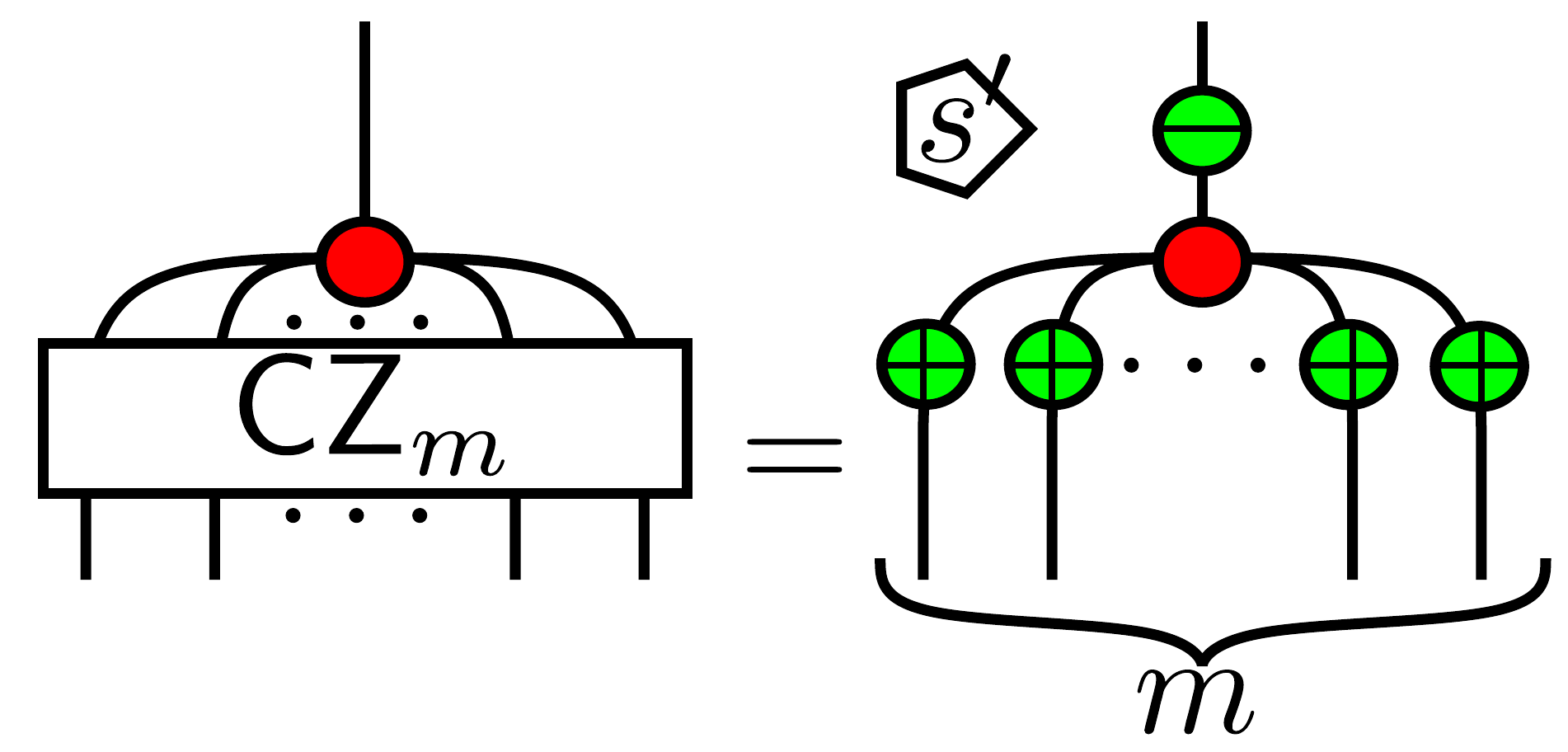}
\end{equation}
for some scalar $s'$,
which means the following equation also holds:
\begin{equation}\label{eq:y-equiv-2}
    \includegraphics[scale=0.2]{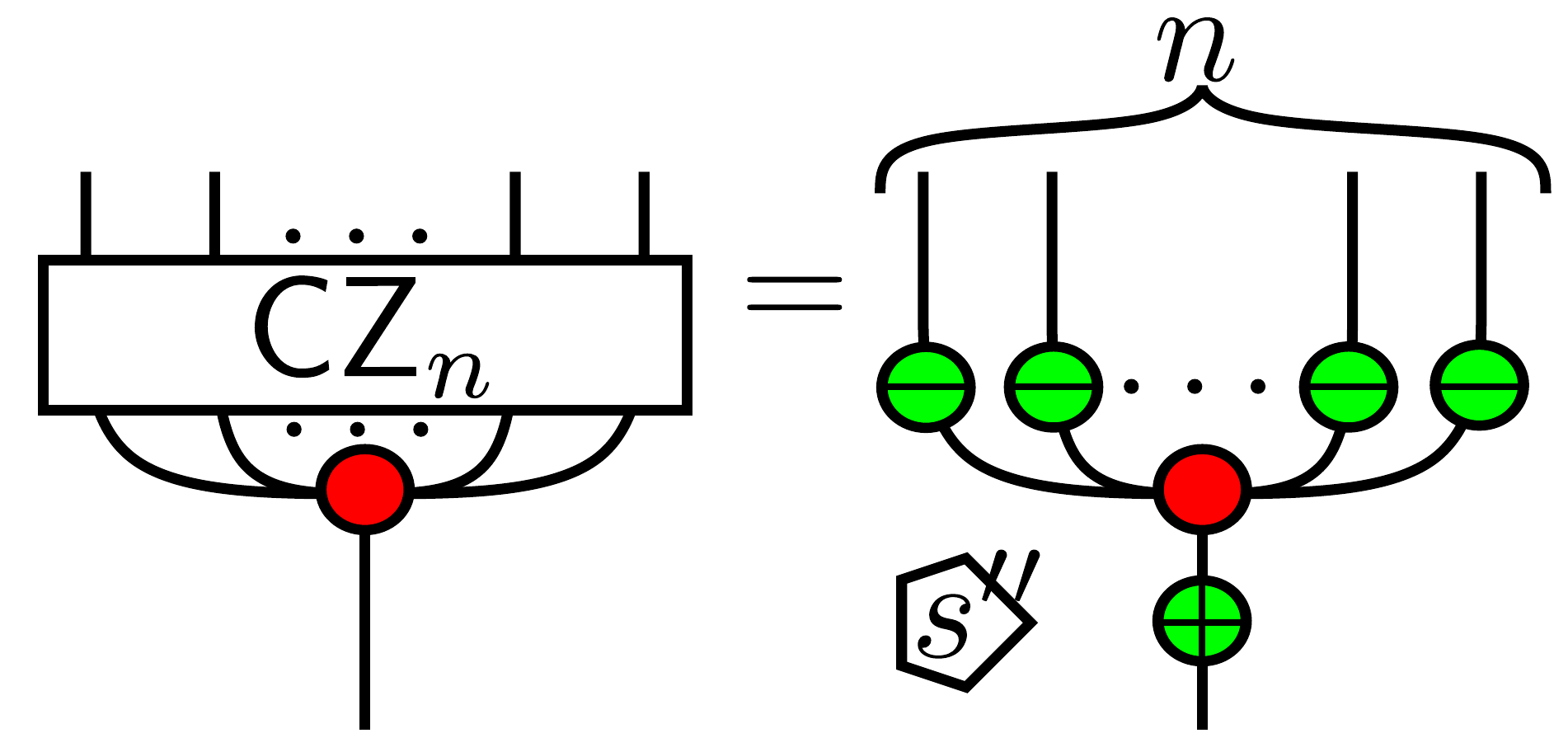}
\end{equation}
for some scalar $s''$.

\noindent Combining Eqs. \ref{eq:y-equiv-1} and \ref{eq:y-equiv-2}, we have:
\begin{equation}\label{eq:y-equiv-final}
    \includegraphics[scale=0.2]{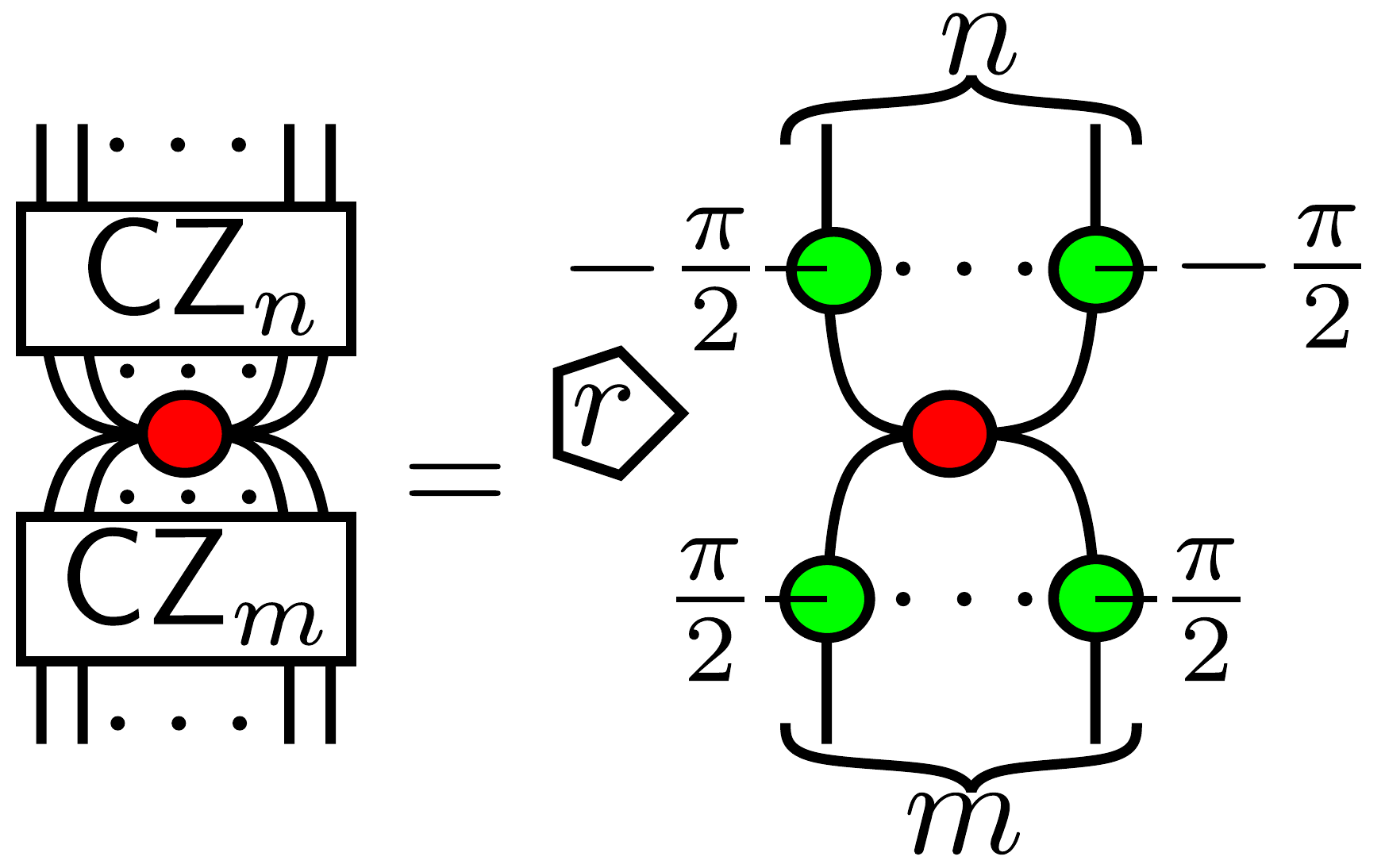}
\end{equation}
for some scalar $r$, which means the two are equivalent up to scalars. So if the collection of the spiders on RHS of the above equation forms a classical structure, then so does the collection of the spiders on LHS, and vice versa.
\end{proof}

\begin{theorem}
The collection of the following spiders for $m,n\in\mathbb{N}\cup\{0\}$ form a classical structure on single qubit:
\begin{equation*}
    \includegraphics[scale=0.2]{images/26-50/42-y-CS.pdf}
\end{equation*}
Furthermore, the classical structure corresponds to the eigenbasis of the Pauli $Y$ operator. Thus, we denote it by $\mathcal{Y}$. 
\begin{proof}
First, we show that the collection of the spider on RHS of Eq. \ref{eq:y-equiv-final} forms a classical structure:
\begin{equation*}
\includegraphics[scale=0.2]{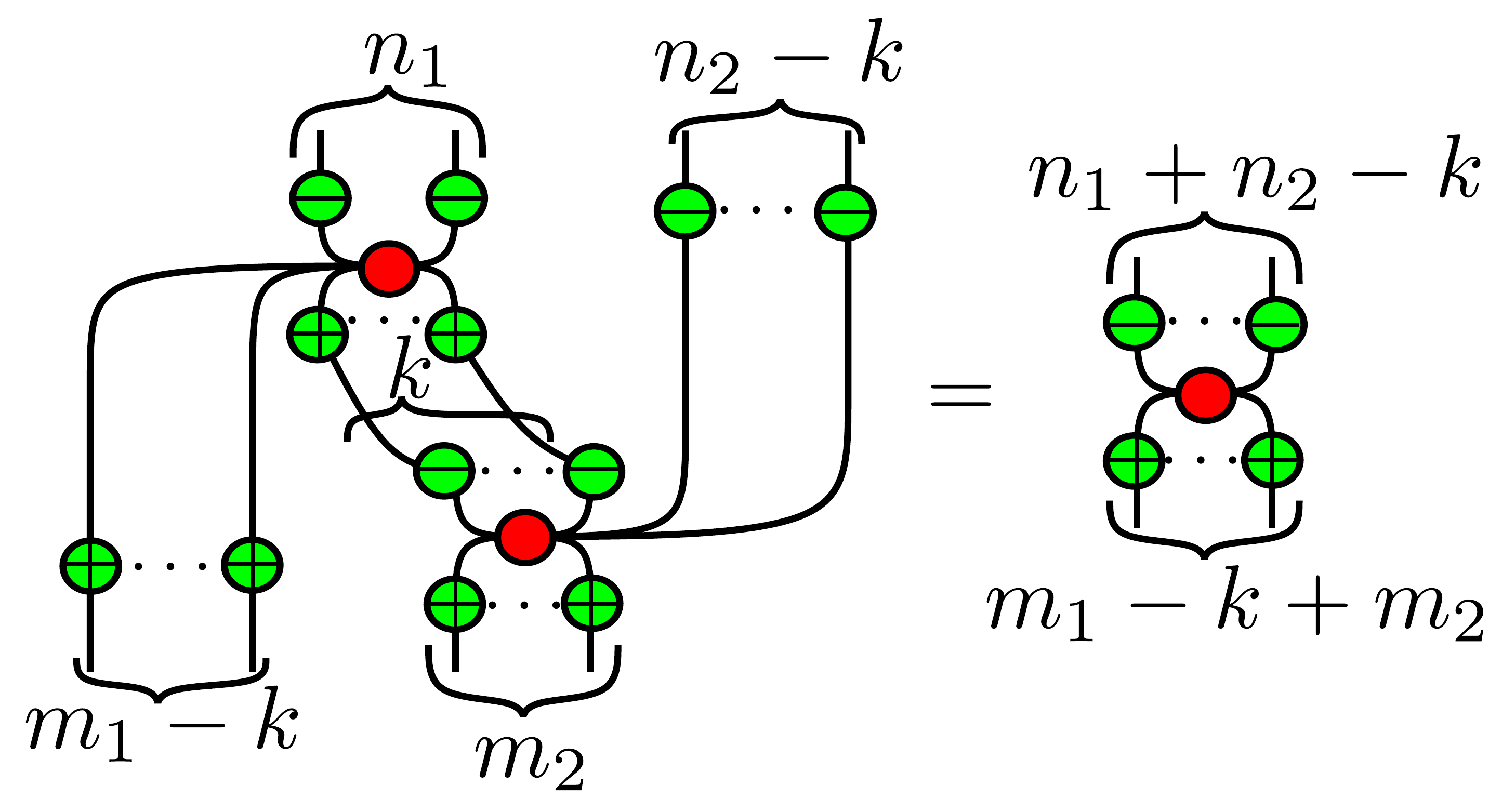}
\end{equation*}
\begin{equation*}
\includegraphics[scale=0.2]{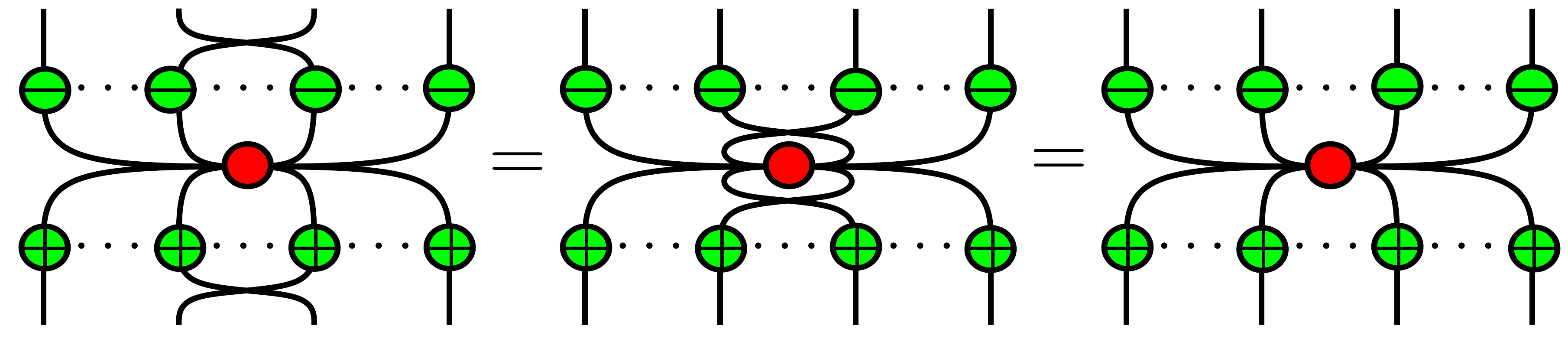}
\end{equation*}
Due to Eq. \ref{eq:y-equiv-final}, the collection of the spiders of the following form is a classical structure:
\begin{equation*}
    \includegraphics[scale=0.2]{images/26-50/42-y-CS.pdf}
\end{equation*}
The aforementioned classical structure corresponds to the Pauli $Y$ operator (in the same context as in \cite{Coecke2013New}). So we shall denote it by $\mathcal{Y}$.
\end{proof}
\end{theorem}

\section{Chapter Summary}

In this chapter:
\begin{itemize}
    \item We introduced strong complementarity which is the basis for a graphical calculus called ZX-calculus;
    \item We provided a brief review of ZX-calculus;
    \item ZX-calculus includes classical structures corresponding to the eigenbases of Pauli $Z$ and Pauli $X$ operators, which we denote by $\mathcal{Z}$ and $\mathcal{X}$, respectively;
    \item Using the generators of ZX-calculus, we constructed the classical structure corresponding to the eigenbasis of the Pauli $Y$ operator, which we denote by $\mathcal{Y}$.
\end{itemize}

\chapter{Composing Classical Structures}\label{sec:composeCS}

% product composition and composition via cup
% for product composition, include commentary on McNulty's work on MU product bases on qudits

Due to its equivalence to orthogonal bases, we may think of a classical structure as a representative of an observable and particular spiders could have special roles. For example, the $2,1$-spider could be interpreted as measurement on a process theory consisting of completely positive maps (CPM) as processes and the state of the system may be mixed --- which is more precise in its description of quantum systems --- and the $1,2$-spider could be interpreted as encoding classical data into a system \cite{Selinger2007a,Coecke2008a,Coecke2007,Coecke2014CategoriesChannels,Coecke2016a}. While we shall not delve deeper into CPM-based process theory, we find its model of quantum-classical interactions to be remarkable. As such, we find it important to study one of its core concepts --- classical structures. In this chapter, we seek to exploit the compositionality of CQM/process theories in order to provide a more precise picture of classical structures on multiple qubits with classical structures on a single qubit as their building blocks. 

\section{Composing Classical Structures}\label{sec:composeCSmultQ}

To construct complementary classical structures on multiple qubits, we first outline the methods for composing complementary classical structures on a single qubit. In addition, we shall justify our chosen unitaries for joining spiders, i.e. the Hadamard and the identity, in the context of obtaining complementary classical structures on multiple qubits. 

\subsection{Separably Composing Spiders}\label{sec:SCCS}

Suppose $\mathcal{A}$ and $\mathcal{B}$ are classical structures on some systems $T_A$ and $T_B$, respectively, where spiders of $\mathcal{A}$ is $\black$ and spiders of $\mathcal{B}$ is $\blue$, respectively. Then the collection of processes of the following type, formed by composing $\black$ and $\blue$, with the same number of input and output legs, in parallel is a classical structure on $T_A\otimes T_B$:

\begin{figure}[!ht]
    \centering
    \includegraphics[scale=0.2]{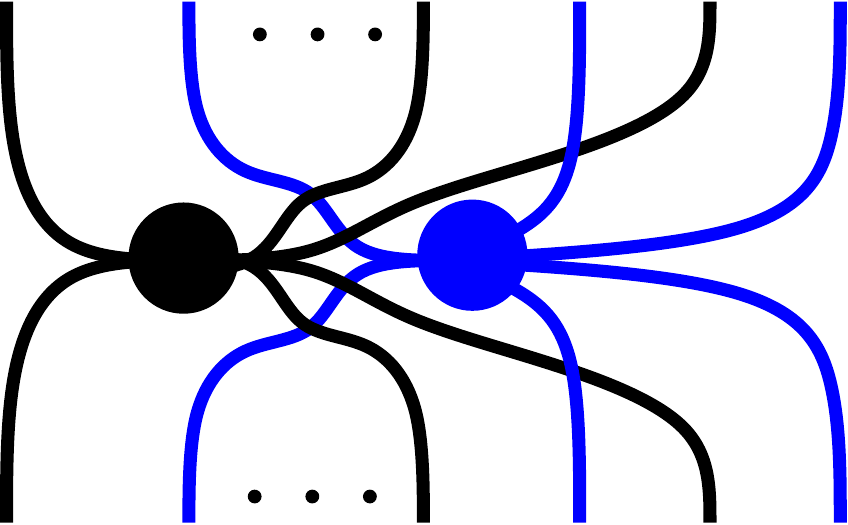}
    \caption{Separably composed spider on a bipartite system. The legs of the spider is $T_A\otimes T_B$.}
    \label{fig:SCspider}
\end{figure}

In Fig. \ref{fig:SCspider}, we used different coloured wires for the legs of the black and blue spiders. This is to distinguish the legs of the two spiders. Each leg still represents a single qubit, and not any other system. In the rest of this thesis, we would sometimes use the same method to distinguish the legs of constituent spiders within a composite spider.  

Since the spiders in Fig. \ref{fig:SCspider} are composed in parallel, the resulting process inherits its constituents' ability to fuse and invariance under permutation of the spiders' legs. 

\begin{definition}\label{def:SC}
A classical structure on a bipartite system $T_A\otimes T_B$ is called \textbf{separably composed} (SC) if its $m,n$-spider is formed by composing the $m,n$-spider of a classical structure on $T_A$ and the $m,n$-spider of a classical structure on $T_B$ in parallel, and the legs are permuted as in Fig. \ref{fig:SCspider}. Suppose the aforementioned classical structures on $T_A$ and $T_B$ are $\mathcal{A}$ and $\mathcal{B}$, respectively. We denote the SC classical structure on $T_A\otimes T_B$ by $\mathcal{AB}$.
\end{definition}

\subsection{Joining Spiders}\label{sec-join-spiders}

The method of composing classical structures in the previous section only gives rise to classical structures with an underlying basis consisting of separable states. To find a more exhaustive collection of classical structures on two qubits, we must devise another way of constructing classical structures via composition, i.e. one that gives rise to a classical structure on two qubits with an underlying basis consisting of entangled states.

To entangle two qubits (or a bipartite system), we need a bipartite process that is non-separable, i.e. a bipartite process such that:
\begin{equation*}
     \includegraphics[scale=0.2]{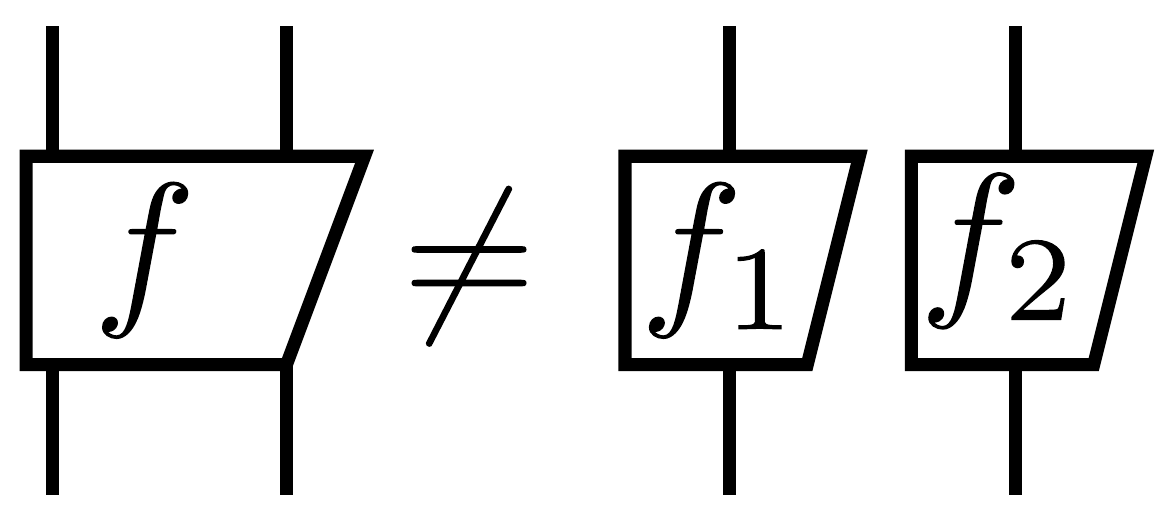}
\end{equation*}
for any processes $f_1$ and $f_2$.

In a separably composed spider on two qubits, a single leg is two qubits or $\mathbb{C}^2\otimes\mathbb{C}^2$. So, to entangle the constituent spiders of a SC spider, we propose composing a bipartite process on each of its legs as in Fig. \ref{fig:join-spiders}.

\begin{figure}[!ht]
    \centering
    \includegraphics[scale=0.2]{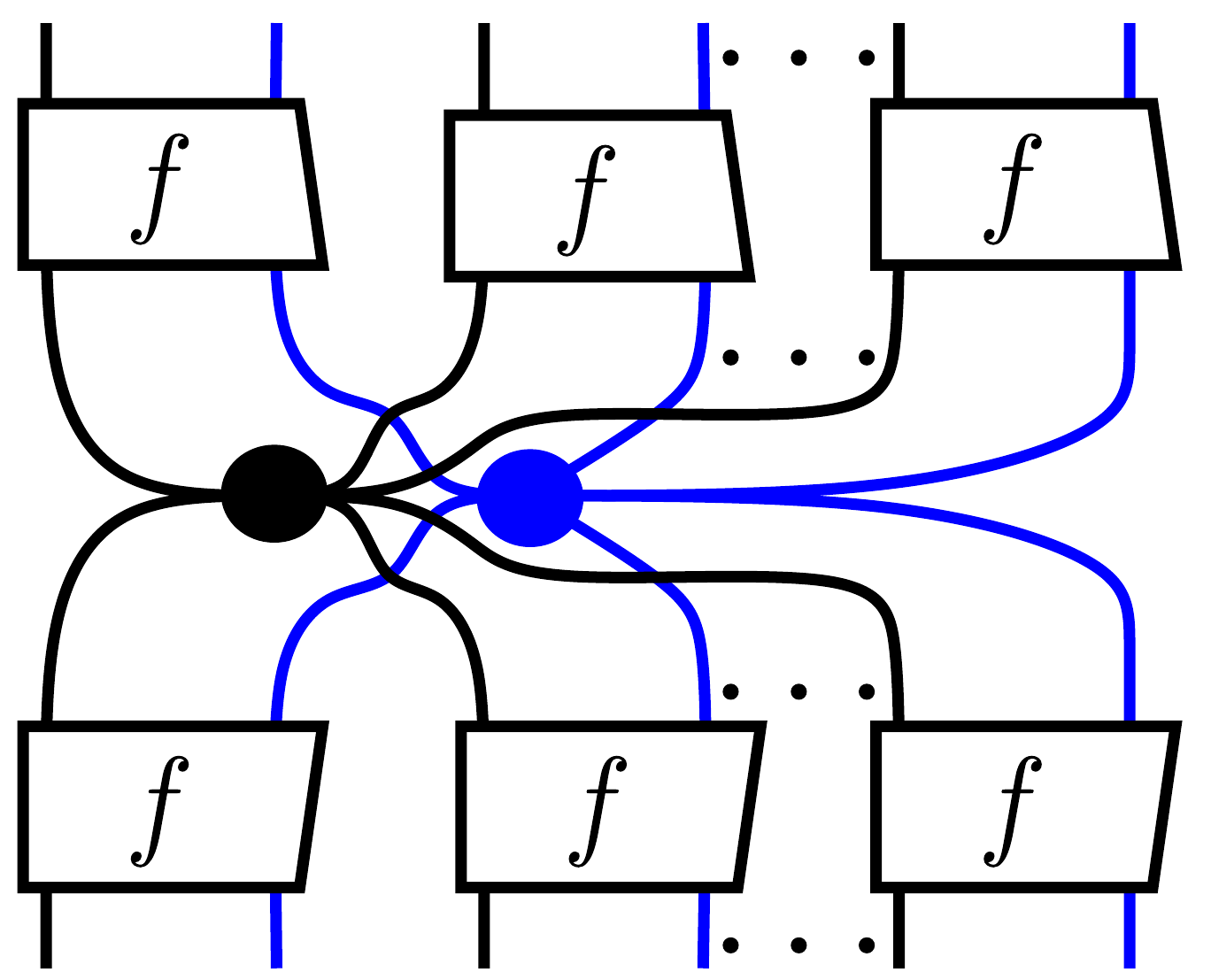}
    \caption{Joining spiders via a non-separable bipartite process.}
    \label{fig:join-spiders}
\end{figure}

The input or output legs of the diagram in Fig. \ref{fig:join-spiders} are invariant under permutation:
\begin{equation*}
    \includegraphics[scale=0.2]{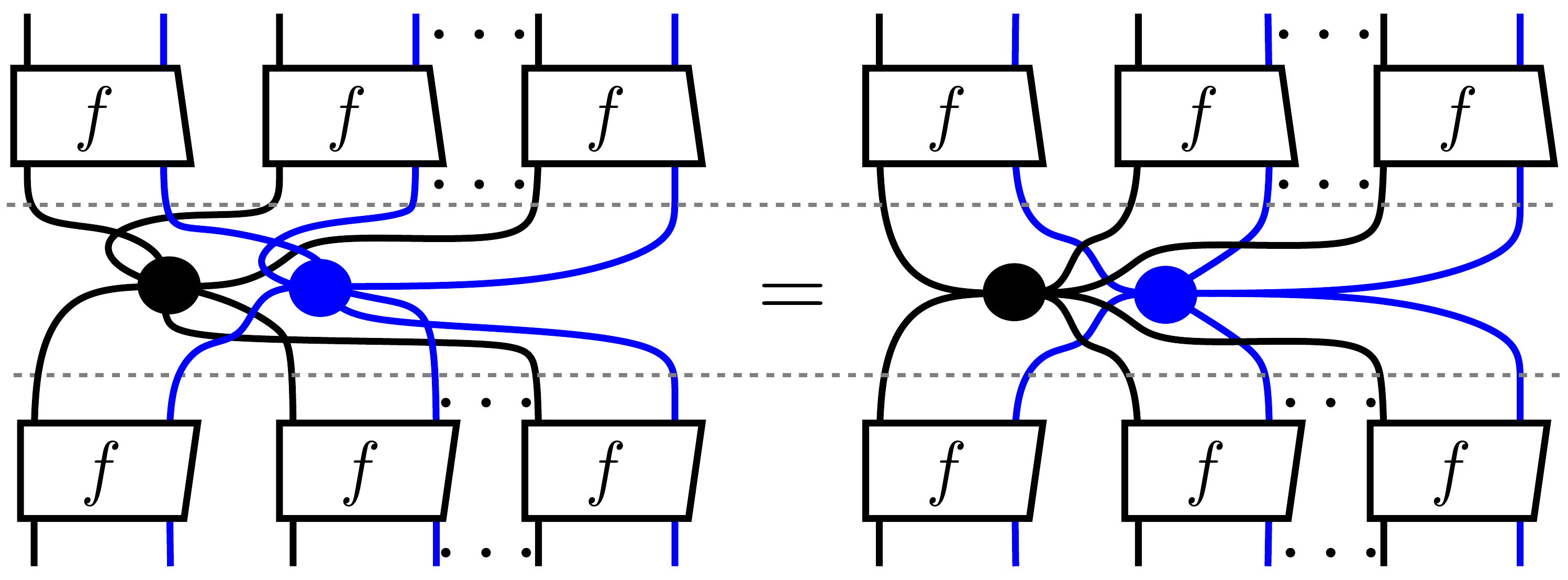}
\end{equation*}
since the legs of spiders within the dashed lines can be permuted.

However, the process $f$ cannot be just be any non-separable bipartite process. 
\begin{equation*}
\includegraphics[scale=0.2]{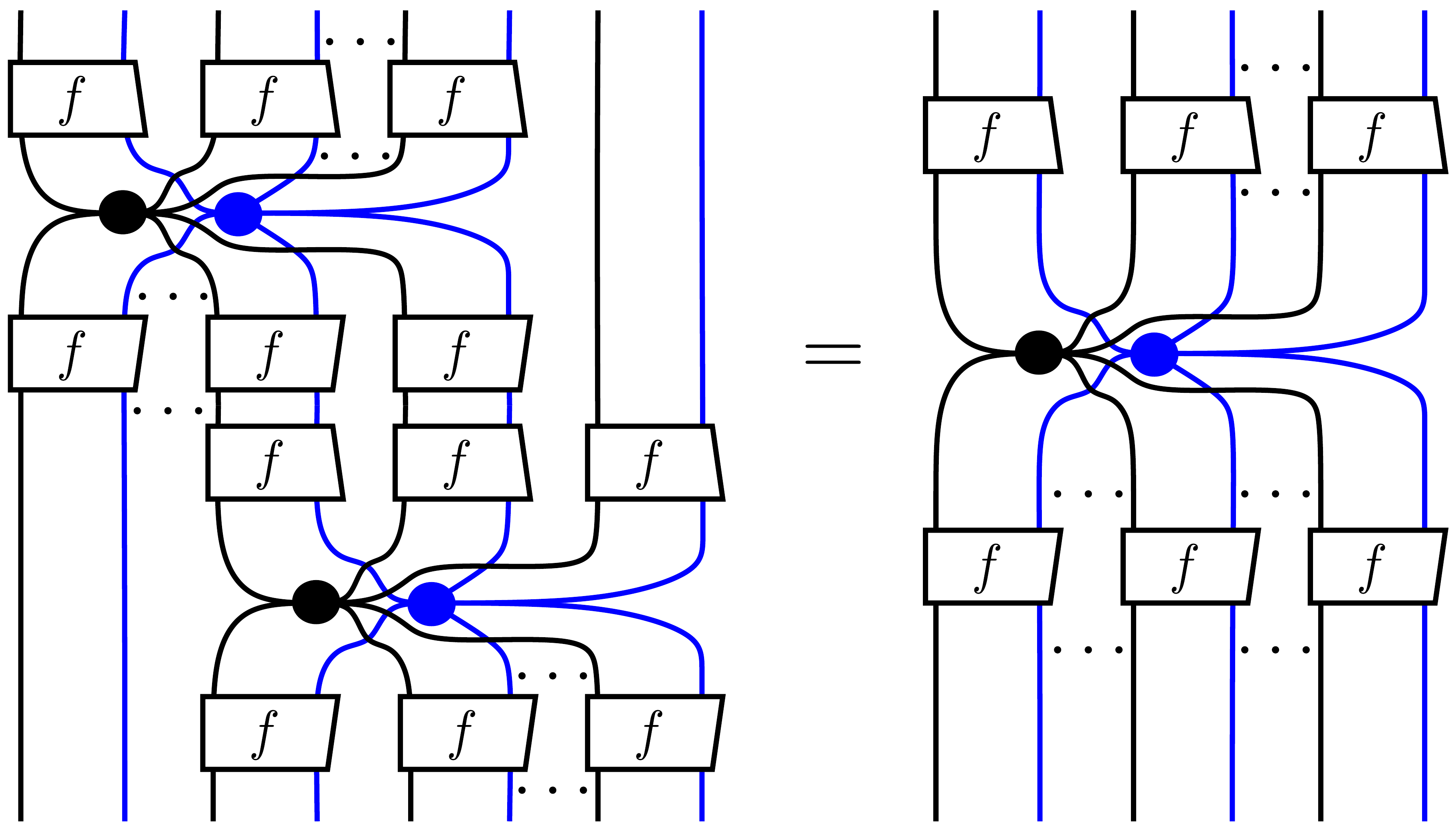}
\end{equation*}
In order for the diagram in Fig. \ref{fig:join-spiders} to satisfy the fusion rule for spiders as in the previous diagram, the process $f$ must be unitary.

\citet{Vidal2004UniversalGates} showed that any special unitary process
%\footnote{we never specified the unitary to be special }
on two qubits can be decomposed into up to three $\mathsf{CNOT}$ gates and single qubit unitaries. That is:
% add one and two CNOT
\begin{equation*}
    \includegraphics[scale=0.2]{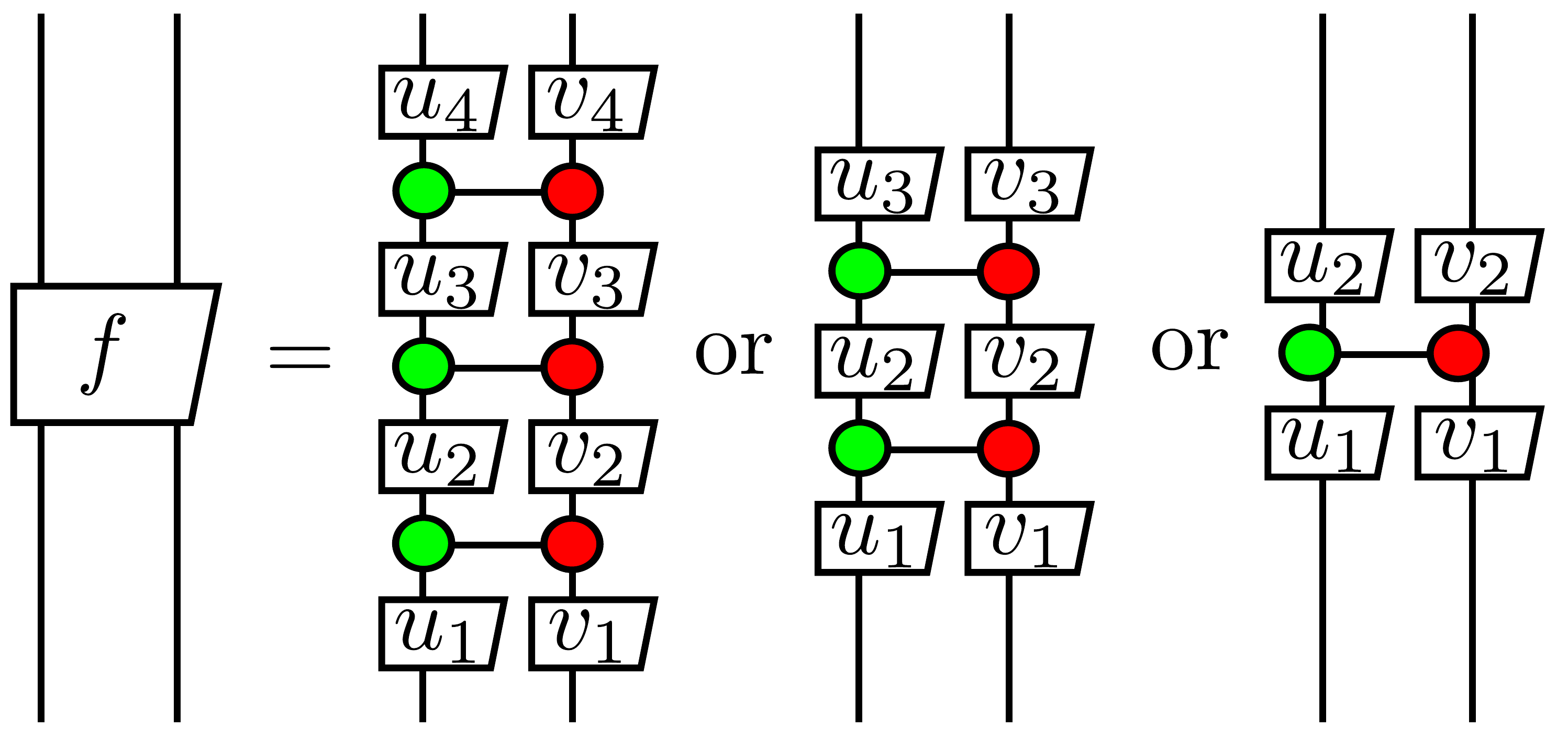}
\end{equation*}

We shall not explore all possible values of $u_j$ and $v_k$, but the decomposition above does give us an idea about the degree of complexity we can impose on the entanglement to be generated. Suppose we choose $\cwgr$ as our entanglement generator since it is the simplest non-separable bipartite process that is unitary based on the previous equation. Then we obtain the following diagram for Fig. \ref{fig:join-spiders}:
\begin{figure}[!ht]
    \centering
    \includegraphics[scale=0.2]{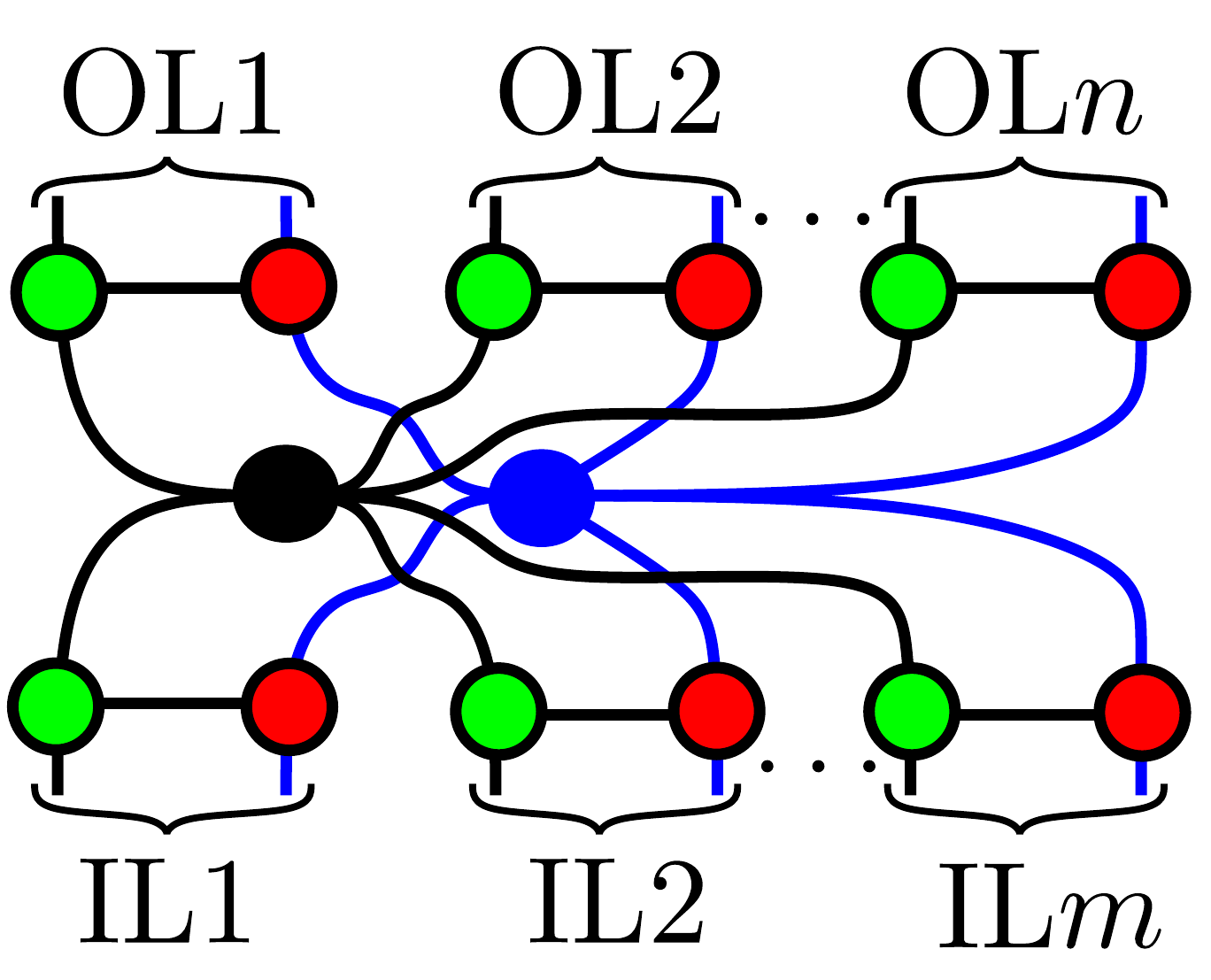}
    \caption{Fig \ref{fig:join-spiders} where $f$ is the $\mathsf{CNOT}$ gate.}
    \label{fig:join-spider-CNOT}
\end{figure}

From the rewrite rules of $ZX$-calculus, we can derive the following\footnote{Recall that scalars are ignored.}:
\begin{equation}\label{eq:gen-bialg}
    \includegraphics[scale=0.2]{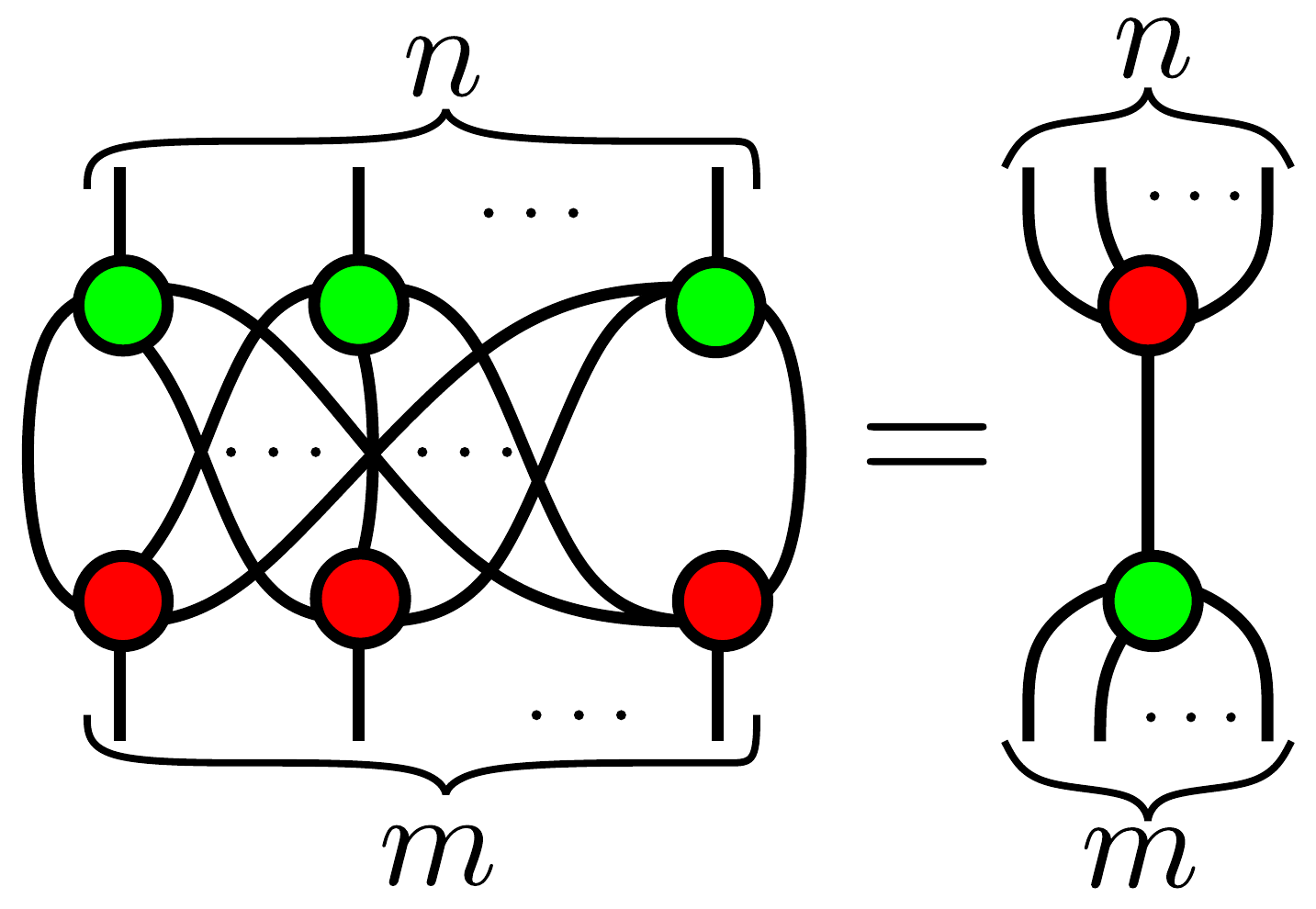}
\end{equation}
That is, for every green-red node pair in the left diagram, there is a single wire between them. We shall use Eq. \ref{eq:gen-bialg} to discuss the various possibilities for the simplest entanglement generator mentioned above.

\begin{enumerate}[label=\textbf{Case \arabic*}:]
\item If $\black=\zspider=\blue$, then we have:
\begin{equation*}
    \includegraphics[scale=0.2]{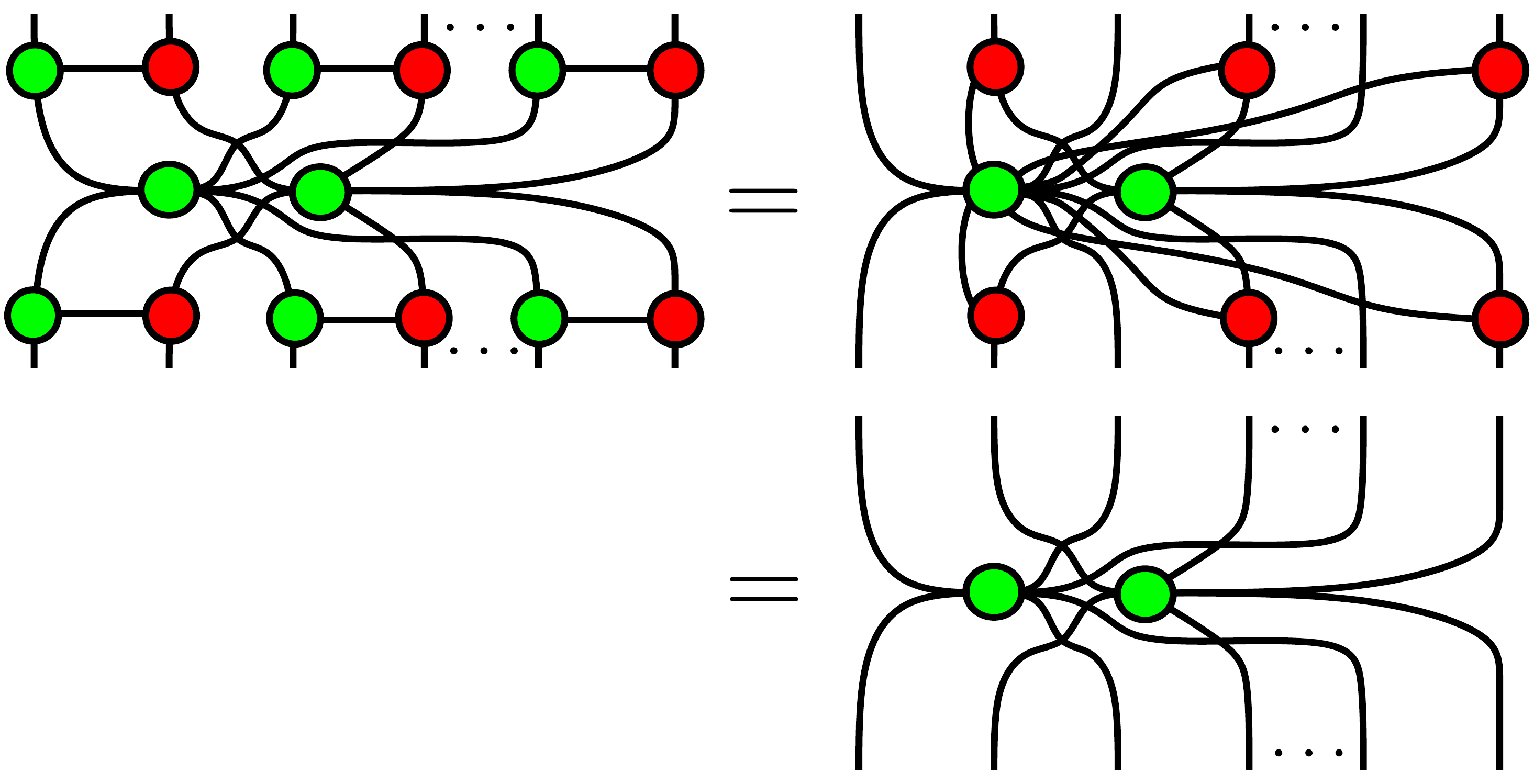}
\end{equation*}
In the middle diagram, there is exactly one wire between every green-red node-pair. So we can use Eq. \ref{eq:gen-bialg} to obtain the second equality.
\item If $\black=\xspider=\blue$, then we have:
\begin{equation*}
    \includegraphics[scale=0.2]{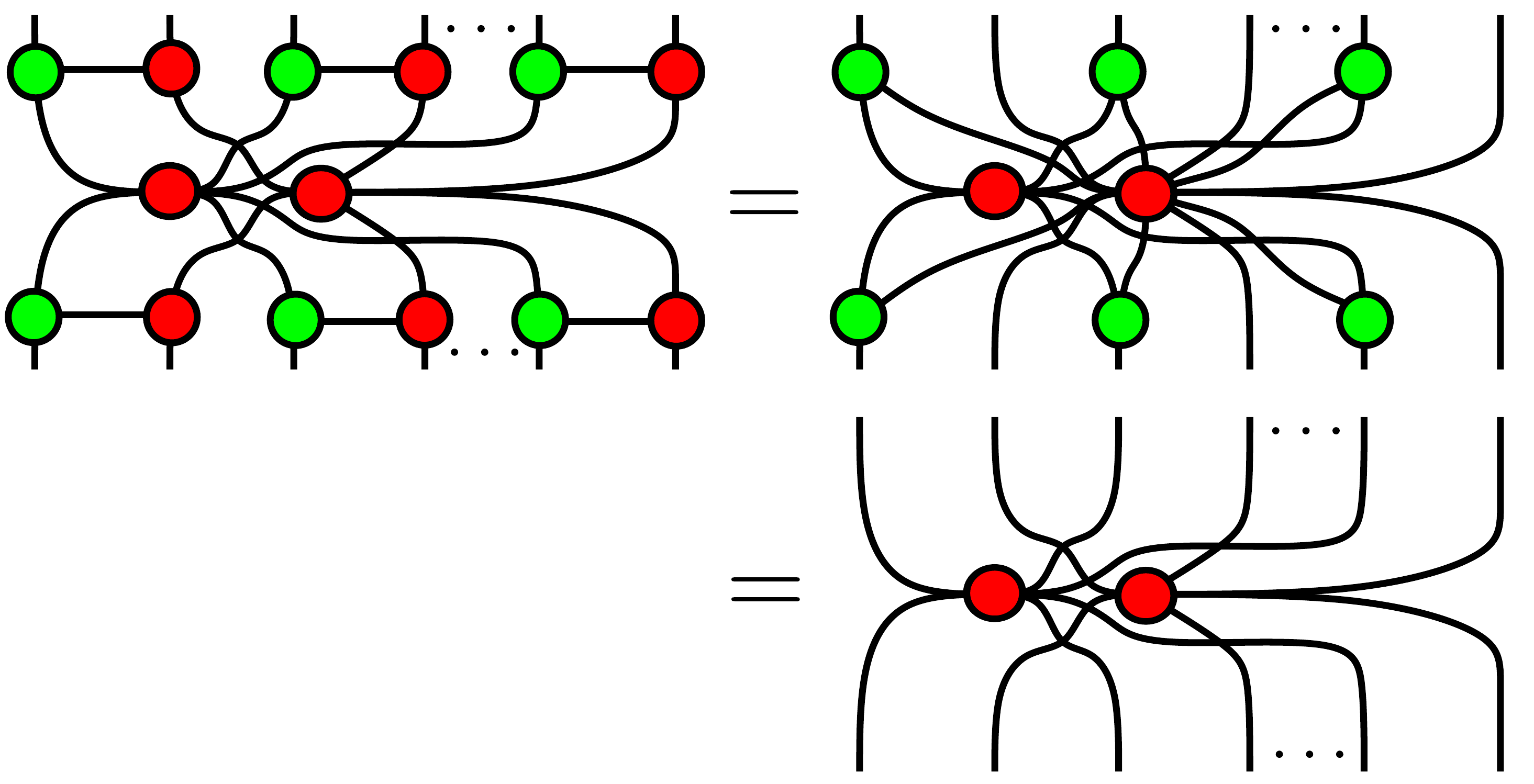}
\end{equation*}
where the third equality is obtained via the same reasoning as Case 1.
\item If $\black=\zspider$ and $\blue=\xspider$:
\begin{equation*}
    \includegraphics[scale=0.2]{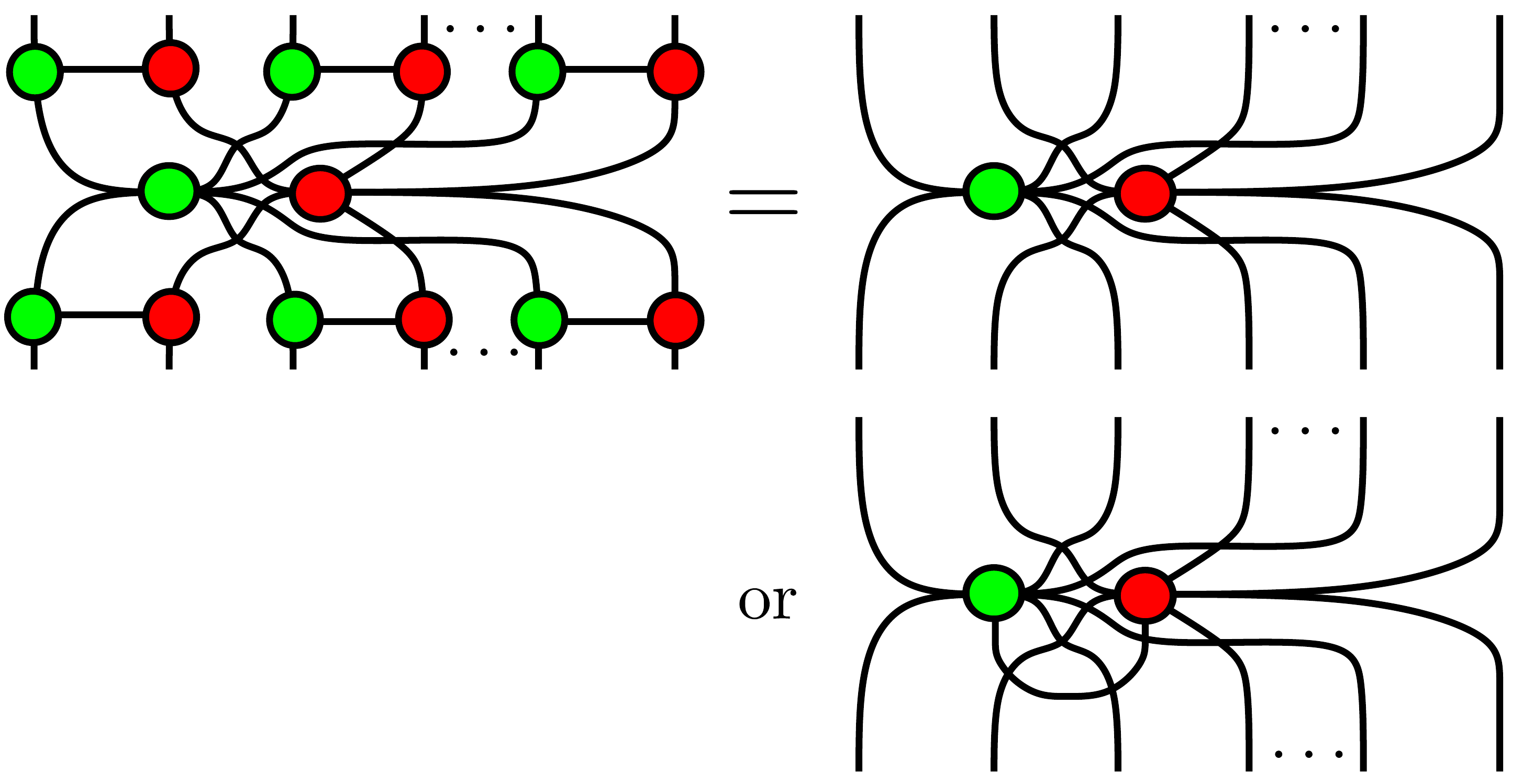}
\end{equation*}
where we obtain a separably composed spider (first diagram on RHS) if $m+n$ is even as in Fig. \ref{fig:join-spider-CNOT}, and we obtain the second diagram on RHS  if $m+n$ is odd. 
\item If $\black=\xspider$ and $\blue=\zspider$, Fig. \ref{fig:join-spider-CNOT} cannot be simplified as in Cases 1 to 3 so we maintain its original appearance:
\begin{equation*}
    \includegraphics[scale=0.2]{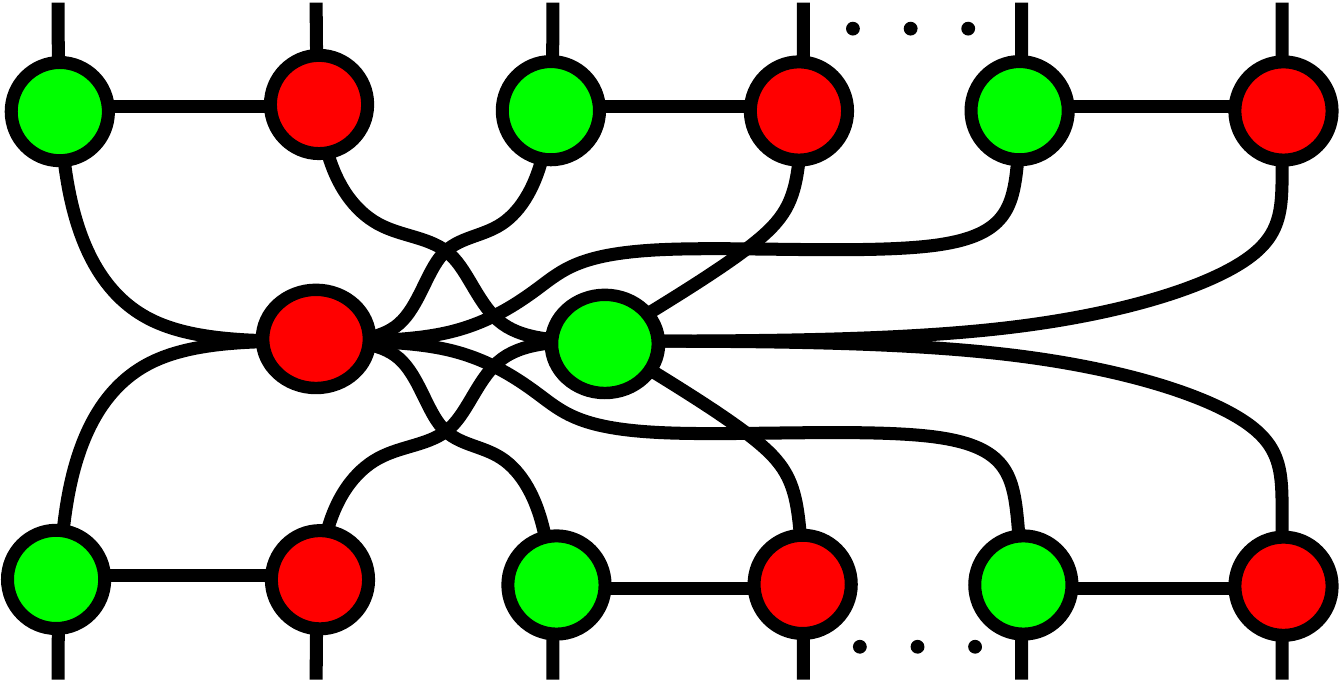}
\end{equation*}
\end{enumerate}

Cases 1 and 2 can immediately be excluded since the resulting spiders from each case are equal to separably composed spiders. For Case 3, half of the spiders produced are separably composed while the other half can be described as separably composed spiders with a wire connecting the two central nodes. If we look at the underlying basis for such a classical structure, it does indeed consist of separable states.

\begin{proposition}\label{prop:case3-SC}
The underlying basis of the classical structure in Case 3 above is equal to, up to scalars, the underlying basis of the classical structure formed by separably composing $\mathcal{X}$ and $\mathcal{Z}$ which results in the classical structure $\mathcal{XZ}$ on two qubits.
\begin{proof}
The underlying basis of $\mathcal{XZ}$ is:
\begin{equation*}
    \includegraphics[scale=0.2]{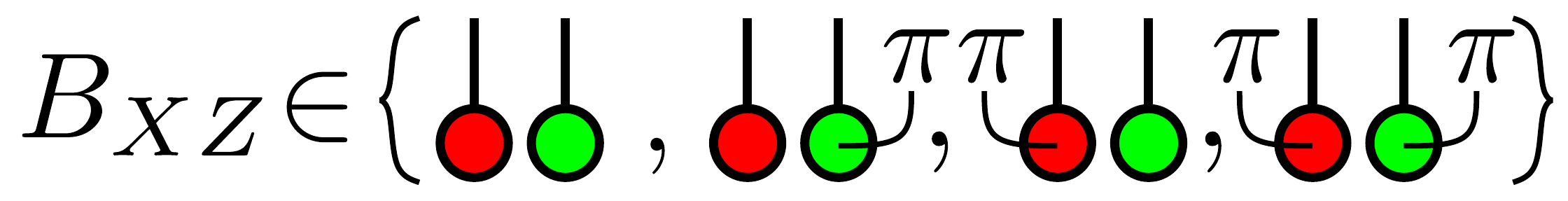}
\end{equation*}
which is equal to $\{\ket{00_X},\ket{01_X},\ket{10_X},\ket{11_X}\}$ up to scalars.

Suppose $\wunit\bunit\in B_{XZ}$. Then:
\begin{eqnarray*}
\includegraphics[scale=0.2]{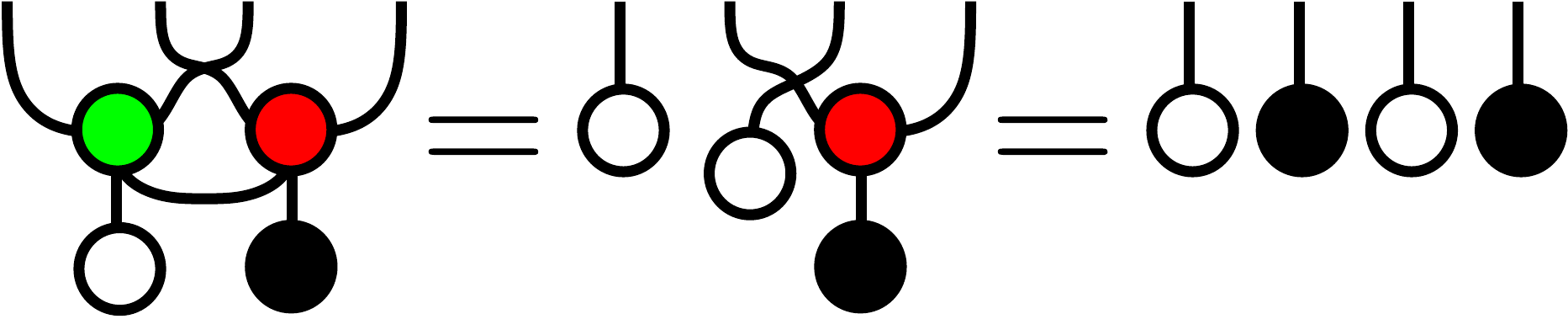}\\
\includegraphics[scale=0.2]{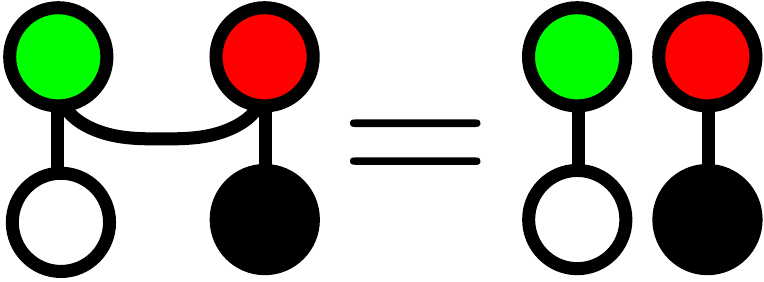}
\end{eqnarray*}
This means that the 2,1-spider and 0,1-spider of the classical structure in Case 3, respectively, copies and deletes $B_{XZ}$ or an orthogonal basis equivalent to $B_{XZ}$ up to scalars. Thus, the underlying basis of the classical structure in Case 3 must be equivalent to $B_{XZ}$ up to scalars.
\end{proof}
\end{proposition}

While a pair of classical structures with underlying bases that contain the same states but with differing global phases are not necessarily equivalent (see Section \ref{sec:discussion} for further discussion), the separability of these bases is preserved, i.e. the classical structure in Case 3 has an underlying basis that is separable. Thus, we are left with the classical structure in Case 4 as a classical structure that is not equivalent to $\mathcal{XZ}$.

\begin{proposition}\label{prop:case4-NS}
The classical structure produced by Case 4 has an underlying basis consisting of non-separable states.
\begin{proof}
For this proposition, we shall use the scalar inclusive version of ZX-calculus as outlined in the previous chapter. First, suppose $B$ is the underlying basis of the classical structure in Case 4 and suppose $B$ consists of separable states. Let $\ket{\psi}\ket{\phi}$ be a normalized state that is proportional to one of the states in $B$. Then we have:
\begin{equation}\label{eq:case4-1}
    \includegraphics[scale=0.2]{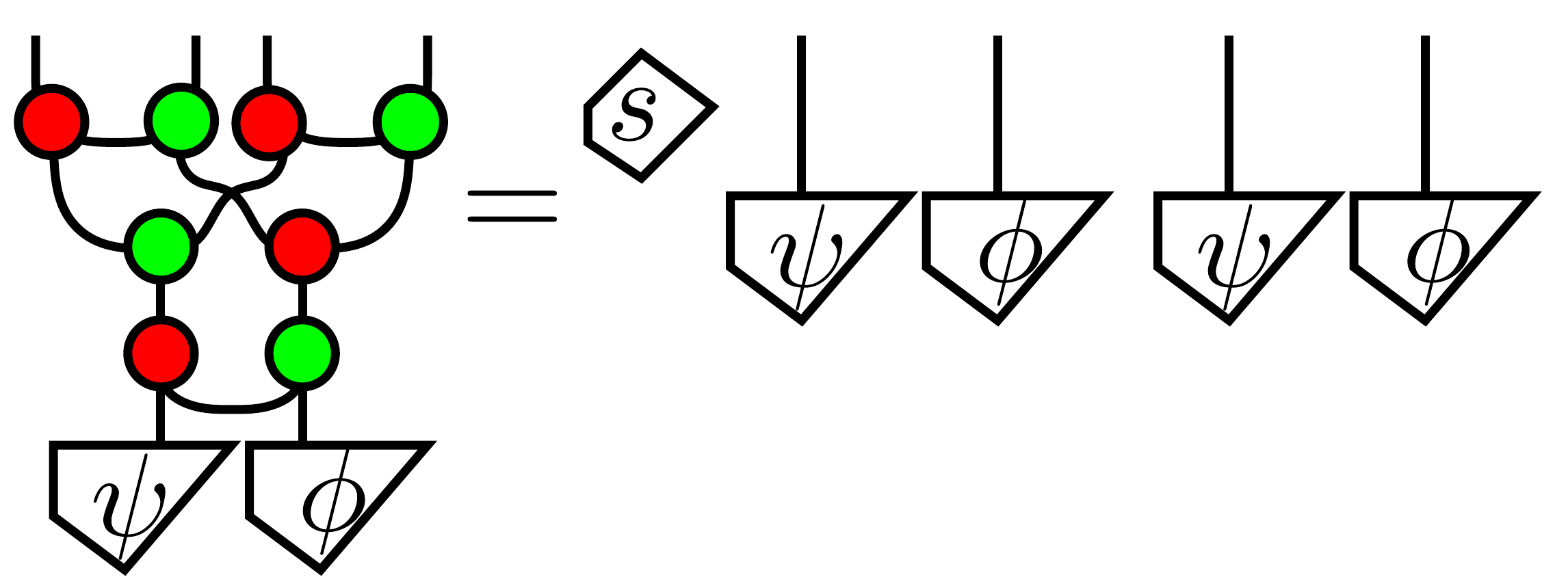}
\end{equation}
for some scalar $s\in\mathbb{C}$.

From the rewrite rules of ZX-calculus, we obtain:
\begin{equation}\label{eq:case4-2}
    \includegraphics[scale=0.2]{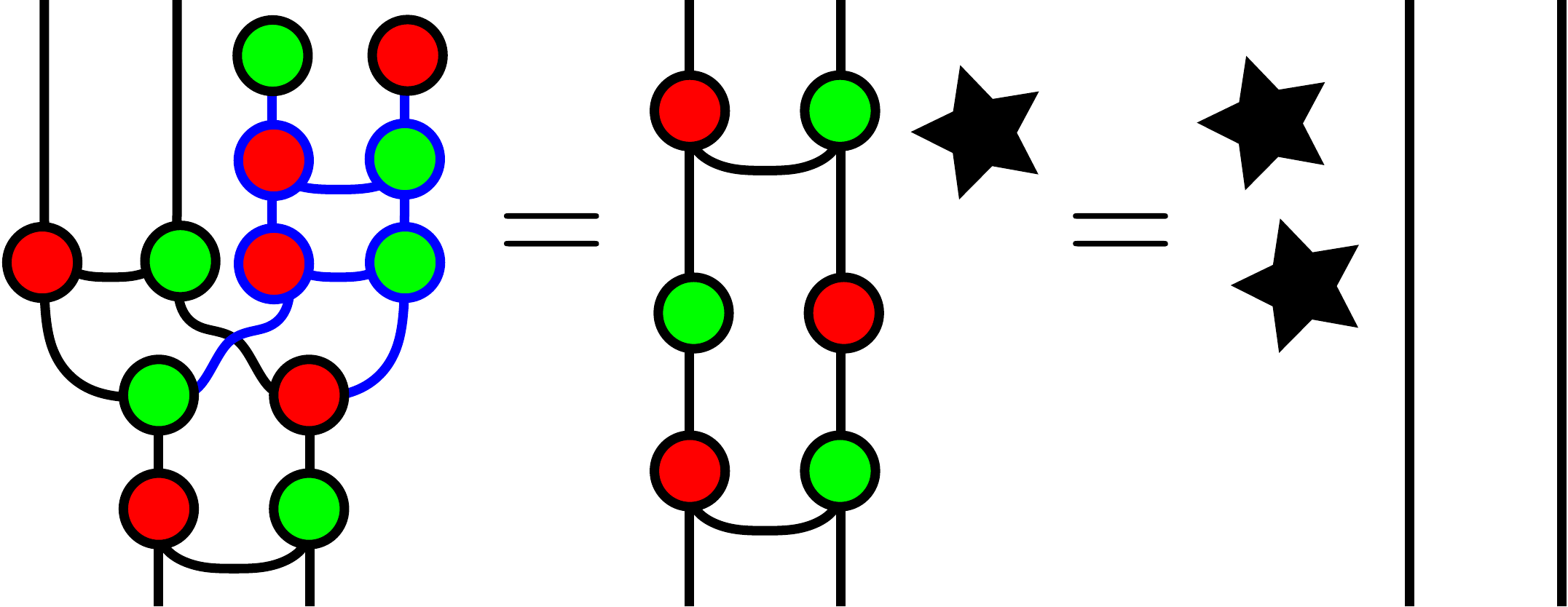}
\end{equation}
where we applied Eqs. \ref{eq:zx1-fuse-green}, \ref{eq:zx2-fuse-red} and \ref{eq:dzx1-hopf} on the part highlighted blue in the first diagram to obtain the first equality. The same equations are then applied to the second diagram to obtain the final diagram.

By Eq. \ref{eq:case4-1} and \ref{eq:case4-2}, we have:
\begin{equation}\label{eq:case4-3}
    \includegraphics[scale=0.2]{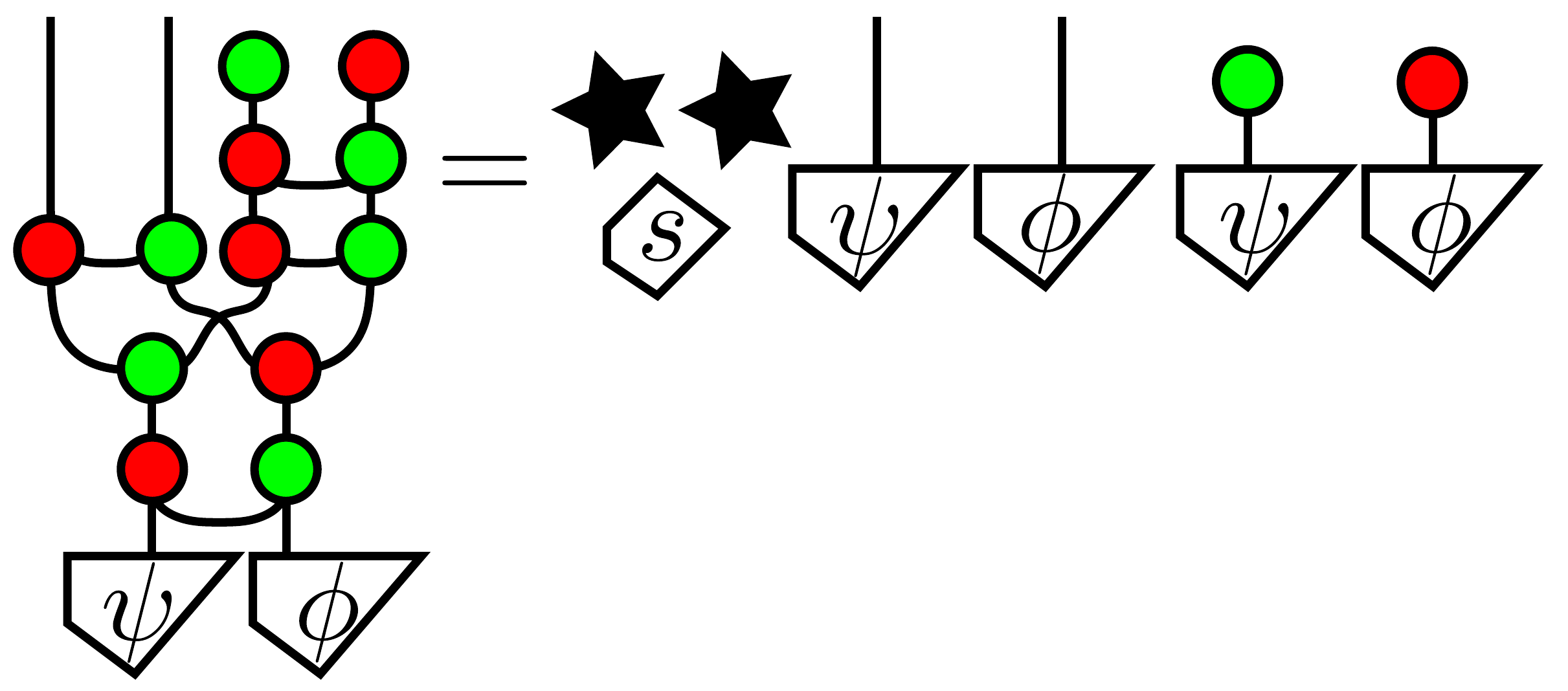}
\end{equation}
Now, the 1,0-spider of the classical structure in Case 4 is:
\begin{equation}\label{eq:case4-4}
    \includegraphics[scale=0.2]{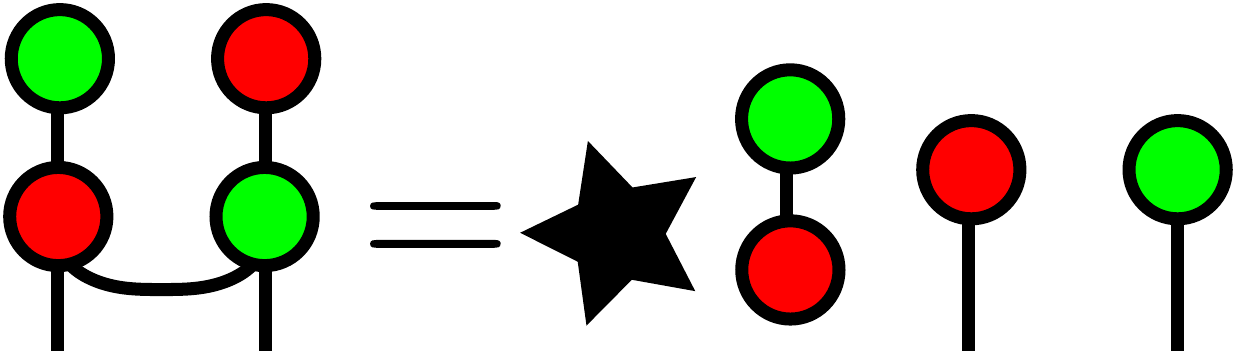}
\end{equation}
due to Eqs. \ref{eq:zx9-copy-green} and \ref{eq:zx10-copy-red}. 
So then:
\begin{equation}\label{eq:case4-5}
    \includegraphics[scale=0.2]{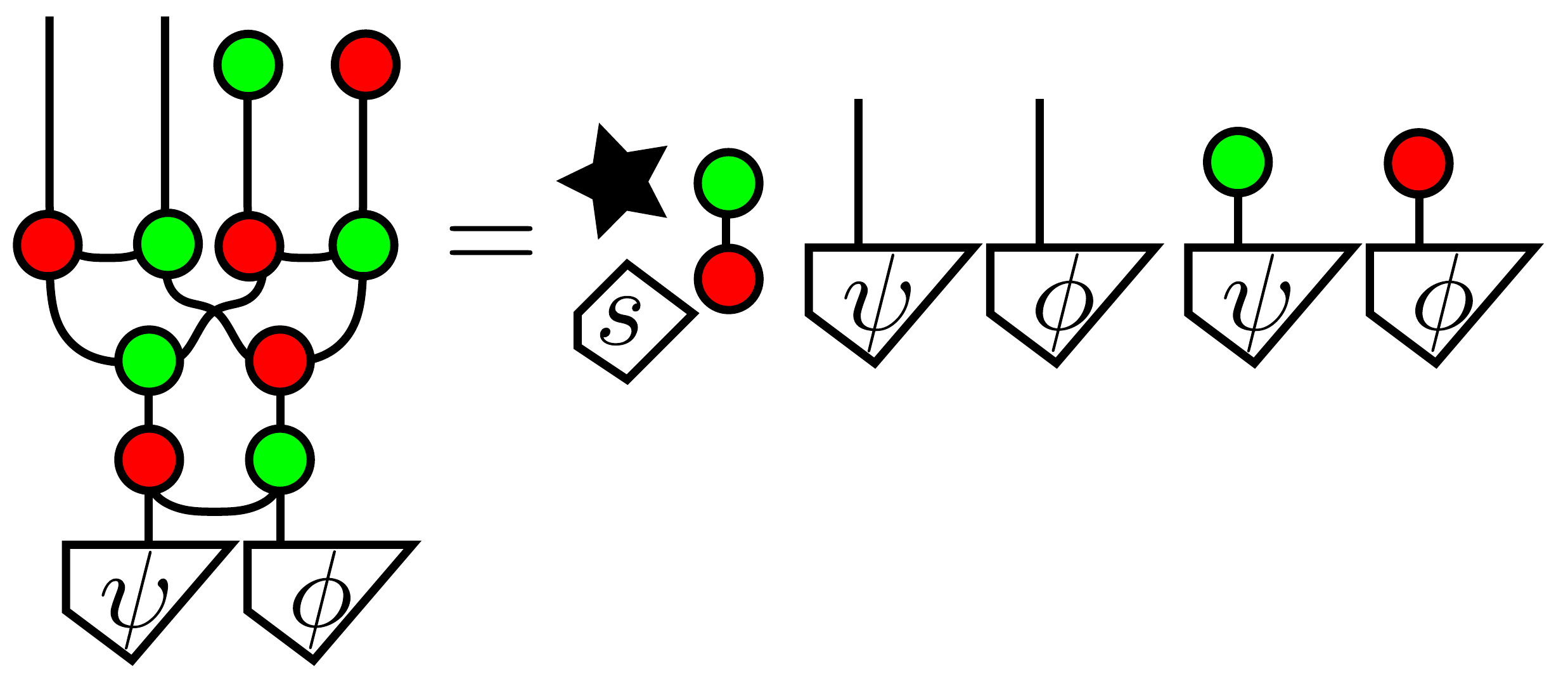}
\end{equation}
Eqs. \ref{eq:case4-3} and \ref{eq:case4-5} give us:
\begin{equation}\label{eq:case4-6}
    \includegraphics[scale=0.2]{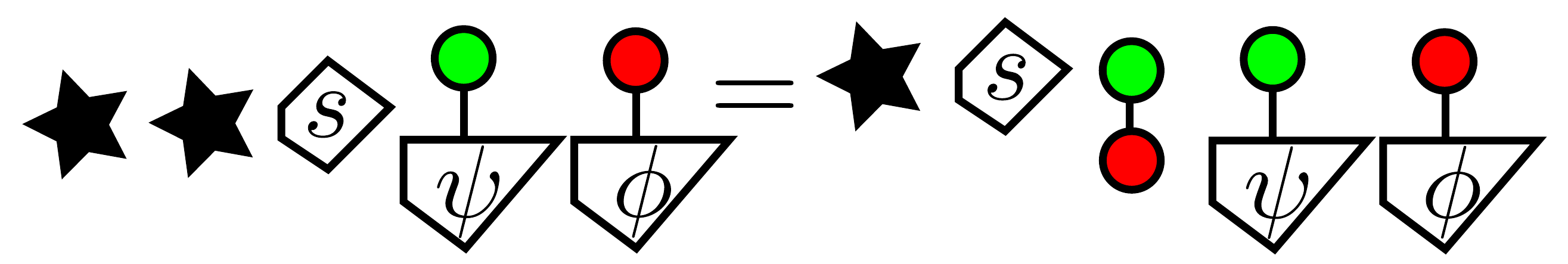}
\end{equation}
So, we have:
\begin{equation}
    \includegraphics[scale=0.2]{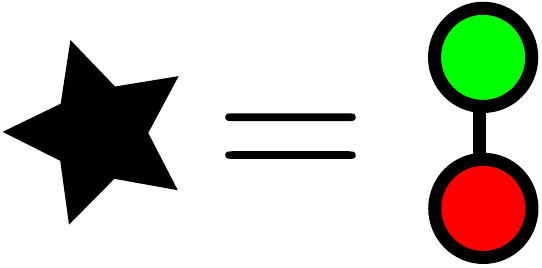}
\end{equation}
which is not true since $\starscalar$ and $\sqrtwo\sqrtwo$ are inverses of each other and neither are equal to 1. 
\end{proof}
\end{proposition}

\subsection{Connecting Wires for Joining Spiders}\label{sec:NSCS}

Proposition \ref{prop:case4-NS} tells us that the classical structure produced by Case 4 above must have an underlying basis consisting of non-separable states. We call it a \textbf{non-separable classical structure on two qubits with constituents $\mathcal{X}$ and $\mathcal{Z}$}, denoted by $\mathcal{X}\diamond\mathcal{Z}$. From $\mathcal{X}\diamond\mathcal{Z}$, we can construct other non-separable classical structures on two qubits with different constituents:
\begin{eqnarray}
\includegraphics[scale=0.2]{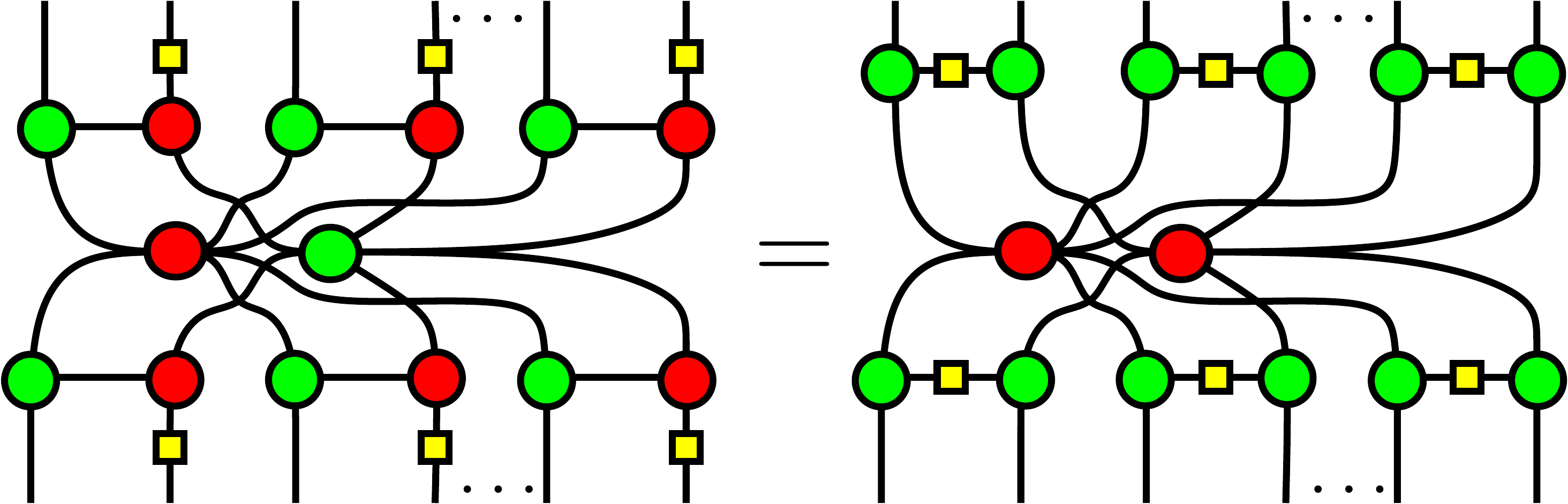} \label{eq:XX}\\
\includegraphics[scale=0.2]{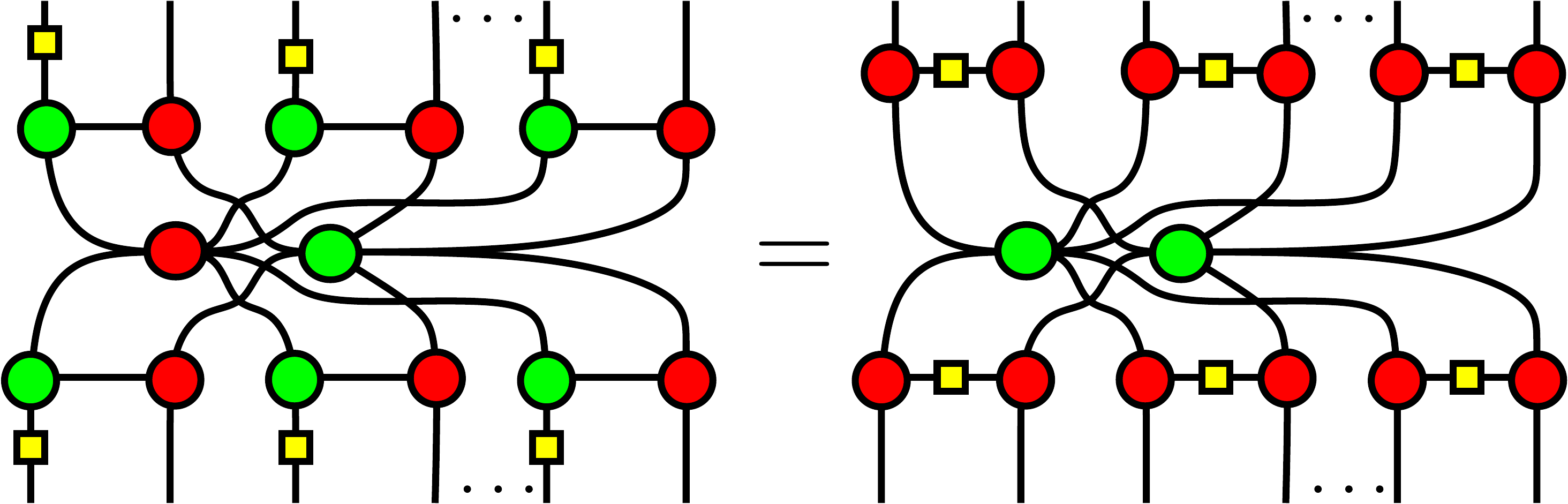} \label{eq:ZZ}\\
\includegraphics[scale=0.2]{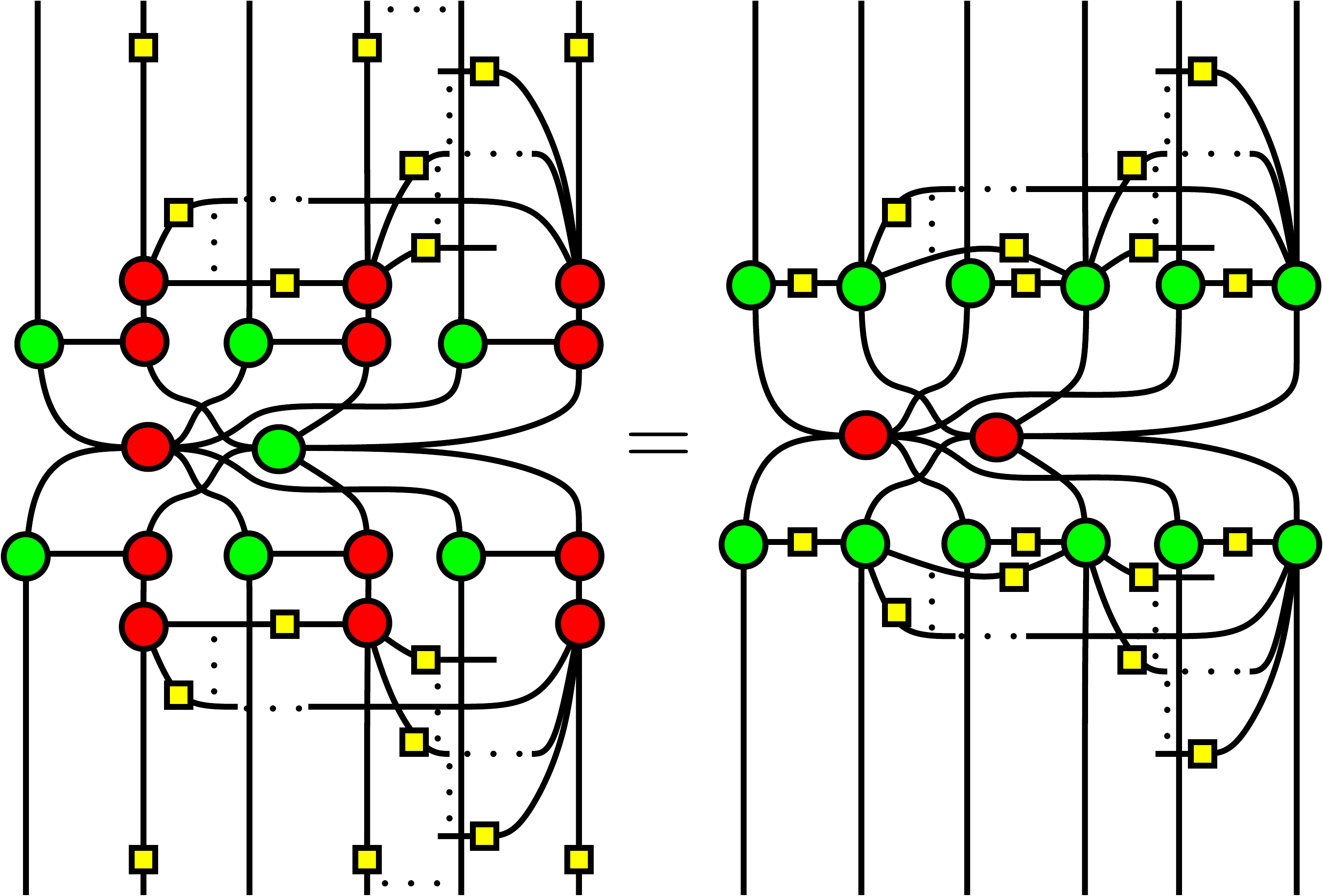} \label{eq:XY}\\
\includegraphics[scale=0.2]{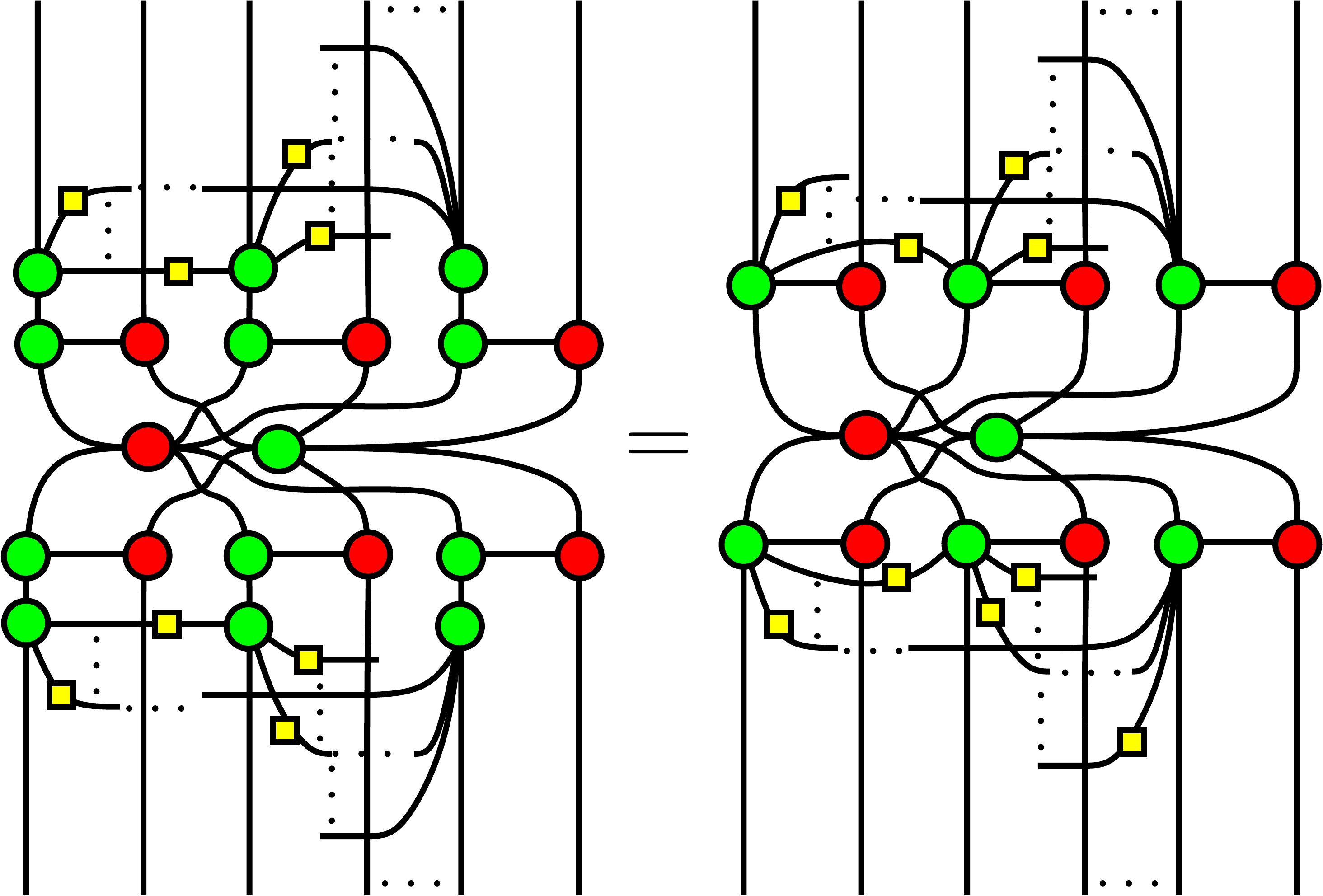} \label{eq:YZ}
\end{eqnarray}

If we compose the spider from Case 4 with $\cwgg$ as in Eq. \ref{eq:XY}, and also compose on its legs $\cwrr$ and $\had$ as in Eq. \ref{eq:YZ}, we shall obtain a non-separable classical structures on two qubits where both of its constituents are $\mathcal{Y}$.

We can permute the constituents of the known spiders. For example:
\begin{equation}\label{eq:ZX}
    \includegraphics[scale=0.2]{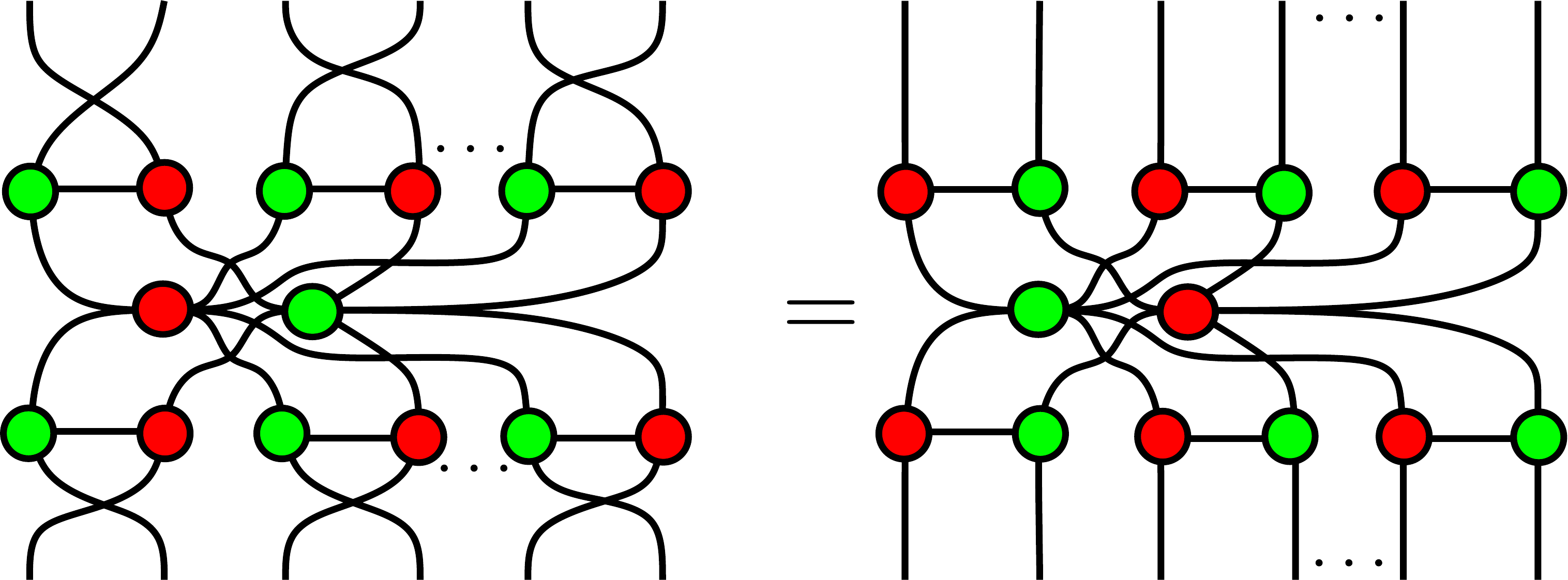}
\end{equation}

Thus, we have obtained non-separable classical structures on two qubits for any combination of constituents from $\{\mathcal{X},\mathcal{Y},\mathcal{Z}\}$.

We call the bipartite process which entangles a pair of classical structures as \textbf{connecting wires}. Table \ref{tab:connecting-wire} shows the connecting wire for each pair of classical structures from $\{\mathcal{X},\mathcal{Y},\mathcal{Z}\}$.

\begin{table}
    \centering
     \caption{Connecting wires for pairs of classical structures, where CS1 is the classical structure on the first qubit and CS2 is the classical structure on the second qubit. The rightmost column gives the notation of the resulting composed classical structure.}
    \label{tab:connecting-wire}
    \begin{longtable}{|c|c|c|c|}\hline
        CS1 & CS2 & Connecting Wire & Notation \\ \hline
        $\mathcal{X}$ & $\mathcal{X}$ & \includegraphics[scale=0.2]{images/symbols/conwire2.pdf} & $\mathcal{X}\diamond\mathcal{X}$\\ \hline
        $\mathcal{X}$ & $\mathcal{Y}$ & \includegraphics[scale=0.2]{images/symbols/conwire2.pdf} & $\mathcal{X}\diamond\mathcal{Y}$\\ \hline
        $\mathcal{X}$ & $\mathcal{Z}$ & \includegraphics[scale=0.2]{images/symbols/conwire1.pdf} & $\mathcal{X}\diamond\mathcal{Z}$\\ \hline
        $\mathcal{Y}$ & $\mathcal{X}$ & \includegraphics[scale=0.2]{images/symbols/conwire2.pdf} & $\mathcal{Y}\diamond\mathcal{X}$\\ \hline
        $\mathcal{Y}$ & $\mathcal{Y}$ & \includegraphics[scale=0.2]{images/symbols/conwire2.pdf} & $\mathcal{Y}\diamond\mathcal{Y}$\\ \hline
        $\mathcal{Y}$ & $\mathcal{Z}$ & \includegraphics[scale=0.2]{images/symbols/conwire1.pdf} & $\mathcal{Y}\diamond\mathcal{Z}$\\ \hline
        $\mathcal{Z}$ & $\mathcal{X}$ & \includegraphics[scale=0.2]{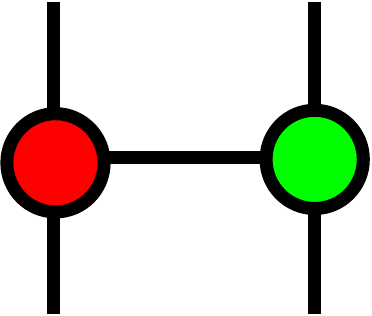} & $\mathcal{Z}\diamond\mathcal{X}$\\ \hline
        $\mathcal{Z}$ & $\mathcal{Y}$ & \includegraphics[scale=0.2]{images/symbols/conwire4.pdf}  & $\mathcal{Z}\diamond\mathcal{Y}$\\ \hline
        $\mathcal{Z}$ & $\mathcal{Z}$ & \includegraphics[scale=0.2]{images/symbols/conwire3.pdf}  & $\mathcal{Z}\diamond\mathcal{Z}$\\
        \hline
    \end{longtable}
\end{table}

\section{Composite Classical Structures on \texorpdfstring{$N$}{TEXT} Qubits}\label{sec:composemultCS}

Based on the constructions in the previous sections, we introduce some notations:

\begin{definition}\label{def:composite-CS}
A \textbf{composite} classical structure on two qubits is formed by composing two classical structures on a single qubit. We call these classical structures on a single qubit as constituents of the composite classical structure. With respect to the present work, a constituent is either $\mathcal{X},\mathcal{Y}$ or $\mathcal{Z}$. 

A \textbf{separably composed classical structure} (SCCS) on two qubits is a composite classical structure which consists of spiders composed separably as in Section \ref{sec:SCCS}. For $\mathcal{O}_1,\mathcal{O}_2\in\{\mathcal{X,Y,Z}\}$, a SCCS with constituents $\mathcal{O}_1$ and $\mathcal{O}_2$ is denoted by $\mathcal{O}_1\mathcal{O}_2$.

A \textbf{non-separable classical structure} (NSCS) on two qubits is a composite classical structure which consists of entangled spiders as in Section \ref{sec:NSCS}. For $\mathcal{O}_1,\mathcal{O}_2\in\{\mathcal{X,Y,Z}\}$, a NSSC with constituents $\mathcal{O}_1$ and $\mathcal{O}_2$ is denoted by $\mathcal{O}_1\diamond\mathcal{O}_2$.
\end{definition}

We can extend Definition \ref{def:composite-CS} to multiple qubits. However, for $N$ qubits where $N>2$, there are more than one pair of qubits. So, the degree of entanglement of a classical structure is more varied. That is, it is no longer necessary for a classical structure to be separable on all qubits or non-separable on all qubits. A classical structure can be separable on one pair of qubits and non-separable on another pair of qubits. For a example, a classical structure on three qubits could be non-separable on the first two qubits, but separable on the second and third qubits and on the third and first qubits. So, we should be able to compose a SCCS and a NSCS to form a classical structure on $N$ qubits where $N>2$.

To show that this is indeed possible, we introduce a diagram which subsumes both SCCS and NSCS:
\begin{equation*}
    \includegraphics[scale=0.2]{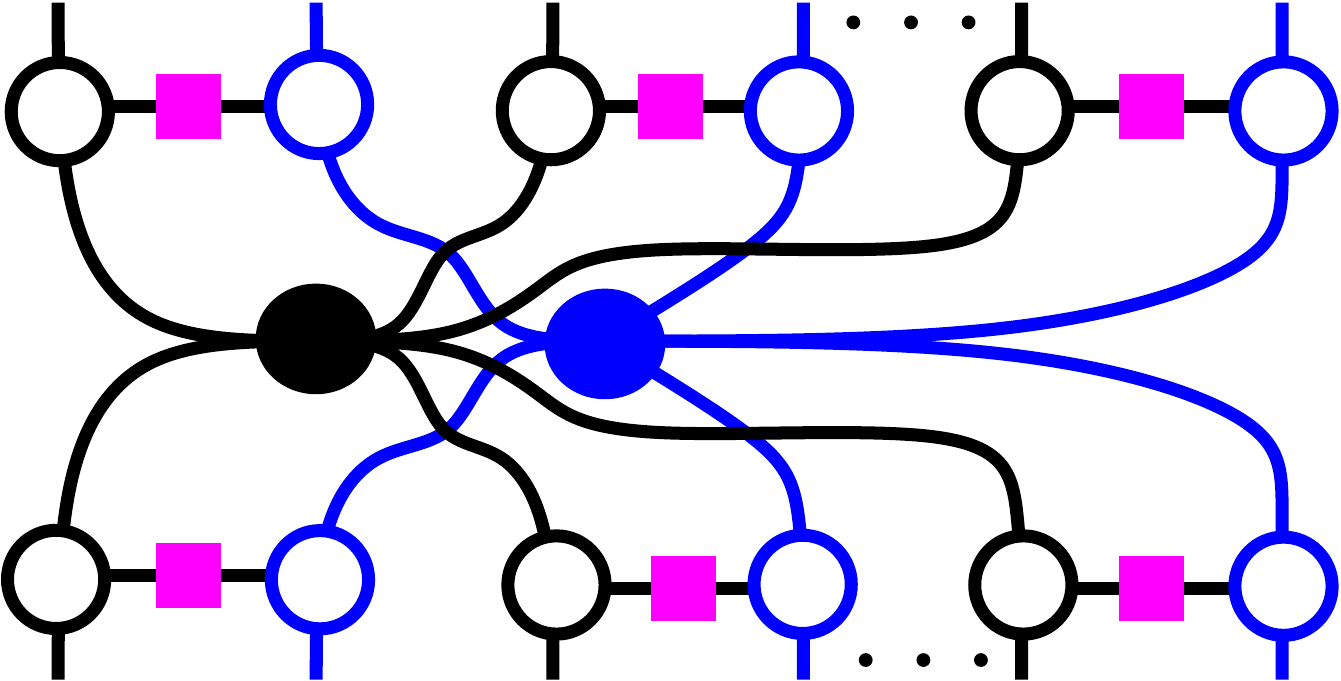}
\end{equation*}
\noindent where:
\begin{equation*}
    \includegraphics[scale=0.2]{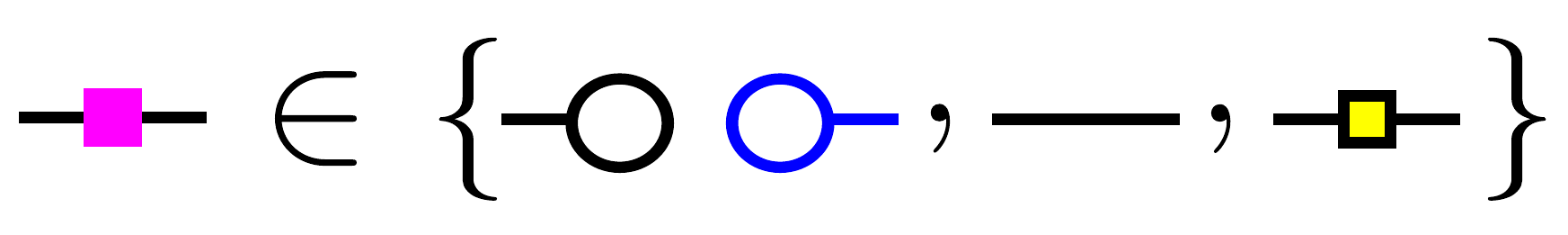}
\end{equation*}
\noindent depending on separability of the classical structure and the spiders $\blackb$ and $\blueb$. 

\begin{figure}[h]
\centering
    \includegraphics[scale=0.2]{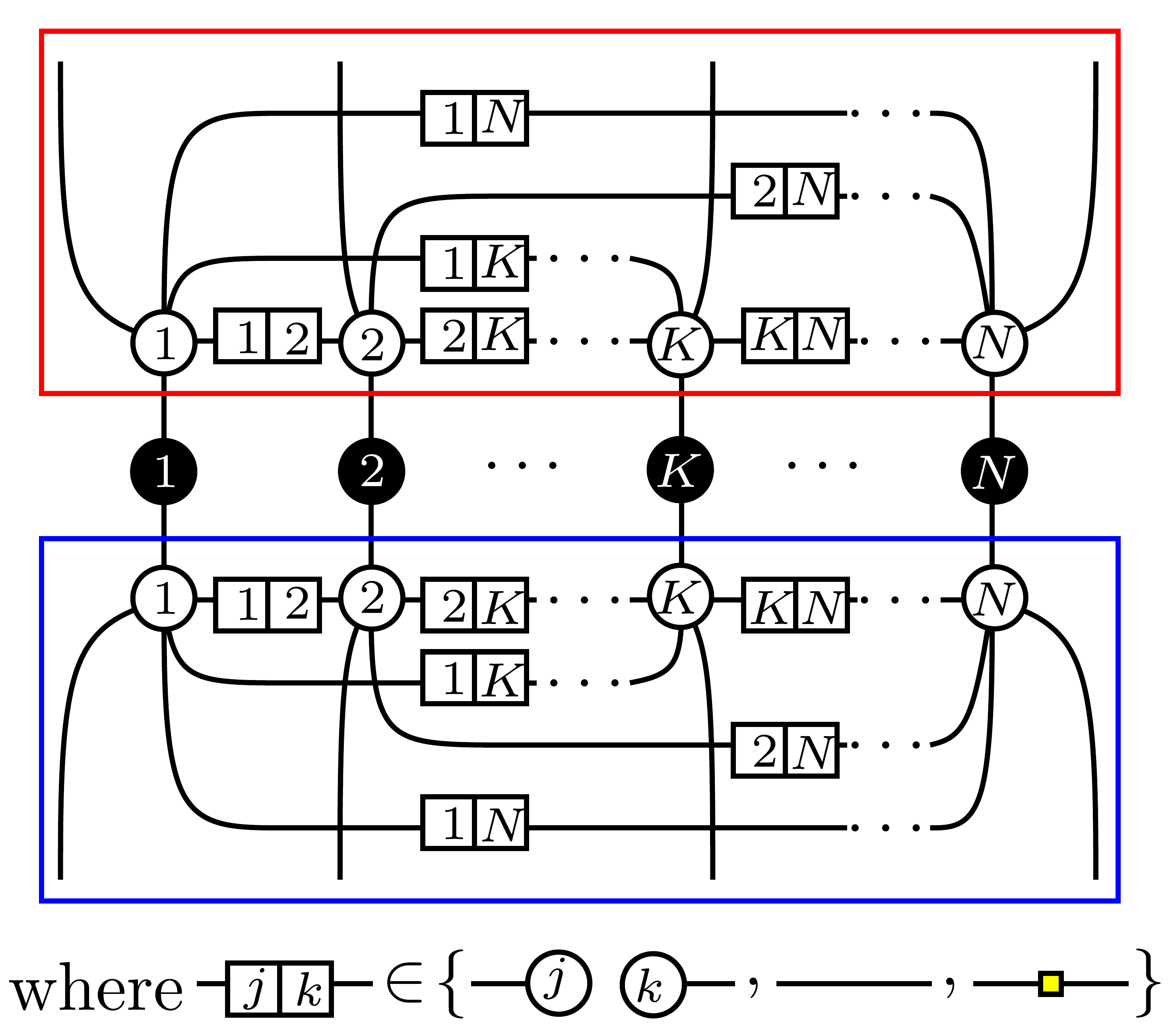}
\caption{$m,n$-spider of a composite classical structure on $N$ qubits. The diagrams in the blue boxes are repeated $m$ times, and the diagrams in the red boxes are repeated $n$ times.}
\label{fig:compositeCS-N}
\end{figure}

Before we proceed further, we shall recall !-boxes, which we view as a convenient tool to draw diagrams for an arbitrary number of qubits. Note that the following definition is a simplified one as not to distract from our main study. For a more detailed description of !-boxes, one may refer to references \cite{Kissinger2014,Merry2014}. 

\begin{definition}\label{def:!-boxes}
Let $D$ be a diagram. A !-box in $D$ is its subdiagram which can be repeated an arbitrary number of times, starting from 0. In a diagram, a !-box is marked by a coloured box around it. !-boxes of the same colour represent the same number of repetitions of those subdiagrams. 
\end{definition}

A simple example of a diagram with !-boxes is a generic spider:
\begin{equation*}
    \includegraphics[scale=0.2]{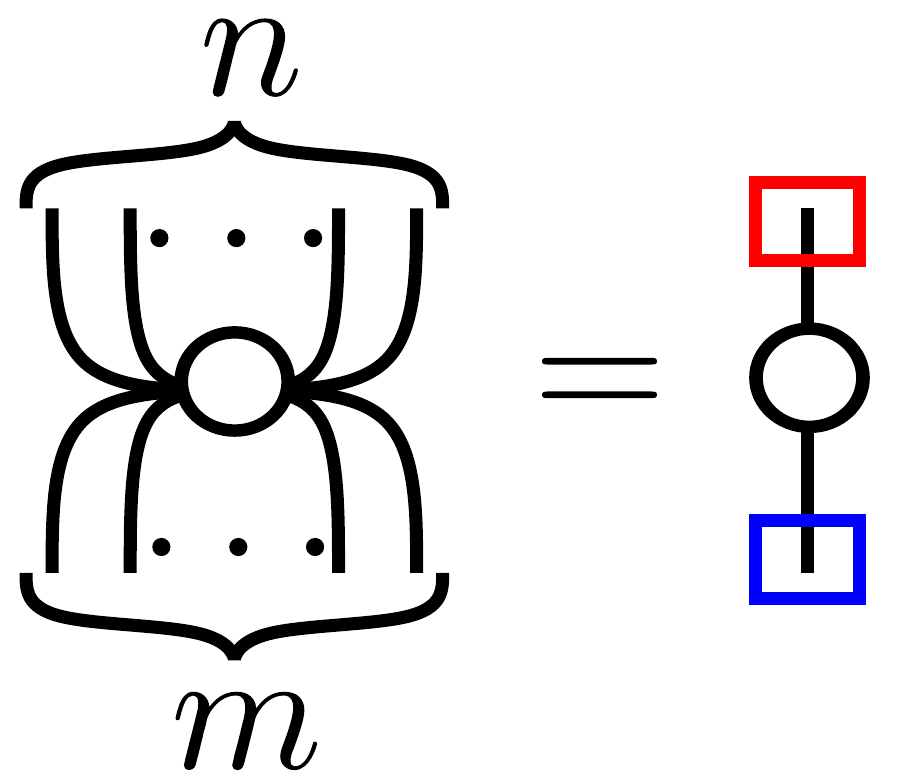}
\end{equation*}
or, a spider belonging to $\mathcal{XZ}$:
\begin{equation*}
    \includegraphics[scale=0.2]{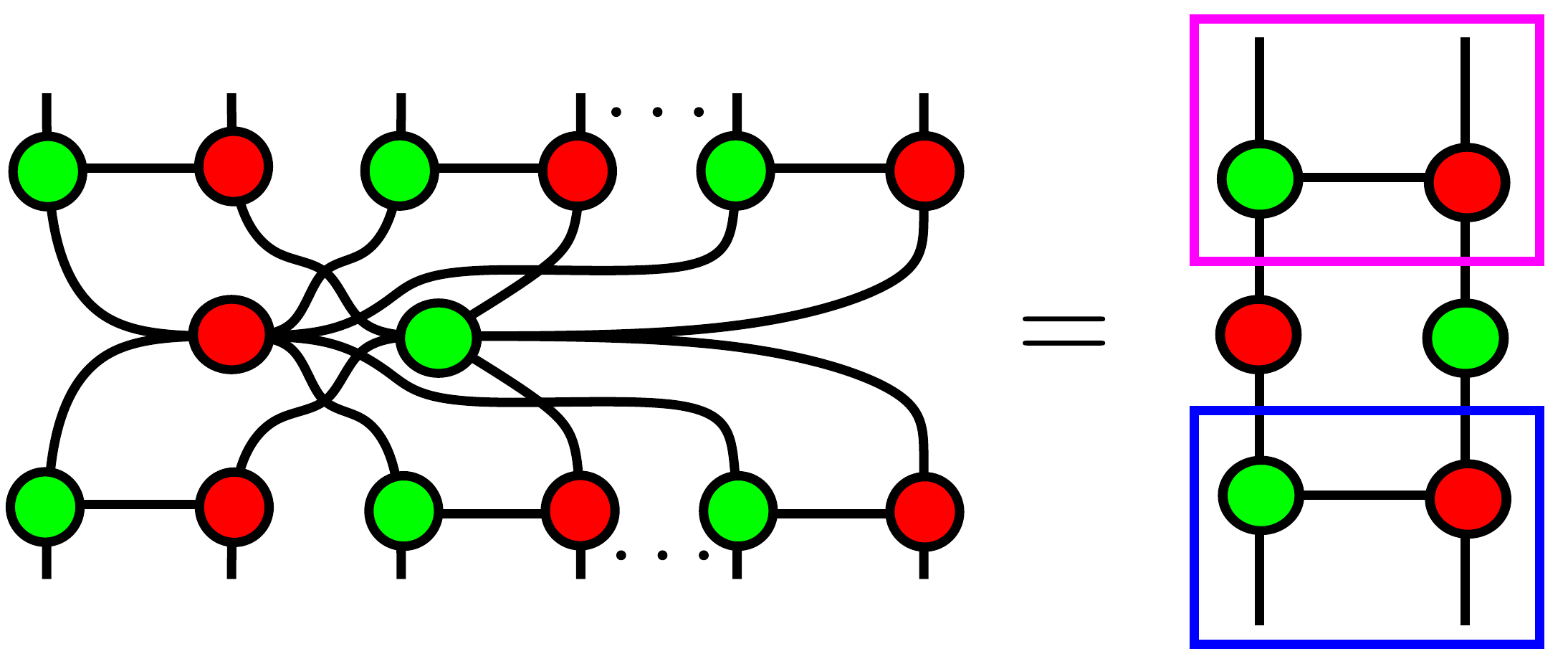}
\end{equation*}

\begin{theorem}\label{thm:compositeCS-N}
The process in Fig. \ref{fig:compositeCS-N}
is a spider of a classical structure on $N$ qubits.
\end{theorem}
\begin{proof}
%We use !-boxes in Fig. \ref{fig:compositeCS-N} to highlight the connecting wire, or its absence between the legs of two spiders, say the black spiders $S_j$ and $S_k$, with $\conwirejk$.
%The bottom blue box and the top red box ensures that the input and output wires of the diagram are permuted so that the invariance of input (output) legs of a spider is satisfied, and it makes clear that each leg of the supposed spider in Fig. \ref{fig:compositeCS-N} is $N$ qubits.

If a diagram of the form in Fig. \ref{fig:compositeCS-N} with $m_1$ input legs and $n_1$ output legs is composed with another diagram of the form in Fig. \ref{fig:compositeCS-N} with $m_2$ input legs and $n_2$ output legs at $p$ legs, then the two diagrams will fuse at the black spiders $S_1,S_2,...,S_N$ and form a diagram of the form Fig. \ref{fig:compositeCS-N} with a total of $m_1+m_2+n_1+n_2-p$ legs since:
\begin{equation*}
    \includegraphics[scale=0.2]{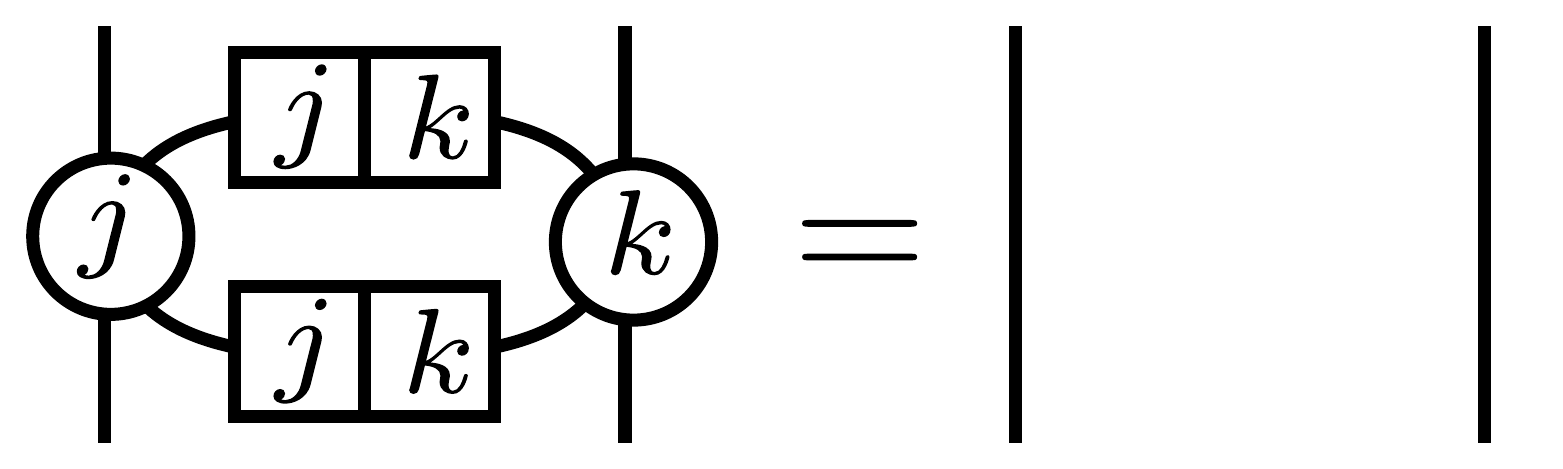}
\end{equation*}
\noindent for any $j,k\in\{1,2,...,N\}$.
\end{proof}

\begin{corollary}
The 0,1-spider of a composite classical structure on $N qubits$ is equivalent (up to scalars) to:
\begin{equation*}
    \includegraphics[scale=0.2]{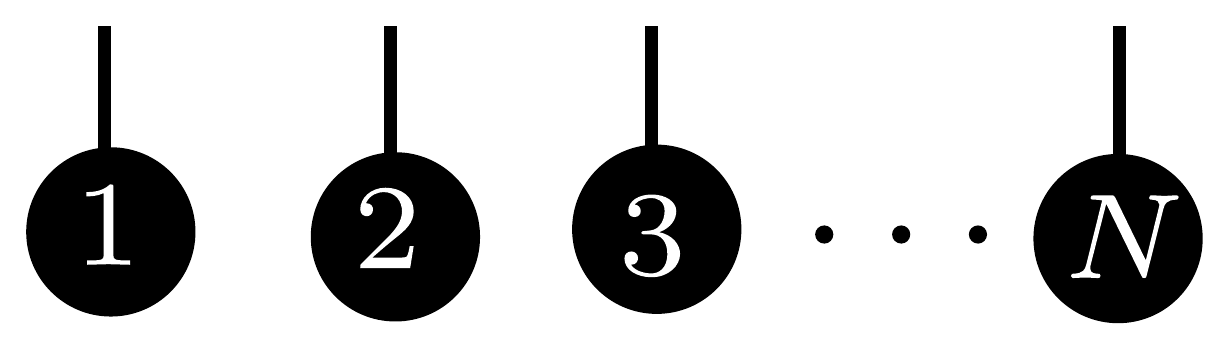}
\end{equation*}
\begin{proof}
From Fig. \ref{fig:compositeCS-N}, the 0,1-spider of a composite classical structure on $N$ qubits is:
\begin{equation*}
    \includegraphics[scale=0.2]{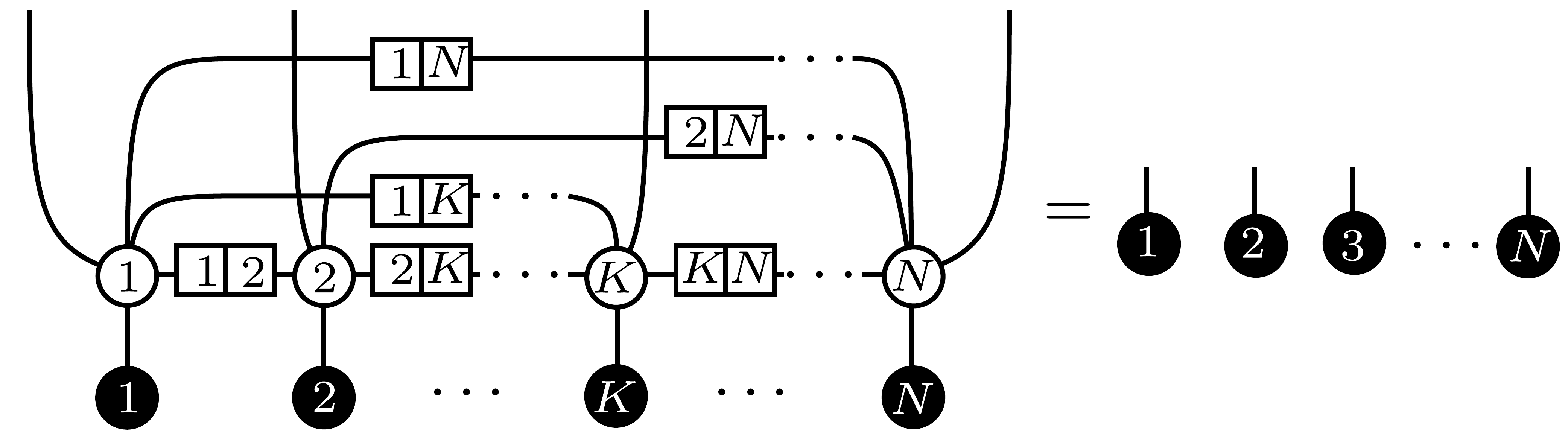}
\end{equation*}
since the white spider labelled $S_j$ copies the black spider labelled $S_j$ by definition. 
\end{proof}
\end{corollary}

\section{Chapter Summary}\label{sec:chapter4-summary}

In this chapter:
\begin{itemize}
    \item We showed how to compose two classical structures on a single qubit to form a composite classical structure on two qubits via two methods:
    \begin{enumerate}
        \item[\textbf{C1}] Separably composing their spiders with matching number of legs, i.e. the spiders composed remain separable or not entangled (see Section \ref{sec:SCCS});
        \item[\textbf{C2}] Joining their spiders with matching number of legs via one of the following connecting wires (see Section \ref{sec:NSCS}):
            \begin{longtable}{c c c c}
            \includegraphics[scale=0.2]{images/symbols/conwire1.pdf} &
            \includegraphics[scale=0.2]{images/symbols/conwire2.pdf} &
            \includegraphics[scale=0.2]{images/symbols/conwire3.pdf} &
            \includegraphics[scale=0.2]{images/symbols/conwire4.pdf}
            \end{longtable}
        to form a non-separable classical structure.  
    \end{enumerate}
    \item We also showed how to construct a composite classical structure on $N$ qubits, where $N>2$, by composing constituents on each pair of qubits in $N$ qubits via Methods \textbf{C1} or \textbf{C2}. 
\end{itemize}

\chapter{Separability of Complementary CS}\label{sec:results}

Classical structures and their spiders provide intuitive images of compatible observables and measurements, and those that are incompatible. The former is depicted by the equation in Eq. \ref{eq:decoherence} which could be interpreted as zero information loss, and the latter could be described by Eq. \ref{eq:complementarity}, where the discontinuity between the input system and the output system can be interpreted as complete information loss on the input system. There is a stronger version of complementarity, called strong complementarity (see Definition \ref{def:strong-complementarity}) \cite{Coecke2012}. An implication of strong complementarity between two classical structures, represented by $\black$ and $\blue$, is the following equation:
\begin{equation}
    \includegraphics[scale=0.2]{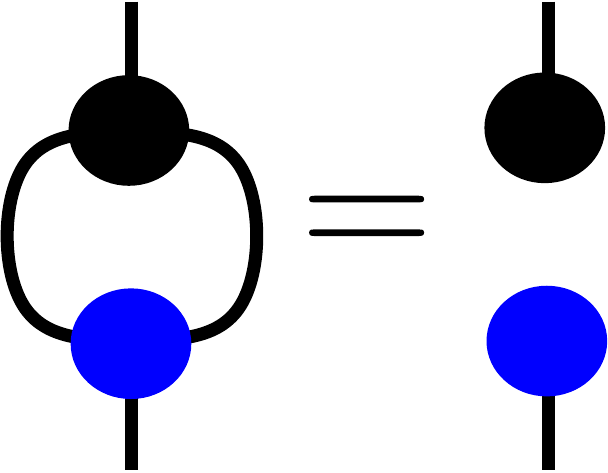}
\end{equation}
which can be thought of as the exact opposite of Eq. \ref{eq:decoherence}.

Previously, we presented two procedures for composing classical structures on a single qubit, namely $\mathcal{X},\mathcal{Y}$ and $\mathcal{Z}$, in order to obtain classical structures on multiple qubits. We devised these procedures so that we may study the entanglement structure of sets of mutually unbiased bases in qubit systems. As was discussed in the previous chapter, the separability of classical structures obtained through these procedure are apparent from the presence of the connecting wires on the legs of their spiders. In this chapter, we shall present the complete sets of complementary classical structures on two and three qubits that can be obtained by composing the complementary classical structures on a single qubit via Methods \textbf{C1} and \textbf{C2} from Section \ref{sec:chapter4-summary}. 

\section{On Notations}

To denote arbitrary spiders, we shall use spiders with nodes that are neither green nor red. Some examples are $\black$,$\blue$,$\purple$,$\pink$,$\orange$, etc. At times, we shall use spiders to represent classical structures and write $\black\in\mathcal{A}$ to mean classical structure $\mathcal{A}$ consists of spiders $\black$.

Within the diagram of a spider belonging to a composite classical structure, there are two types of spiders. The first one is the central spiders which belong to constituents of the composite classical structure. We shall use coloured nodes with their interior filled in along with their outlines. The second one is the spiders which form part of the connecting wires on the legs of the central spiders. For these spiders, we shall use nodes with only their outlines coloured while their interiors are blank. The colour of the outline of such spider depends on the central spider to which it is connected. For example, if the central spider is $\black$, then the spider on its leg is $\blackb$.

Fig. \ref{fig:gen-CCS2Q} is an example of an arbitrary composite classical structure on two qubits.
\begin{figure}[h]
    \centering
    \includegraphics[scale=0.2]{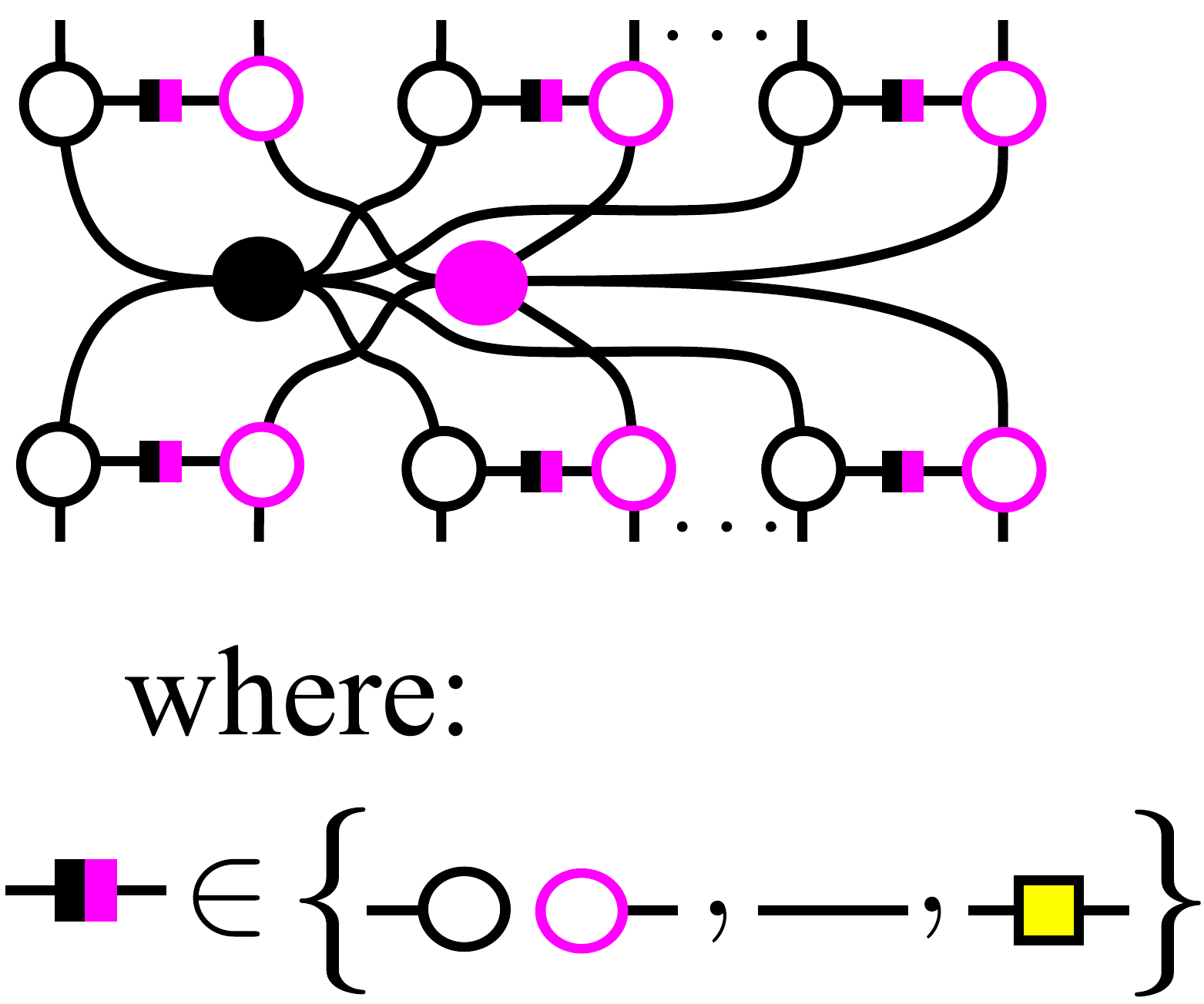}
    \caption{An arbitrary composite classical structure on two qubits.}
    \label{fig:gen-CCS2Q}
\end{figure}
The first option for $\cwbkp$ is given to include separably composed classical structures. In this case, $\blackb$ and $\pinkb$ can belong to either $\mathcal{X}$ or $\mathcal{Z}$. Due to Eqs. \ref{eq:zx1-fuse-green} and \ref{eq:zx2-fuse-red}, the leg nodes would simply disappear. If the composite classical structure is non-separable, then
the diagram:
\begin{equation*}
    \includegraphics[scale=0.2]{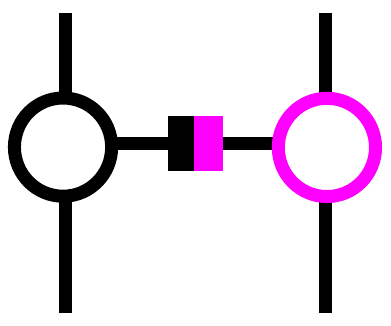}
\end{equation*}
is a connecting wire. This means that $\cwbkp$, $\blackb$, and $\pinkb$ depend on $\black$ and $\pink$ as prescribed in Table \ref{tab:connecting-wire}. For example, if $\black\in\mathcal{X}$ and $\pink\in\mathcal{Y}$, then $\blackb,\pinkb\in\mathcal{Z}$ and $\cwbkp=\had$.

\section{Complementarity Diagrams in Qubit Systems}\label{sec:CD-intro} 

Two classical structures are complementary if they satisfy the following equation:
\begin{equation}\label{eq:complementarity}
    \includegraphics[scale=0.2]{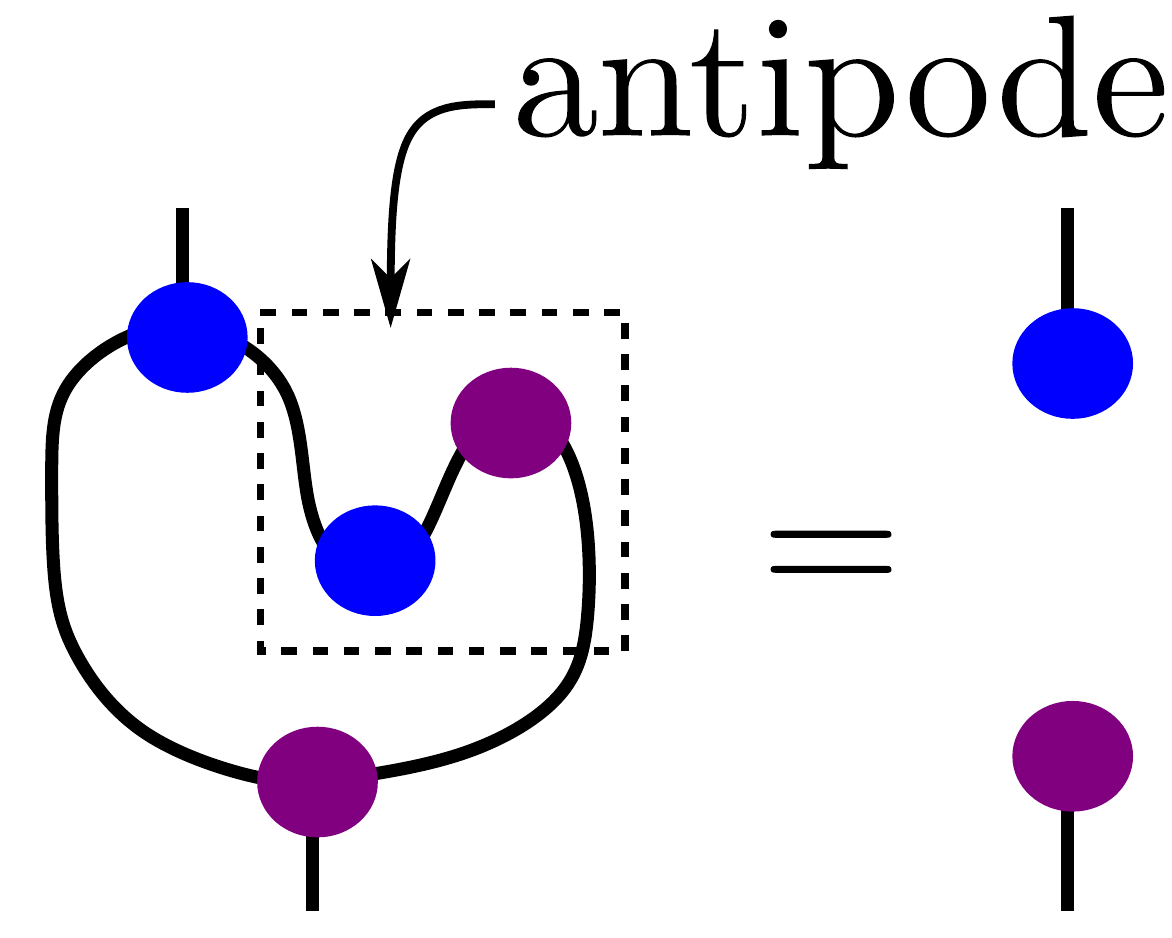}
\end{equation}
We call the above equation the \textbf{complementarity condition}, and the LHS diagram in the equation as the \textbf{complementarity diagram}, abbreviated as CD, between $\blue$ and $\purple$. 

From Eq. \ref{eq:complementarity}, we define the antipode between $\blue$ and $\purple$ as follows:
\begin{equation}
    \includegraphics[scale=0.2]{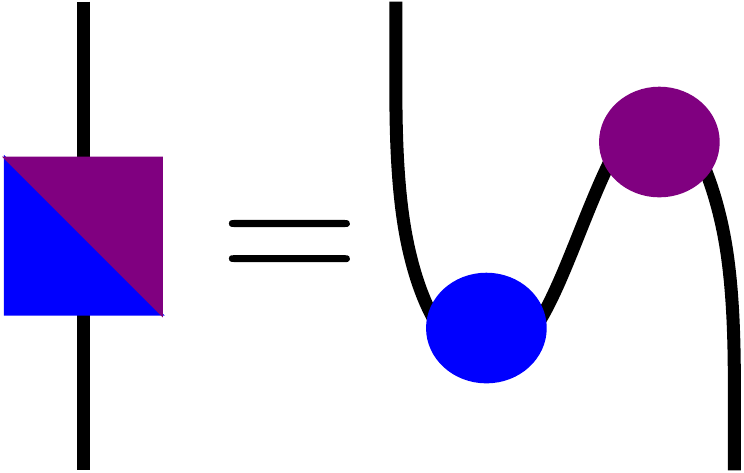}
\end{equation}
\begin{proposition}\label{prop:antipode}
Suppose $\mathcal{A},\mathcal{B},\mathcal{C},\mathcal{D}\in\{\mathcal{X},\mathcal{Y},\mathcal{Z}\}$. Let $\blue\in\mathcal{A}$, $\pink\in\mathcal{B}$, $\purple\in\mathcal{C}$ and $\orange\in\mathcal{D}$. For a complementarity diagram between a composite classical structure with constituents $\mathcal{A}$ and $\mathcal{B}$, and a composite classical structure with constituents $\mathcal{C}$ and $\mathcal{D}$, the following is true:
\begin{equation}
    \includegraphics[scale=0.2]{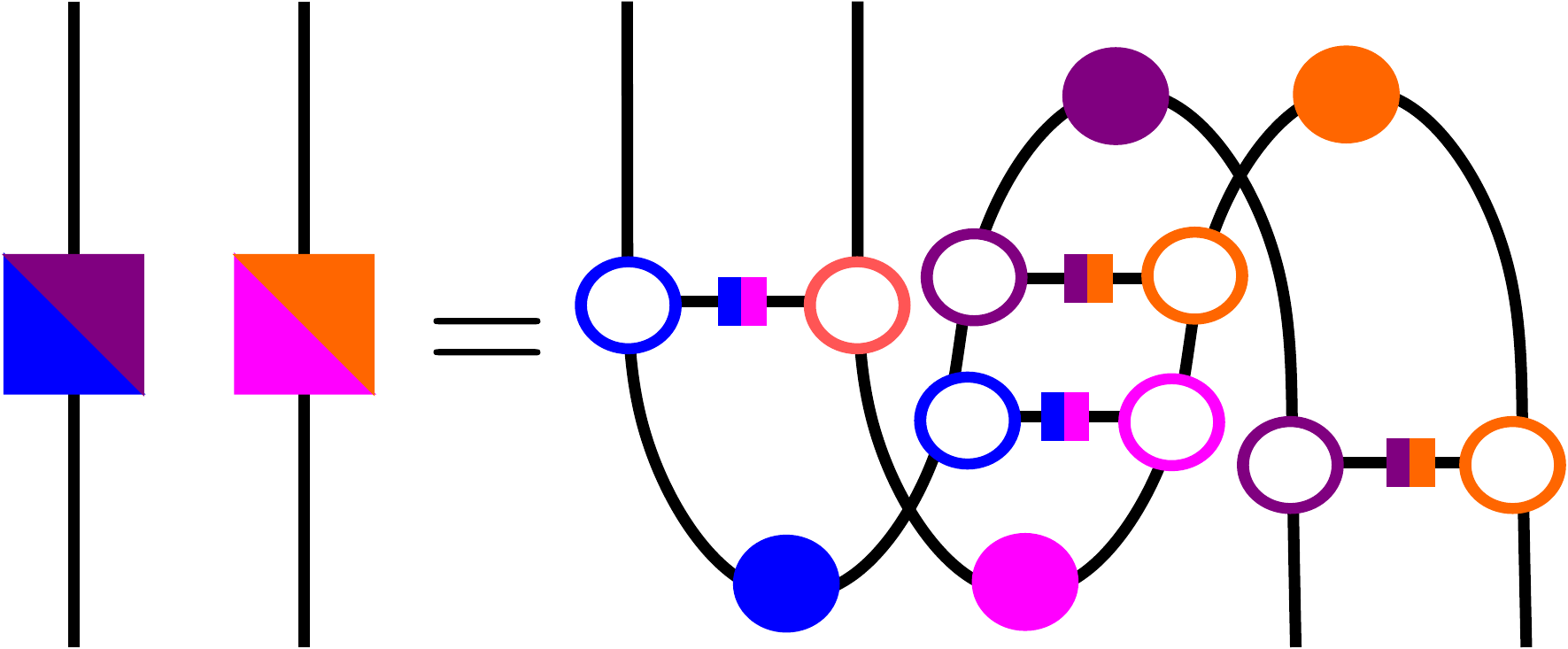}
\end{equation} 
\begin{proof}
For the 0,2-spiders of $\mathcal{X,Y}$ and $\mathcal{Z}$, we have:
\begin{equation}
\includegraphics[scale=0.2]{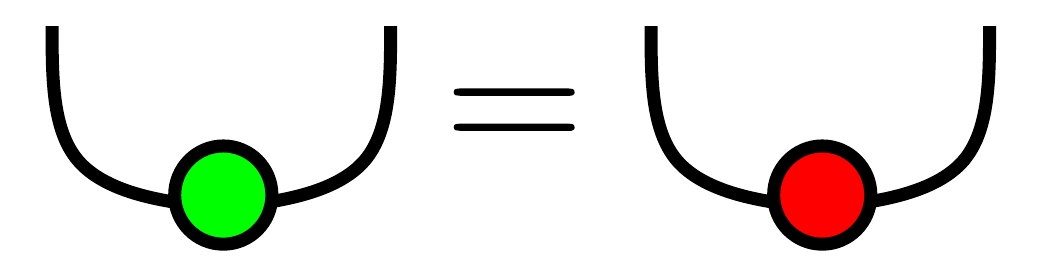}\label{eq:cupZX}
\end{equation}
\begin{equation}
\includegraphics[scale=0.2]{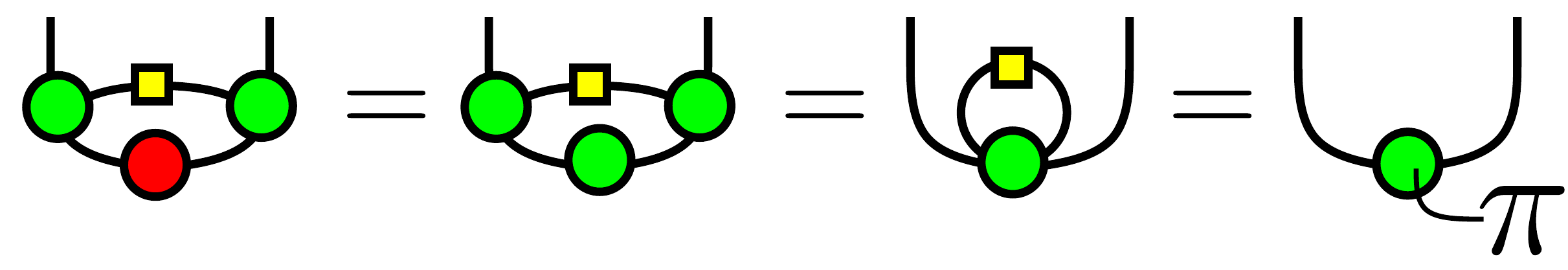}\label{eq:cupY}
\end{equation}
Therefore:
\begin{equation}
\includegraphics[scale=0.2]{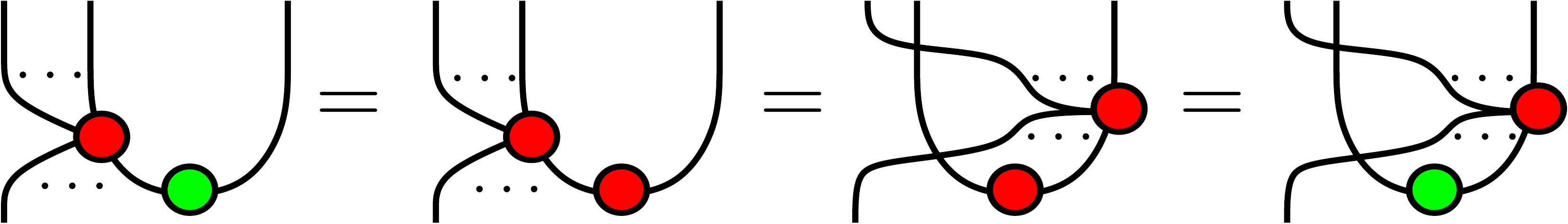}
\end{equation}
\begin{equation}
\includegraphics[scale=0.2]{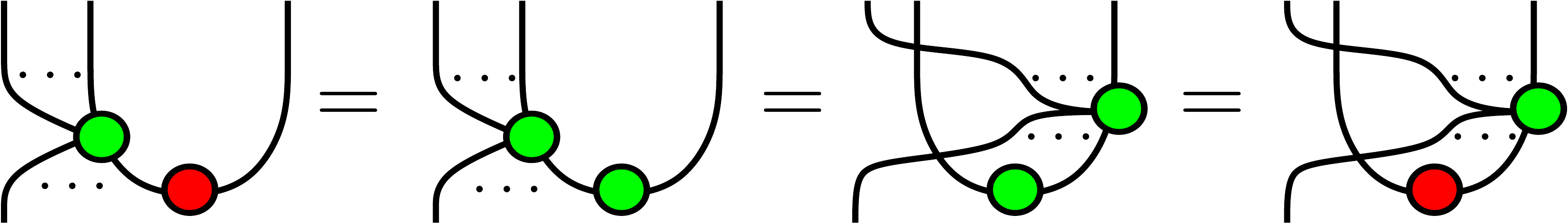}
\end{equation}
\begin{equation}
\includegraphics[scale=0.2]{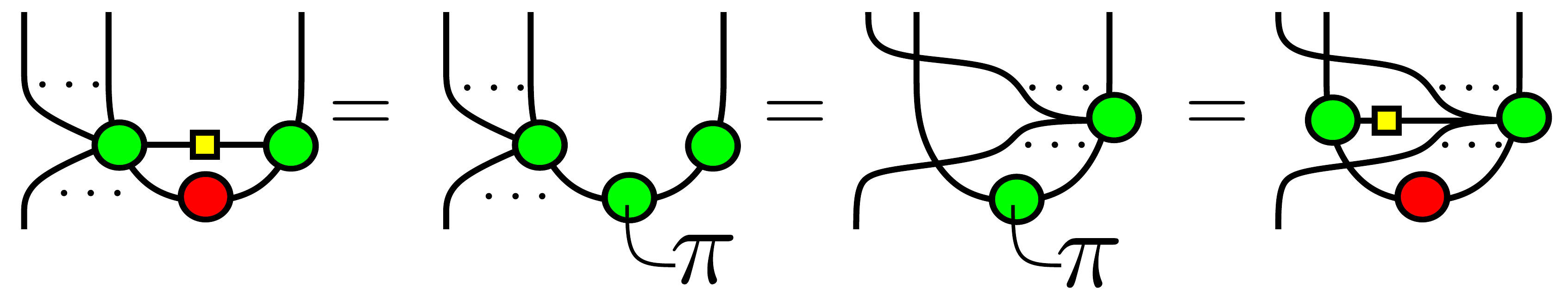}
\end{equation}
This means that for a classical structure that is either $\mathcal{X},\mathcal{Y}$ or $\mathcal{Z}$, represented by $\black$, we have:
\begin{equation}
    \includegraphics[scale=0.2]{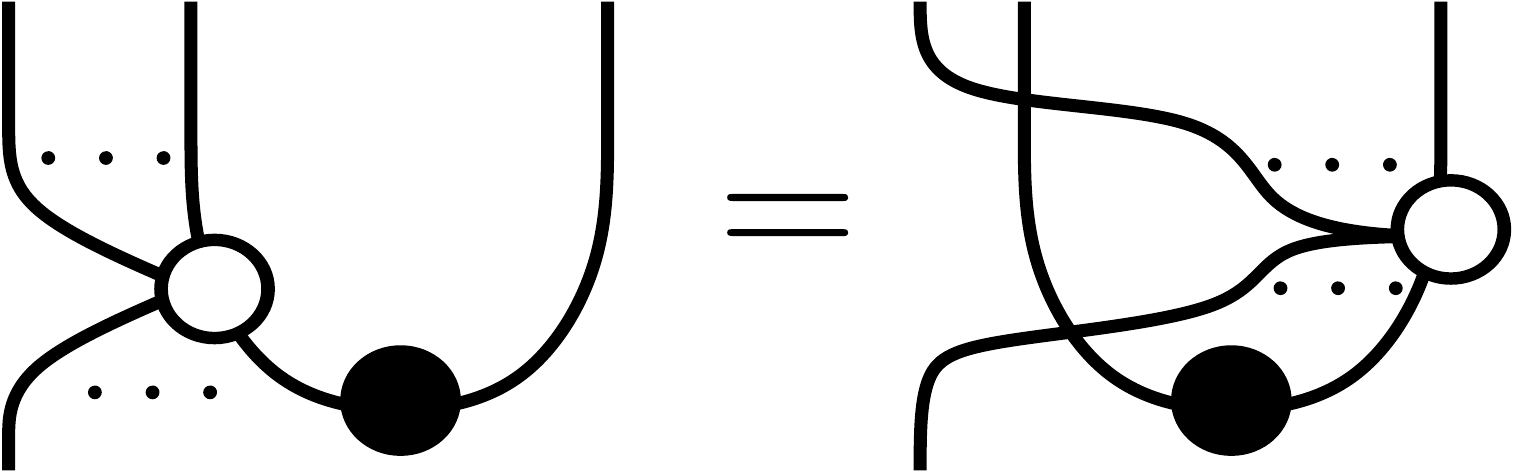}
\end{equation}
Thus:
\begin{equation*}
    \includegraphics[scale=0.2]{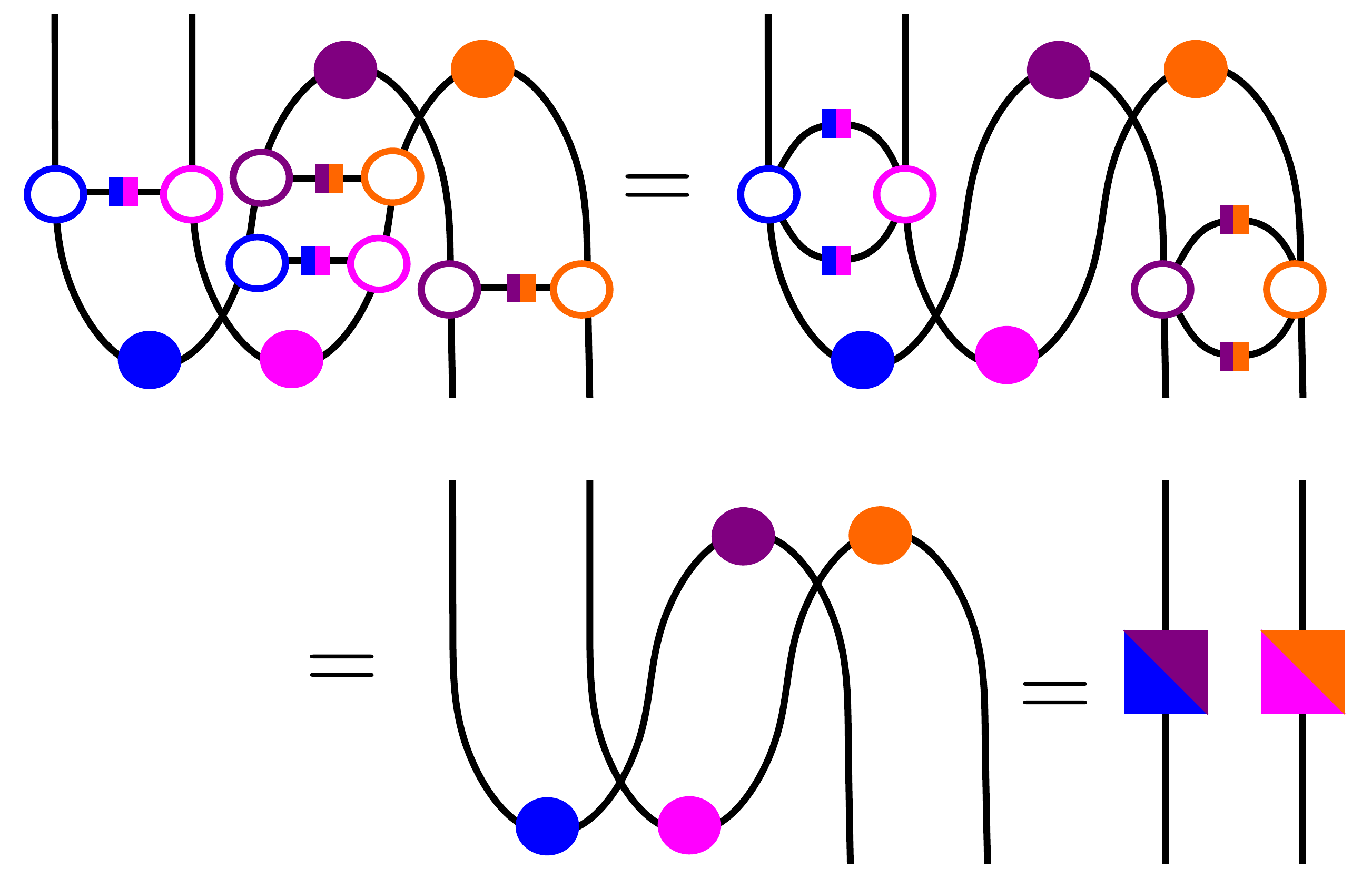}
\end{equation*}
\end{proof}
\end{proposition}

\begin{proposition}\label{prop:fuse-antipode}
Eqs \ref{eq:cupZX} and \ref{eq:cupY} also imply that a complementarity diagram between a pair of classical structures in $\{\mathcal{X},\mathcal{Y},\mathcal{Z}\}$ can be simplified as follows:
\begin{equation}\label{eq:fuse-antipode}
    \includegraphics[scale=0.2]{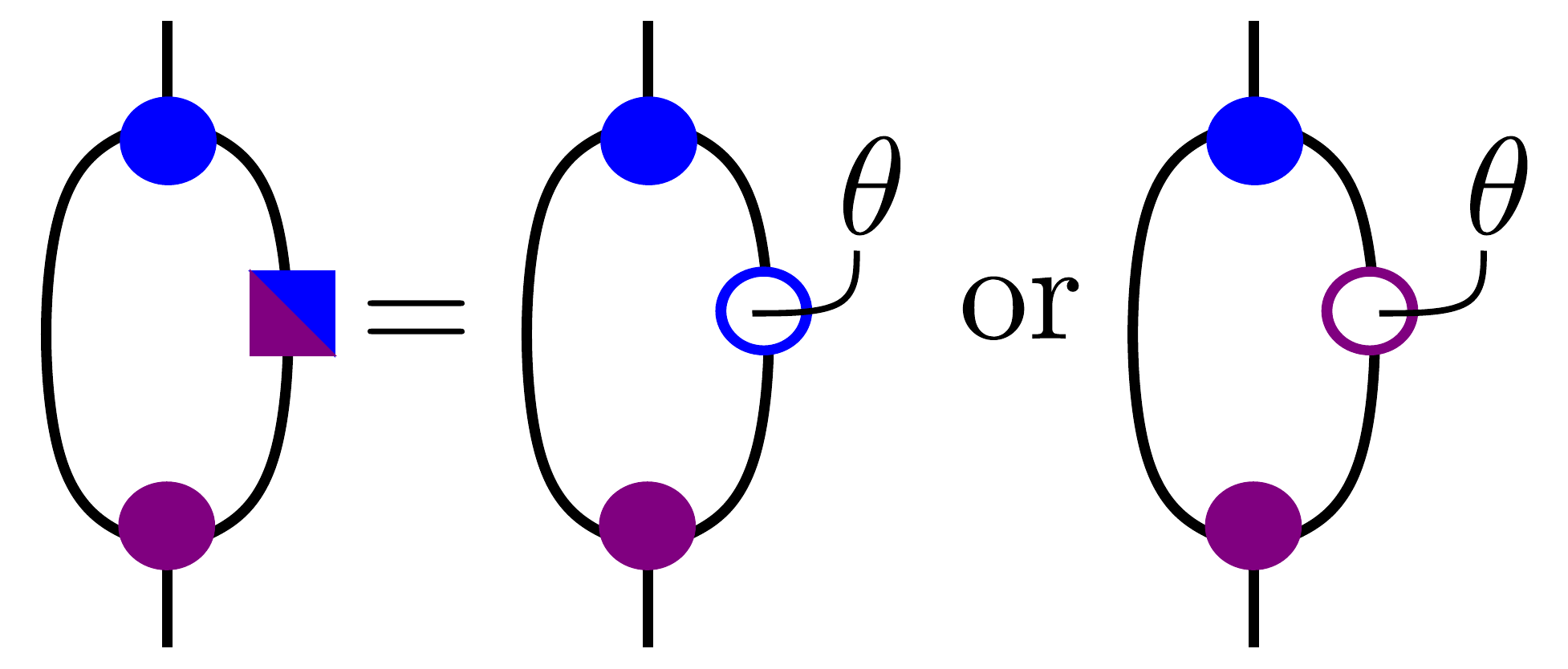}
\end{equation}
where $\theta=0$ if $\blue$ and $\purple$ both belong to either $\mathcal{X}$ or $\mathcal{Z}$ or they both belong to $\mathcal{Y}$,
%, which means that both options on RHS are equal to each other. LHS is equal to the first option on RHS and $\theta=\pi$ if $\blue\in\mathcal{Y}$ and $\purple$ belong to either $\mathcal{X}$ or $\mathcal{Z}$, and LHS is equal to the second option on RHS 
and $\theta=\pi$ if $\blue$ belongs to either $\mathcal{X}$ or $\mathcal{Z}$ and $\purple\in\mathcal{Y}$ or vice versa.
\begin{proof}
If $\blue$ and $\purple$ belong to either $\mathcal{X}$ or $\mathcal{Z}$, then both sides of Eq. \ref{eq:fuse-antipode} are equal to:
\begin{equation*}
    \includegraphics[scale=0.2]{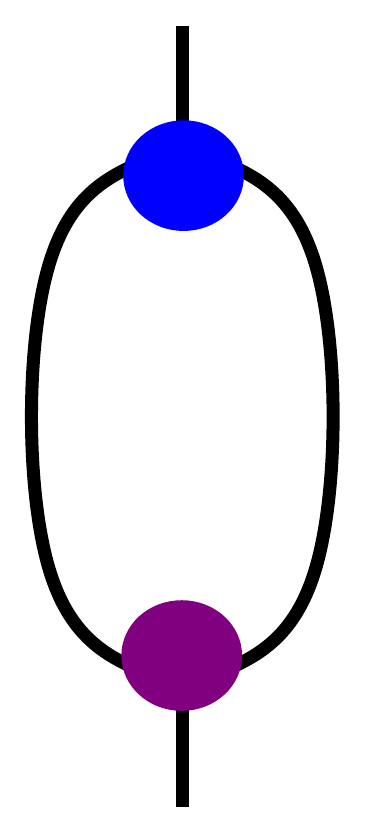}
\end{equation*}
due to eq. \ref{eq:cupZX} for LHS and due to the fact that the 1,1-spider with zero phase of any classical structure is the identity for RHS.

If both $\blue$ and $\purple$ belong to $\mathcal{Y}$, then the antipode must be the identity as the antipode between any classical structure and itself is the identity.

If $\blue\in\mathcal{Y}$ and $\purple$ belong to either $\mathcal{X}$ or $\mathcal{Z}$, then LHS of eq. \ref{eq:fuse-antipode} is equal to:
\begin{equation*}
    \includegraphics[scale=0.2]{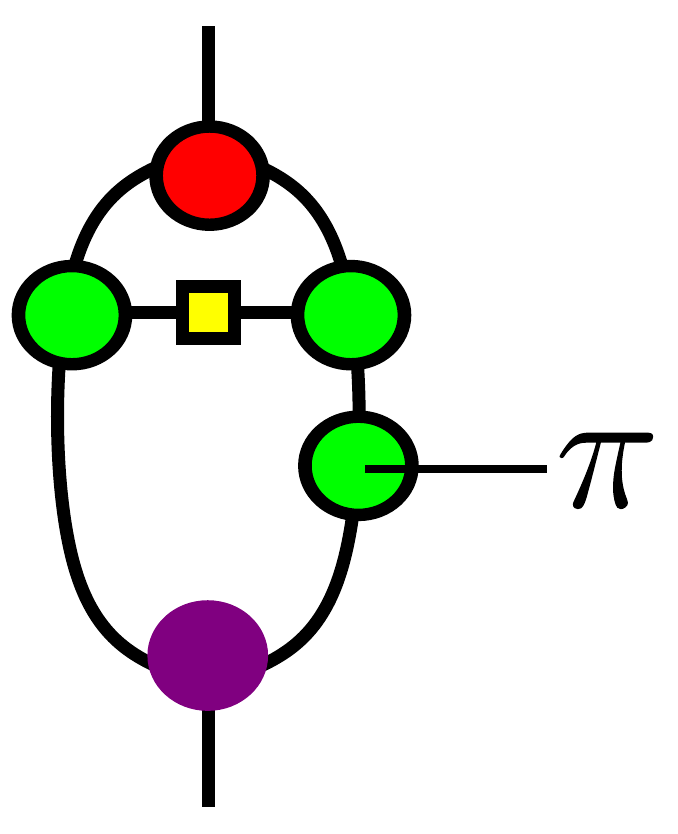}
\end{equation*}
If $\purple\in\mathcal{Y}$ and $\blue$ either belongs to $\mathcal{X}$ or $\mathcal{Z}$, then LHS of eq. \ref{eq:fuse-antipode} is equal to a similar diagram as the adjoint of the diagram above, except the purple node is replaced with a blue node.
\end{proof}
\end{proposition}

As a consequence of Propositions \ref{prop:antipode} and \ref{prop:fuse-antipode}, there are three types of complementarity diagrams we may obtain between composite classical structures on two qubits that are produced by Method \textbf{C1} or \textbf{C2}, or both. These are the diagrams below:
\begin{longtable}{ccc}
    \includegraphics[scale=0.2]{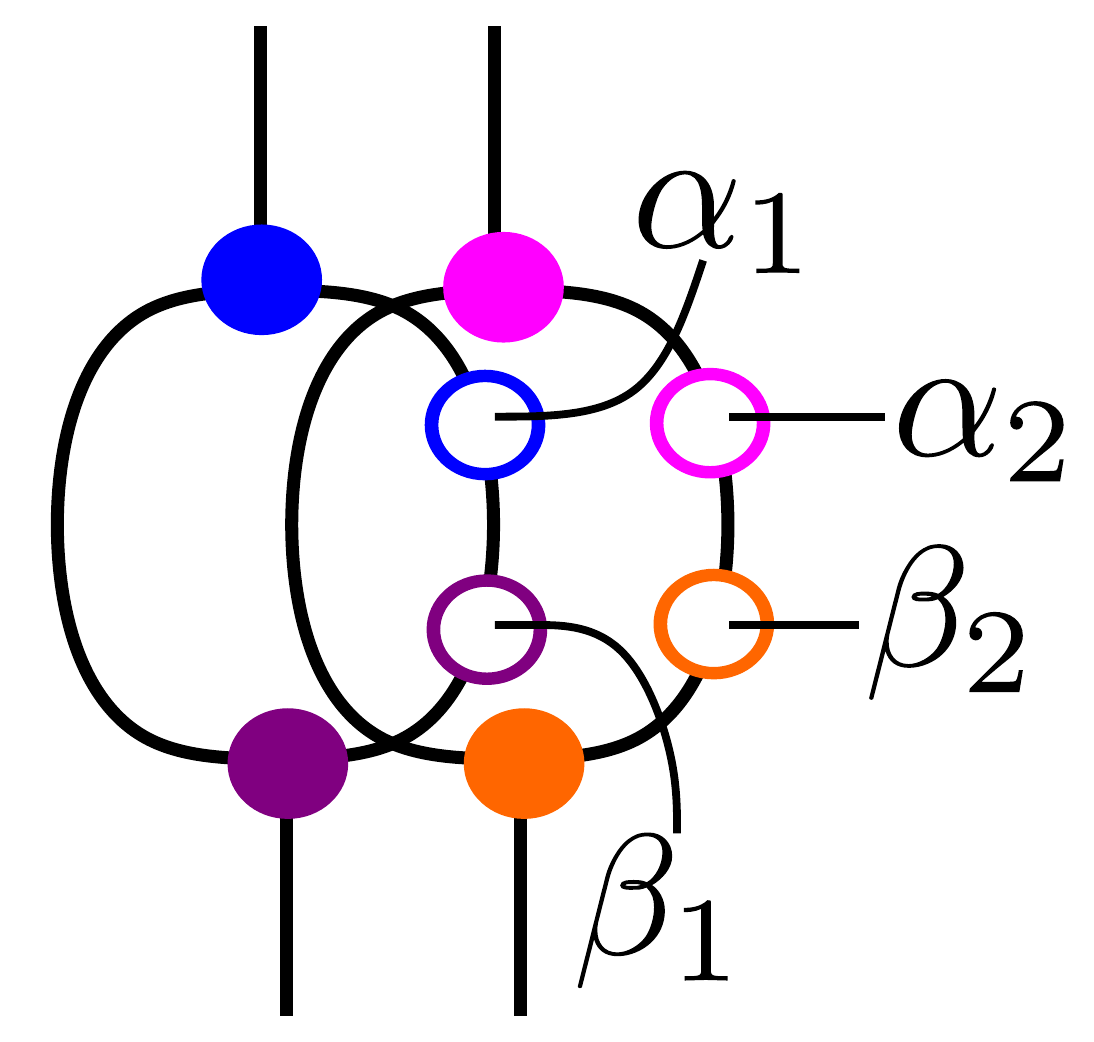} &
    \includegraphics[scale=0.2]{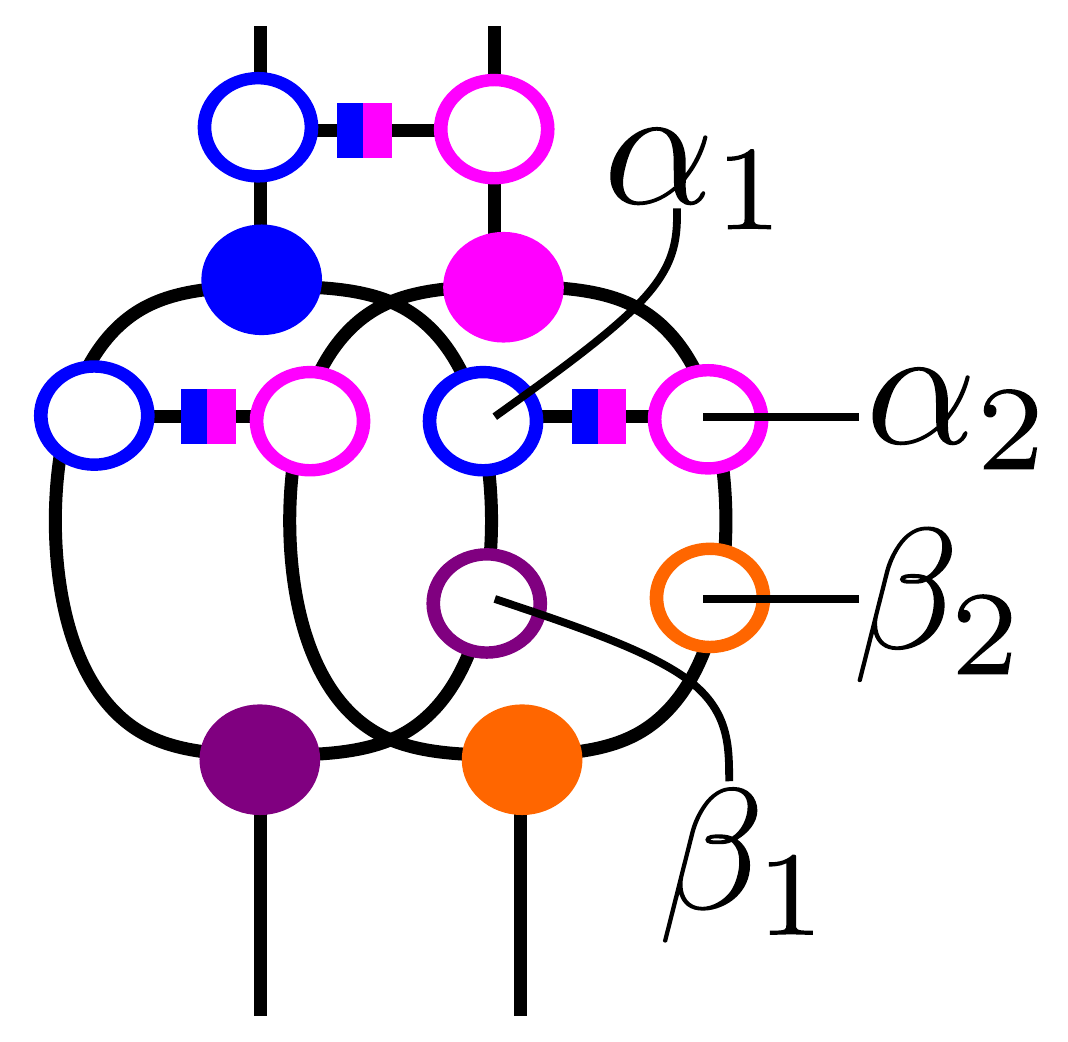} &
    \includegraphics[scale=0.2]{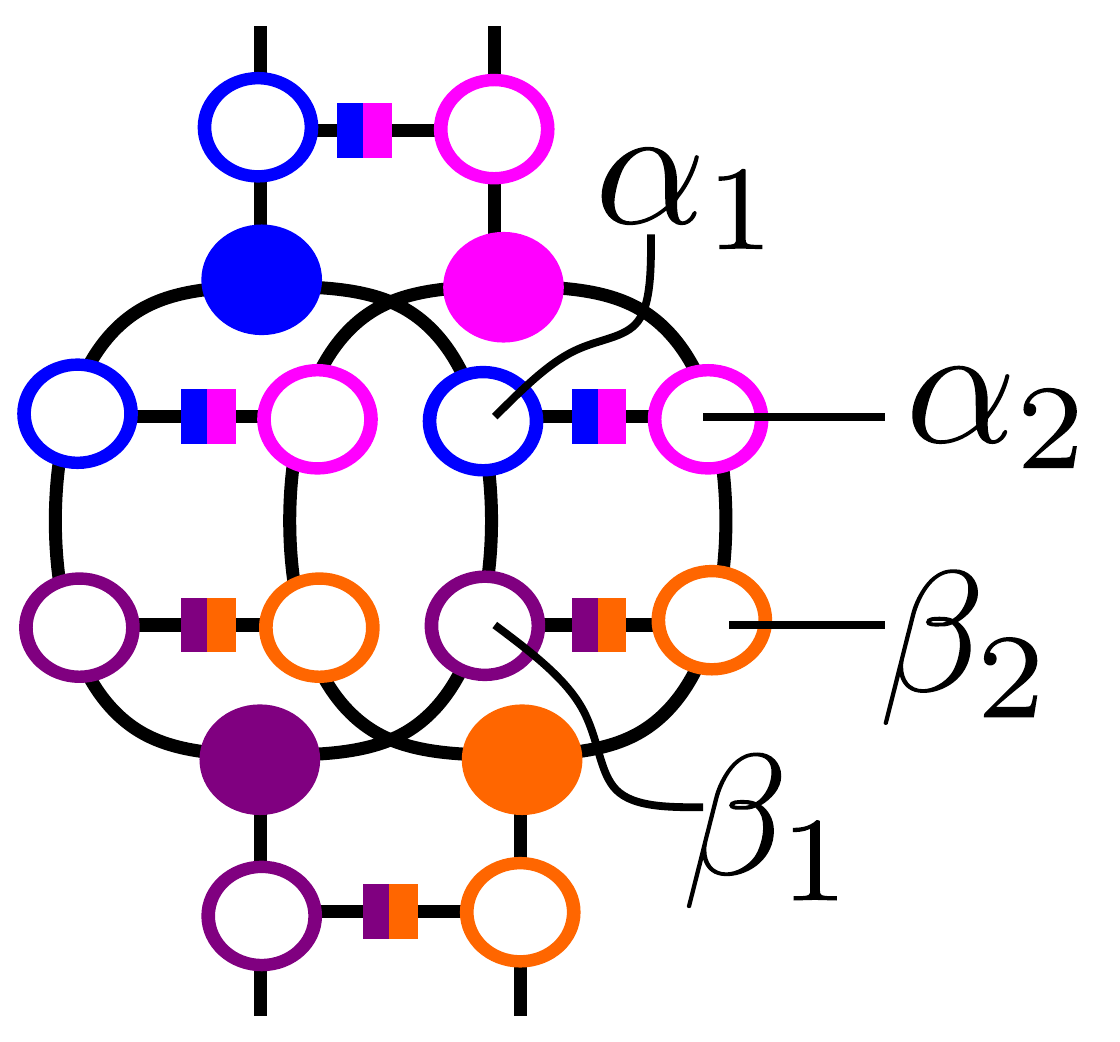}
\end{longtable}
\noindent where $\alpha_i,\beta_j\in\{0,\pi\}$ for $i,j\in\{1,2\}$. 

Each of the diagrams above satisfies the complementarity condition if it is equal to the following diagram:
\begin{equation*}
    \includegraphics[scale=0.2]{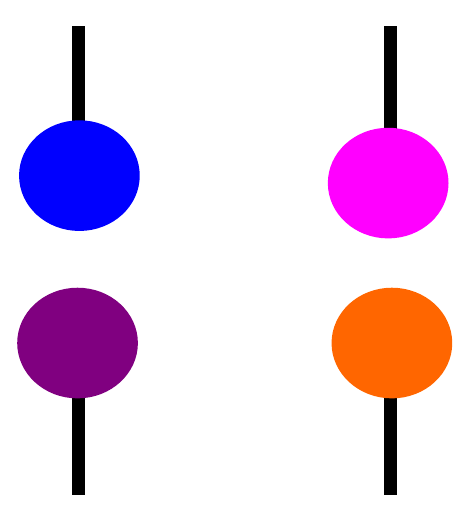}
\end{equation*}
However, we will not be able to rewrite the preceding diagrams using the rules of ZX-calculus unless we fix the classical structures to which $\blue,\pink,\purple$ and $\orange$ belong. This would be quite tedious since we would have to consider all the possibilities for each spider and connecting wires in the diagrams above, and even more so when determining the complementarity between a pair of composite classical structures on $N$ qubits where $N>2$.

A simpler way to go about this is to partition all the complementarity diagrams into equivalence classes. In doing so, we only need to check the complementarity of a representative diagram for each class. 

\begin{theorem}\label{thm:comp-diagram-equivalence}
For $N$ qubits, let $\mathsf{D}_N$ be the set of all complementarity diagrams between a pair of composite classical structures on $N$ qubits as defined in Definition \ref{def:composite-CS}. Suppose the relation $r\in\mathsf{D}_N\times\mathsf{D}_N$ is defined as follows: $D_1rD_2$ if and only if there exists a unitary $U=P\circ V_1\otimes V_2\otimes\cdots\otimes V_N$ where $V_j\in\{\text{id}_{\mathbb{C}^2},\had\}$ and $P$ is a permutation on $N$ qubits, such that $D_2=U^\dagger\circ D_1\circ U$ or $D_2=U^\dagger\circ D_1^\dagger\circ U$. Then $r$ is an equivalence. 
\begin{proof}
For a diagram $D\in\mathsf{D}_N$, $DrD$ since the identity of $(\mathbb{C}^2)^{\otimes N}$ for any $N\in\mathbb{N}$ is unitary.

Suppose for $D_1,D_2\in\mathsf{D}$, $D_1rD_2$. Then there is some unitary $U$ on $N$ qubits such that $D_2=U^\dagger\circ D_1\circ U$ or $D_2=U^\dagger\circ D_1^\dagger\circ U$. This means that $D_1=U\circ D_2 U^\dagger$ or $D_1=U\circ D_2^\dagger \circ U^\dagger$. So, $D_2rD_1$.

Suppose for $D_1,D_2,D_3\in\mathsf{D}$, $D_1rD_2$ and $D_2rD_3$. Then $D_2=U_1^\dagger\circ D_1\circ U_1$ or $D_2=U_1^\dagger\circ D_1^\dagger\circ U_1$, or $D_3=U_2^\dagger\circ D_2\circ U_2$ or $D_3=U_2^\dagger\circ D_2^\dagger\circ U_2$ for some unitaries $U_1, U_2$ on $N$ qubits. So  $D_3$ must be equal to either $U_3^\dagger\circ D_1\circ U_3$ or $U_3^\dagger\circ D_1^\dagger\circ U_3$ where $U_3\in\{U_2\circ U_1,U_2\circ U_1^\dagger, U_2^\dagger\circ U_1,U_2^\dagger\circ U_1^\dagger\}$. So, $D_1rD_3$.
\end{proof}
\end{theorem}

The proof for Theorem \ref{thm:comp-diagram-equivalence} uses the property of the unitary instead of its specific definition in the statement, i.e. as a composition of a permutation on $N$ qubits, $\text{id}_{\mathbb{C}^2}$ and $\had$. So the theorem would still hold even if we are not as precise about the unitary as we were in the theorem's statement. However, since we have limited the constituents of the composite classical structures in the present work to $\mathcal{X},\mathcal{Y}$ and $\mathcal{Z}$, only a unitary as defined in Theorem \ref{thm:comp-diagram-equivalence} can transform a diagram in $\mathsf{D}_N$ and into another diagram in the set.

\begin{definition}\label{def:slice}
Let $\mathsf{D}_N$ be defined as in Theorem \ref{thm:comp-diagram-equivalence}. A diagram in $\mathsf{D}_N$ can be split into $N$ slices. Each slice takes the following form:
\begin{equation*}
    \includegraphics[scale=0.2]{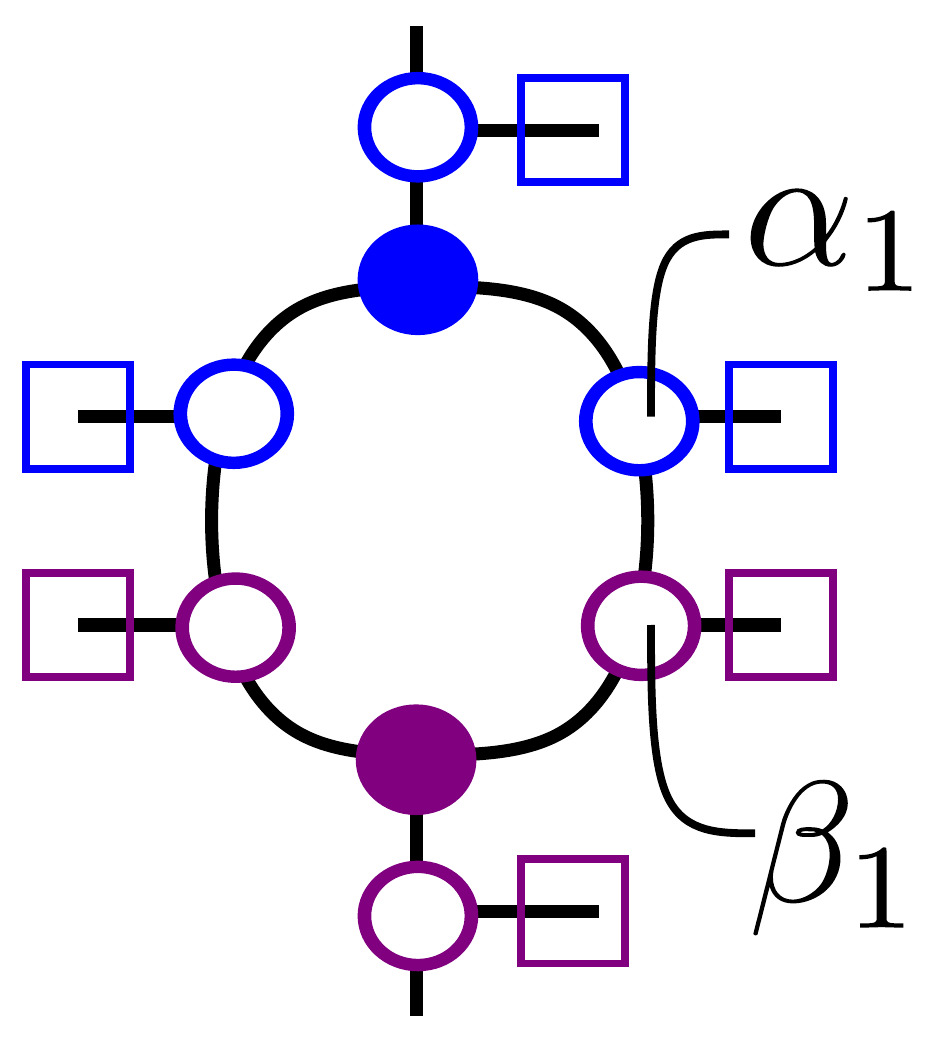}
\end{equation*}
\end{definition}

We have used !-boxes (see Definition \ref{def:!-boxes}) to indicate the connecting wires on the top and bottom spiders in the diagram. Above, !-boxes with the same colour must contain the same number of wires. 

Suppose for $D_1,D_2\in\mathsf{D}_N$, $D_1=U^\dagger\circ D_2\circ U$ where $U=V_1\otimes\cdots\otimes V_N$ such that for some $m\in\{1,...,N\}$, $V_m=\had$ and $V_j=\text{id}_{\mathbb{C}^2}$ for $j\not=m$. Then $D_1$ is equivalent to $D_2$. So the slice above can be represented by the following diagrams:
\begin{equation*}
    \includegraphics[scale=0.2]{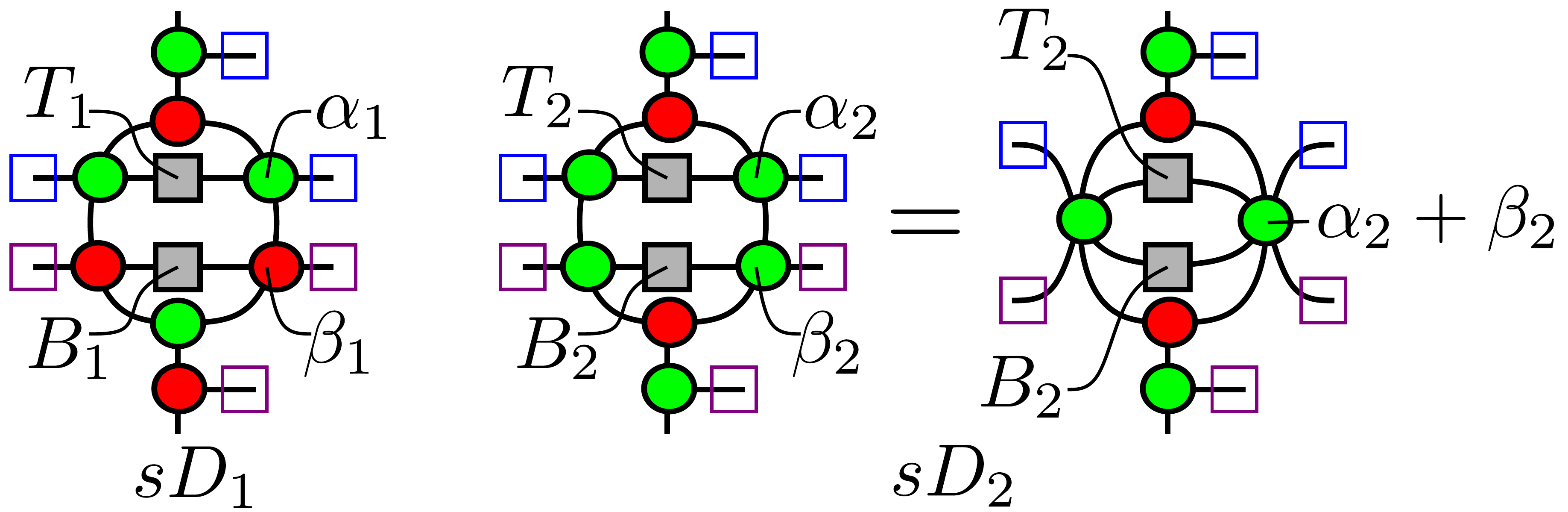}
\end{equation*}
where:
\begin{equation*}
    \includegraphics[scale=0.2]{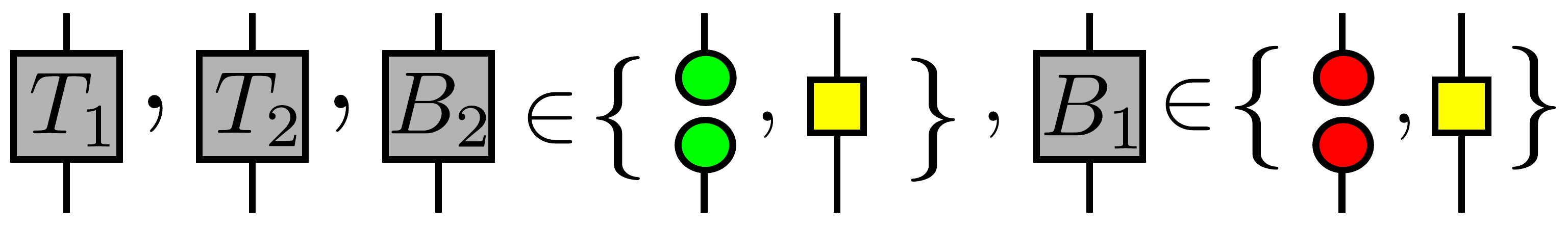}
\end{equation*}
For $sD_1$ we have allowed for bottom spider to take the form:
\begin{equation*}
    \includegraphics[scale=0.2]{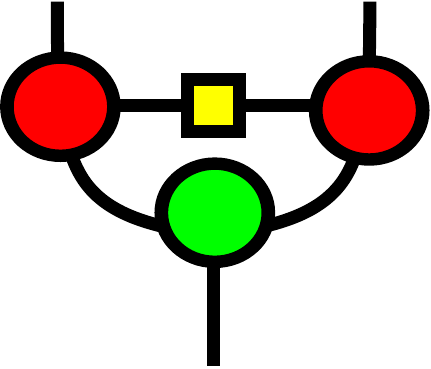}
\end{equation*}
However, we have not introduced a classical structure with such a spider. We do this because the connecting wires on the top and bottom spiders are not interchangeable. So $sD_1$ includes the following diagram:
\begin{equation*}
    \includegraphics[scale=0.2]{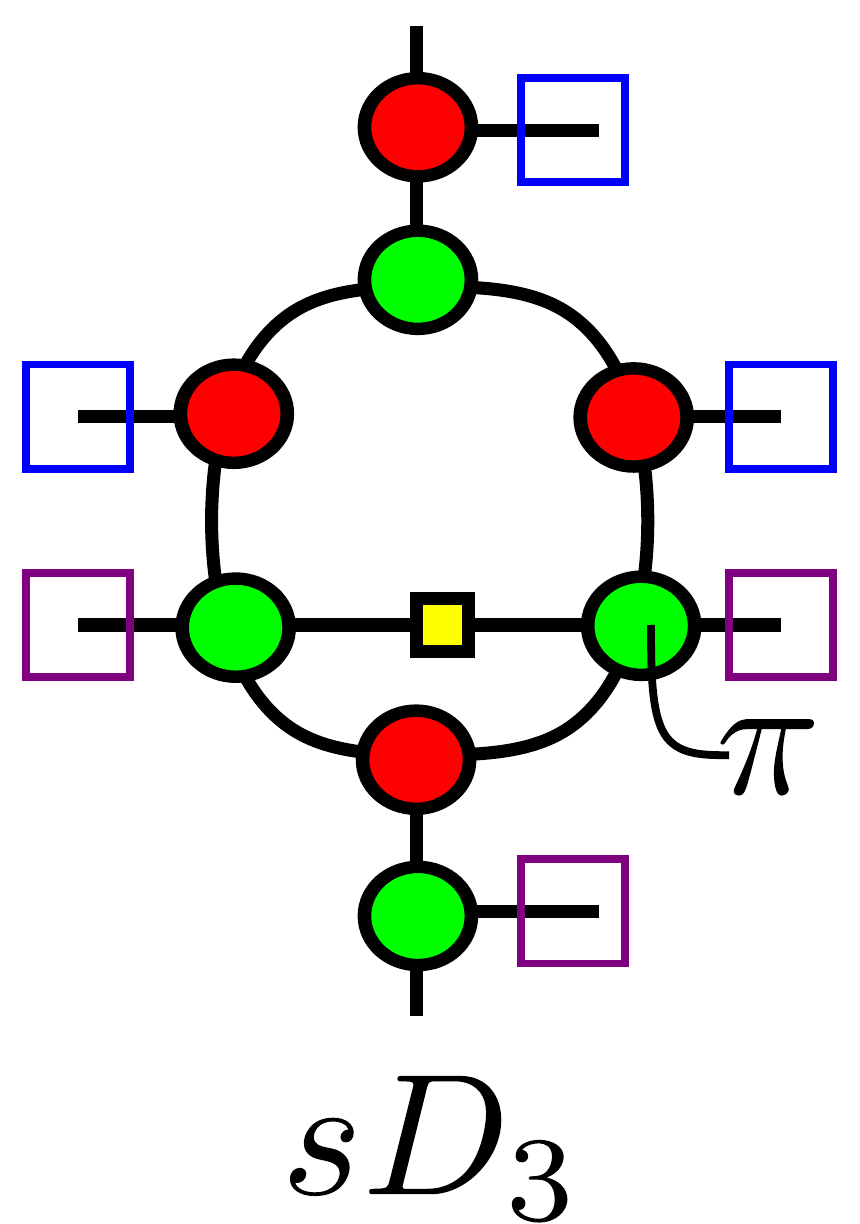}%unitary-proof-type-3
\end{equation*}

There are three forms which $sD_1$ can take:
\begin{equation*}
    \includegraphics[scale=0.2]{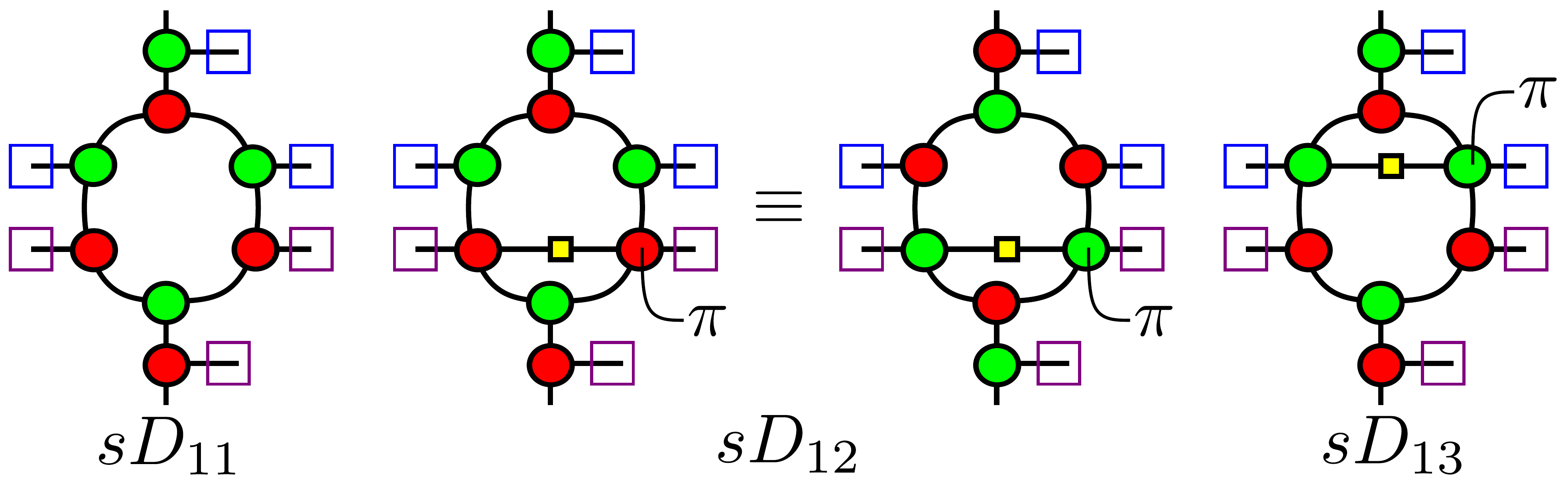}
\end{equation*}

And $sD_2$ can take two forms:
\begin{equation*}
    \includegraphics[scale=0.2]{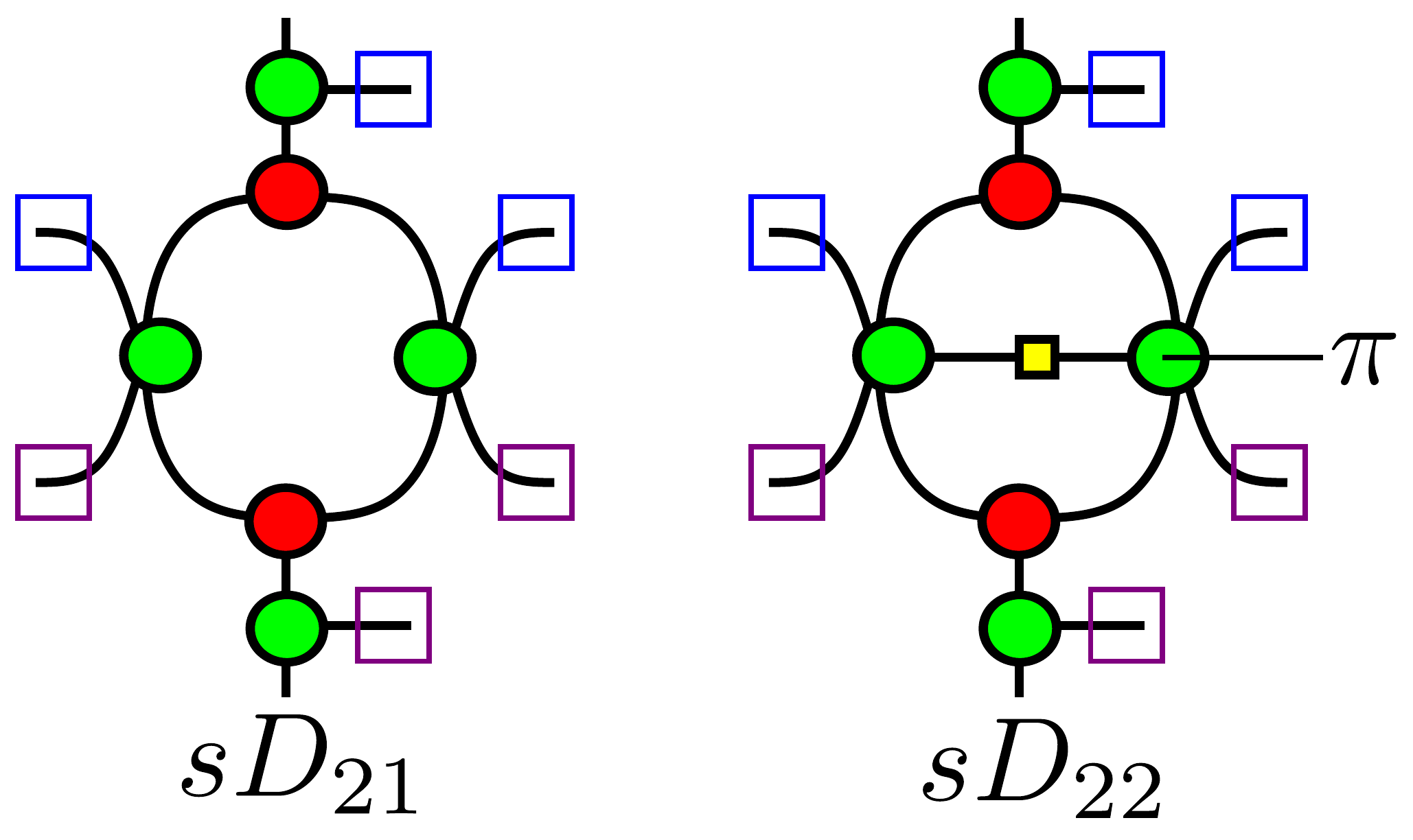}
\end{equation*}

To check complementarity between all pairs of composite classical structures on $N$ qubits, it is sufficient to look at composite complementarity diagrams consisting of $sD_{11},sD_{12},sD_{13},sD_{21}$ and $sD_{22}$.

\begin{proposition}\label{prop:slice-special-case}
Let $m$ be the number of wires contained in the blue !-box and $n$ be the number wires contained in the purple !-box. Then $sD_{11}=sD_{12}$ if $m=0$ and $sD_{11}=sD_{13}$ if $n=0$.
\begin{proof}
Suppose $m=0$. Then $sD_{11}$ can be rewritten as:
\begin{equation*}
    \includegraphics[scale=0.2]{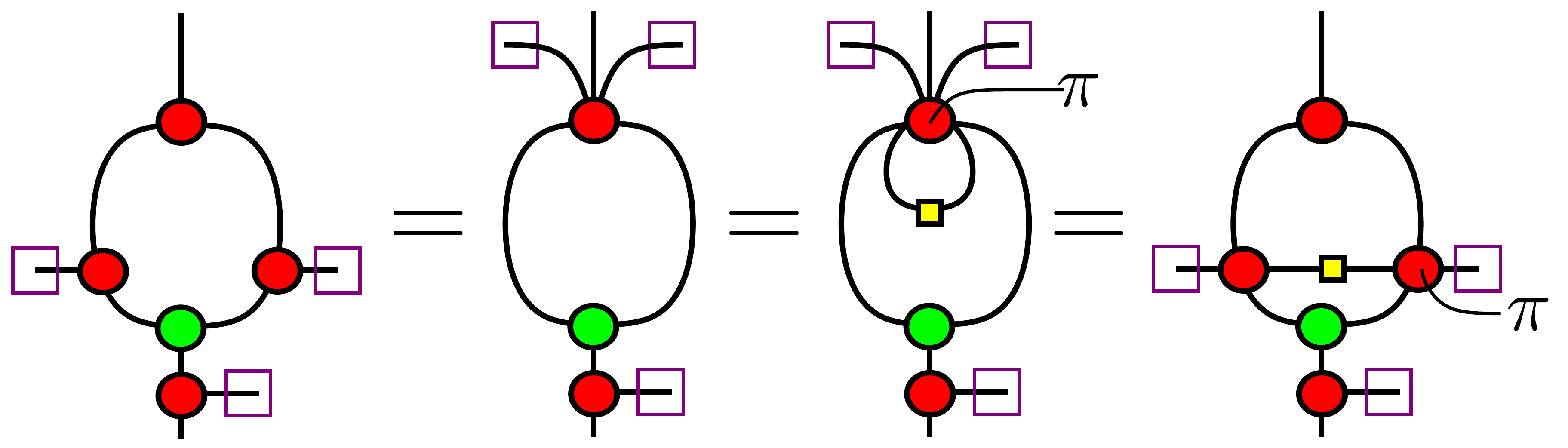}
\end{equation*}
$sD_{11}=sD_{13}$ if $n=0$ can be proven in a similar way.
\end{proof}
\end{proposition}

Proposition \ref{prop:slice-special-case} tells us that $sD_{12}$ if $m=0$ and $sD_{13}$ if $n=0$ can be subsumed by $sD_{11}$.  

\begin{definition}
For slices $sD_{11}, sD_{12}, sD_{13}, sD_{21}$ and $sD_{22}$, let $m$ be the number of wires contained in the blue !-box and $n$ be the number of wires contained in the purple !-box.
The generating set of $\mathsf{D}_N$, denoted as $\mathsf{G}_N$, is a subset of $\mathsf{D}_N$ such that each slice of a complementarity diagram in $\mathsf{G}_N$ is either 
$sD_{11}$, $sD_{21}$ and $sD_{22}$ for any $m,n\in\mathbb{N}\cup\{0\}$, $sD_{12}$ for $m\in\mathbb{N}, n\in\mathbb{N}\cup\{0\}$, and $sD_{13}$ for $n\in\mathbb{N}, m\in\mathbb{N}\cup\{0\}$.  
\end{definition}

To find maximal complete sets of mutually complementary classical structures on $N$ qubits, we shall form a graph where the vertices are the composite classical structures on $N$ qubits. A pair of vertices/classical structures have an edge between them if they are complementary. However, it is highly inefficient to check the complementarity of every pair of classical structures on $N$ qubits. For example, there are 171 unique pairs of composite classical structure on two qubits, and there are 23,435 unique pairs of composite classical structures on three qubits. So instead, we shall check the complementarity condition of a smaller set of complementarity diagrams which generate all complementarity diagrams on $N$ qubits, i.e. complementarity diagrams consisting of the slices mentioned previously. Then we shall encode the information of those complementarity diagrams which satisfy the complementarity condition as matrices called `names' (to be defined), and gathered to form a test set.   

\begin{definition}\label{def:name-CS}
Let $\mathcal{A}$ be a composite classical structure on $N$ qubits. The name of $\mathcal{A}$ is a $(N+1)\times N$ matrix, $A$, with entries $A_{pq}$, defined as follow:
\begin{itemize}
    \item The first row of $A$ consists of information on the constituents of $\mathcal{A}$, i.e.
    \begin{equation}
        A_{1q}=
        \begin{cases}
        0 & \text{if the } q \text{-th constituent of } \mathcal{A} \text{ is } \mathcal{X}\\ 
        1 & \text{if the } q \text{-th constituent of } \mathcal{A} \text{ is } \mathcal{Y}\\
        Z & \text{if the } q \text{-th constituent of } \mathcal{A} \text{ is } \mathcal{Z}
        \end{cases}
    \end{equation}
    \item The $p,q$-th entry of $A$ where $p>1$ gives the information on the separability of $\mathcal{A}$ on the $(p-1)$-th and $q$-th qubits, i.e. for $p>2$,
    \begin{equation}
        A_{pq}=
        \begin{cases}
        0 & \text{if there is no connecting wire between}\\
        & \text{the } (p-1) \text{-th and } q \text{-th qubits in }\mathcal{A}\\
        1 & \text{if } A_{1p},A_{1q}\in\{0,1\}\\
        i & \text{if } A_{1p}=A_{1q}=Z\\
        j & \text{if } A_{1q}\in\{0,1\}\text{ and }A_{1p}=Z\\
        k & \text{if } A_{1q}=Z\text{ and }A_{1p}\in\{0,1\}
        \end{cases}
    \end{equation}
\end{itemize}
\end{definition}

To represent a complementarity diagram between a pair of classical structures, we define a binary operation between the names of those classical structures.

\begin{definition}\label{def:name-CD}
Let $\mathcal{A}$ and $\mathcal{B}$ be composite classical structures on $N$ qubits with names $A$ and $B$, respectively. The name of the complementarity diagram between $\mathcal{A}$ and $\mathcal{B}$ is denoted by $A*B$, and is defined as entry-wise addition modulo 2 of $A$ and $B$, i.e. $(A*B)_{pq}=A_{pq}+_{\text{mod }2}B_{pq}$.  
\end{definition}

\noindent Below is an example of a complementarity diagram (diagram on the left) and its name (the matrix on the right):
\begin{equation*}
    \includegraphics[scale=0.2]{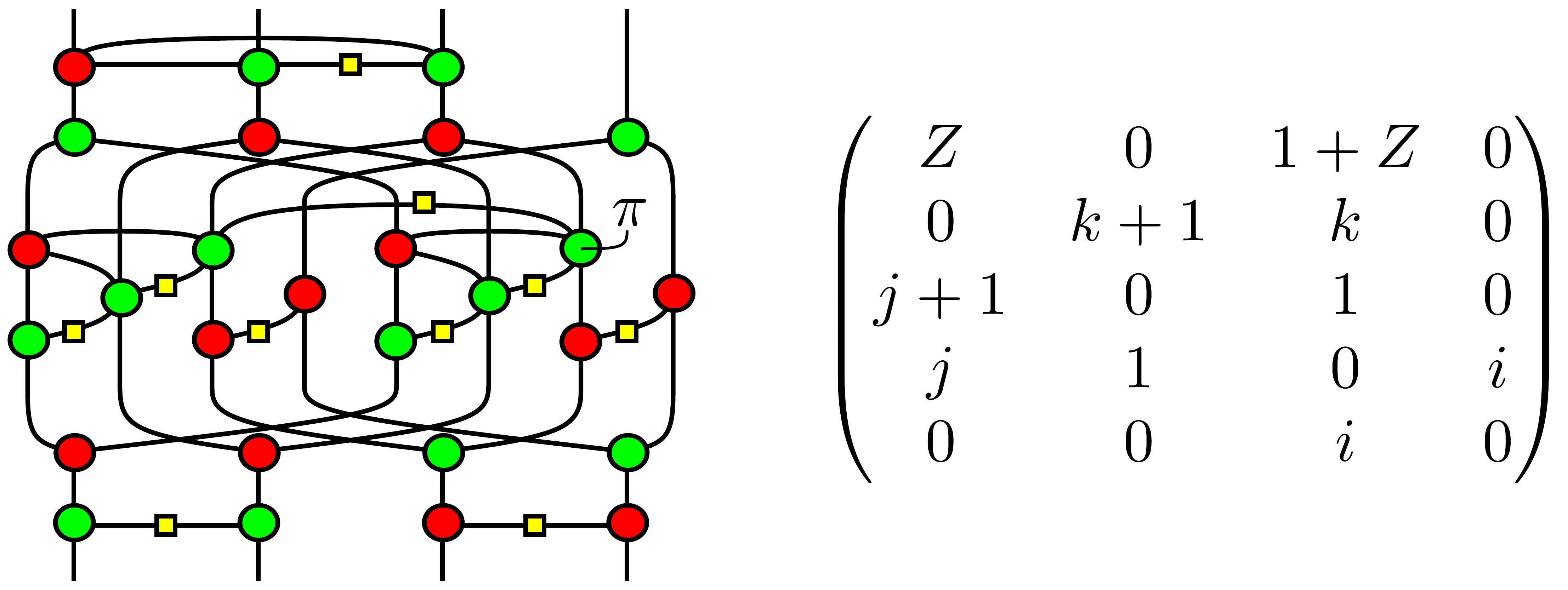}
\end{equation*}

Partitioning complementarity diagrams on $N$ qubits with respect to the equivalence defined in Theorem \ref{thm:comp-diagram-equivalence}, we shall obtain equivalence classes where we can say that \textbf{every member of a class satisfies the complementarity condition if a single member of the class satisfies the complementarity condition}. 

\begin{corollary}\label{cor:CD-equivalence}
Let $|P|$ be a partition of $\mathsf{D}_N$ with respect to $r$ in Theorem \ref{thm:comp-diagram-equivalence}. Let $C\in|P|$ and $D_0\in C$. Suppose $D_0$ satisfies the complementarity condition. Then all $D\in C$ satisfies the complementarity condition.
\begin{proof}
If $D_0$ satisfies the complementarity condition, then:
\begin{equation*}
    \includegraphics[scale=0.2]{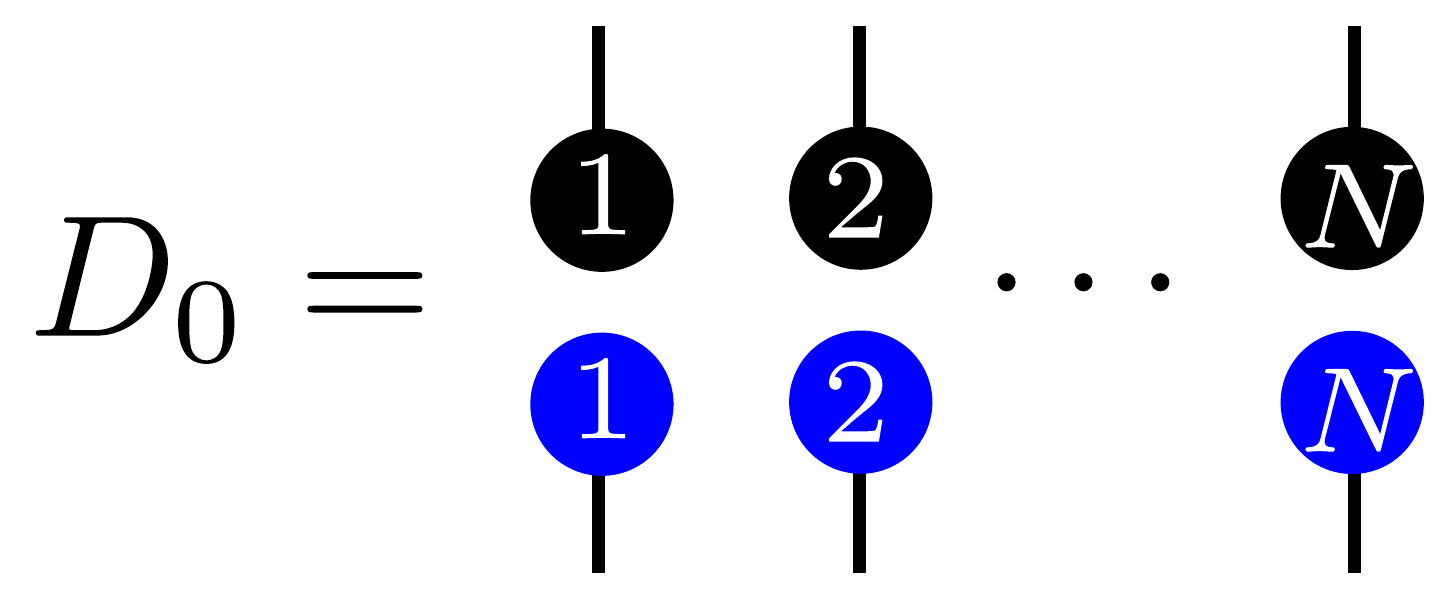}
\end{equation*}
Applying any permutation or $\had$ on any of the qubits to $D_0$, we shall obtain a diagram of the following form:
\begin{equation*}
    \includegraphics[scale=0.2]{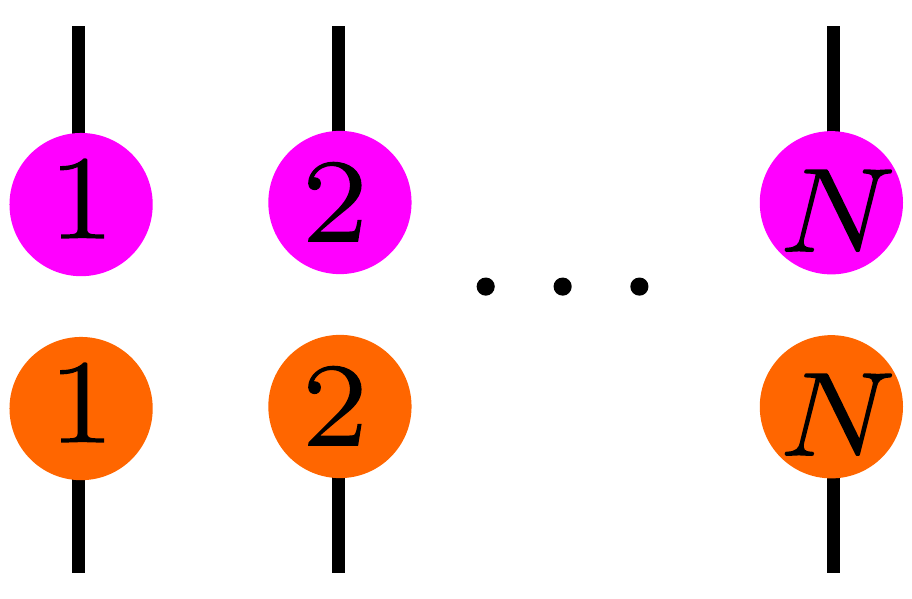}
\end{equation*}
which tells us that any other complementarity diagram belonging to $C$ also satisfies the complementarity condition. 
\end{proof}
\end{corollary}

\begin{longtable}{|p{2cm}|p{8.5cm}|}
\caption{Correspondence between criteria of equivalence for complementarity diagrams and their names.}
\\ \hline
\textbf{Equivalence Criteria} & \textbf{Change in Name}\\ \hline
\endhead
Adjoint & No change since $*$ is commutative.\\ \hline
The $p$-th and $q$-th slices in a complementarity diagram are permuted & The $p$-th and $q$-th columns of the name are permuted, along with the $(p+1)$-th and $(q+1)$-th rows are permuted.\\ \hline
Apply $\had$ on the $p$-th slice & 
\begin{itemize}
    \item No change in first entry of $p$-th column;
    \item The entries in the second to $N$-th entry of the $p$-th column are tranformed in the following way:
    \begin{eqnarray*}
        0 \mapsto 0\\
        1 \mapsto k\\
        i \mapsto j\\
        j \mapsto i\\
        k \mapsto 1
    \end{eqnarray*}
    \item The entries in the $(p+1)$-th row are transformed in the following way:
    \begin{eqnarray*}
    0 \mapsto 0\\
    1 \mapsto j\\
    i \mapsto k\\
    j \mapsto 1\\
    k \mapsto i
    \end{eqnarray*}
\end{itemize} \\ \hline
Slices that represent different pairs of constituents &
\begin{equation*}
\includegraphics[scale=0.2]{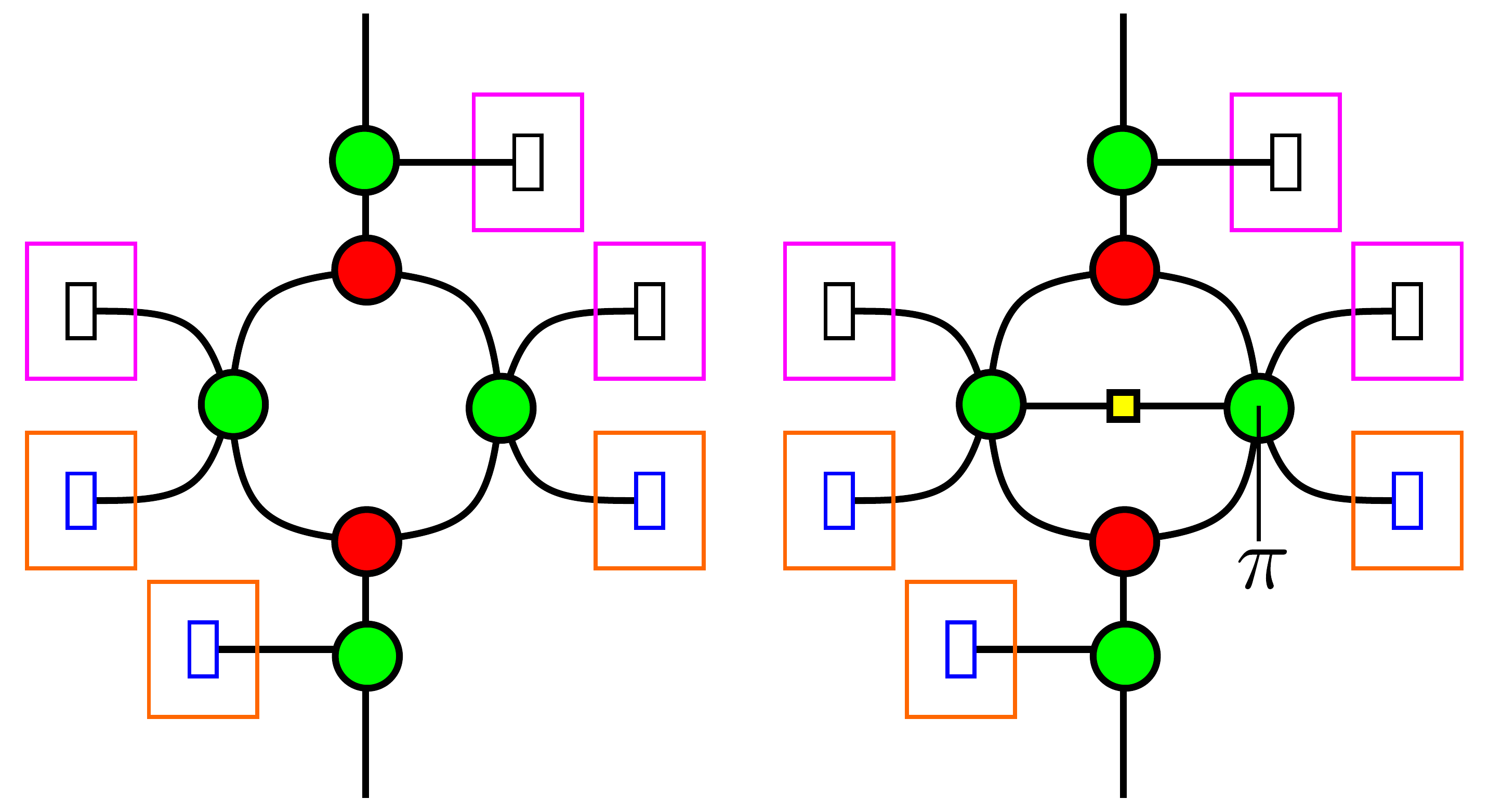}
\end{equation*}\\
 & No change in name.\\ \hline
Slices that represent different pairs of constituents &
\begin{equation*}
    \includegraphics[scale=0.2]{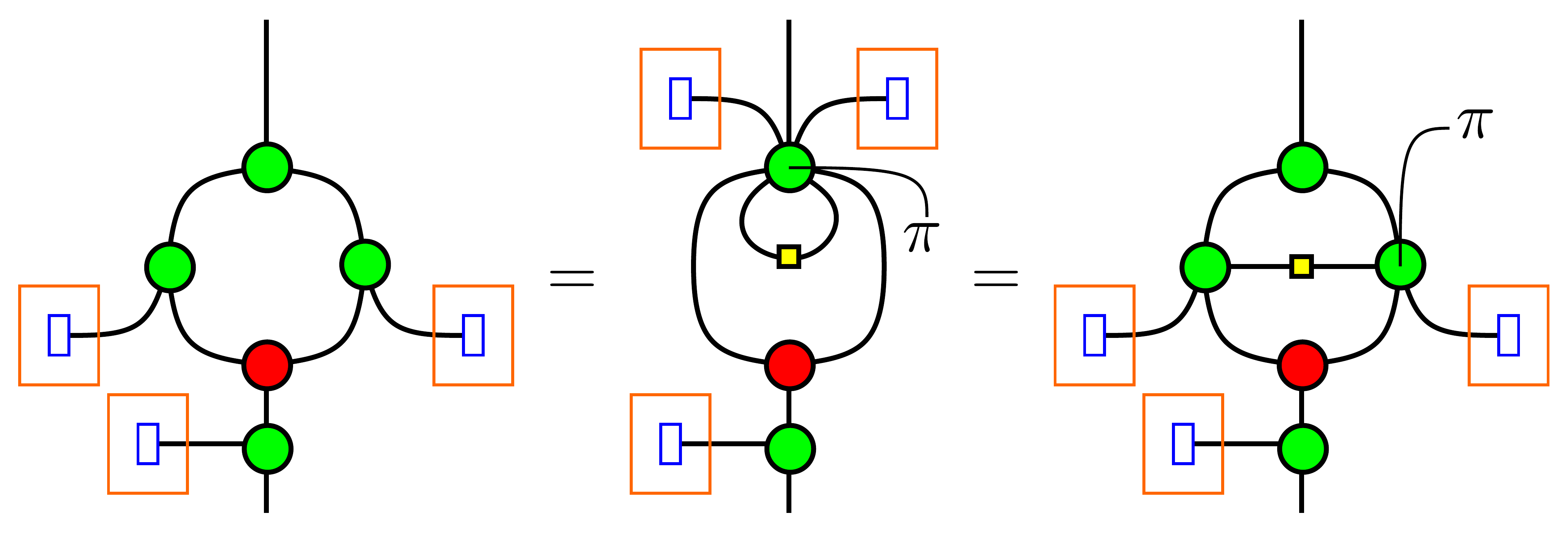}
\end{equation*}\\ 
& First entry of the corresponding column in the name is transformed in the following way:
$$Z\mapsto 1+Z$$\\
\hline
\label{tab:equiv-CD}
\end{longtable}

For some of the criteria in Theorem \ref{thm:comp-diagram-equivalence} and consequently Corollary \ref{cor:CD-equivalence}, the names of complementarity diagrams do not change but there are others where the names are transformed. In Table \ref{tab:equiv-CD}, we shows this correspondence.

Finally, we have all the tools we need to devise a strategy for finding maximal complete sets of mutually complementary classical structures on $N$ qubits. The following is an outline of our strategy:
% strategy to find maximal complete set of complementary CS
% How to generate test
% rewrite theorem-CDname to be more general

\begin{enumerate}
    \item Check whether or not each member of $\mathsf{G}_N$ satisfies the complementarity condition;
    \item Gather all members of $\mathsf{G}_N$ which satisfy the complementarity condition into a set denoted as $\mathsf{G}_N^\text{pass}$;
    \item Identify the name of each member in $\mathsf{G}_N^\text{pass}$ and collect them in a set denoted as $\mathsf{P}_N$;
    \item Expand $\mathsf{P}_N$ to include other names of complementarity diagrams within the same equivalence class (see Table \ref{tab:equiv-CD}), and denote the expanded set as $\mathsf{T}_N$;
    \item Generate the names of the complementarity diagrams between all pairs of composite CS on $N$ qubits;
    \item Check whether or not each name obtained from the previous step is a member of $\mathsf{T}_N$;
    \begin{itemize}
        \item If a name is a member of $\mathsf{T}_N$, then the pair of composite CS it represents are complementary;
        \item If not, then the pair of composite CS the name represents are not complementary;
    \end{itemize}
    \item Form a graph where its vertices are composite CS on $N$ qubits, and for each pair of vertices, there exists an edge if, as composite CS, the vertices are complementary, and we denote this graph as $\mathbf{G}_N$;
    \item Search for complete subgraphs in $\mathbf{G}_N$ and each subgraph will give us a complete set of mutually complementary CS on $N$ qubits.
\end{enumerate}

In the procedure above, we used Quantomatic \cite{Quantomatic} to compute the results for step 1 and we used Mathematica \cite{WolframComputing} to compute the results for steps 4 to 8. 

In the next two sections, we shall apply this strategy to search for maximal complete sets of mutually complementary CS on two and three qubits. To simplify and improve readability of complementarity diagrams, we shall forego the connecting wires on the input and output wires of the diagrams and any phases contained in them. For the latter, whether or not a complementarity diagram contains a phase is apparent. That is, there is a phase in a particular slice of a complementarity diagram when only the constituent on the top of the slice or only the one at the bottom is $\mathcal{Y}$. Otherwise, there is no phase in the slice. For the former, we can indeed prove that a complementarity diagram satisfies the complementarity condition if and only if that same complementarity diagram without connecting wires on its input and output wires also satisfies the complementarity condition.

\begin{proposition}\label{prop:in-out-conwire}
A complementarity diagram on $N$ qubits satisfies the complementarity condition if and only if that same complementarity diagram without connecting wires on its input and output wires also satisfies the complementarity condition.
\end{proposition}
\begin{proof}
Let $D$ be a complementarity diagram on $N$ qubits. Suppose $D$ satisfies the complementarity condition. Then, the $p$-th and $q$-th slices of the diagram take the following form:
\begin{equation*}
    \includegraphics[scale=0.2]{images/51-75/68-comp-cond.pdf}
\end{equation*}
Meanwhile, the diagram $D$ without connecting wires on its input and output wires, denoted by $D'$, has the following as its $p$-th and $q$-th slices:
\begin{equation*}
    \includegraphics[scale=0.2]{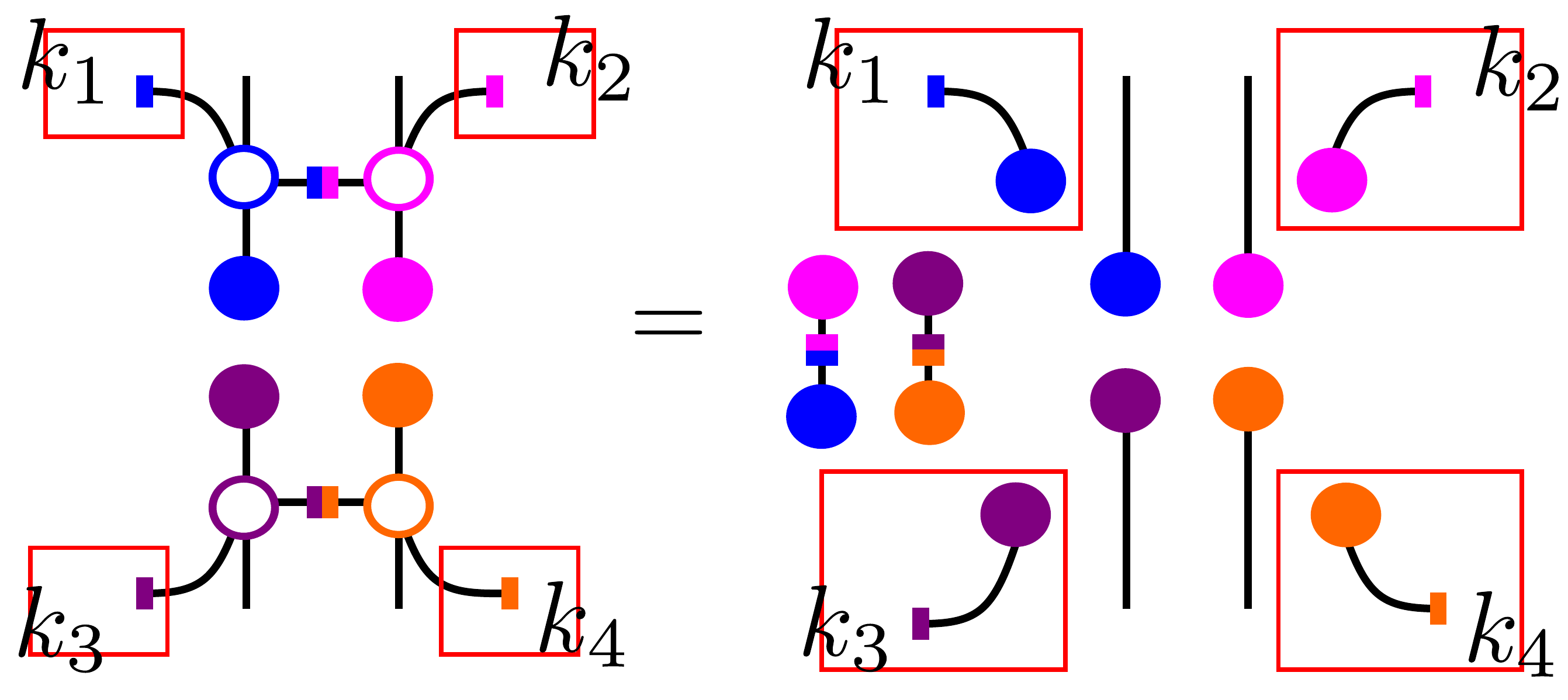}
\end{equation*}
The labels of the red !-boxes indicate that the repetitions within the boxes may vary in frequency. 

Since we have taken the pair of slices to be arbitrary, we can apply the above result to all pairs of slices in $D'$. Thus, any pair of slices in $D'$ must take the following form:
\begin{equation*}
    \includegraphics[scale=0.2]{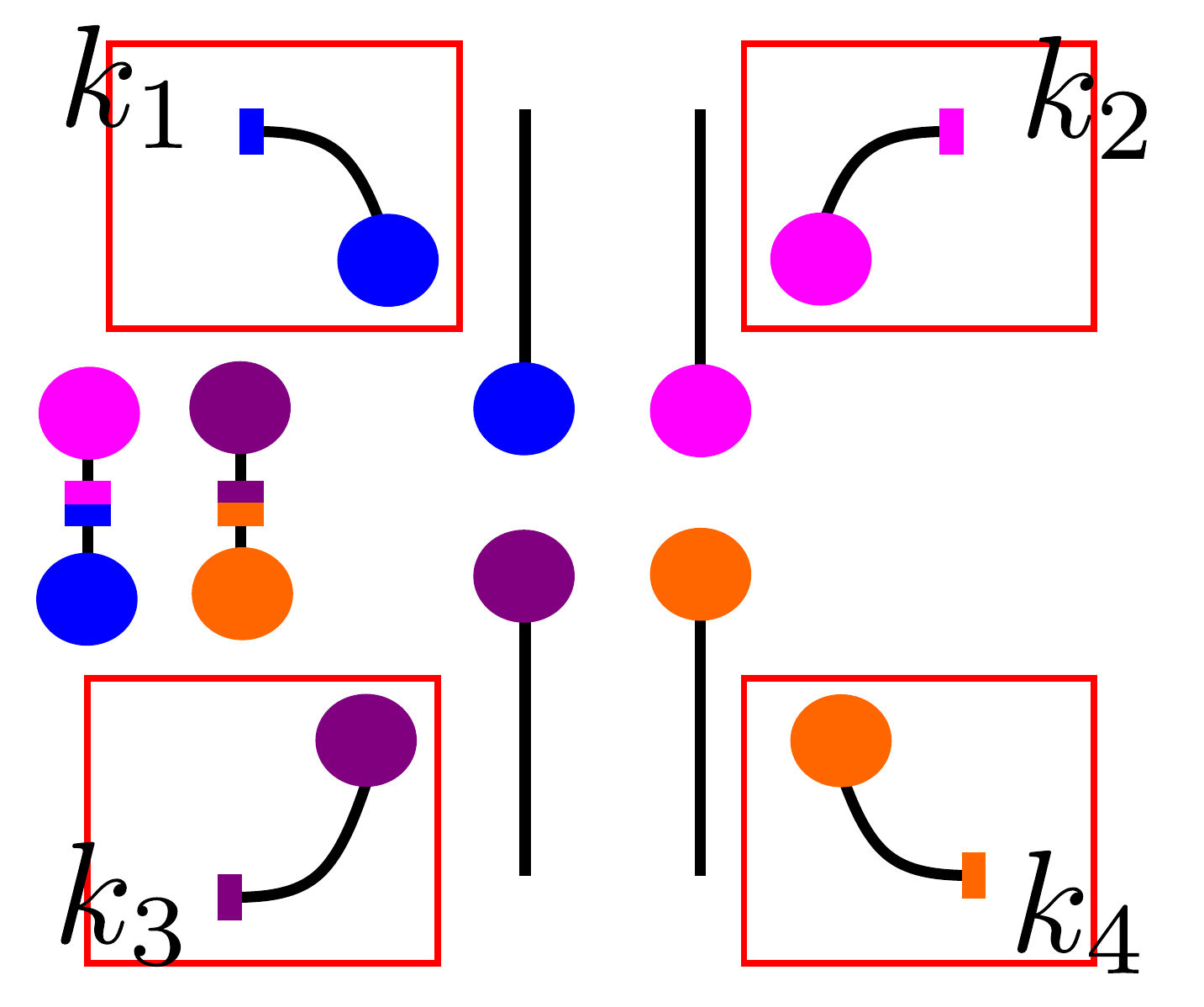}
\end{equation*}
which means that $D$ and $D'$ are equivalent up to scalars.

The converse can be proven in a similar way. 
\end{proof}

\section{Complementary CS on Two Qubits}\label{sec:comp2Q}

\begin{figure}[!ht]
    \centering
    \begin{longtable}{cccc}
    \includegraphics[scale=0.2]{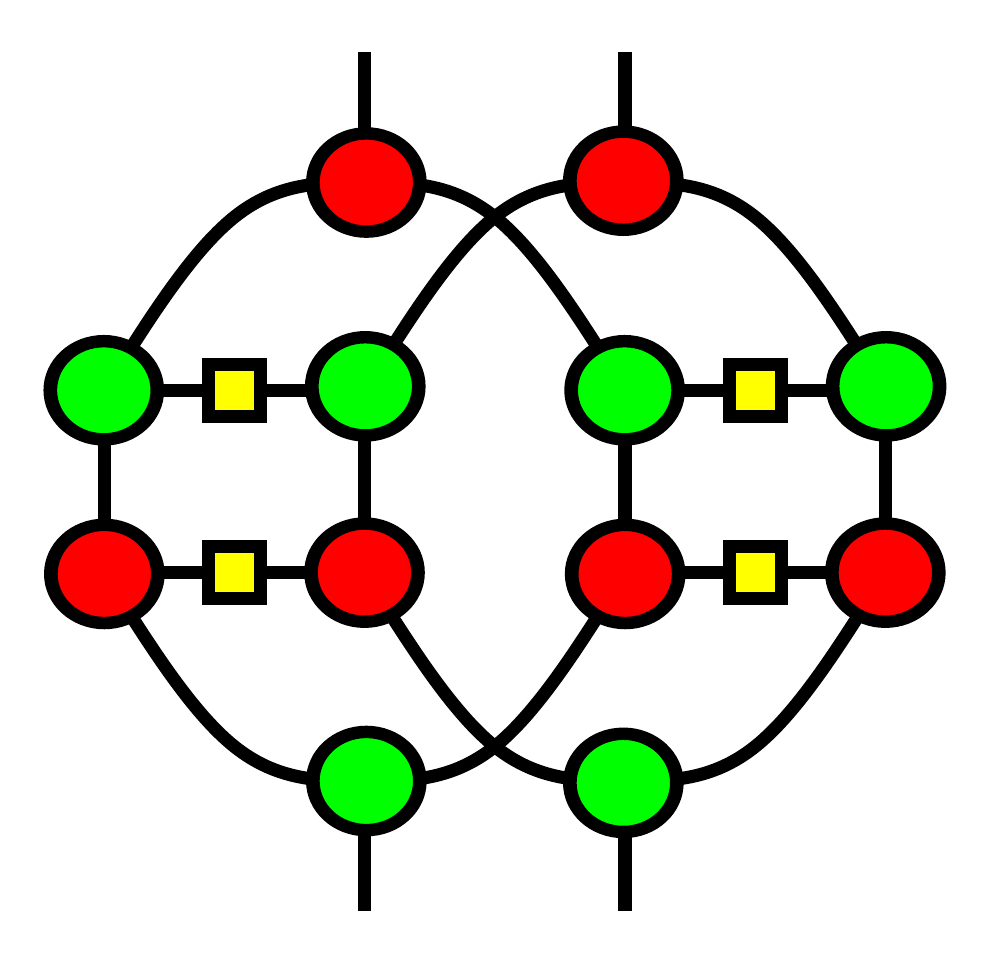} &
    \includegraphics[scale=0.2]{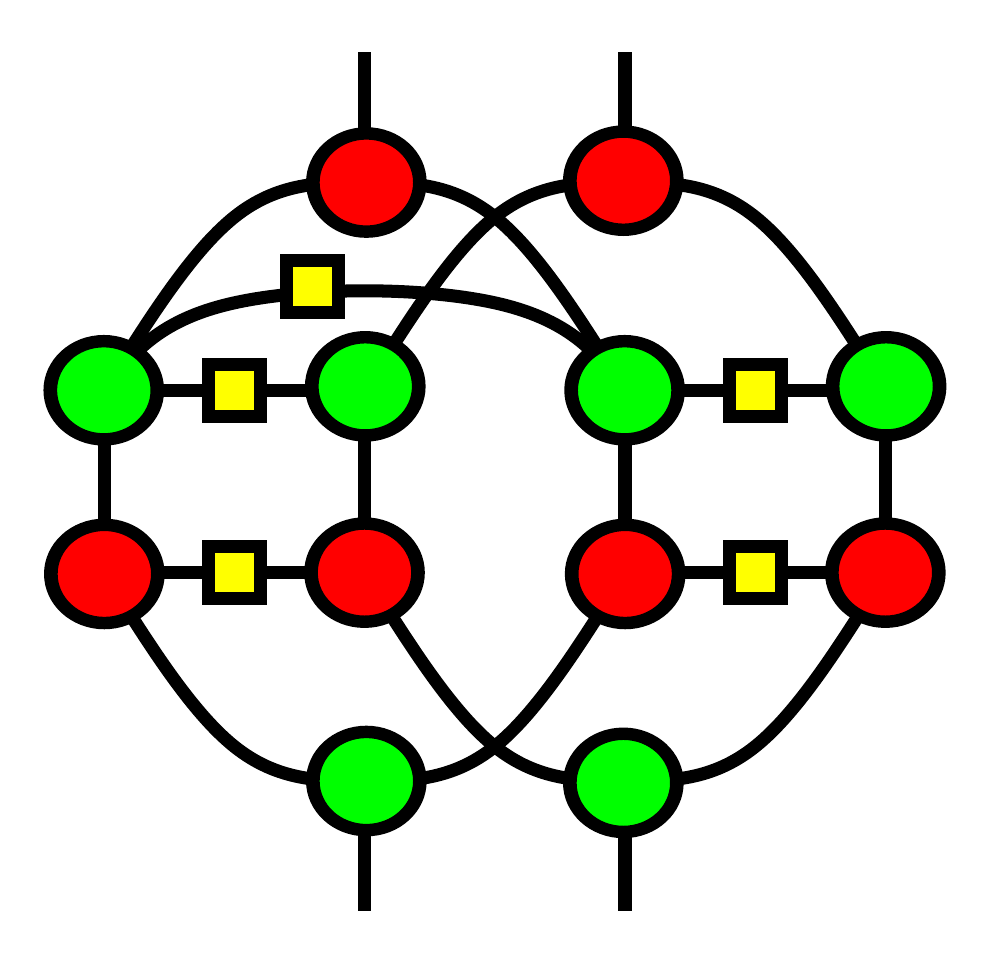} &
    \includegraphics[scale=0.2]{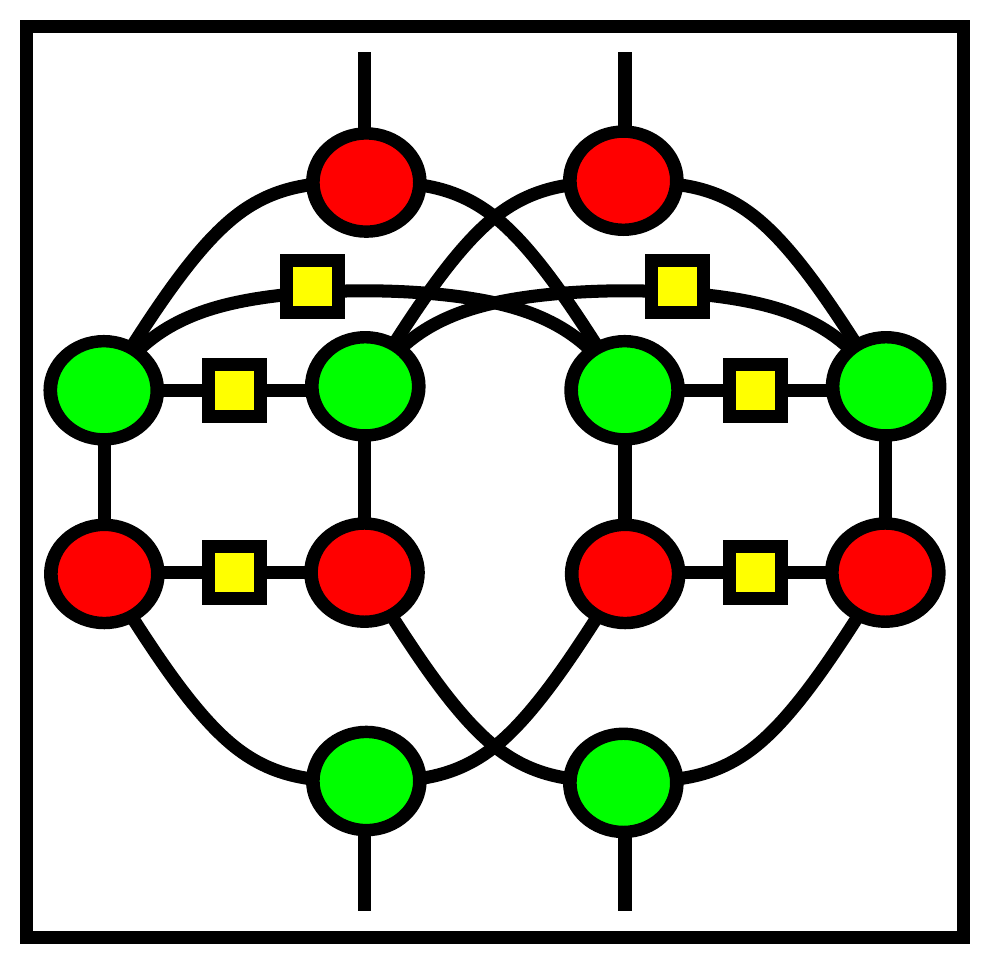} &
    \includegraphics[scale=0.2]{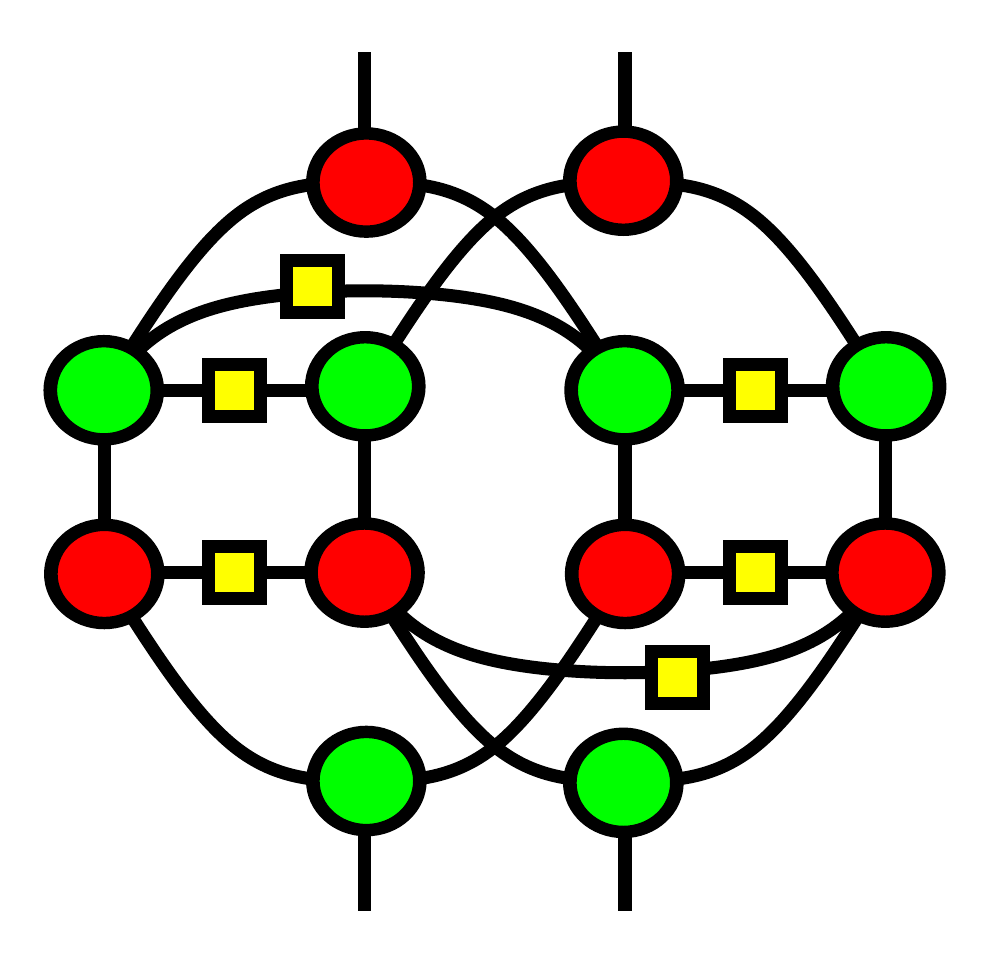}\\
    \includegraphics[scale=0.2]{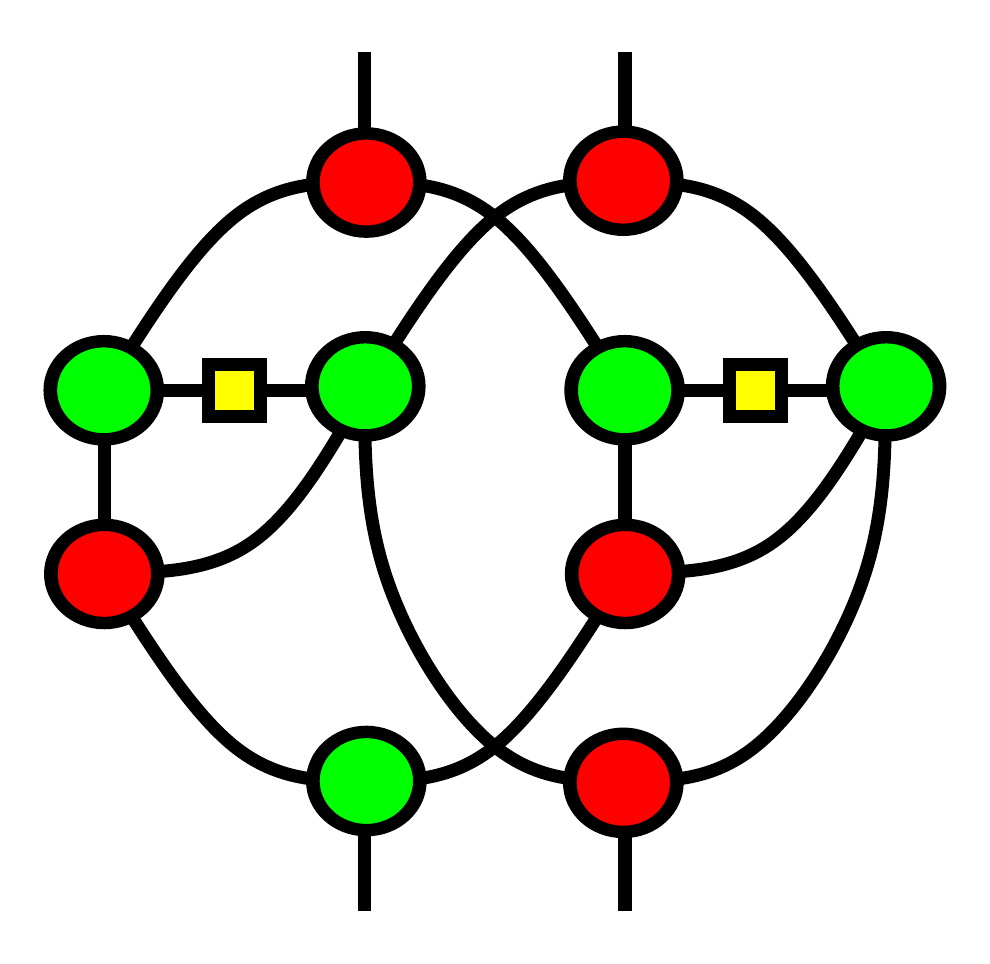} &
    \includegraphics[scale=0.2]{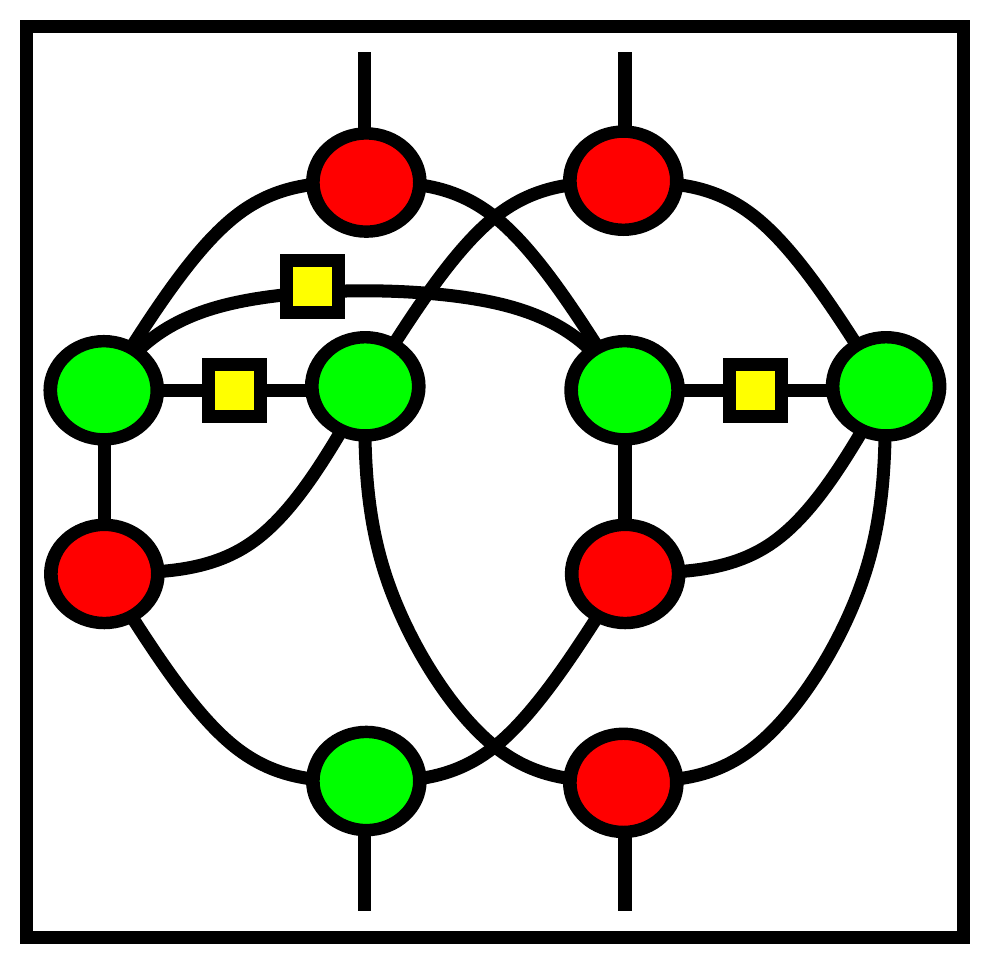} &
    \includegraphics[scale=0.2]{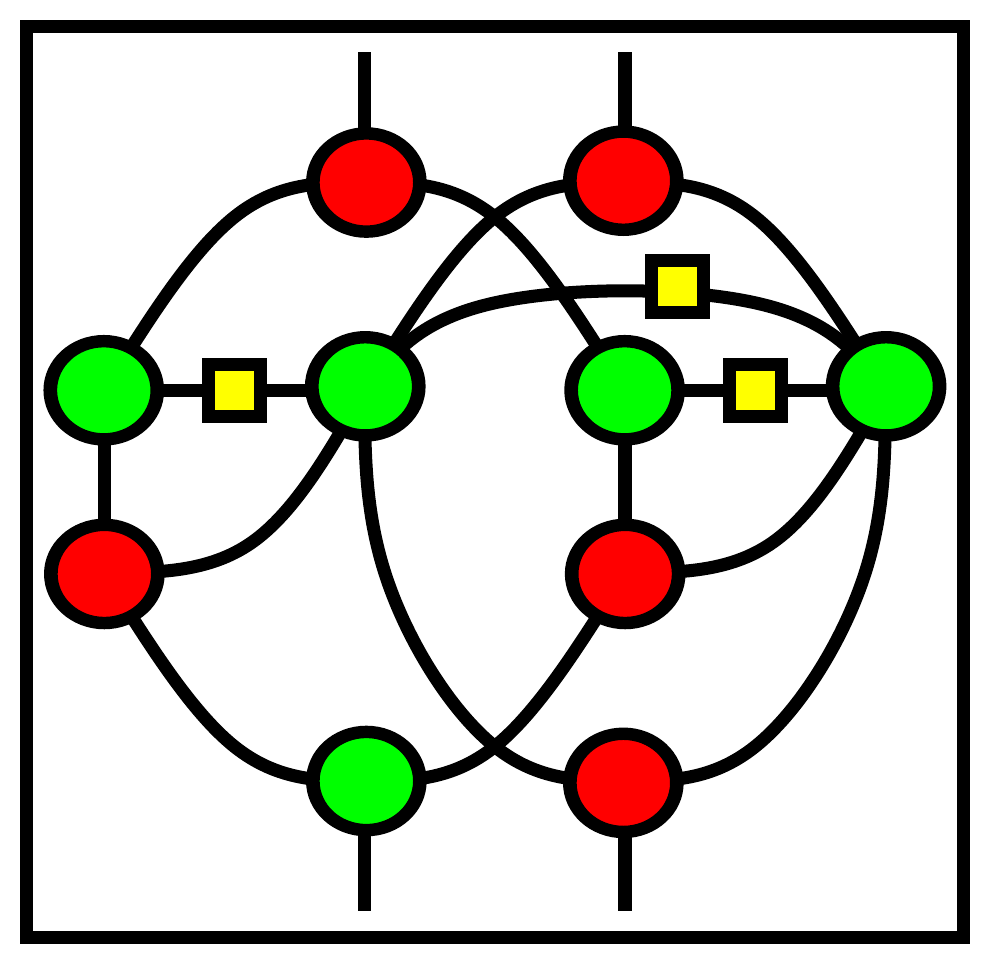} &
    \includegraphics[scale=0.2]{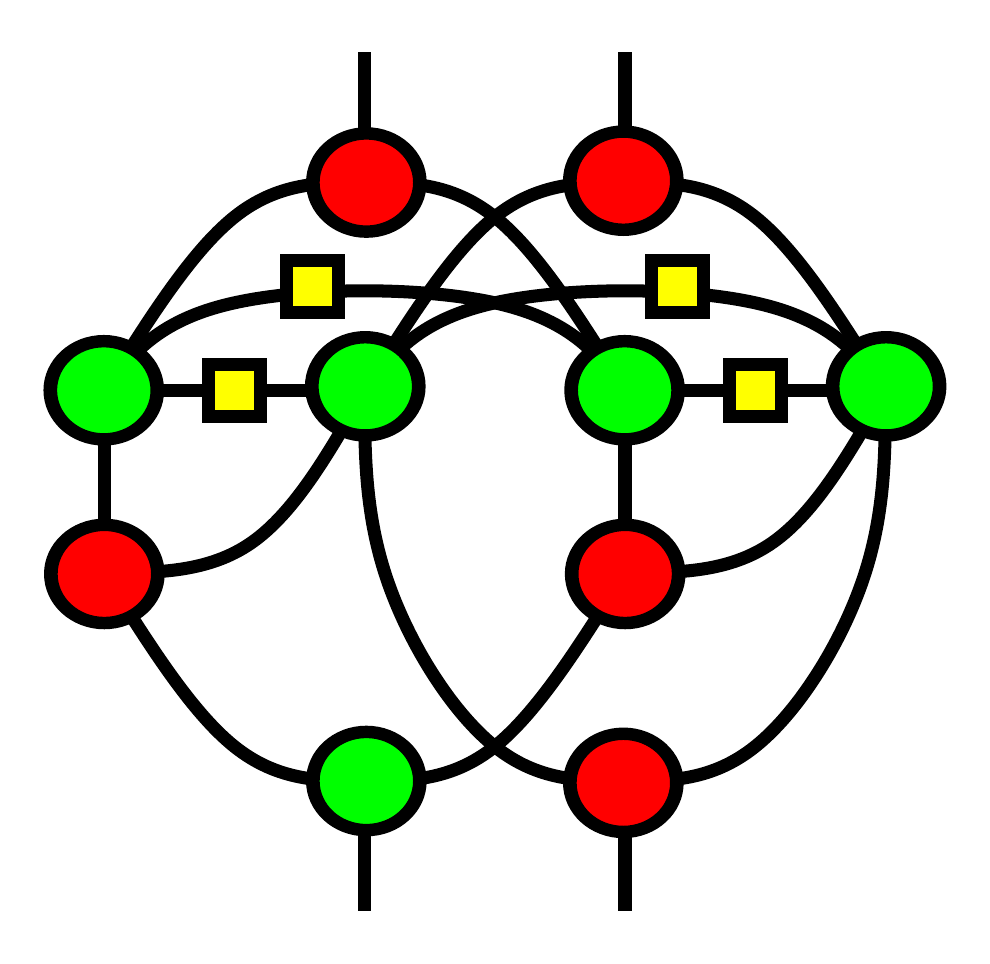}\\
    \includegraphics[scale=0.2]{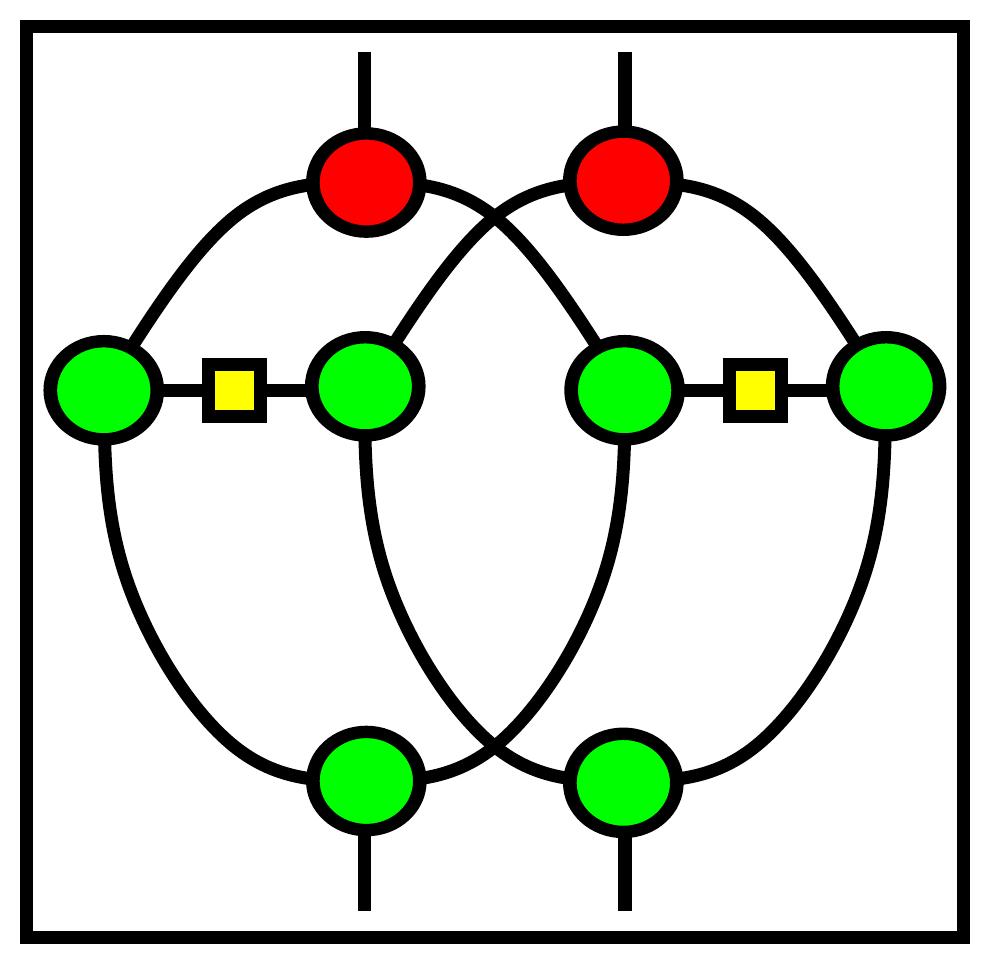} &
    \includegraphics[scale=0.2]{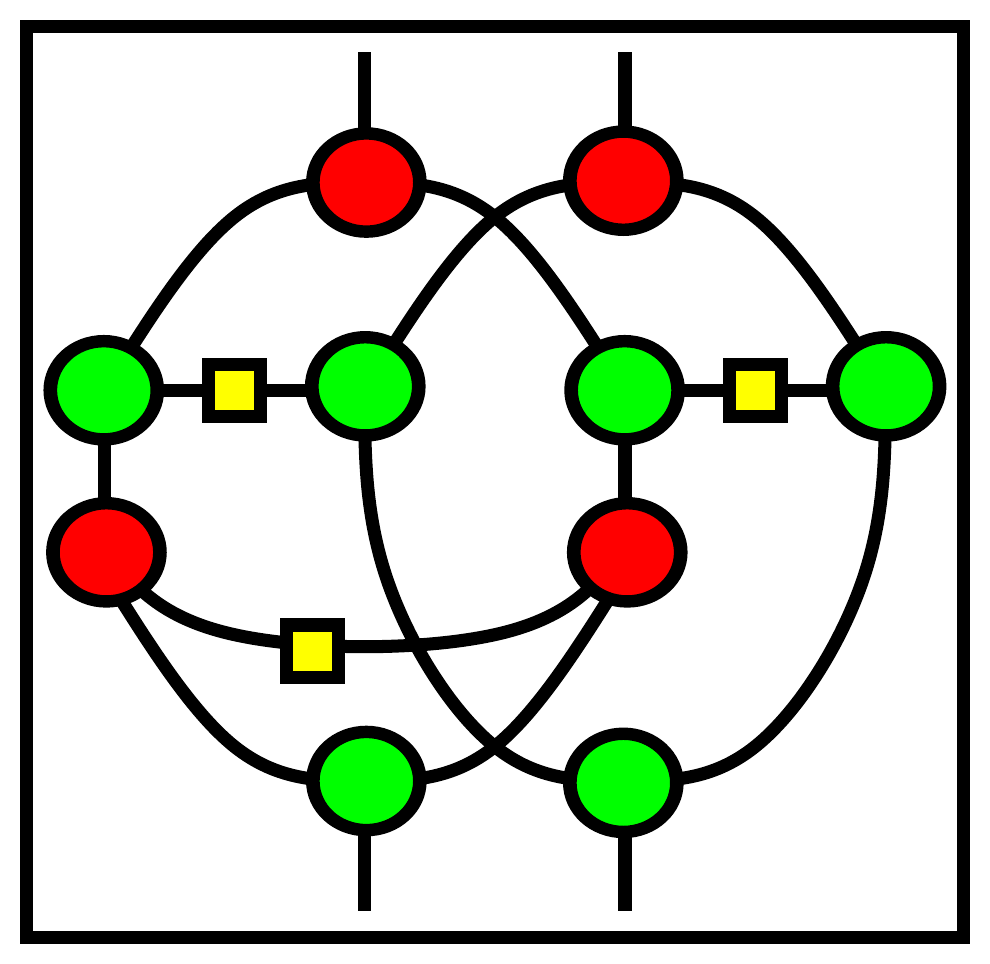} &
    \includegraphics[scale=0.2]{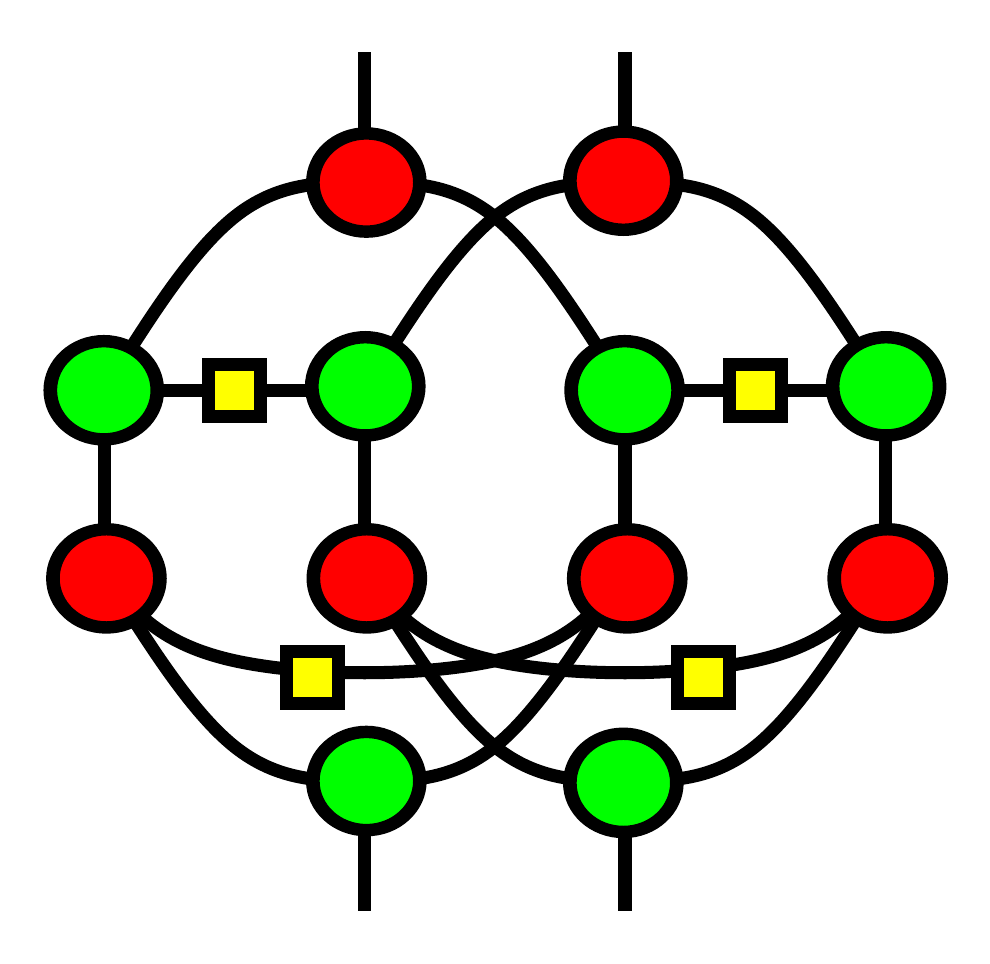} &
    \includegraphics[scale=0.2]{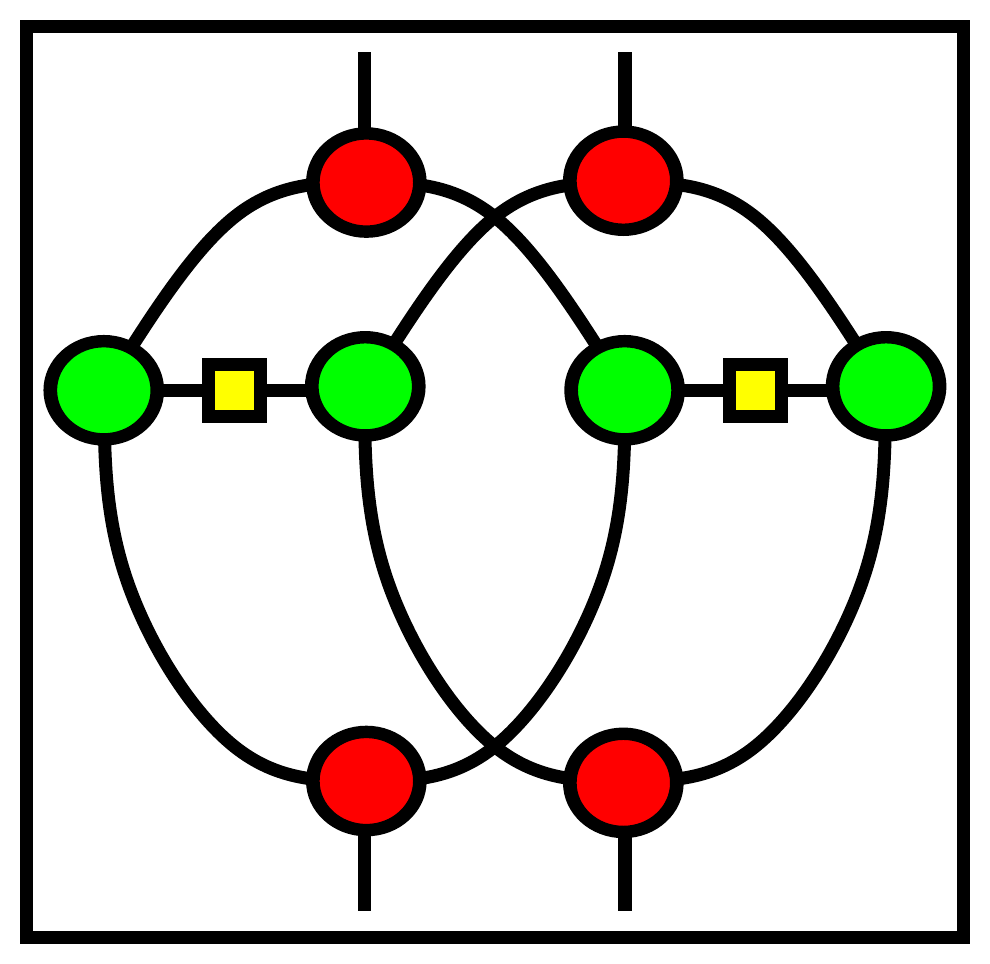}\\
    \includegraphics[scale=0.2]{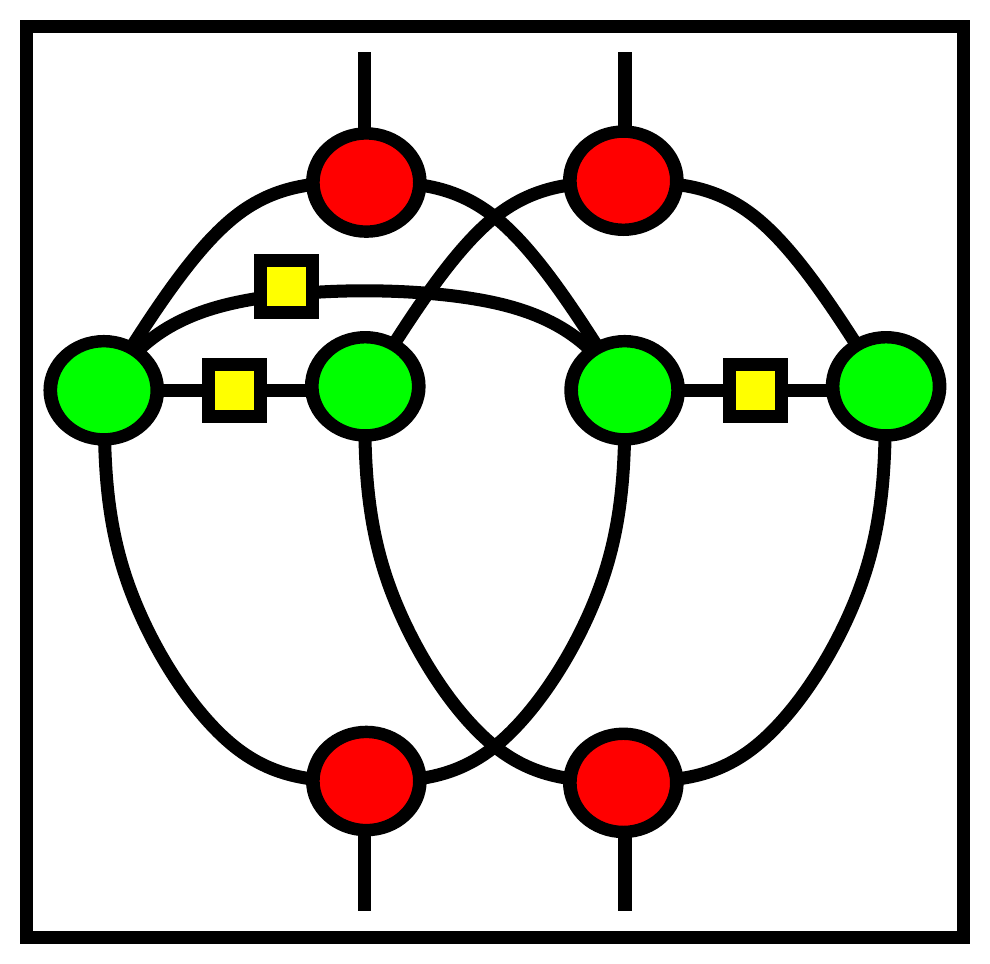} &
    \includegraphics[scale=0.2]{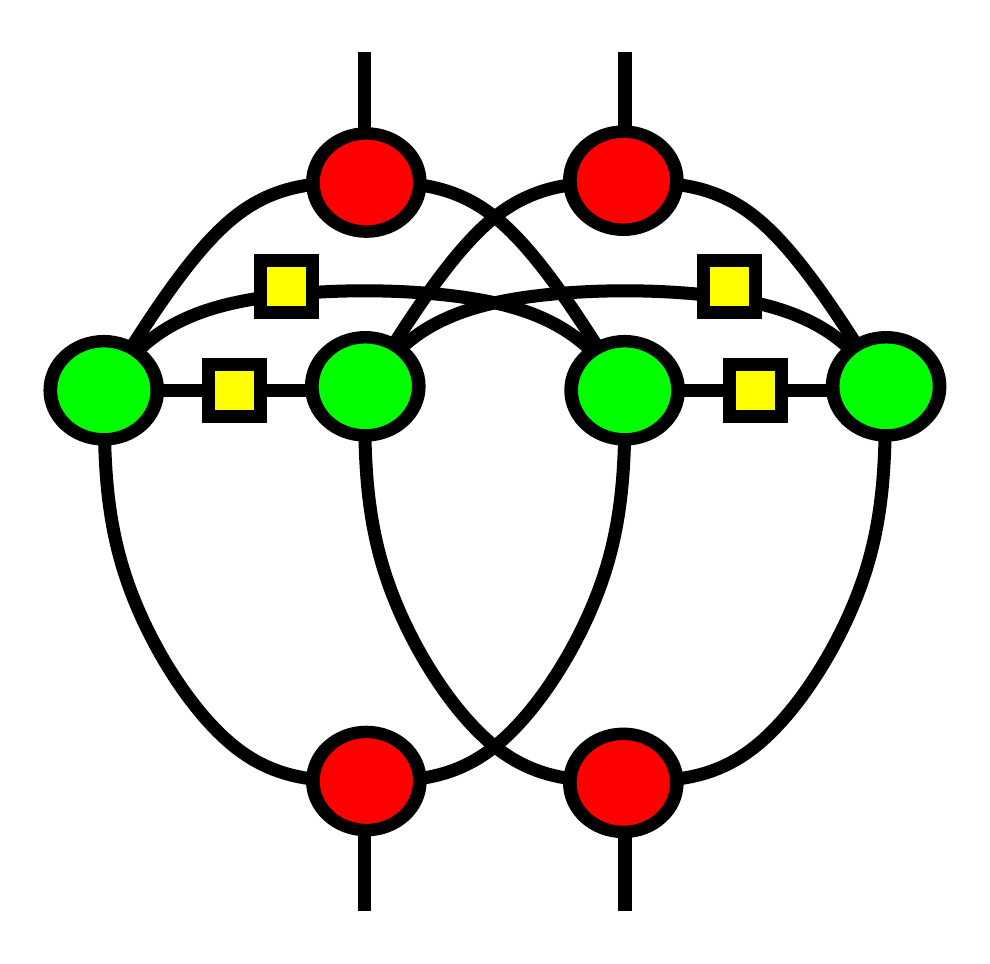} &
    \includegraphics[scale=0.2]{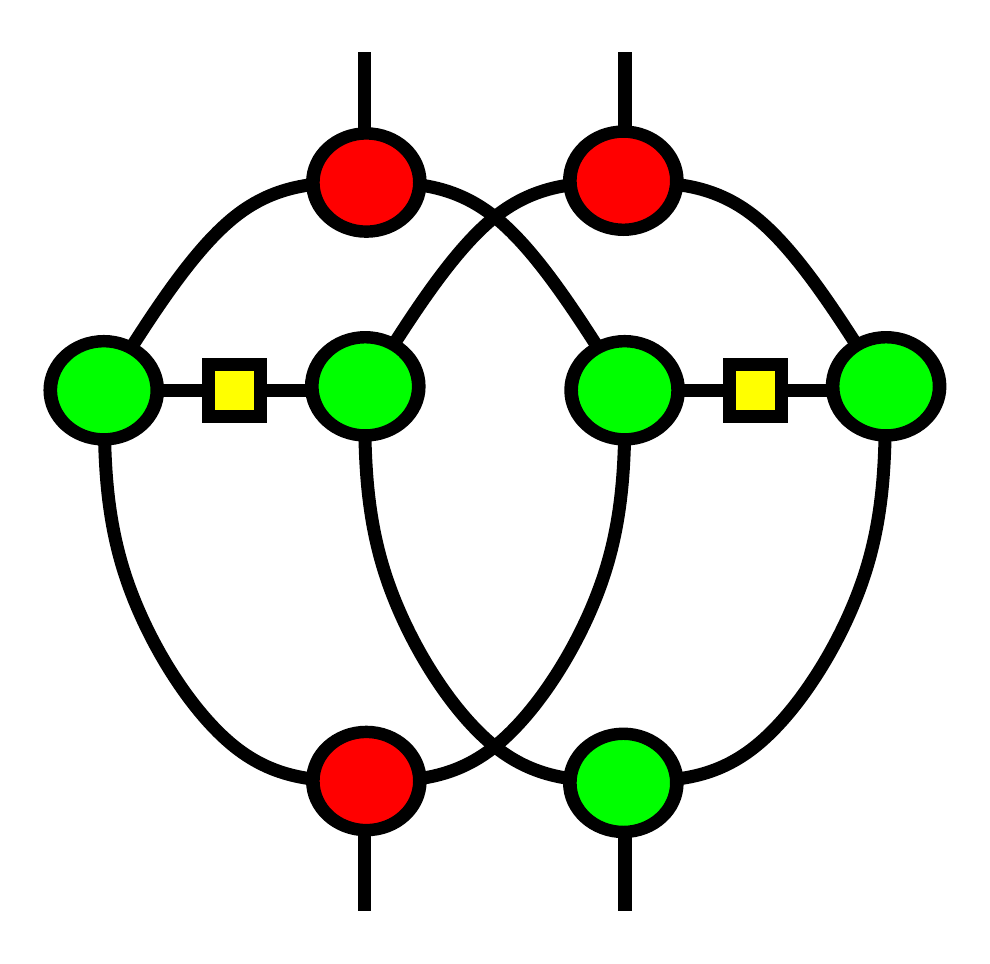} &
    \includegraphics[scale=0.2]{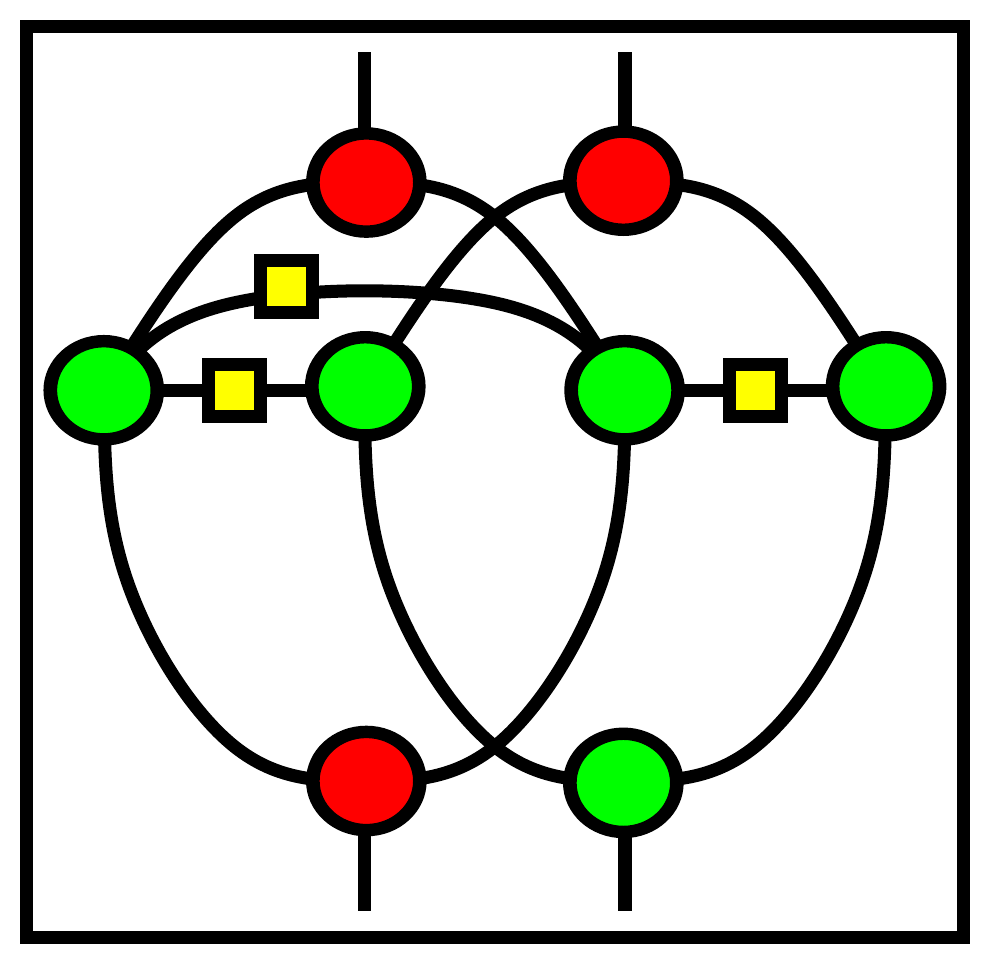}\\
     &
    \includegraphics[scale=0.2]{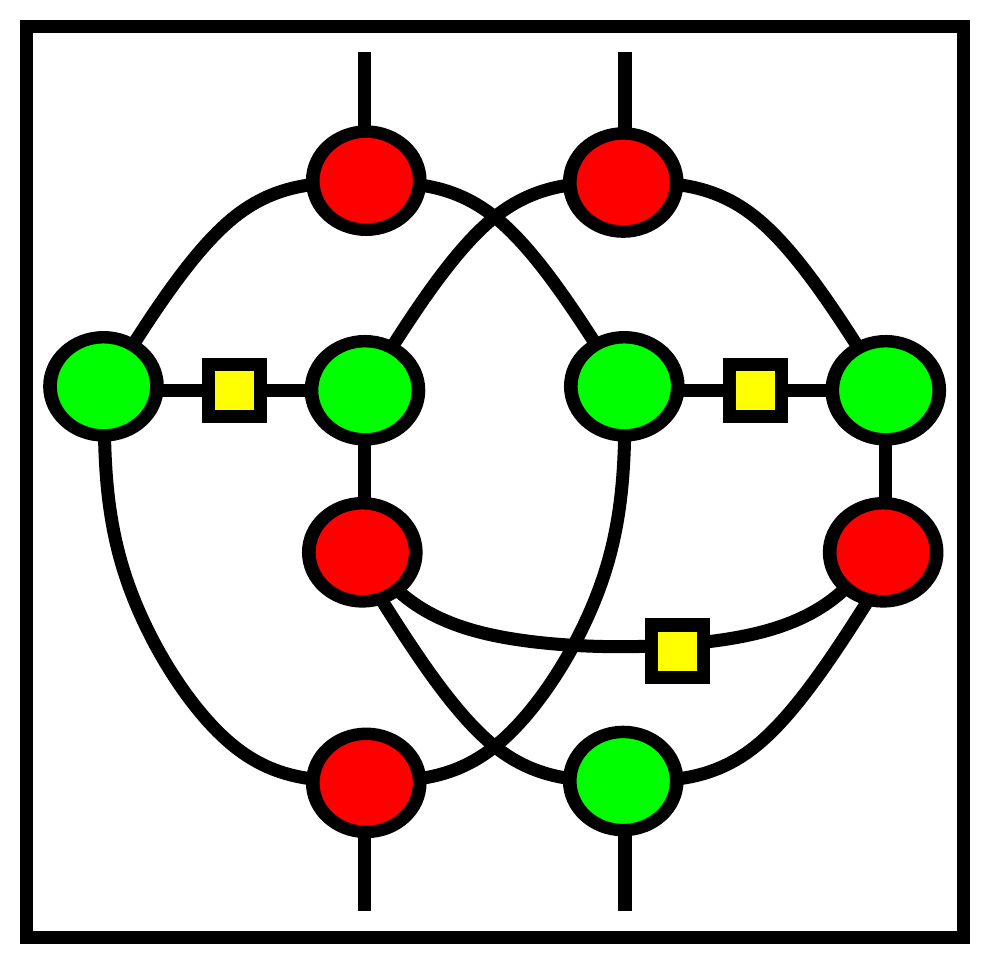} &
    \includegraphics[scale=0.2]{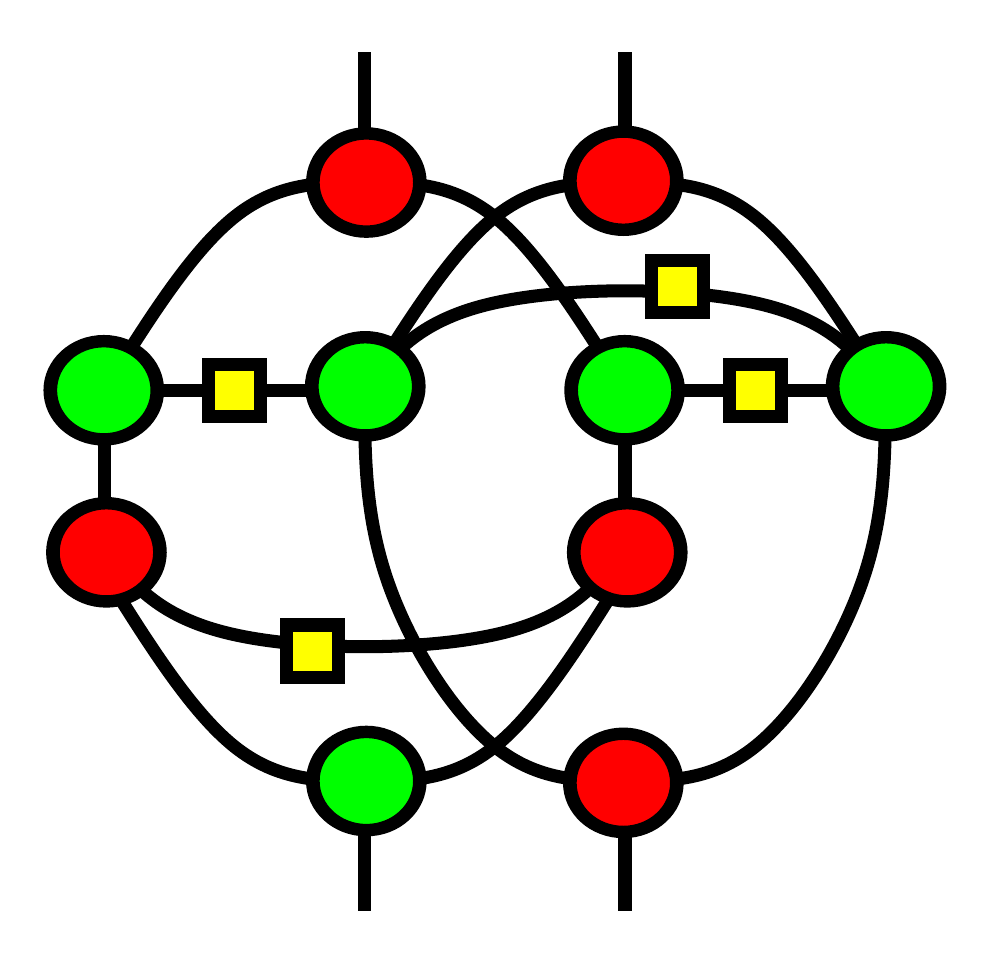} &
    \end{longtable}
    \caption{Entangled Complementarity Diagrams in $\mathsf{G}_2$}
    \label{fig:CD-2Q}
\end{figure}

Fig. \ref{fig:CD-2Q} gives the members of $\mathsf{G}_2$ which are entangled, and those complementarity diagrams within boxes in the figure are the ones which satisfy the complementarity condition, i.e. they are members of $\mathsf{G}_2^\text{pass}$.

To complete $\mathsf{G}_2^\text{pass}$, we need to include those complementarity diagrams which are separable. This can be easily obtained from complementarity diagrams on a single qubit, of which there are three (see Fig. \ref{fig:CD-1Q}). The middle and rightmost diagrams in Fig. \ref{fig:CD-1Q} are complementarity diagrams which satisfy the complementarity condition. We can separably compose them to obtain separable complementarity diagrams on $N$ qubits which satisfy the complementarity condition. For two qubits, those complementarity diagrams are given in Fig. \ref{fig:sep-CDpass-2Q}.

\begin{figure}[!ht]
\centering
    \includegraphics[scale=0.2]{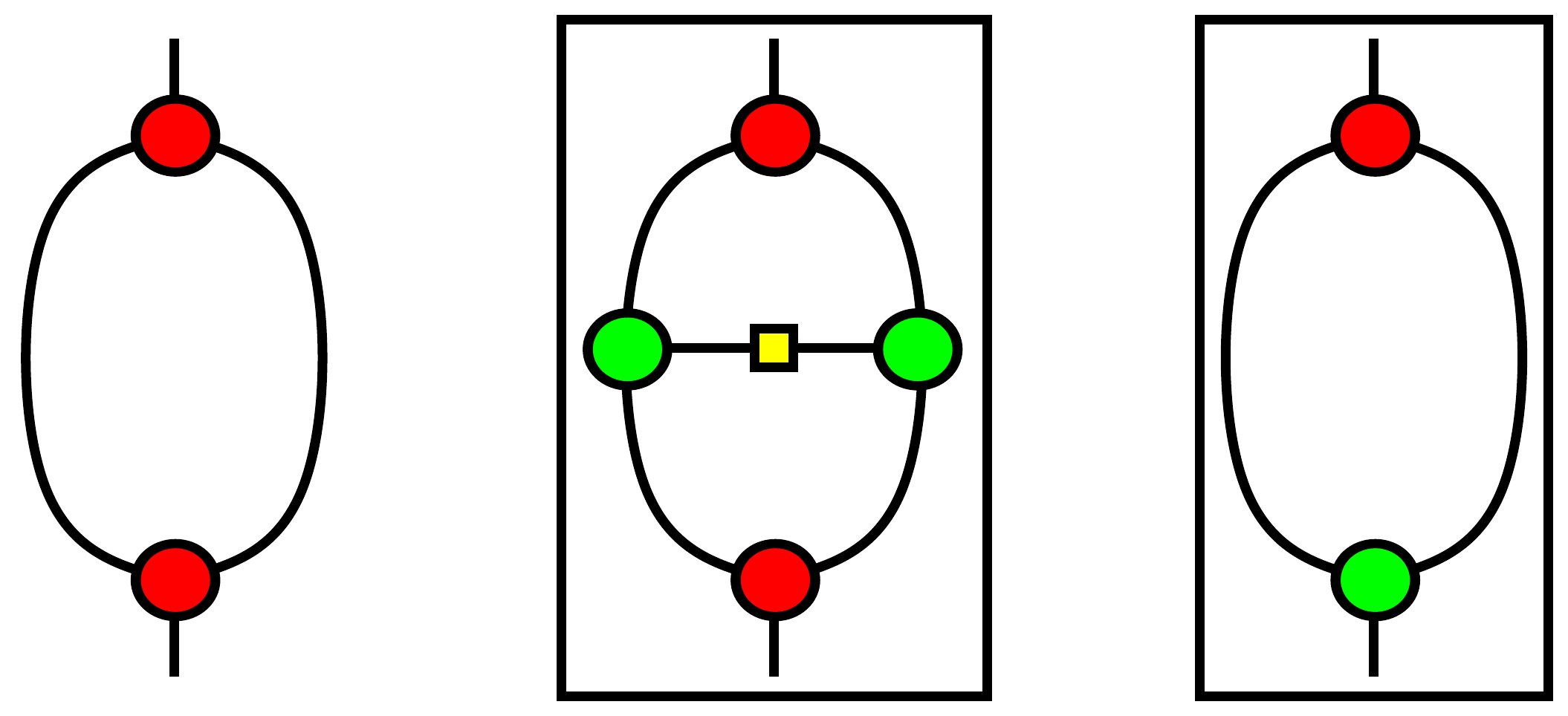}
		\caption{Complementarity diagrams on a single qubit.}
    \label{fig:CD-1Q}
\end{figure}

\begin{figure}[!ht]
    \centering
    \includegraphics[scale=0.2]{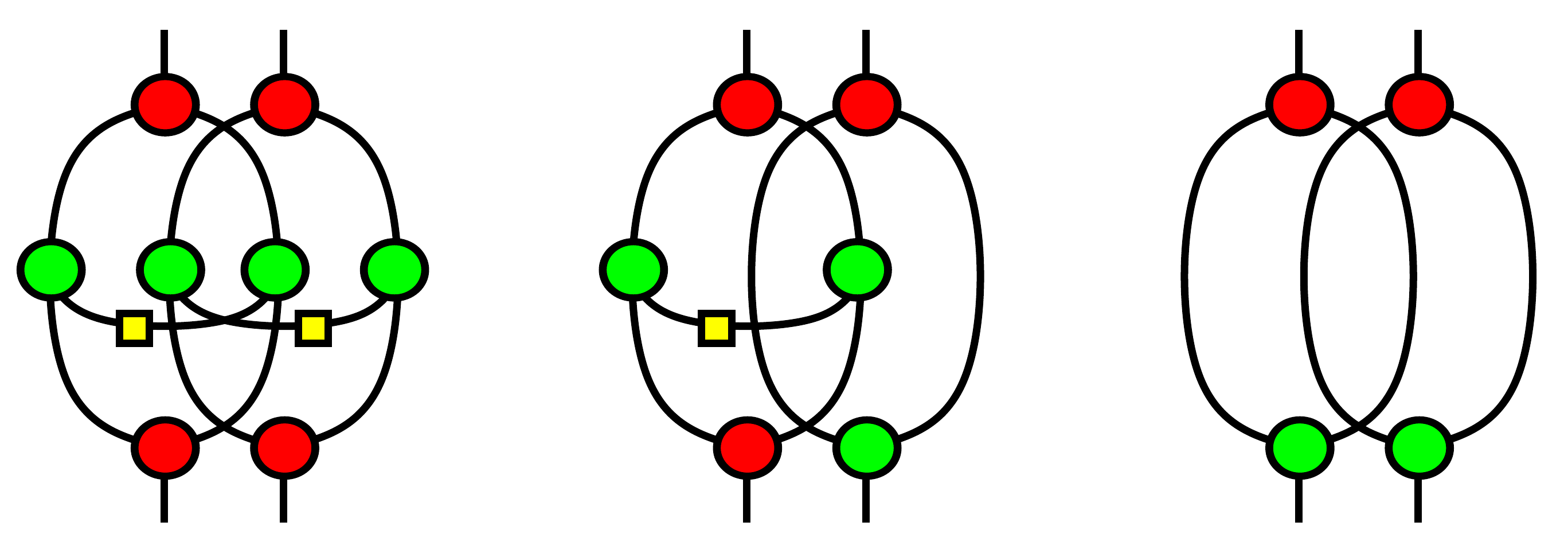}
    \caption{Separable Complementarity Diagrams in $G_2^\text{pass}$}
    \label{fig:sep-CDpass-2Q}
\end{figure}

The three diagrams in Fig. \ref{fig:sep-CDpass-2Q} along with the boxed complementarity diagrams in Fig. \ref{fig:CD-2Q} are all the members of $\mathsf{G}_2^\text{pass}$. From $\mathsf{G}_2^\text{pass}$, we obtain $\mathsf{P}_2$:

\begin{align*}
\mathsf{P}_2 & = &
\left\{
\begin{pmatrix}
1+Z & 1+Z\\
0 & 1+i\\
1+i & 0
\end{pmatrix},
\begin{pmatrix}
1+Z & 0\\
0 & 1+j\\
1+k & 0
\end{pmatrix},
\begin{pmatrix}
Z & 1\\
0 & 1+j\\
1+k & 0
\end{pmatrix},
\begin{pmatrix}
Z & Z\\
0 & 1\\
1 & 0
\end{pmatrix},
\begin{pmatrix}
1+Z & Z\\
0 & j\\
k & 0
\end{pmatrix}\right., 
\\
& &
\quad \left.
\begin{pmatrix}
0 & 0\\
0 & 1\\
1 & 0
\end{pmatrix},
\begin{pmatrix}
1 & 0\\
0 & 1\\
1 & 0
\end{pmatrix},
\begin{pmatrix}
1 & Z\\
0 & 1\\
1 & 0
\end{pmatrix},
\begin{pmatrix}
0 & 1+Z\\
0 & k\\
j & 0
\end{pmatrix},
\begin{pmatrix}
1 & 1\\
0 & 0\\
0 & 0
\end{pmatrix},
\begin{pmatrix}
1 & Z\\
0 & 0\\
0 & 0
\end{pmatrix},
\begin{pmatrix}
Z & Z\\
0 & 0\\
0 & 0
\end{pmatrix}
\right\}
\end{align*}

Performing steps 5 to 8 in the previous section, we obtain the graph $\mathbf{G}_2$ and its complete subgraphs, which are given in Fig. \ref{fig:complete-sub-g2}. Recall that the vertices of $\mathbf{G}_2$ are composite classical structures on two qubits, and edges in the graph are between those pairs of classical structures which are complementary. 

%\begin{figure}[!ht]
    %\centering
   % \includegraphics[scale=0.5]{images/graph/g2.pdf}
    %\caption{Graph $\mathbf{G}_2$}
    %\label{fig:g2}
%\end{figure}

\begin{figure}[!ht]
    \centering
\begin{longtable}{cccc}
\includegraphics[scale=0.25]{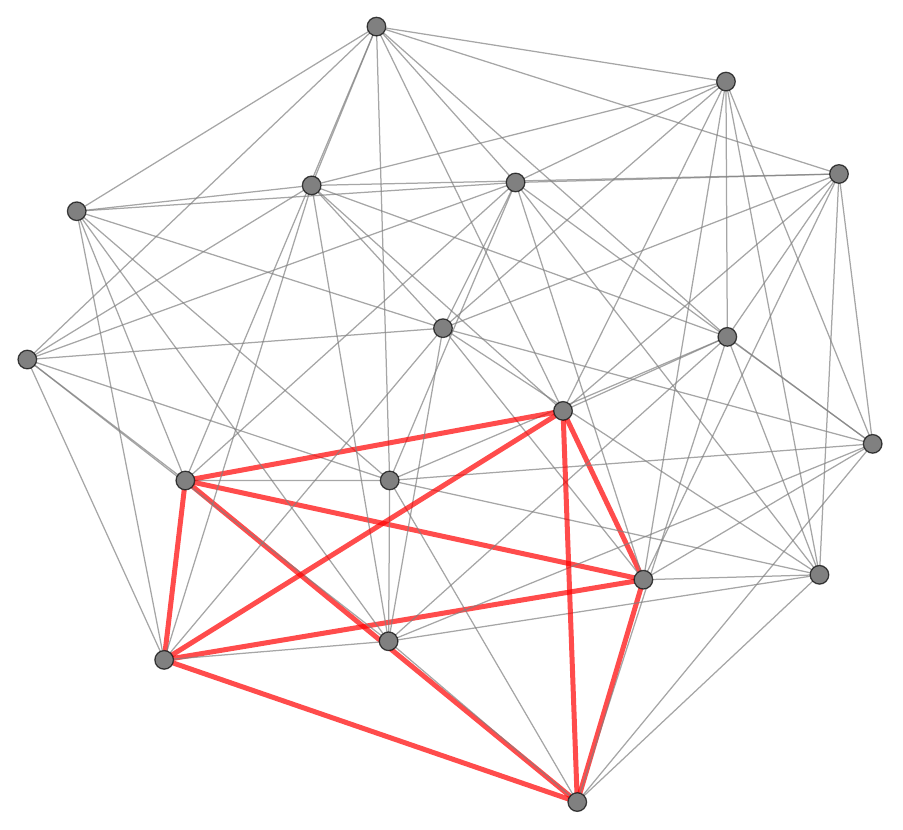} &
\includegraphics[scale=0.25]{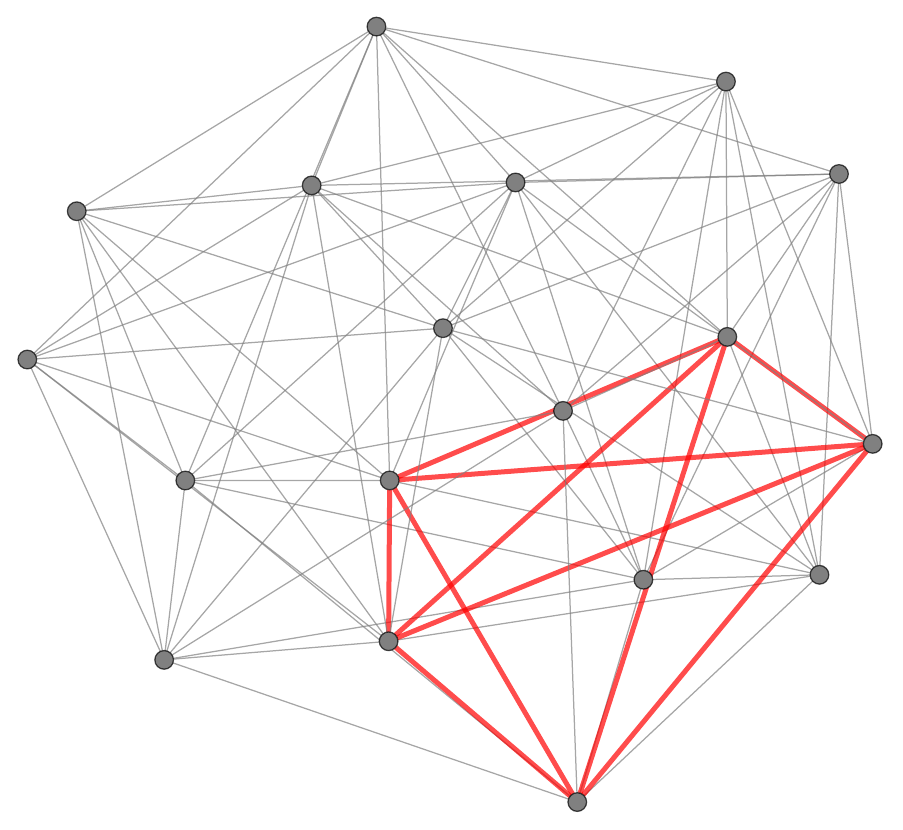} &
\includegraphics[scale=0.25]{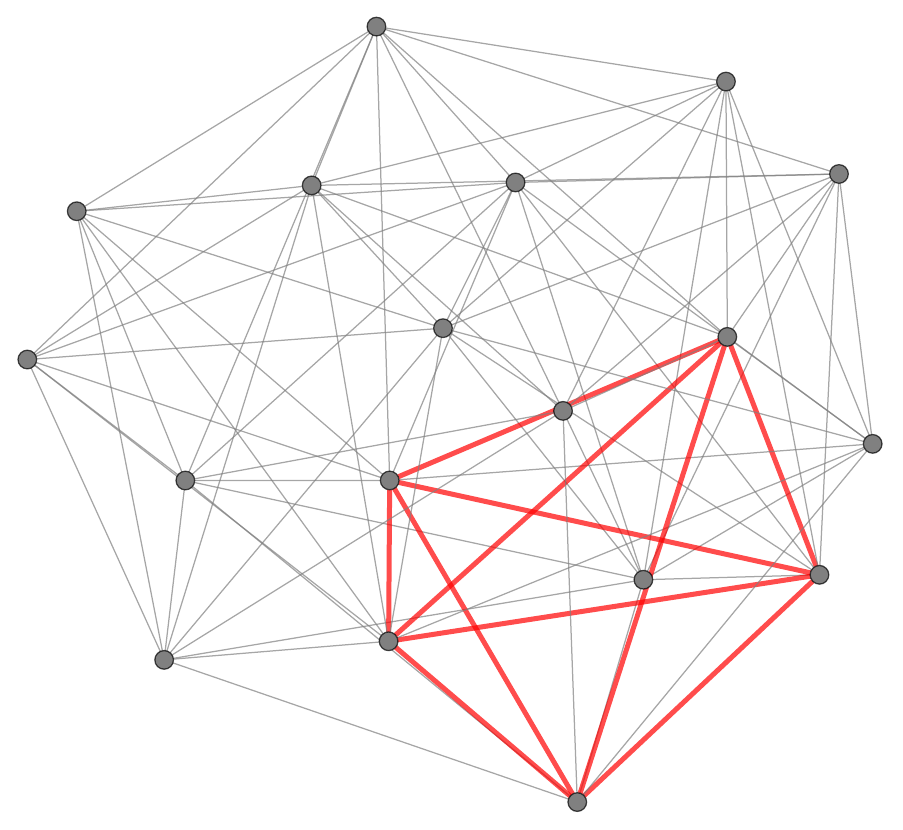} & 
\includegraphics[scale=0.25]{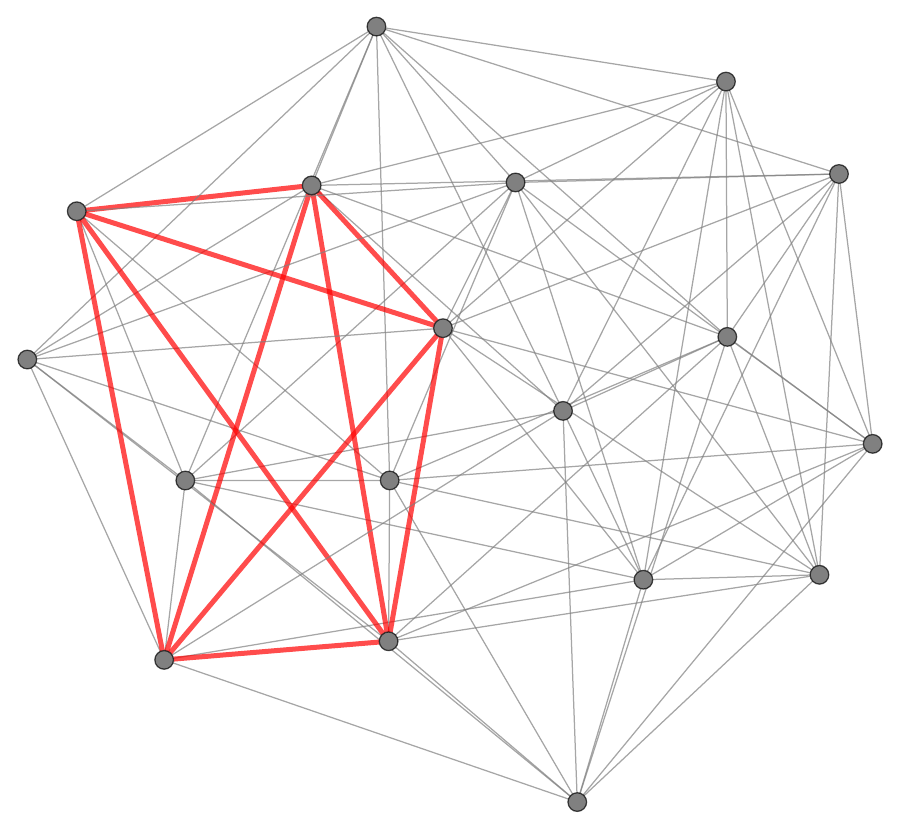} \\
\includegraphics[scale=0.25]{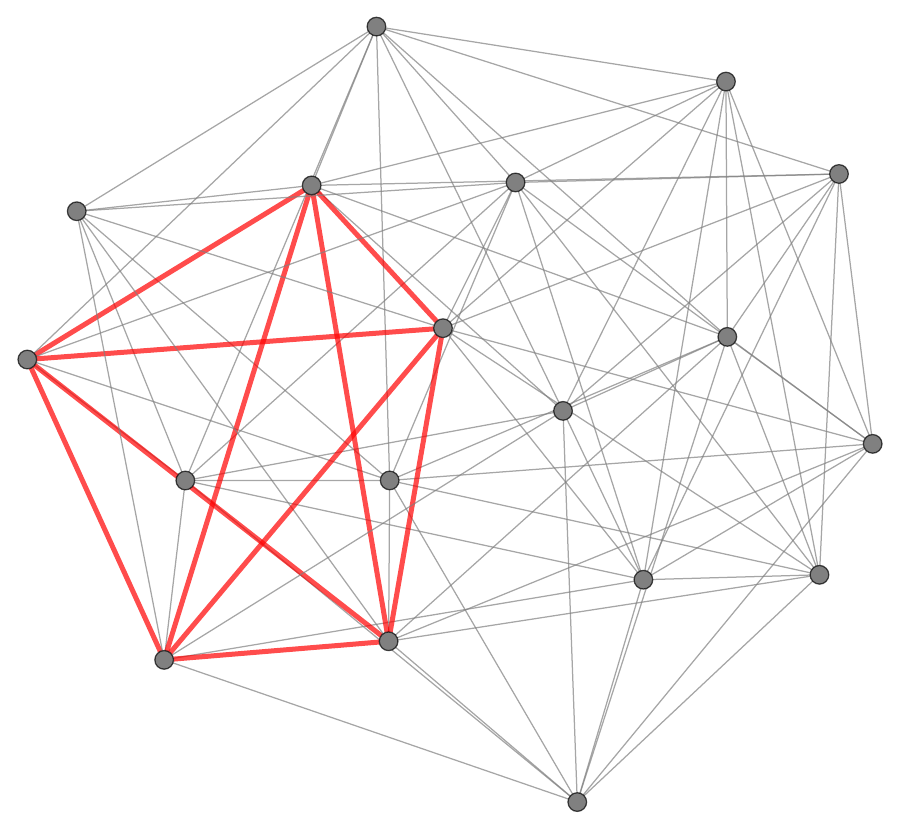} &
\includegraphics[scale=0.25]{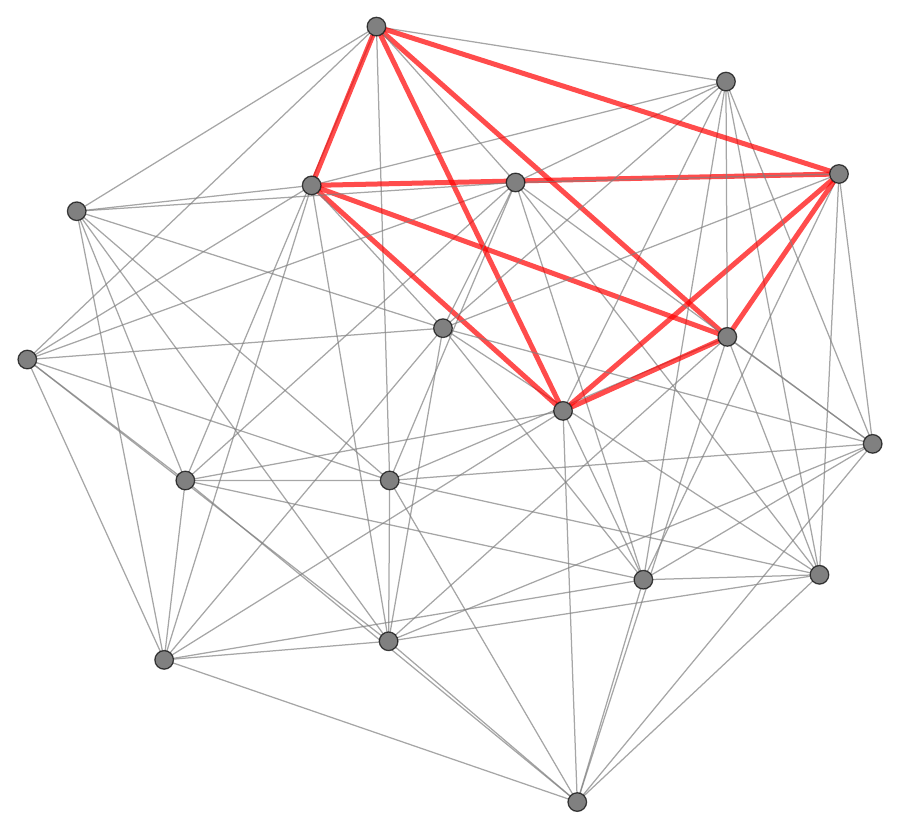} &
\includegraphics[scale=0.25]{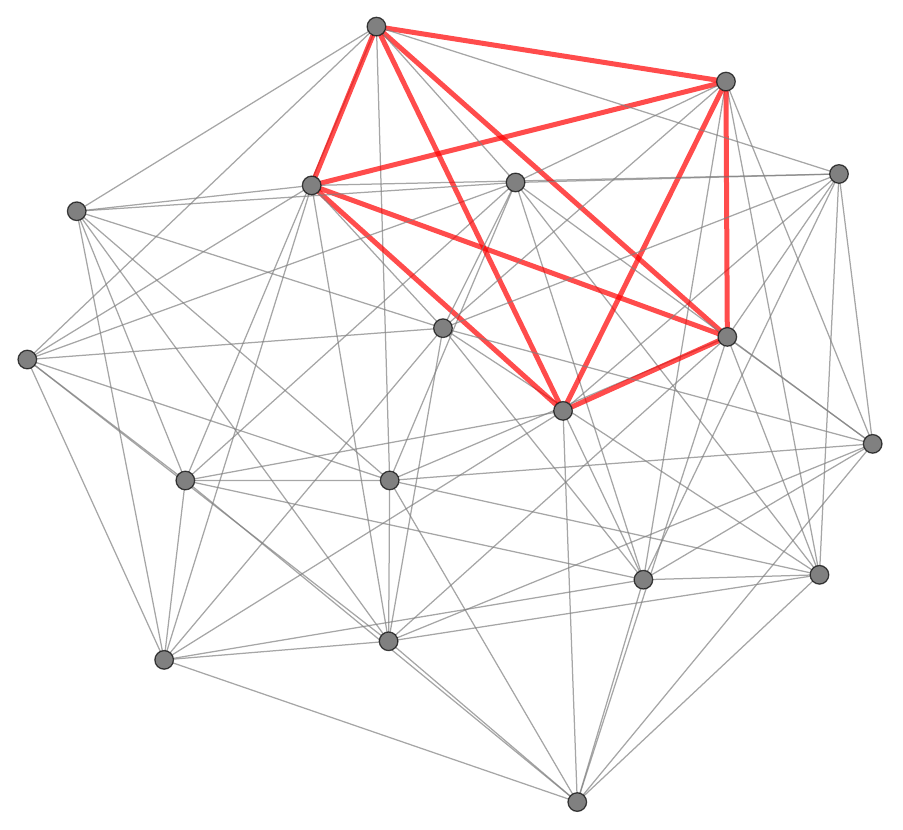} &
\includegraphics[scale=0.25]{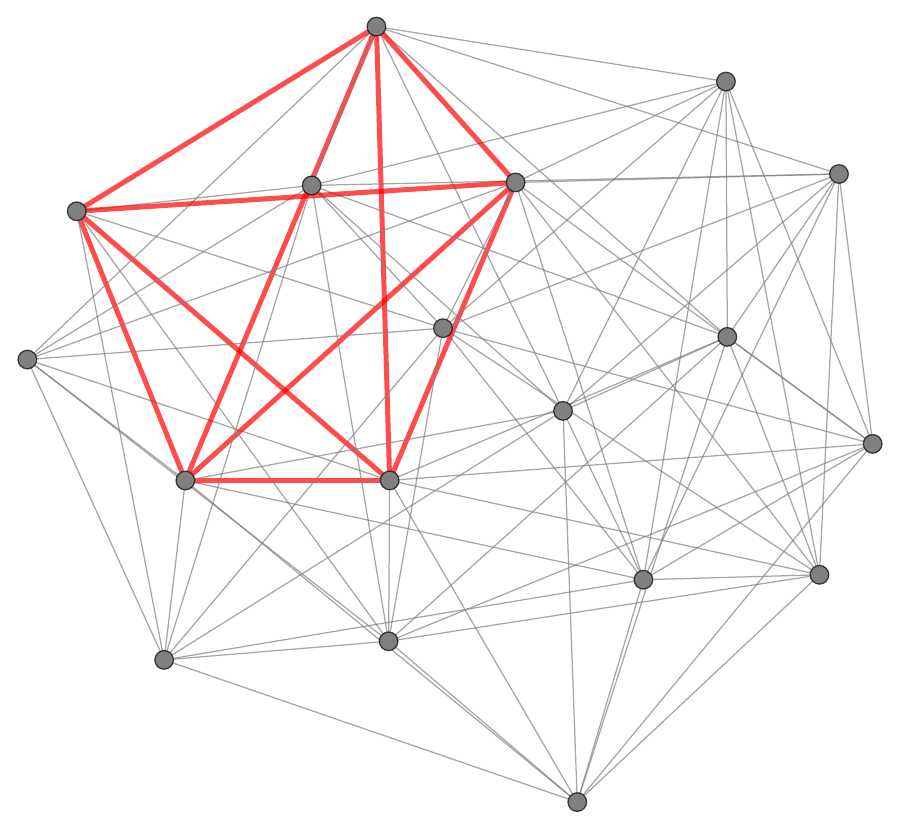} \\ 
\includegraphics[scale=0.25]{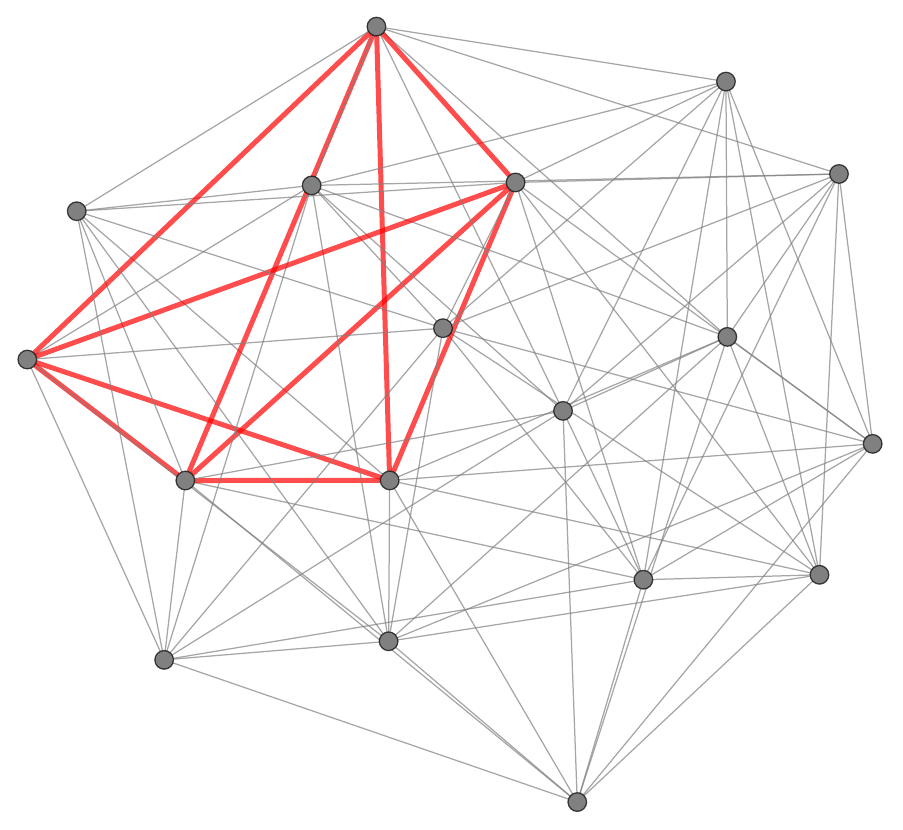} & 
\includegraphics[scale=0.25]{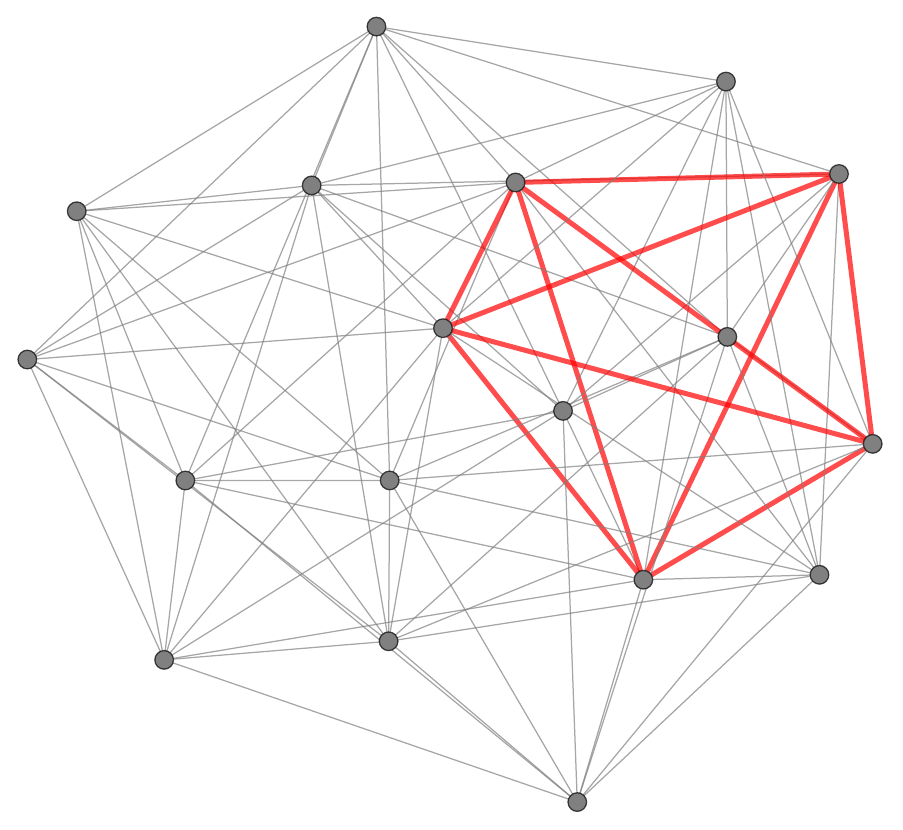} &
\includegraphics[scale=0.25]{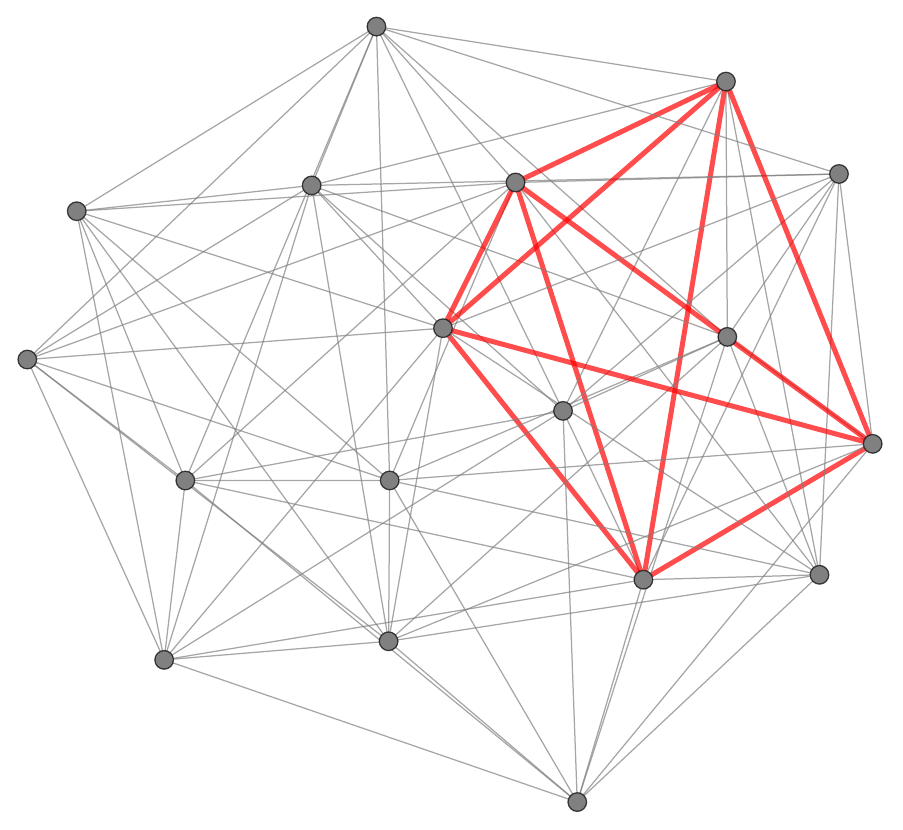} &
\includegraphics[scale=0.25]{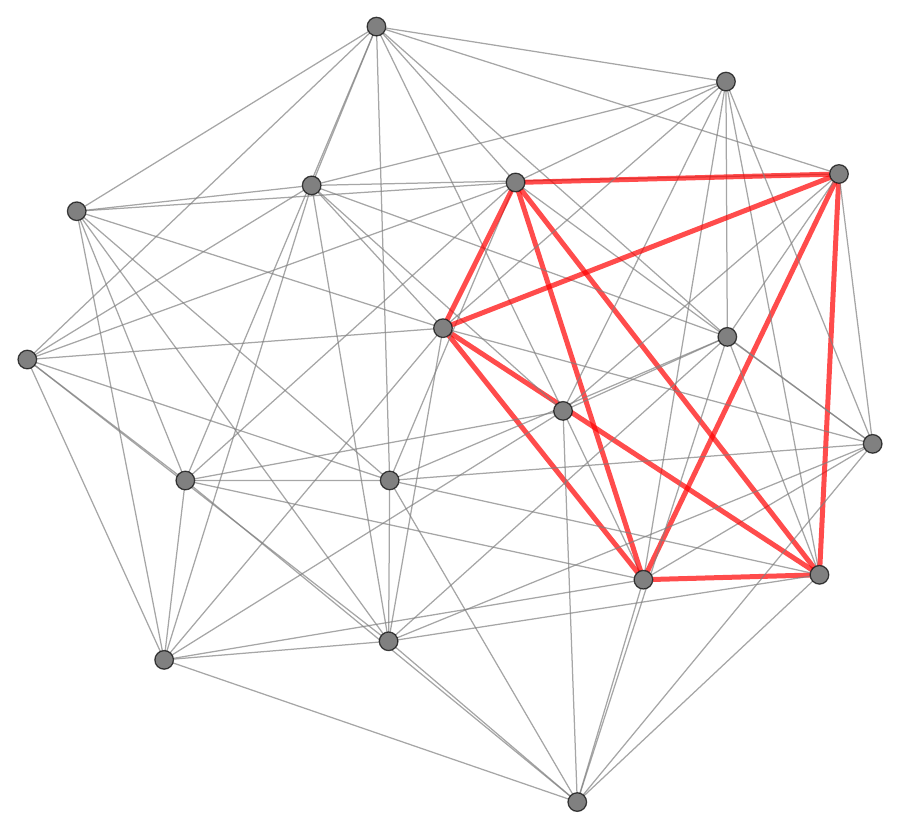} \\
& \includegraphics[scale=0.25]{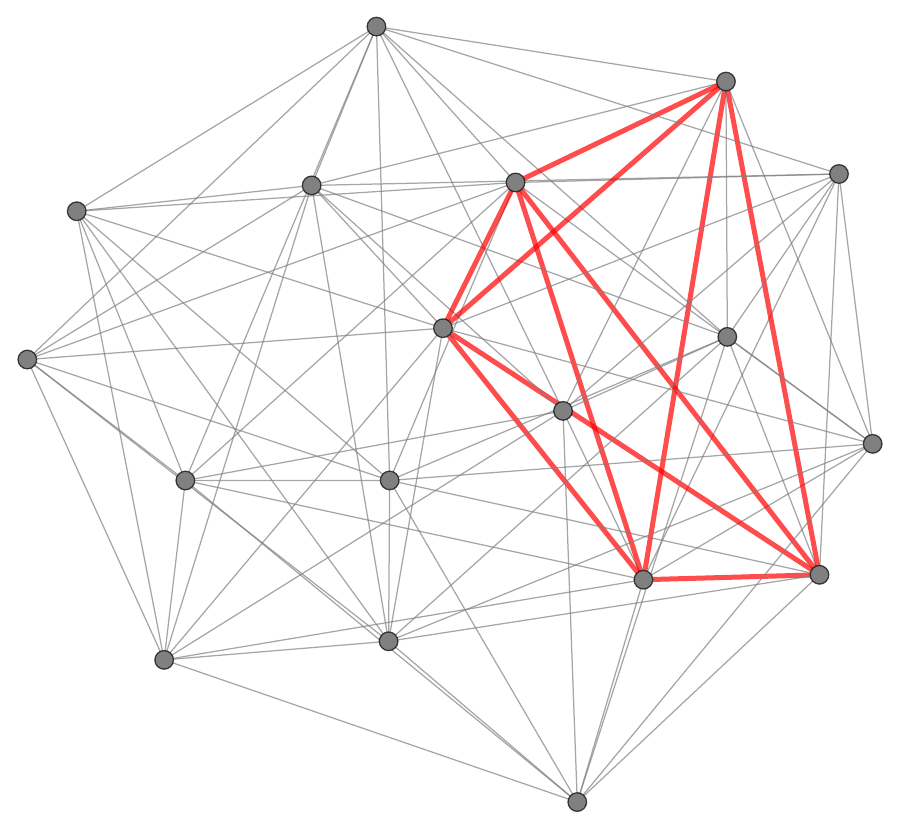} &&
\end{longtable}
    \caption{Complete subgraphs of $\mathbf{G}_2$}
    \label{fig:complete-sub-g2}
\end{figure}

As maximal complete sets of mutually complementary classical structures on two qubits, the complete subgraphs in Fig. \ref{fig:complete-sub-g2} are given as the following sets of names of classical structures:
\begin{eqnarray*}
\left\{
\begin{pmatrix}0&0\\0&0\\0&0\end{pmatrix},
\begin{pmatrix}0&1\\0&1\\1&0\end{pmatrix},
\begin{pmatrix}1&0\\0&1\\1&0\end{pmatrix},
\begin{pmatrix}1&1\\0&0\\0&0\end{pmatrix},
\begin{pmatrix}Z&Z\\0&0\\0&0\end{pmatrix}
\right\}\\
\left\{
\begin{pmatrix}0&1\\0&0\\0&0\end{pmatrix},
\begin{pmatrix}1&0\\0&0\\0&0\end{pmatrix},
\begin{pmatrix}1&1\\0&1\\1&0\end{pmatrix},
\begin{pmatrix}Z&Z\\0&0\\0&0\end{pmatrix},
\begin{pmatrix}Z&Z\\0&i\\i&0\end{pmatrix}
\right\}\\
\left\{
\begin{pmatrix}0&0\\0&1\\1&0\end{pmatrix},
\begin{pmatrix}0&1\\0&0\\0&0\end{pmatrix},
\begin{pmatrix}1&0\\0&0\\0&0\end{pmatrix},
\begin{pmatrix}1&1\\0&1\\1&0\end{pmatrix},
\begin{pmatrix}Z&Z\\0&0\\0&0\end{pmatrix}
\right\}\\
\left\{
\begin{pmatrix}0&Z\\0&0\\0&0\end{pmatrix},
\begin{pmatrix}1&0\\0&0\\0&0\end{pmatrix},
\begin{pmatrix}1&0\\0&1\\1&0\end{pmatrix},
\begin{pmatrix}Z&0\\0&j\\k&0\end{pmatrix},
\begin{pmatrix}Z&1\\0&0\\0&0\end{pmatrix}
\right\}\\
\left\{
\begin{pmatrix}0&Z\\0&0\\0&0\end{pmatrix},
\begin{pmatrix}0&Z\\0&k\\j&0\end{pmatrix},
\begin{pmatrix}1&0\\0&0\\0&0\end{pmatrix},
\begin{pmatrix}1&0\\0&1\\1&0\end{pmatrix},
\begin{pmatrix}Z&1\\0&0\\0&0\end{pmatrix}
\right\}\\
\left\{
\begin{pmatrix}0&Z\\0&0\\0&0\end{pmatrix},
\begin{pmatrix}1&1\\0&0\\0&0\end{pmatrix},
\begin{pmatrix}1&1\\0&1\\1&0\end{pmatrix},
\begin{pmatrix}Z&0\\0&0\\0&0\end{pmatrix},
\begin{pmatrix}Z&1\\0&j\\k&0\end{pmatrix}
\right\}\\
\left\{
\begin{pmatrix}0&Z\\0&0\\0&0\end{pmatrix},
\begin{pmatrix}1&1\\0&0\\0&0\end{pmatrix},
\begin{pmatrix}1&1\\0&1\\1&0\end{pmatrix},
\begin{pmatrix}1&Z\\0&k\\j&0\end{pmatrix},
\begin{pmatrix}Z&0\\0&0\\0&0\end{pmatrix}
\right\}\\
\left\{
\begin{pmatrix}0&1\\0&0\\0&0\end{pmatrix},
\begin{pmatrix}0&1\\0&1\\1&0\end{pmatrix},
\begin{pmatrix}1&Z\\0&0\\0&0\end{pmatrix},
\begin{pmatrix}Z&0\\0&0\\0&0\end{pmatrix},
\begin{pmatrix}Z&0\\0&j\\k&0\end{pmatrix}
\right\}\\
\left\{
\begin{pmatrix}0&1\\0&0\\0&0\end{pmatrix},
\begin{pmatrix}0&1\\0&1\\1&0\end{pmatrix},
\begin{pmatrix}0&Z\\0&k\\j&0\end{pmatrix},
\begin{pmatrix}1&Z\\0&0\\0&0\end{pmatrix},
\begin{pmatrix}Z&0\\0&0\\0&0\end{pmatrix}
\right\}\\
\left\{
\begin{pmatrix}0&0\\0&0\\0&0\end{pmatrix},
\begin{pmatrix}1&Z\\0&0\\0&0\end{pmatrix},
\begin{pmatrix}Z&1\\0&0\\0&0\end{pmatrix},
\begin{pmatrix}Z&1\\0&j\\k&0\end{pmatrix},
\begin{pmatrix}Z&Z\\0&i\\i&0\end{pmatrix}
\right\}\\
\left\{
\begin{pmatrix}0&0\\0&0\\0&0\end{pmatrix},
\begin{pmatrix}1&Z\\0&0\\0&0\end{pmatrix},
\begin{pmatrix}1&Z\\0&k\\j&0\end{pmatrix},
\begin{pmatrix}Z&1\\0&0\\0&0\end{pmatrix},
\begin{pmatrix}Z&Z\\0&i\\i&0\end{pmatrix}
\right\}\\
\left\{
\begin{pmatrix}0&0\\0&0\\0&0\end{pmatrix},
\begin{pmatrix}0&0\\0&1\\1&0\end{pmatrix},
\begin{pmatrix}1&Z\\0&0\\0&0\end{pmatrix},
\begin{pmatrix}Z&1\\0&0\\0&0\end{pmatrix},
\begin{pmatrix}Z&1\\0&j\\k&0\end{pmatrix}
\right\}\\
\left\{
\begin{pmatrix}0&0\\0&0\\0&0\end{pmatrix},
\begin{pmatrix}0&0\\0&1\\1&0\end{pmatrix},
\begin{pmatrix}1&Z\\0&0\\0&0\end{pmatrix},
\begin{pmatrix}1&Z\\0&k\\j&0\end{pmatrix},
\begin{pmatrix}Z&1\\0&0\\0&0\end{pmatrix}
\right\}
\end{eqnarray*}

\section{Complementary CS on Three Qubits}\label{sec:CCS-3Q}

For three qubits, a complementarity diagram formed by composite classical structures $\mathcal{A}$ and $\mathcal{B}$ --- where the constituents of $\mathcal{A}$ have spiders $\blue, \purple$ and $\orange$, and constituents of $\mathcal{B}$ have spiders $\pink$, $\brown$ and $\cyan$ --- satisfies the complementarity condition if it is equal to the following diagram:
\begin{equation*}
    \includegraphics[scale=0.2]{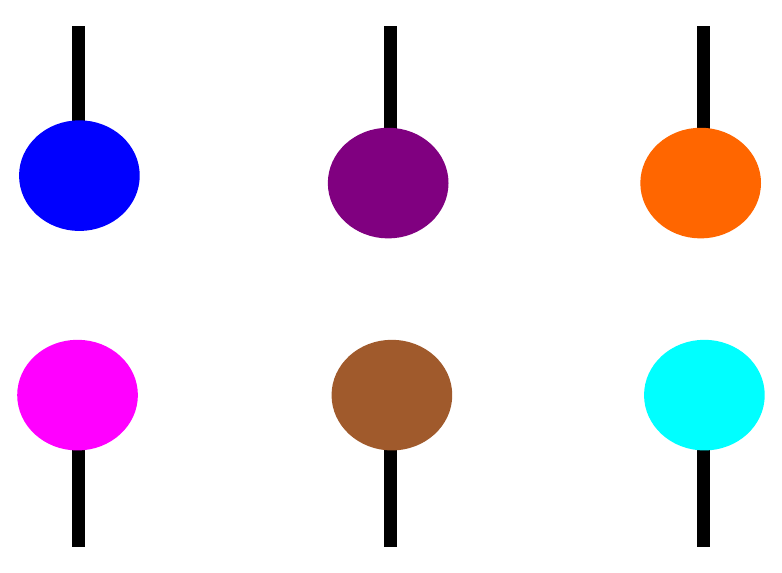}
\end{equation*}

To obtain LHS of the complementarity condition (Eq. \ref{eq:complementarity}), we find complementarity diagrams on three qubits which consists of the following slices:
\begin{equation*}
    \includegraphics[scale=0.2]{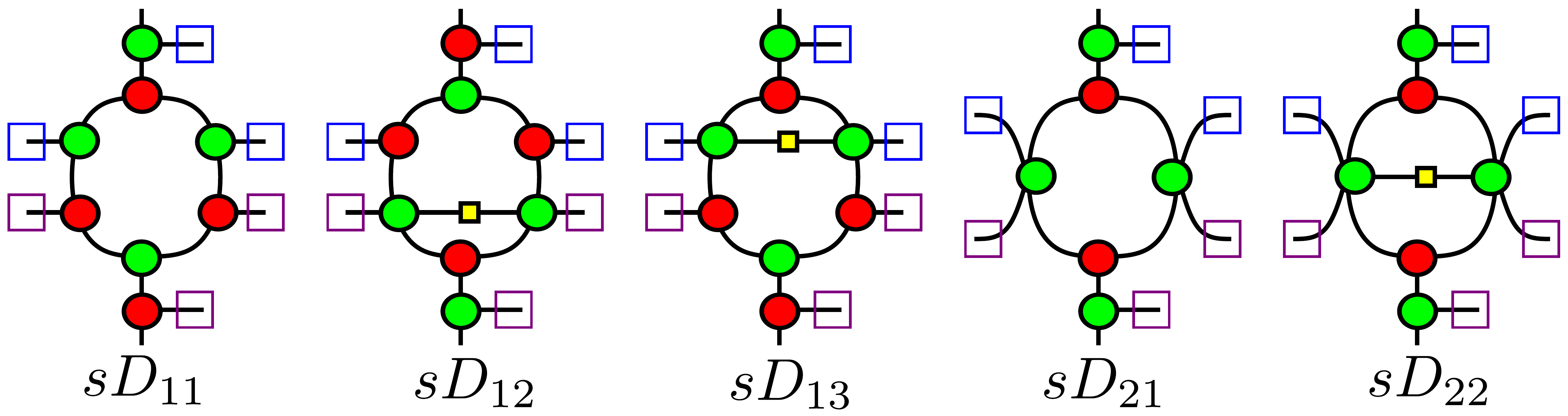}
\end{equation*}

However, we already know the complementarity diagrams on a single qubit and two qubits which generate complementarity diagrams on those respective number of qubits. So, we can ignore, for the moment, those complementarity diagrams on three qubits which are separable on all qubits and those which are biseparable, i.e. complementarity diagrams which are separable on two pairs of qubits and non-separable on one pair of qubits. 

Out of the non-separable complementarity diagrams on three qubits which contain $sD_{11}, sD_{12}, sD_{13}, sD_{21},$ and $sD_{22}$ as subdiagrams, we can eliminate those with names containing at least one column for which the first entry is $1+Z$ and the other entries of the column do not belong in $\{i,k,1+i,1+k,j+i,j+k\}$. We can further reduce the complementarity diagrams by keeping one representative of complementarity diagrams which are equivalent up to permutation on qubits, and discarding the rest. Such a reduced complementarity diagrams on three qubits form a subset of $\mathsf{G}_3$ which contain members which are non-separable on all qubits. We denote this set as $\mathsf{G}_3^\text{ent}$. The members of the $\mathsf{G}_3^\text{ent}$ are given in Section \ref{sec:entCD3Q}. The matrices within boxes in Section \ref{sec:entCD3Q} are names of complementarity diagrams belonging to $\mathsf{G}^\text{ent}_3\cap\mathsf{G}^\text{pass}_3$. We can obtain the rest of the members of $\mathsf{G}^\text{pass}_3$ from $\mathsf{G}^\text{pass}_1$ and $\mathsf{G}^\text{pass}_2$.   

The following are the separable complementarity diagrams on three qubits which satisfy the complementarity condition:

\begin{equation*}
    \includegraphics[scale=0.18]{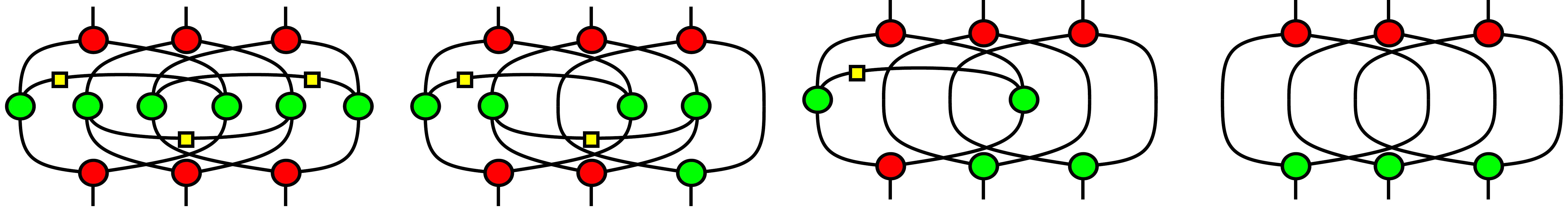}
\end{equation*}

For the biseparable complementarity diagrams on three qubits, we can simply separably compose those complementarity diagrams on a single qubit and those on two qubits, shown below, which satisfy the complementarity condition:
\begin{equation*}
    \includegraphics[scale=0.2]{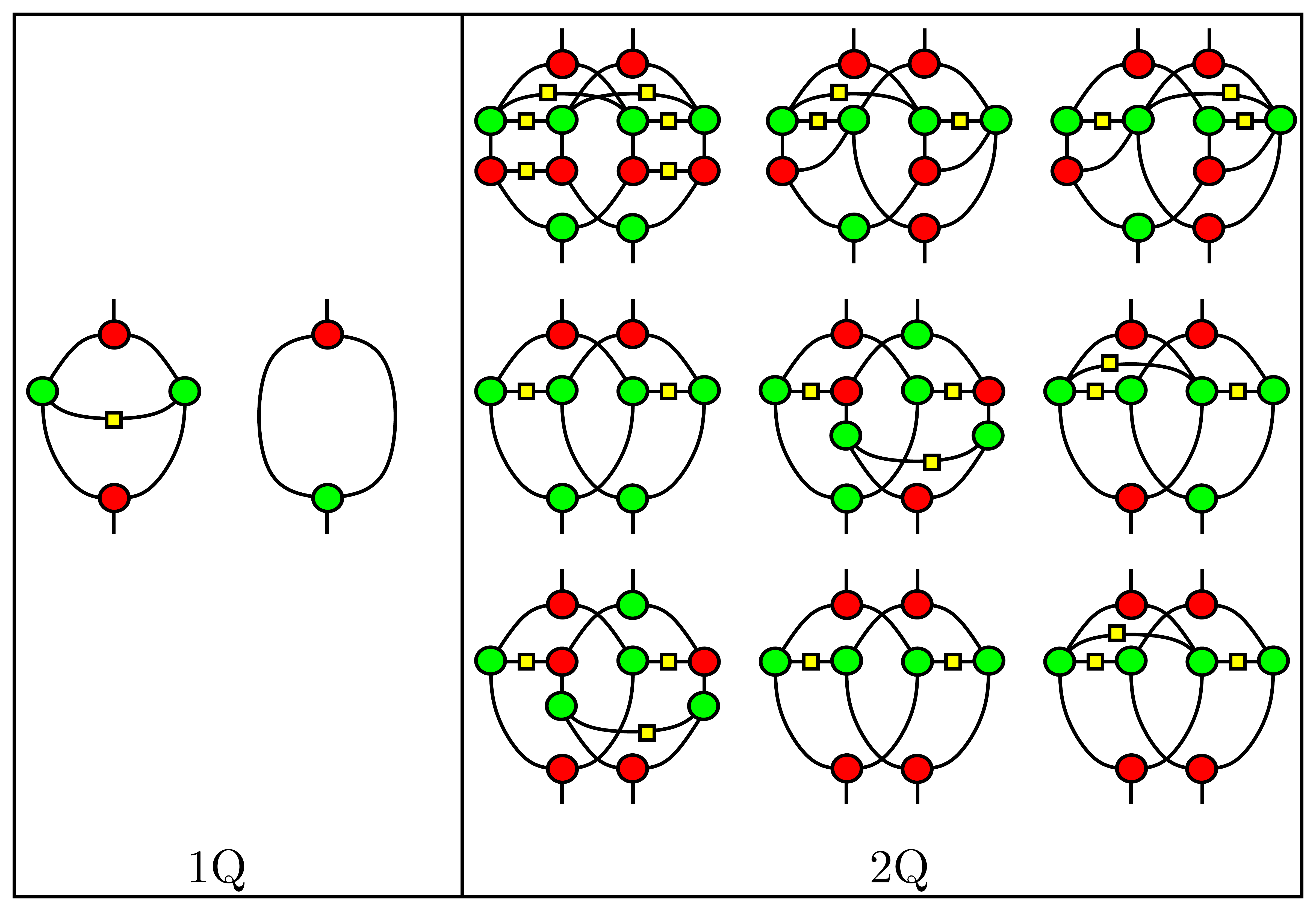}
\end{equation*}

So, in total, $\mathsf{G}^\text{pass}_3$ has 251 members. With this, we are able to obtain $\mathsf{P}_3$, and consequently, $\mathsf{T}_3$. Then we can check complementarity between every pair of composite classical structures on three qubits and form the graph $\mathbf{G}_3$. Via our computation using Mathematica, we found 32,448 complete sets of mutually complementary classical structures on three qubits. All such sets have entanglement configurations in accordance to the results in reference \cite{Romero2005}. Here, we give one example for each entanglement configuration. By entanglement configuration, we mean $(n_\text{SC},n_\text{BS},n_\text{NS})$ of a complete set of mutually complementary classical structures where $n_\text{SC}$ is the number of separably-composed classical structures in the set, $n_\text{BS}$ is the number of biseparable classical structures in the set, and $n_\text{NS}$ is the number of non-separable classical structures in the set. For three qubits, there are four entanglement configurations: (3,0,6), (2,3,4), (1,6,2), and (0,9,0) \cite{Romero2005}. 

An example of a complete set of mutually complementary classical structures for (3,0,6) is:
\input{list/eg306}

An example of a complete set of mutually complementary classical structures for (2,3,4) is:
\input{list/eg234}

An example of a complete set of mutually complementary classical structures for (1,6,2) is:
\input{list/eg162}

An example of a complete set of mutually complementary classical structures for (0,9,0) is:
\input{list/eg090}

\section{Chapter Summary}

In this chapter,
\begin{itemize}
    \item We introduced complementarity diagrams in Section \ref{sec:CD-intro} and presented its general form for composite classical structures on two qubits;
    \item Generalizing to $N$ qubits, we showed how to partition complementarity diagrams into equivalence classes so that we may search for those diagrams which satisfy the complementarity diagrams more efficiently;
    \item We also highlighted those complementarity diagrams which can be formed by more than a single pair of composite classical structures by splitting a complementarity diagram into slices which correspond to a subdiagram on an individual qubit; 
    \item Finally, we introduced a procedure for obtaining every pair of complementary classical structures for the case of two and three qubits (see Section \ref{sec:entCD3Q});
    \item Once we have determined all pairs of complementary classical structures, we form a graph with the classical structures as the vertices and there is an edge between two classical structures if they are complementary. From this graph, we obtain the maximal complete sets of mutually complementary classical structures. 
\end{itemize}

\chapter{Conclusion and Future Work}\label{sec:discussion}

In this thesis, we devised a procedure to compose classical structures on qubits. The two main ingredients of our procedure is the set of three mutually complementary classical structures on a single qubit --- the constituents which are denoted by $\mathcal{X},\mathcal{Y}$ and $\mathcal{Z}$ --- and bipartite processes which `stitch together' these constituents, i.e. connecting wires. Our procedure makes explicit the separability of the classical structures with respect to their underlying bases. We then showed that any pair of composite classical structures, obtained via our procedure, which are complementary produce a complementarity diagram that is separable on all qubits. Once we have established this condition for complementarity, we generated complementarity diagrams for the case of two and three qubits.

For two qubits, we generated entangled complementarity diagrams from 18 complementarity diagrams (see Fig. \ref{fig:CD-2Q}). We need not check all separable CD on two qubits since we can obtain those which satisfy the complementarity condition by separably composing CD on one qubit which satisfy the complementarity condition. Similarly, we can generate entangled complementarity diagrams on three qubits from 470 CD (see Section \ref{sec:entCD3Q}), and we can obtain separable and biseparable CD on three qubits which satisfy the complementarity condition from the CD on one and two qubits which satisfy the complementarity condition. For each of the two cases, we formed a graph with classical structures on the relevant number of qubits as its vertices and there is an edge between a pair of classical structures if they are complementary. From this graph, we obtain complete sets of mutually complementary classical structures by searching for complete subgraphs in the graph. In total, we found 13 maximal complete sets of mutually complementary classical structures on two qubits, and 32,448 maximal complete sets of mutually complementary classical structures on three qubits. Furthermore, the entanglement configuration of these sets agree with the results by \citet{Romero2005}.  

There is a steep increase of  
these generating complementarity diagrams from two qubits to three qubits. We suspect this trend would continue as the number of qubits grows. This does not quite work in our favour as we seek to simplify the search for complete sets of mutually complementary classical structures. However, relative to the set of all pairs of classical structures on three qubits, the size of the generating set of complementarity diagrams is small, i.e. 470 generating complementarity diagrams vs. 23,220 pairs of classical structures. So our procedure does significantly reduce the number of pairs of classical structures to be checked. 

In the following sections, we shall propose a strategy to make our procedure more efficient. This includes a discussion on composing connecting wires and classical structures with underlying bases containing equivalent states.  

\section{Composing Connecting Wires}

The following are examples of maximal complete sets of complementary classical structures on three qubits with entanglement configurations (1,6,2) and (3,0,6), respectively:
\input{list/egXY162}
\input{list/egXY306fr162}

We can obtain one of the sets above from the other by `deleting' or `generating' entanglement between the first and second constituents in each of the classical structures in either set. Entanglement between two constituents can be deleted due to Proposition \ref{prop:hopf-had}, which implies the following:
\begin{equation}\label{eq:cw-delete}
    \includegraphics[scale=0.2]{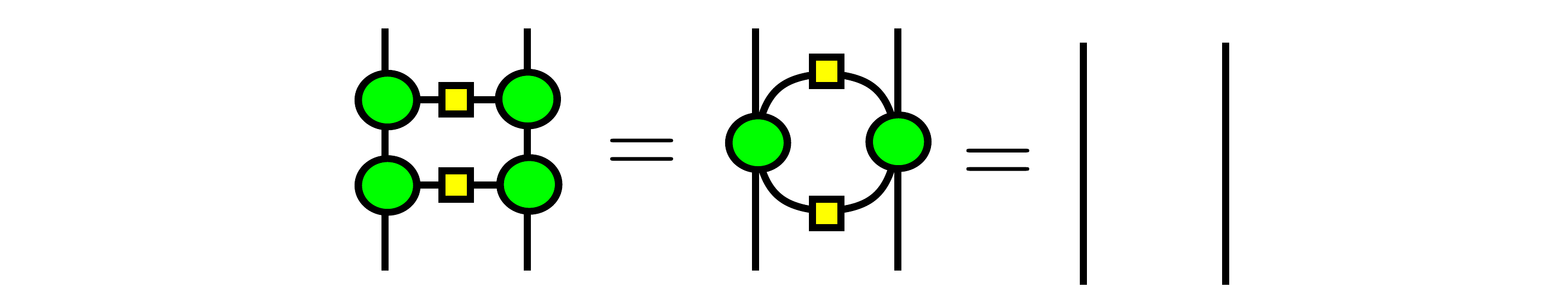}
\end{equation}
Using the previous equation, we obtain the following equation:
\begin{equation}
    \includegraphics[scale=0.2]{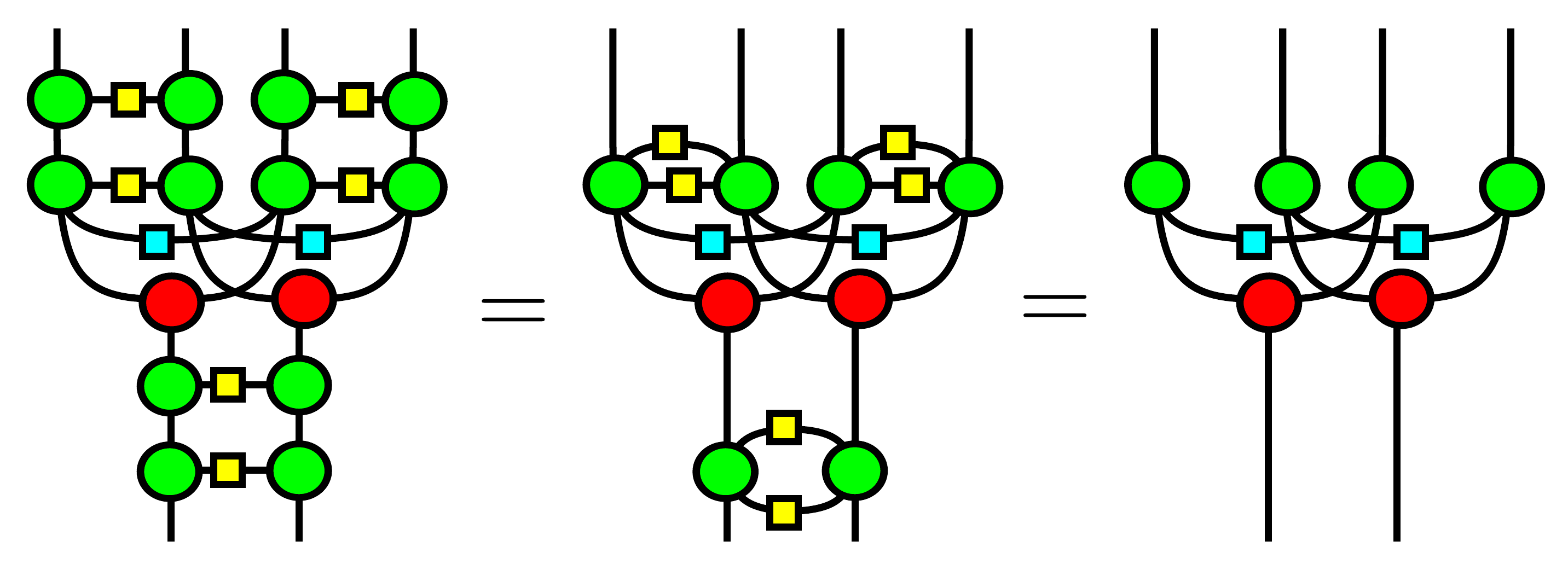} \label{eq:cw-xy-delete}
\end{equation}
where:
\begin{equation*}
    \includegraphics[scale=0.2]{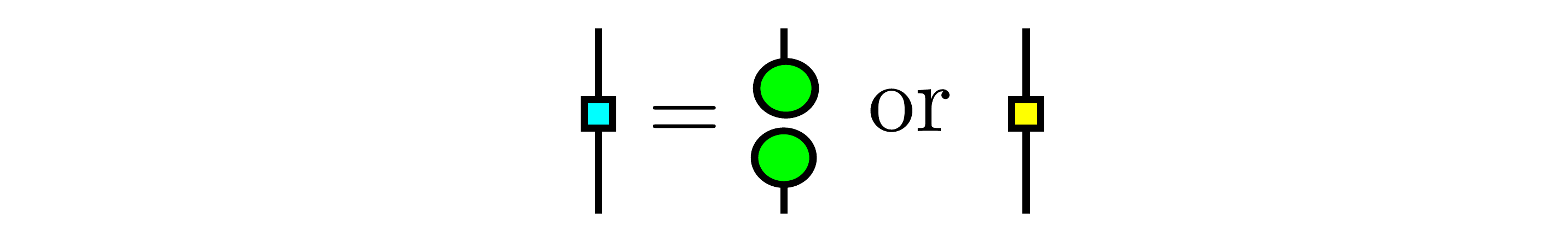}
\end{equation*}
We can generate entanglement between two qubits in the following way:
\begin{equation}\label{eq:cw-xy-generate}
    \includegraphics[scale=0.2]{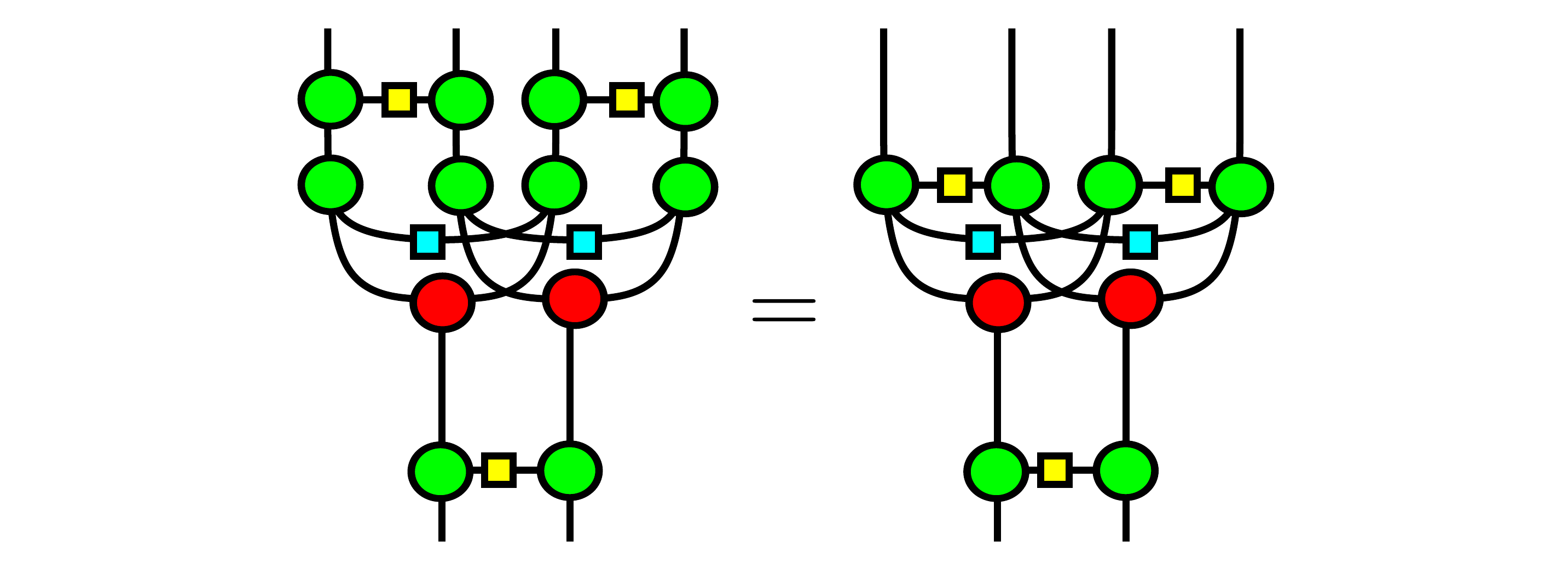}
\end{equation}
where:
\begin{equation*}
    \includegraphics[scale=0.2]{images/76-100/86-condition-xy.pdf}
\end{equation*}
So, entanglement is deleted and generated between two constituents via $\cwgg$.
However, Eqs. \ref{eq:cw-xy-delete} and \ref{eq:cw-xy-generate} do not include $\mathcal{ZZZ}$ which is present in both sets in the example above. In both equations, entanglement is generated or deleted between two constituents via $\cwgg$. For $\mathcal{ZZZ}$, entanglement cannot be generated by $\cwgg$ between any pair of its constituents due to the following:
\begin{equation}
    \includegraphics[scale=0.2]{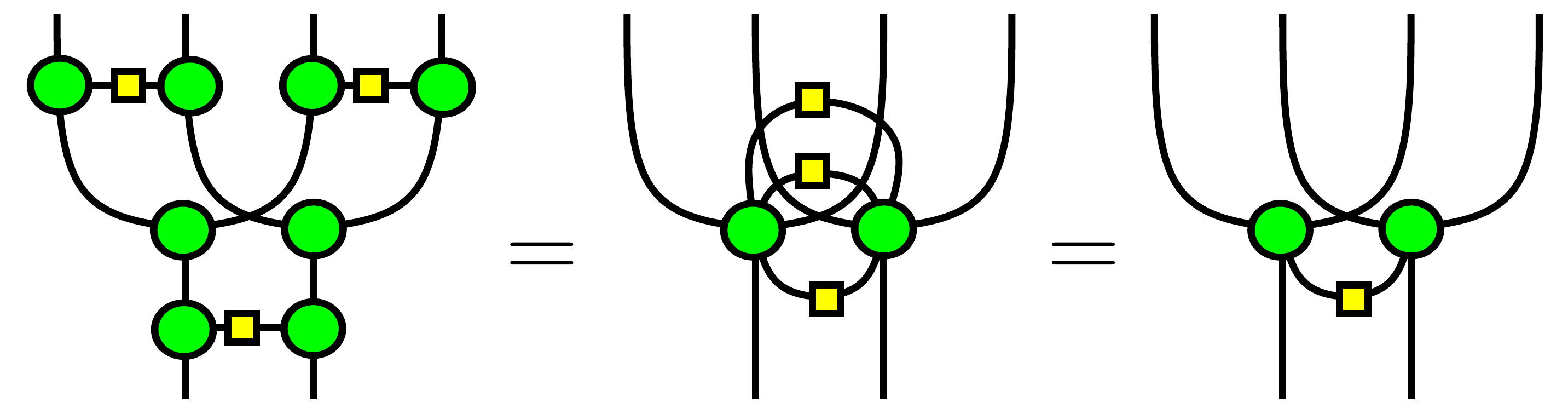} \label{eq:cw-z-same}
\end{equation}
The equivalence between $\mathcal{Z}\mathcal{Z}$ and a classical structure on two qubit which contains the spider on RHS of Eq. \ref{eq:cw-z-same} is given by Proposition \ref{prop:case3-SC}

We can describe the transformation between the two sets as a mapping between names. The transformation between the set with configuration (1,6,2) to the set with configuration (3,0,6), and vice versa, is described by the following mapping: 
\begin{equation*}
\begin{pmatrix}
c_1 & c_2 & c_3\\
0 & e_{1,2} & e_{1,3}\\
e_{1,2} & 0 & e_{2,3}\\
e_{1,3} & e_{2,3} & 0
\end{pmatrix}
\mapsto
\begin{cases}
\begin{pmatrix}
c_1 & c_2 & c_3\\
0 & e_{1,2}+1 & e_{1,3}\\
e_{1,2}+1 & 0 & e_{2,3}\\
e_{1,3} & e_{2,3} & 0
\end{pmatrix}
& \text{ if } c_1,c_2\in\{0,1\}
\\
\\
\begin{pmatrix}
c_1 & c_2 & c_3\\
0 & e_{1,2} & e_{1,3}\\
e_{1,2} & 0 & e_{2,3}\\
e_{1,3} & e_{2,3} & 0
\end{pmatrix} & \text{otherwise}
\end{cases}
\end{equation*}
where + is addition modulo 2. 

Unfortunately, the deletion and generation of entanglement allowed by Eqs. \ref{eq:cw-xy-delete} and \ref{eq:cw-xy-generate} are only applicable to complete sets of complementary classical structures which do not contain a composite classical structure with at least one constituent that is $\mathcal{Z}$, excluding $\mathcal{ZZZ}$. In fact, for the entanglement configuration (0,9,0), such a complete set of complementary classical structures does not exist. A question one might pose is: can we expand this technique to other types of constituents and/or connecting wires? As of now, we do not have a definitive answer for this, but it is of immense interest to us. Such a technique could simplify the search for complementary pairs of composite classical structures that we presented in the previous section, or ideally, it could give us every maximal complete set of complementary classical structures from only a single set of complementary classical structures. 
But before we get ahead of ourselves, we discuss some results which could serve as a jumping-off point for devising such a technique.  

\begin{proposition}\label{prop:cw-swap}
Composing $\cwgg$ and $\cwrr$, we obtain the following:
\begin{equation}\label{eq:cw-swap}
    \includegraphics[scale=0.2]{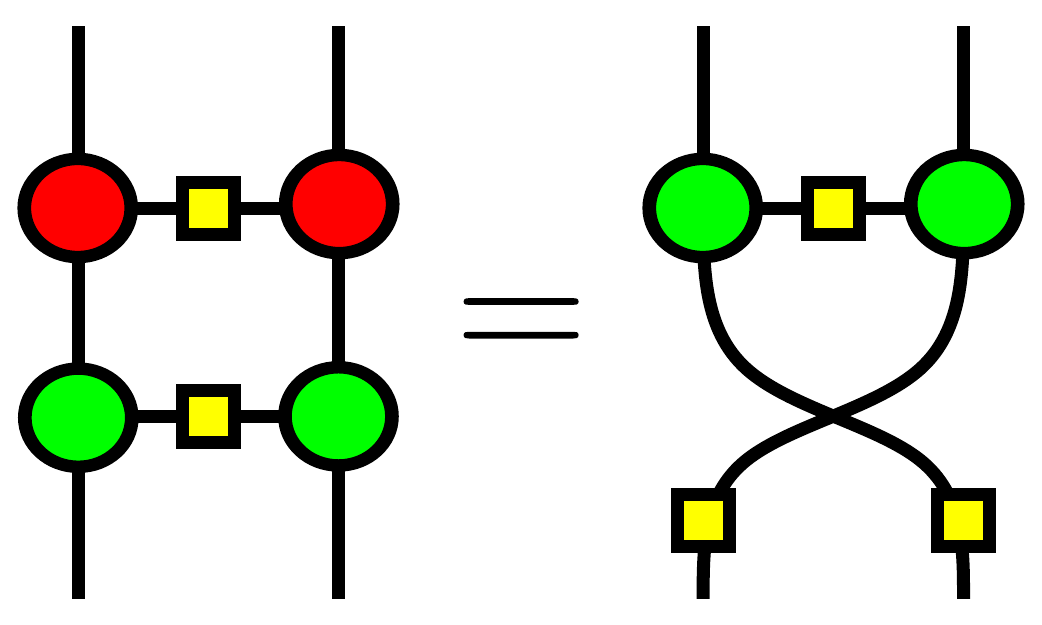}
\end{equation}
That is, the composition above results in $\cwgg$ where the input qubits are transformed by $\had$ and swapped with each other. 
\end{proposition}
\begin{proof}
\begin{equation*}
    \includegraphics[scale=0.2]{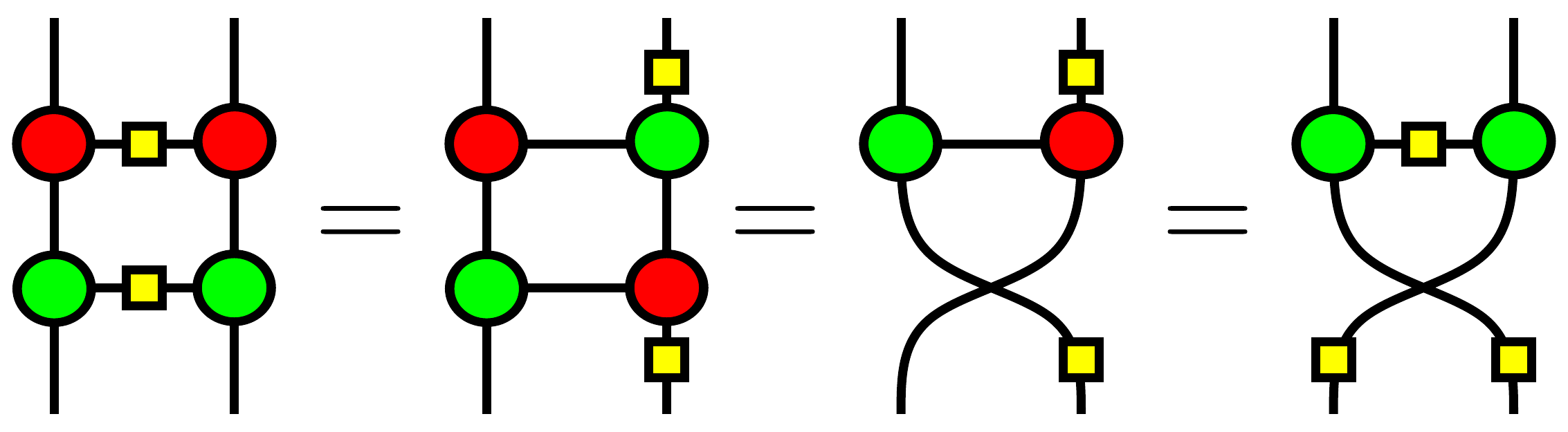}
\end{equation*}
\end{proof}

The previous proposition results in the following:
\begin{equation*}
    \includegraphics[scale=0.2]{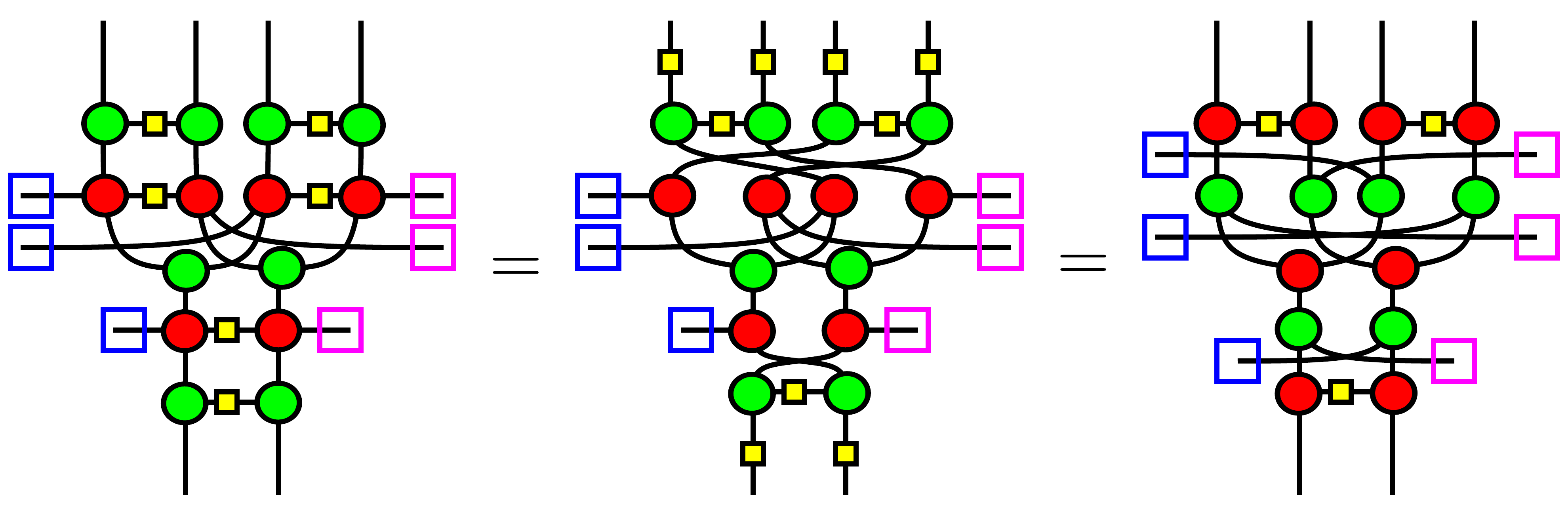}
\end{equation*}
So composing $\cwgg$ on the legs of an entangled pair of constituents $\mathcal{Z}\diamond\mathcal{Z}$ in any composite classical structure transforms the constituents and swaps their positions. 

Other types of composed connecting wires to consider are:
\begin{equation*}
    \includegraphics[scale=0.2]{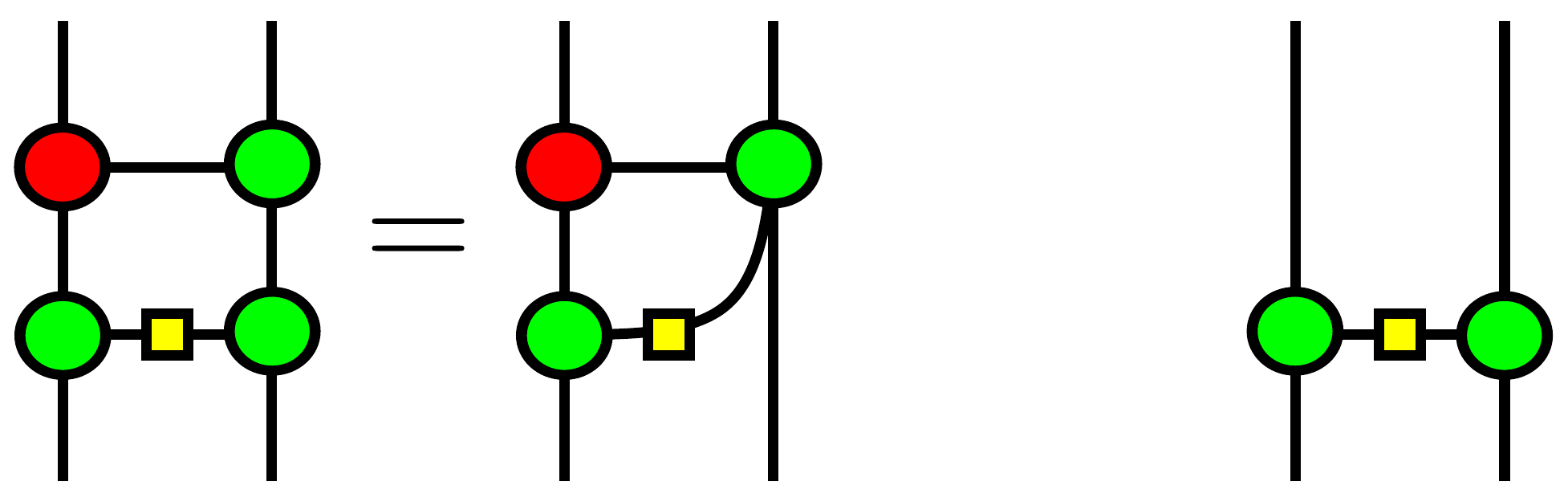}
\end{equation*}
The diagram on the right considers a pair of separable constituents for which the legs are composed with $\cwgg$. We have considered such a case for constituents $\mathcal{XX},\mathcal{XY},\mathcal{YX}$ and $\mathcal{YY}$ in Eqs. \ref{eq:cw-xy-delete} and \ref{eq:cw-xy-generate}. So we need only consider the composed connecting wire on the right for cases where the constituents are $\mathcal{XZ},\mathcal{YZ}$ and $\mathcal{ZZ}$.

Before we proceed further, we shall discuss the underlying bases of composite classical structures in an explicit manner, rather than the abstract way that we have done so far. We hope to use classical structures without referring to their underlying bases in future works, but for now, we shall content ourselves with proofs that utilize their underlying bases. 

% you already proved the equivalent Y classical structure in Theorem \ref{thm:y-equivalence}

\begin{theorem}
% take pair of constituents instead of the whole dang thing
A state $\ket{\psi}$ belongs to the underlying basis of a composite classical structure on $N$ qubits, i.e. a classical structure consisting of the following spiders:
\begin{equation*}
    \includegraphics[scale=0.15]{images/51-75/58-gen-composite-N.pdf}
\end{equation*}
if and only if $\ket{\psi}$ takes the following form:
\begin{equation*}
    \includegraphics[scale=0.2]{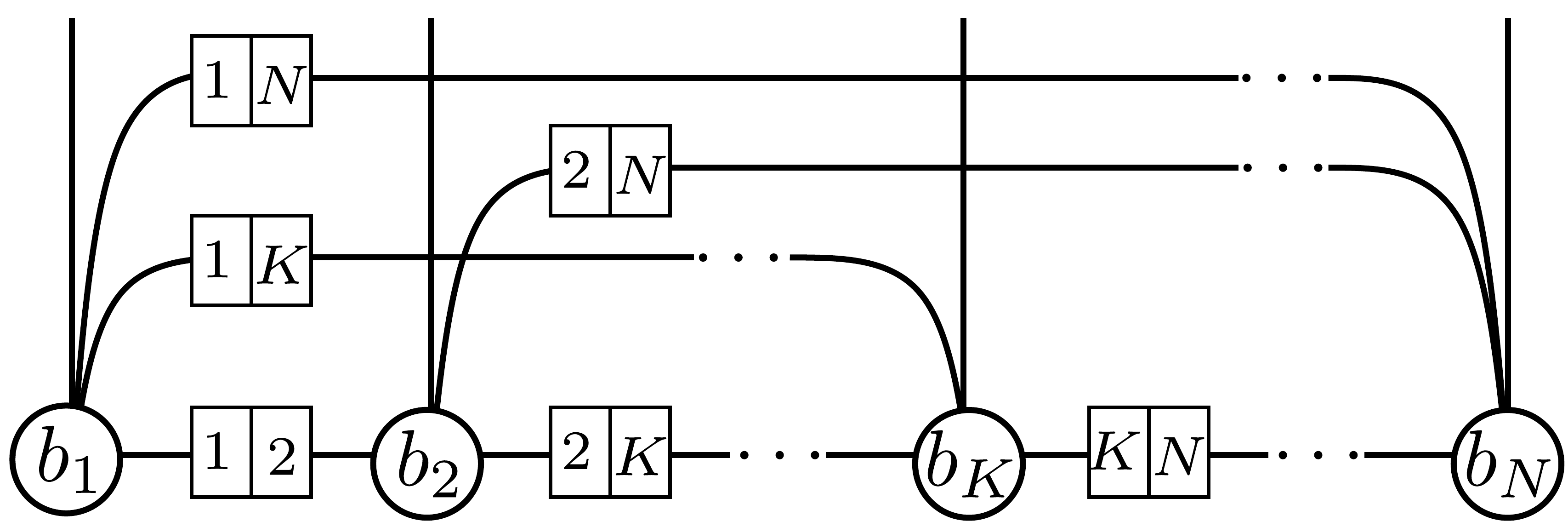}
\end{equation*}
where the white spider labelled $b_j$ is the white $S_j$ spider with phase $\beta_j$, and the white $S_j$ 0,1-spider with phase $\beta_j$ belongs to the underlying basis of the classical structure consisting of the black $S_j$ spider.

To get a more precise picture of the state above, we can look at a pair of slices of the state, say the $j$-th and $k$-th slices, along with the connecting wire between them:
\begin{equation*}
    \includegraphics[scale=0.3]{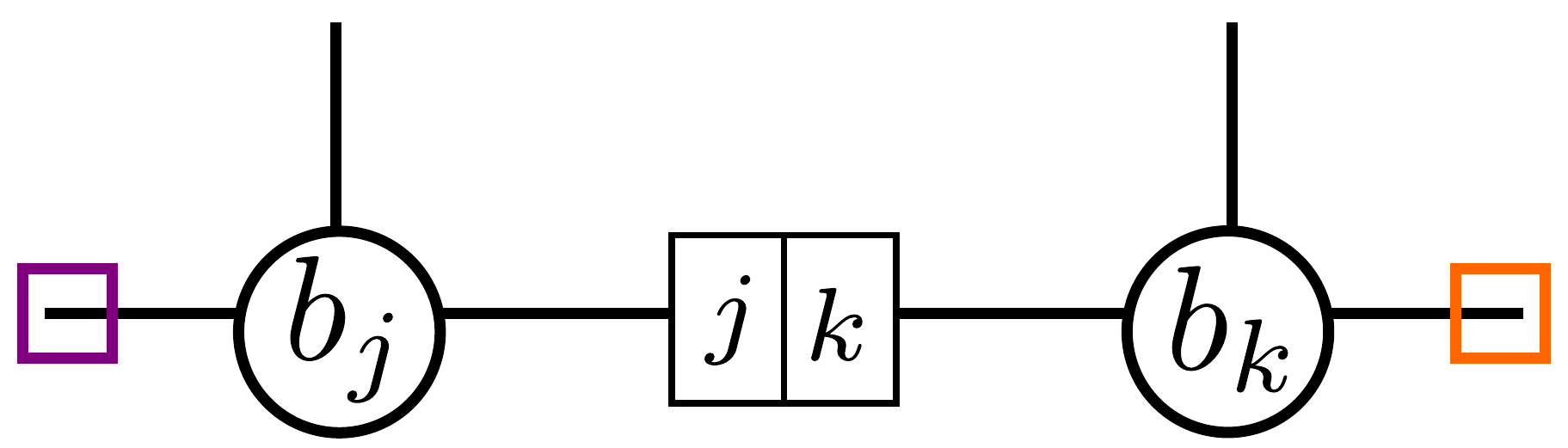}
\end{equation*}
\end{theorem}
\begin{proof}
In this proof, we shall utilize the equivalent spiders of $\mathcal{Y}$ from Theorem \ref{thm:y-equivalence}, and so $\mathcal{Y}$ in a composite classical structure is:
\begin{equation*}
    \includegraphics[scale=0.3]{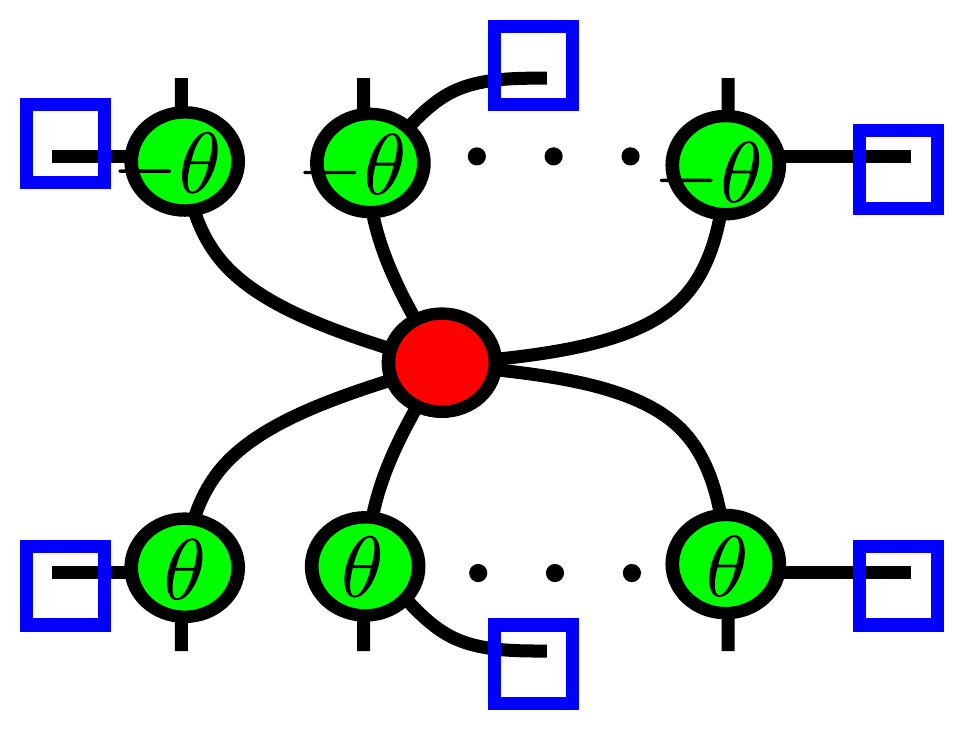}
\end{equation*}
where $\theta=\frac{\pi}{2}$.

The underlying basis of $\mathcal{Y}$ is:
\begin{equation*}
    \includegraphics[scale=0.35]{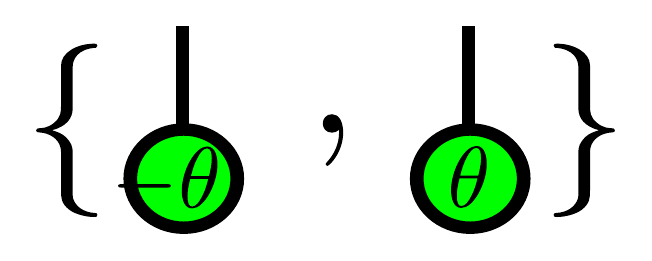}
\end{equation*}
since the following holds;
\begin{equation*}
    \includegraphics[scale=0.3]{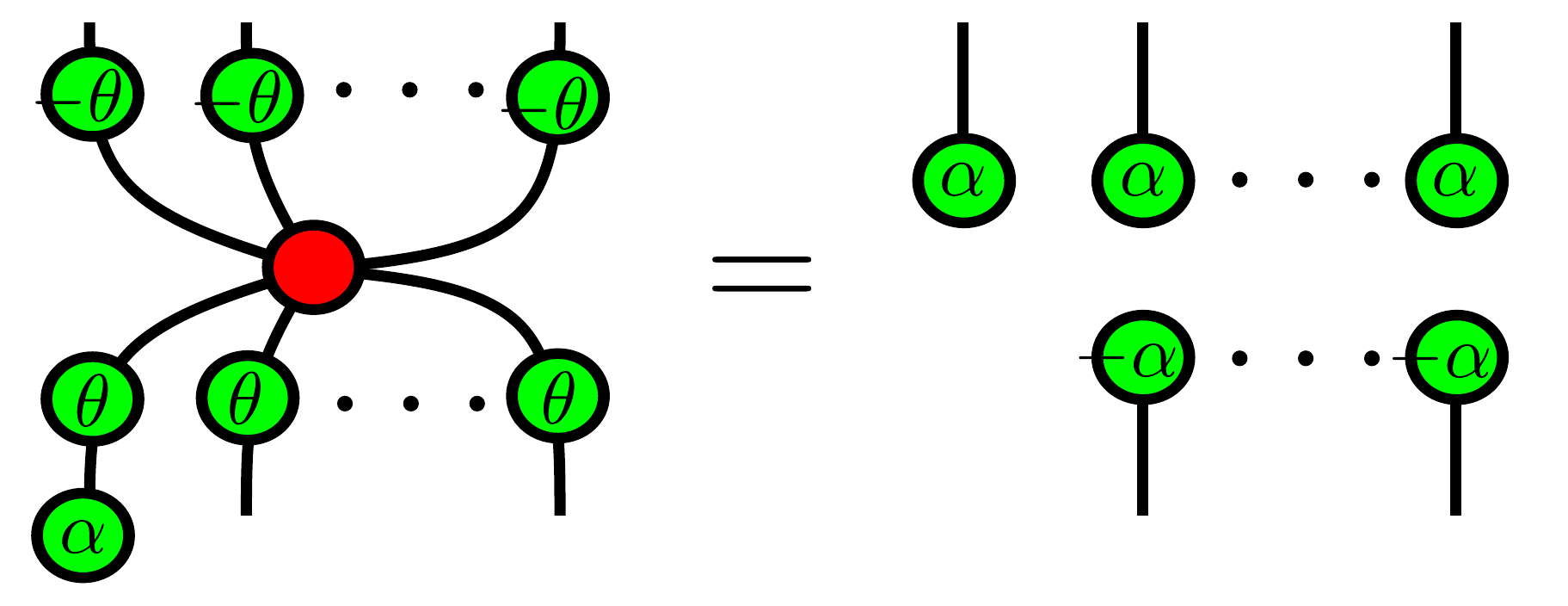}
\end{equation*}
where $\alpha\in\{\frac{\pi}{2},-\frac{\pi}{2}\}$. 

Furthermore, recall that the underlying basis of $\mathcal{Z}$ is:
\begin{equation*}
    \includegraphics[scale=0.3]{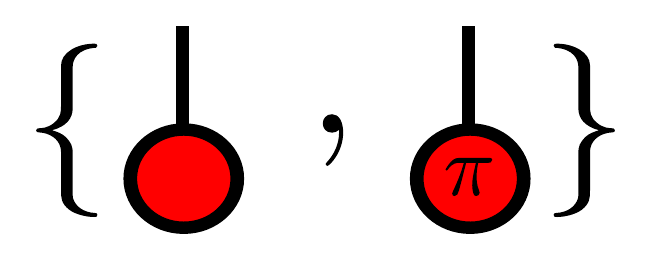}
\end{equation*} 
and the underlying basis of $\mathcal{X}$ is:
\begin{equation*}
    \includegraphics[scale=0.35]{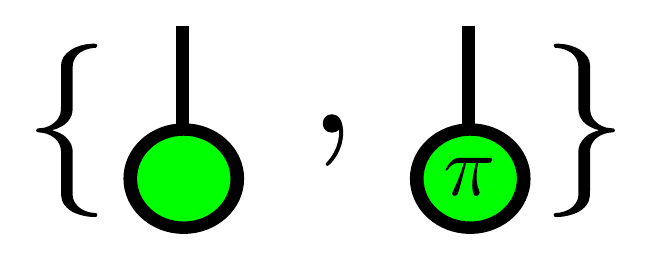}
\end{equation*}

Now, a state $\bunit$ belongs to the underlying basis of a classical structure with spiders $\white$ if:
\begin{equation*}
    \includegraphics[scale=0.2]{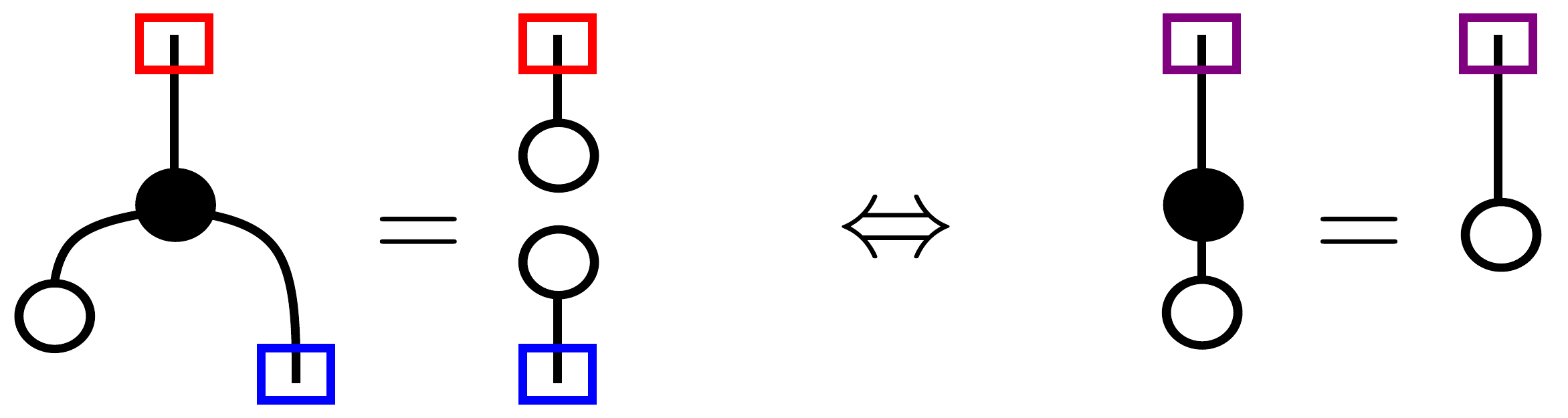}
\end{equation*}
We will show that the theorem's statement is true by proving the equation on the right.

For the $1,n$-spider of a composite classical structure on $N$ qubits, the following is the diagram containing the $j$-th and $k$-th slices along with the connecting wire between them:
\begin{equation*}
    \includegraphics[scale=0.2]{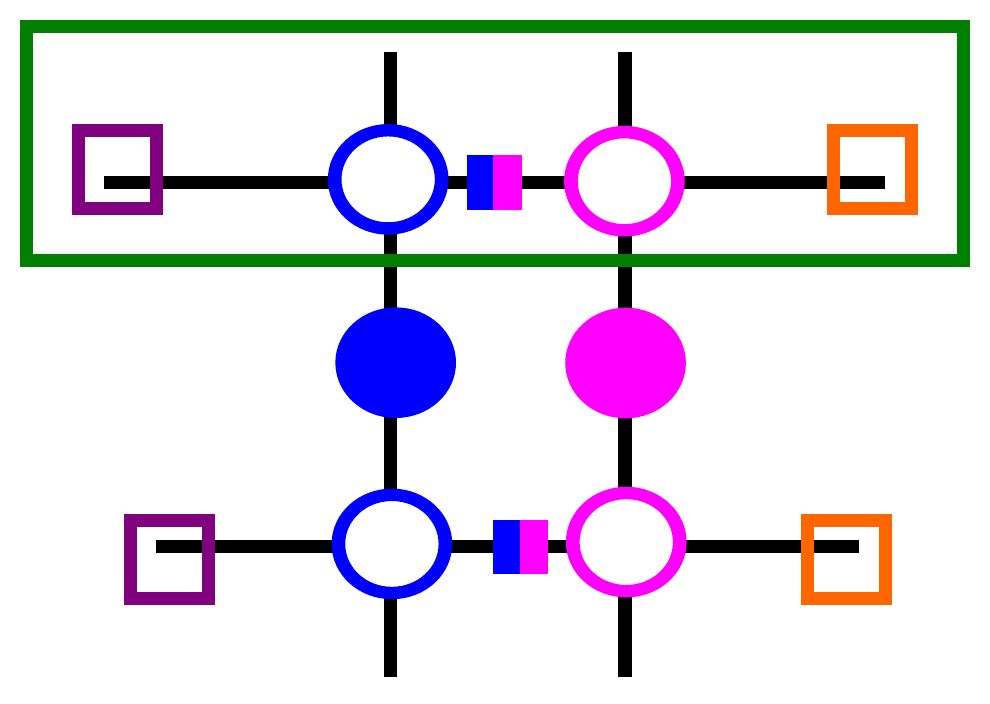}
\end{equation*}
Moreover,  due to the form that states of the underlying bases of $\mathcal{X}$, $\mathcal{Y}$ and $\mathcal{Z}$ take (see above), the $j$-th and $k$-th slices of the state in the theorem's statement can be drawn as follows:
\begin{equation*}
    \includegraphics[scale=0.2]{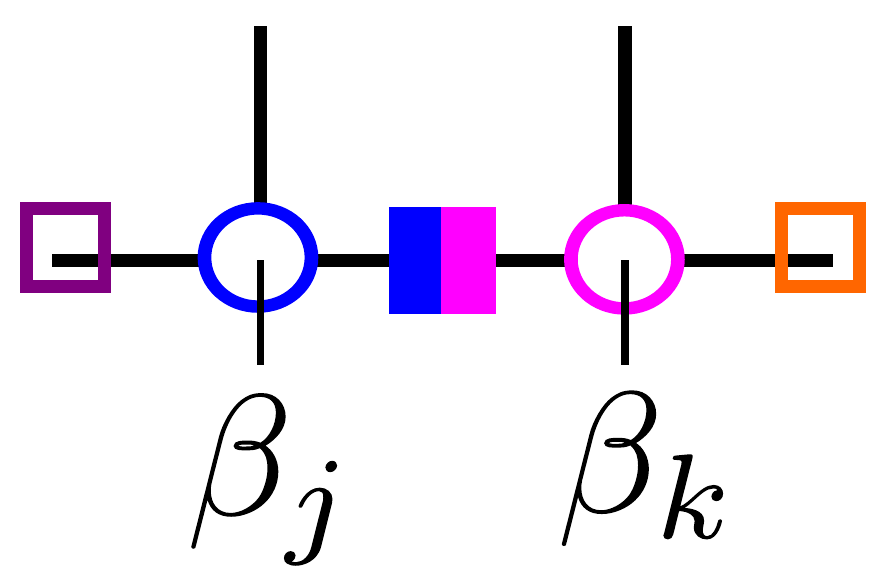}
\end{equation*}
where $\beta_j,\beta_k\in\{0,\pi,\frac{\pi}{2},-\frac{\pi}{2}\}$ depending on which one of $\mathcal{X},\mathcal{Y}$ and $\mathcal{Z}$ that $\blue$ and $\pink$ belong to. 

So, composing the previous two diagrams above, we have: 
\begin{equation*}
    \includegraphics[scale=0.2]{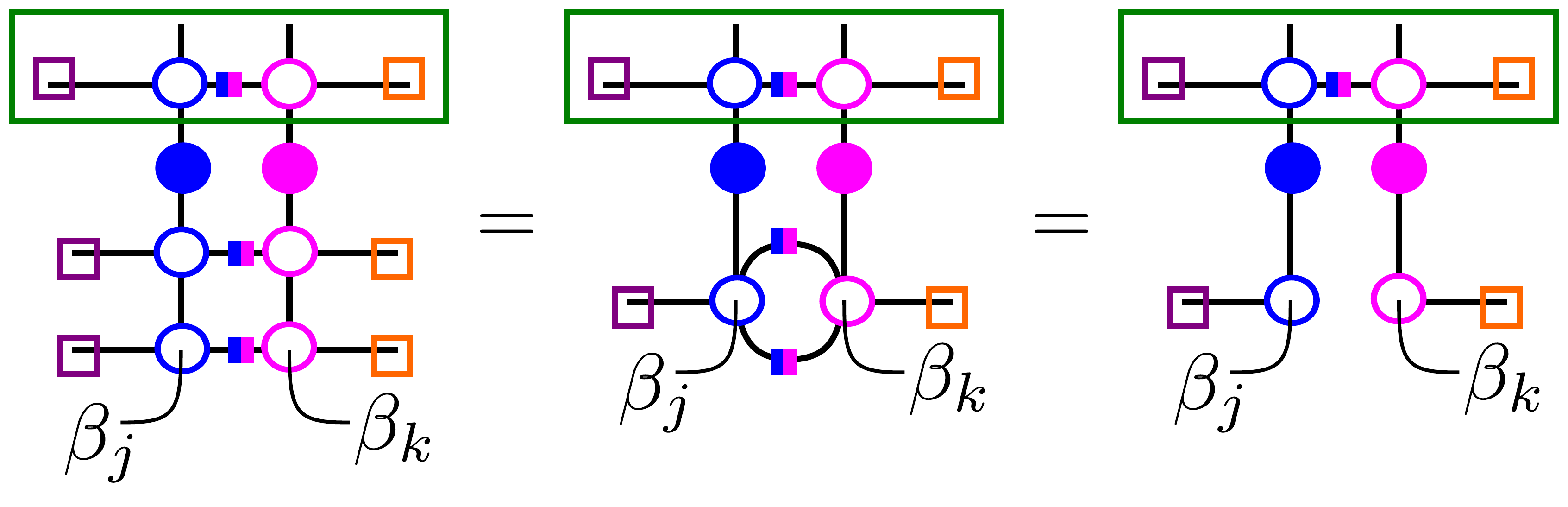}
\end{equation*}
However, since the pair of slices above are arbitrary, we can say the same for other pairs of slices, and so the rightmost diagram can be reduced to the following:
\begin{equation*}
    \includegraphics[scale=0.2]{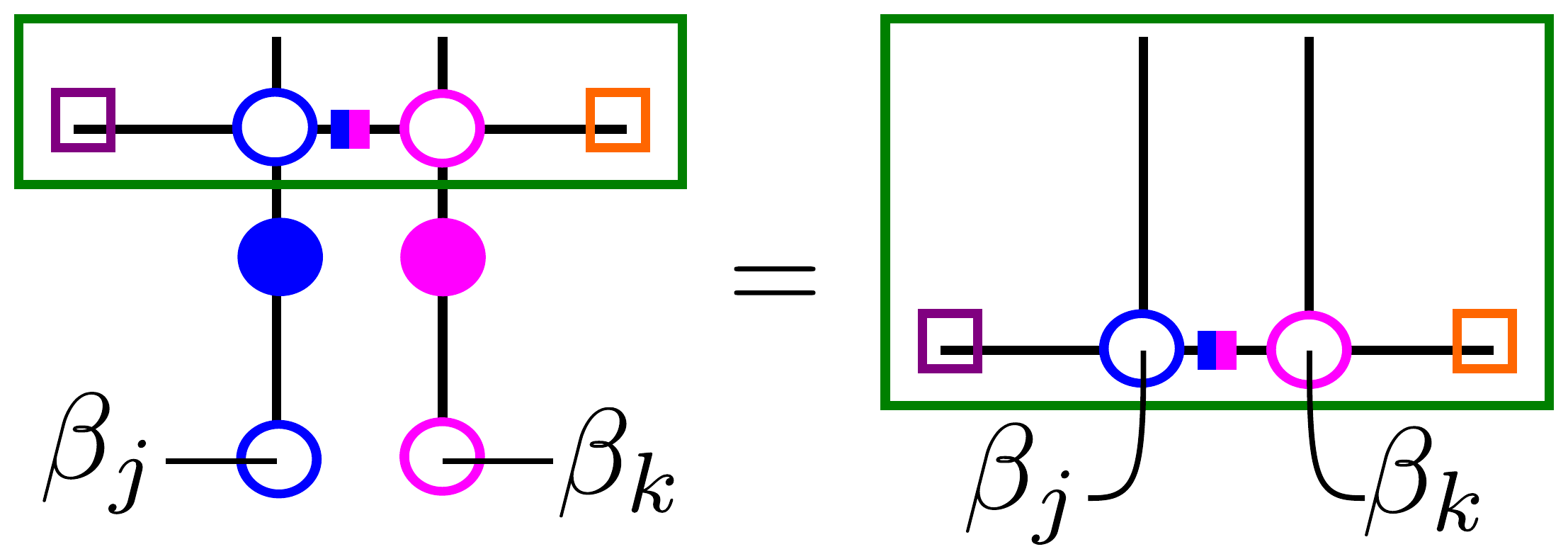}
\end{equation*}

To prove the converse, we only need to consider the number of states that can be obtained from:
\begin{equation*}
    \includegraphics[scale=0.2]{images/76-100/90-statement-1.pdf}
\end{equation*}
Let us call this state $\ket{\psi}$. As mentioned in the theorem's statement, the white 0,1-spider labelled $b_j$ belongs to the underlying basis of the classical structure consisting of the black spider $S_j$, and for any $j$, such a basis contains two states due to it being a basis for a Hilbert space of dimension 2. So there are $2^N$ unique states that $\ket{\psi}$ can take, which is the exact number of members in a basis for $N$ qubits. We have shown that all $\ket{\psi}$ belong to the underlying basis of the composite classical structure consisting of spiders as in the theorem's statement. Therefore, there cannot exist a state which does not take the form $\ket{\psi}$ in the basis. Otherwise, the size of the basis would be greater than $2^N$ which is not true.   
\end{proof}

By determining the underlying basis of a composite classical structure, we are now able to compare them with the underlying bases of some classical structures consisting of spiders joined via composite connecting wires. Below is an example of such a comparison. 

\paragraph{Example 1}

\begin{equation*}
    \includegraphics[scale=0.18]{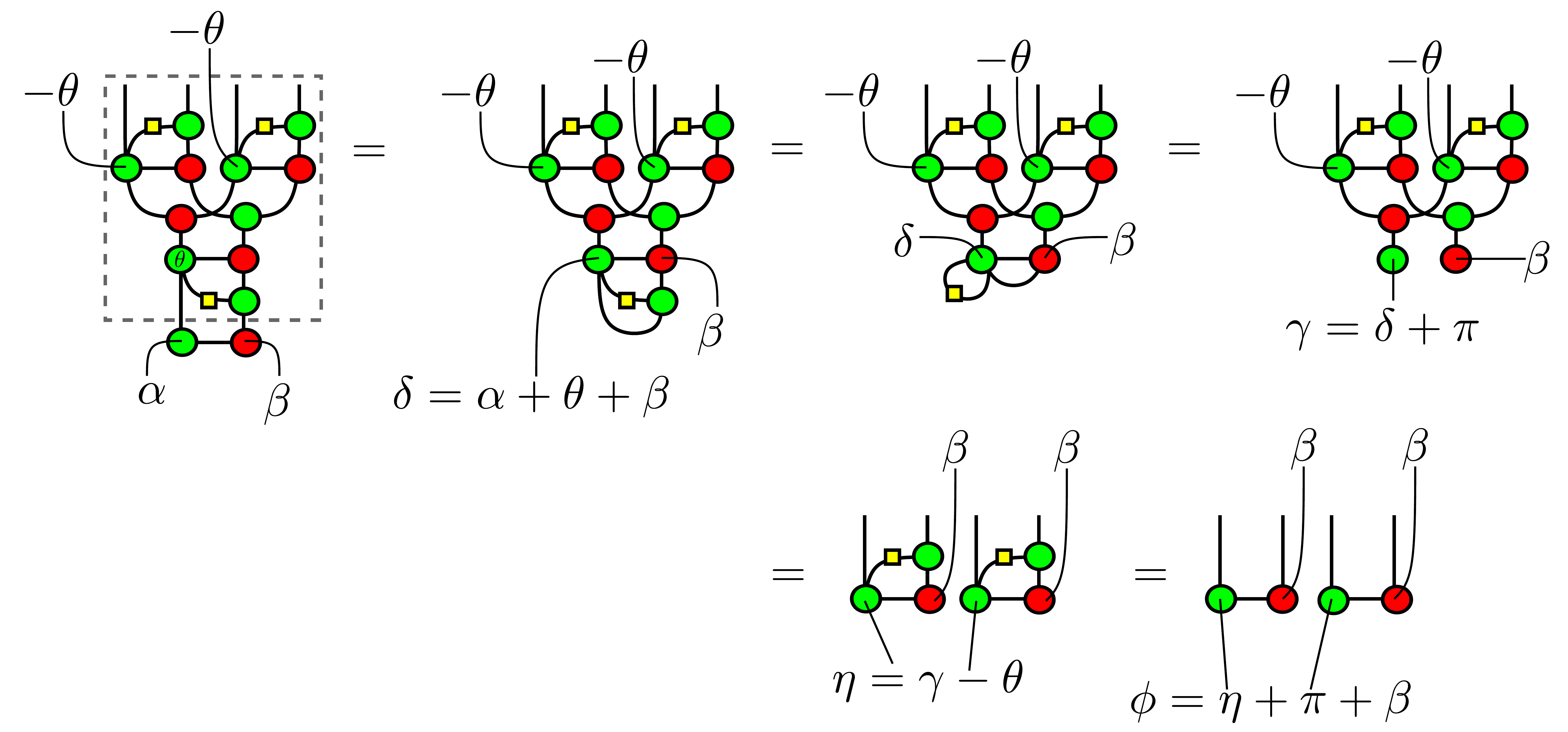}
\end{equation*}
where
\begin{equation*}\label{eq:cond-eg-2}
\alpha =
\begin{cases}
 0\text{ or }\pi & \text{if }\theta=0\\
 \frac{\pi}{2}\text{ or }-\frac{\pi}{2} & \text{ if } \theta=\frac{\pi}{2}
\end{cases}
\end{equation*}
\begin{equation*}\label{eq:cond-eg-3}
    \beta\in\{0,\pi\}
\end{equation*}
Thus,
\begin{equation*}
    \phi=\eta+\pi+\beta=\gamma-\theta+\pi+\beta=\delta+\pi-\theta+\pi+\beta=\alpha+\theta+\beta-\theta+\beta=\alpha
\end{equation*}
since $\alpha,\beta$ and $\theta$ are treated as angles $[-\pi,\pi]$ and addition between these variables are modulo $2\pi$.

So the underlying basis of the classical structure consisting of the 1,2-spider (within the dashed-line box) above, is the same as the one for $\mathcal{X}\diamond\mathcal{Z}$.

Another similar example is:

\paragraph{Example 2}

\begin{equation*}
    \includegraphics[scale=0.2]{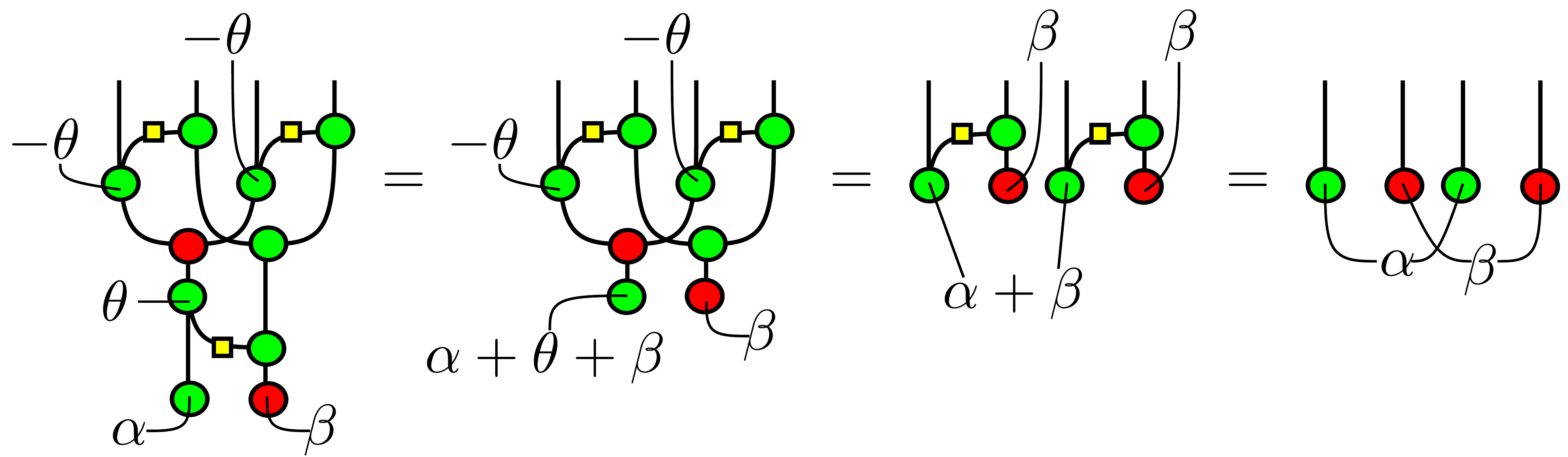}
\end{equation*}
where $\alpha,\beta$ and $\theta$ are as above.

The underlying basis of the classical structure consisting of the 2,1-spider above is then the same as $\mathcal{X}\mathcal{Z}$.

For the rest of the examples in this section, we only state the underlying bases of the corresponding classical structures without proof as the proofs are consume a lot of space. These proofs are straight forward albeit tedious. 

In the following two examples, composing $\cwgg$ on the first and second qubits to the spider does not change the original entanglement of the classical structure. 
%Notice that the original entanglement --- i.e. the $\cwgg$ on the first and third qubits --- could be swapped with the composed $\cwgg$ on the first and second qubits. From Example 2, we know that this does not generate an entanglement between the first and second qubits. Thus, no changes occur to the underlying basis of the original classical structure. 

\begin{longtable}{ccc}
\textbf{Example 3} & &\\
\\
& \includegraphics[scale=0.2]{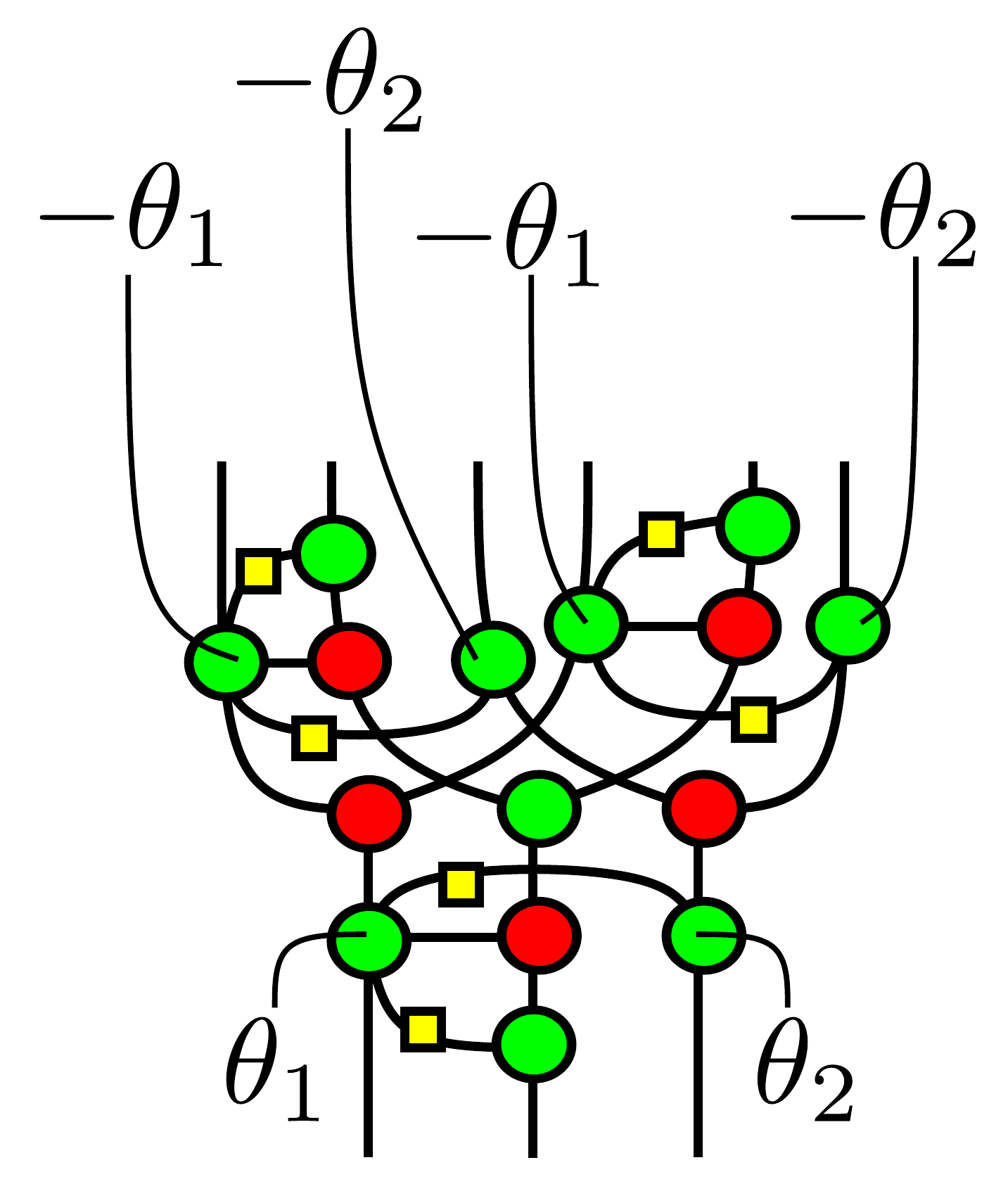} &
\includegraphics[scale=0.2]{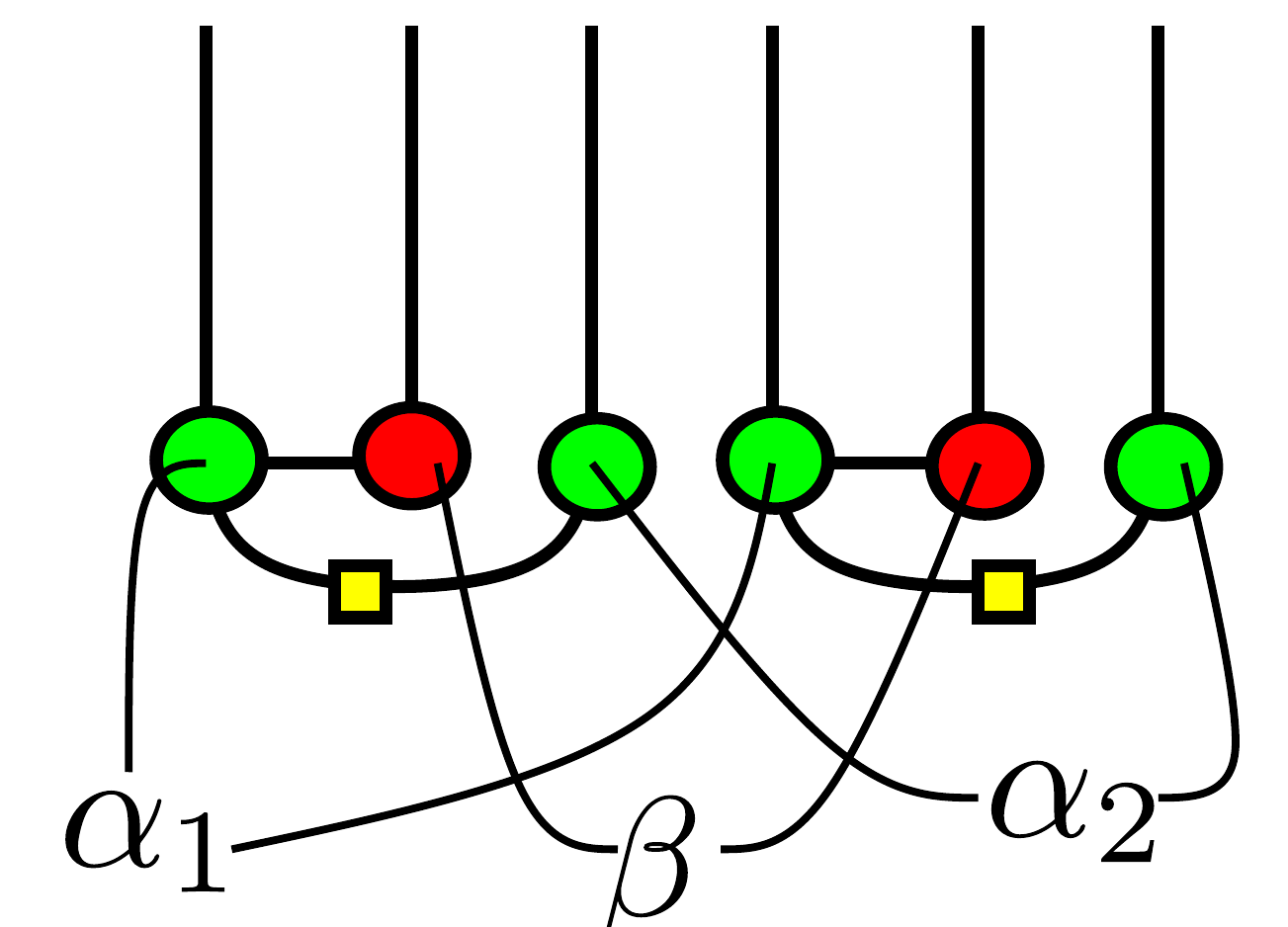}\\
& 2,1-spider & Underlying basis
\end{longtable}

\begin{longtable}{ccc}
\textbf{Example 4} & &\\
\\
& \includegraphics[scale=0.2]{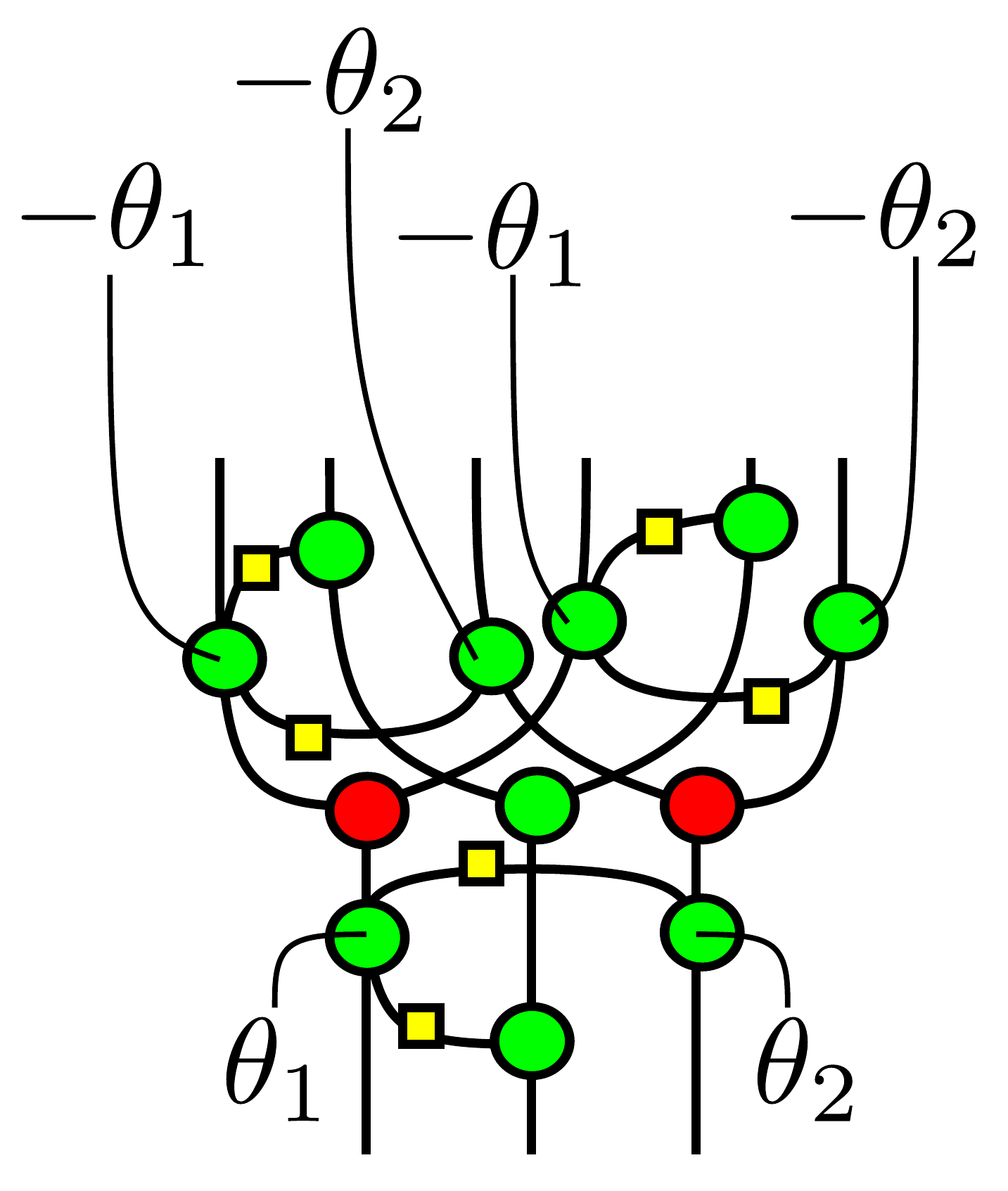} &
\includegraphics[scale=0.2]{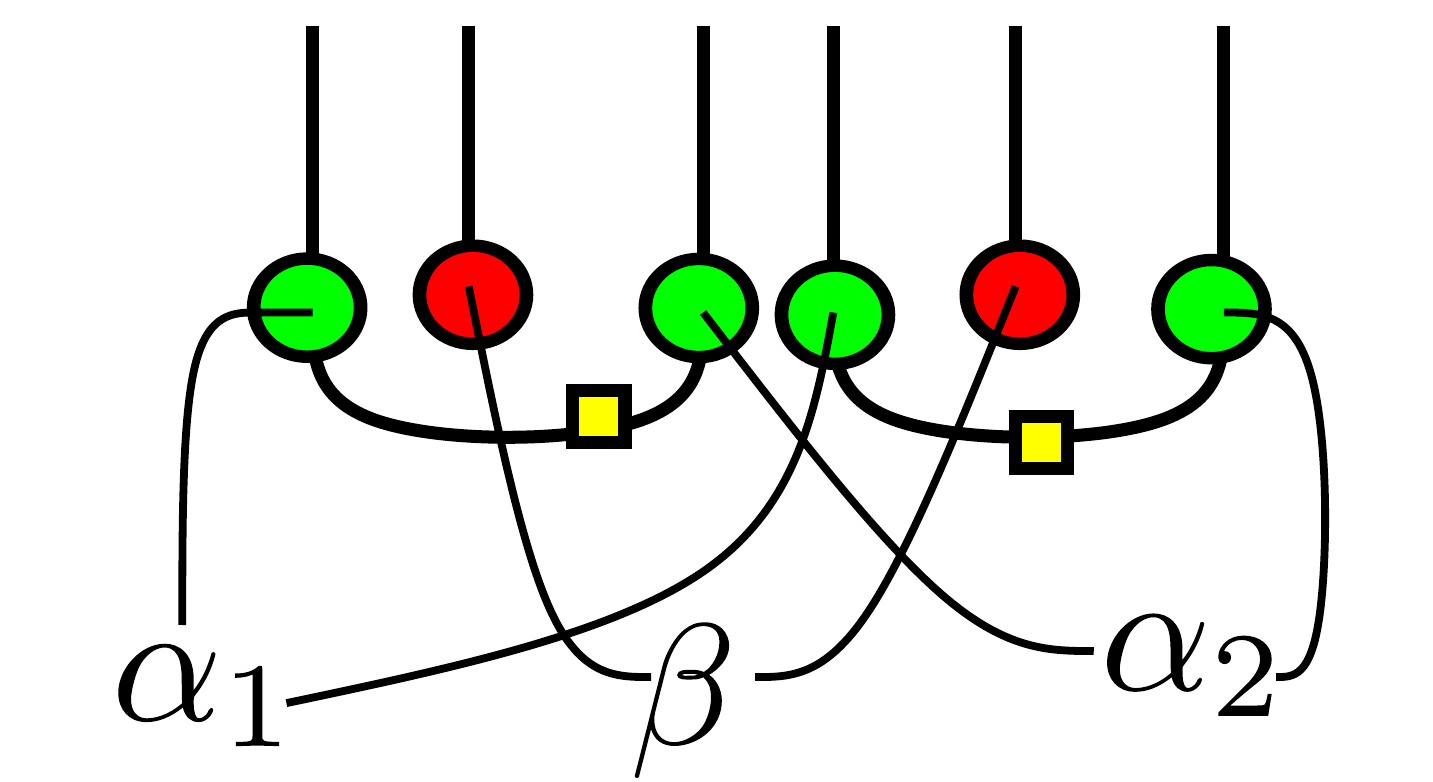}\\
& 2,1-spider & Underlying basis
\end{longtable}

%As with Example 3, the composed $\cwgg$ on the first and second qubits can be swapped with the original $\cwgg$ on the first and third qubits. 
However, it is not true in general that composing $\cwgg$ to connecting wires will not change the entanglement of a classical structure. There are examples where the same procedure generates new entanglement or deletes an old one. Interestingly, the new separability is not between the first and second qubits, and instead, entanglement is generated or deleted between the first and third qubits.

In the forthcoming examples,
\begin{align*}
    \alpha_j=
    \begin{cases}
     0\text{ or }\pi & \text{ if } \theta_j=0\\
     \frac{\pi}{2}\text{ or }-\frac{\pi}{2} & \text{ if } \theta_j=\frac{\pi}{2}
    \end{cases}\\
   \beta\in\{0,\pi\} 
\end{align*}

\begin{longtable}{ccc}
\textbf{Example 5} & &\\
\\
&\includegraphics[scale=0.2]{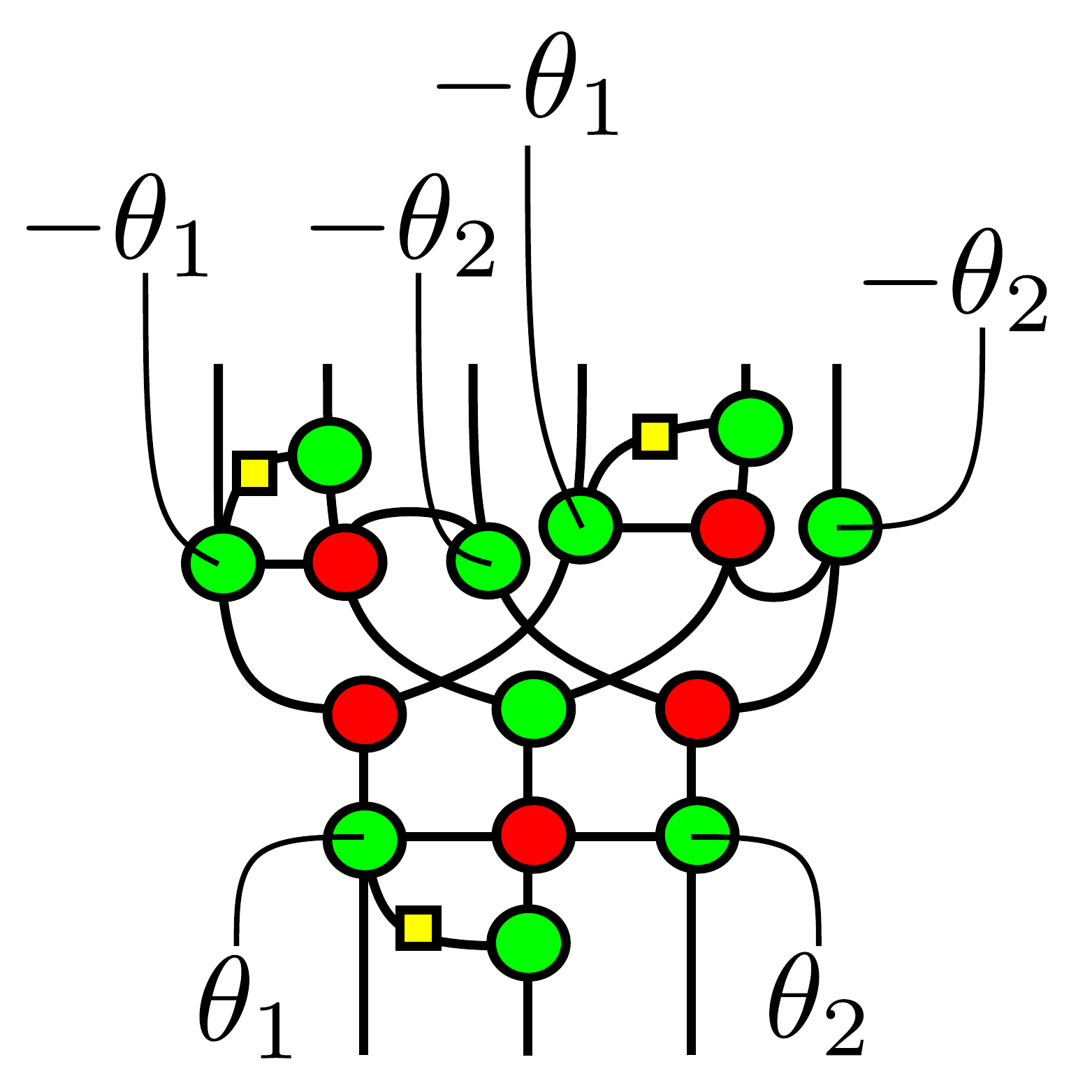} &\includegraphics[scale=0.2]{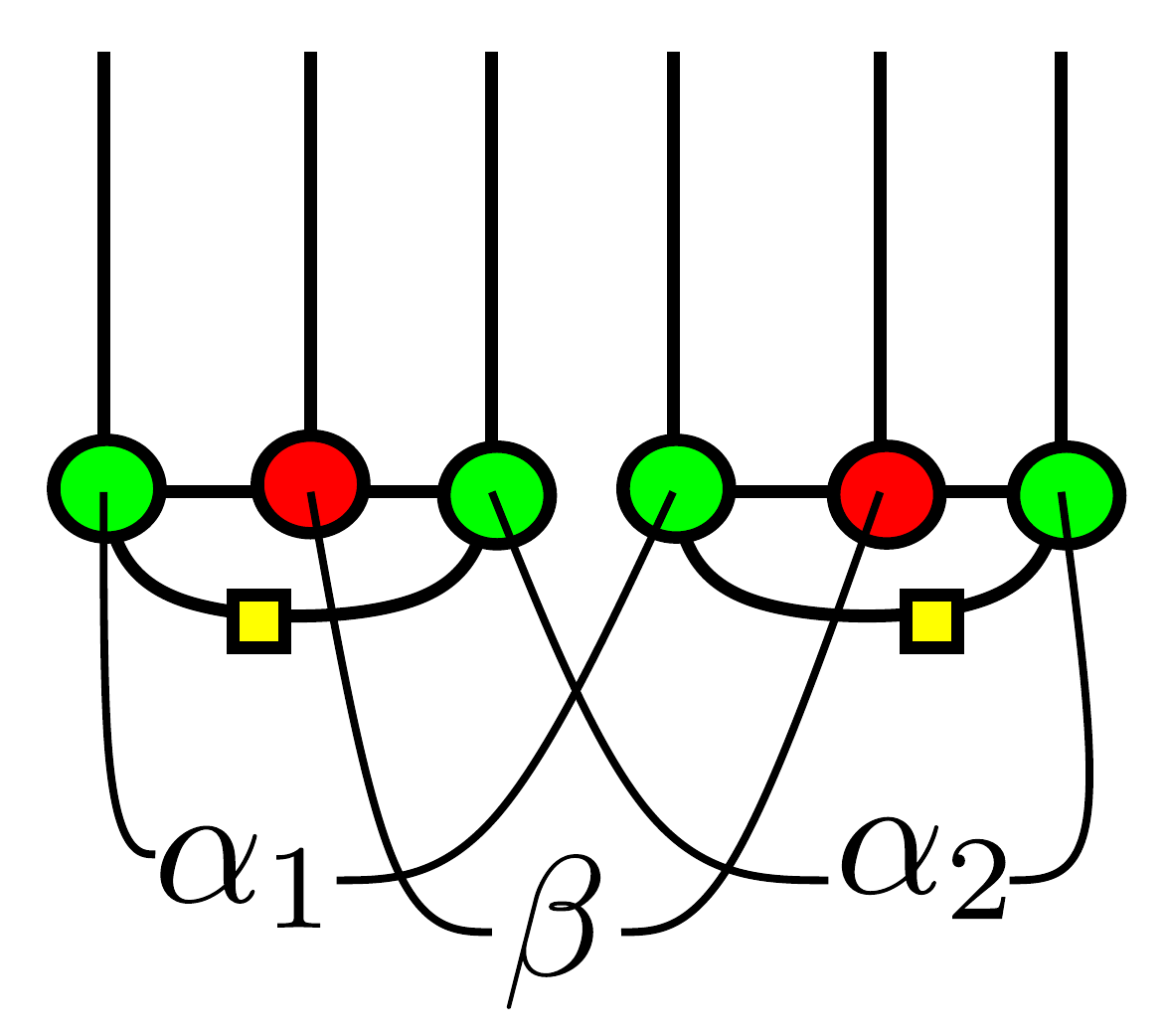}\\
&2,1-spider & Underlying basis\\
\\
\textbf{Example 6} & &\\
\\
&\includegraphics[scale=0.2]{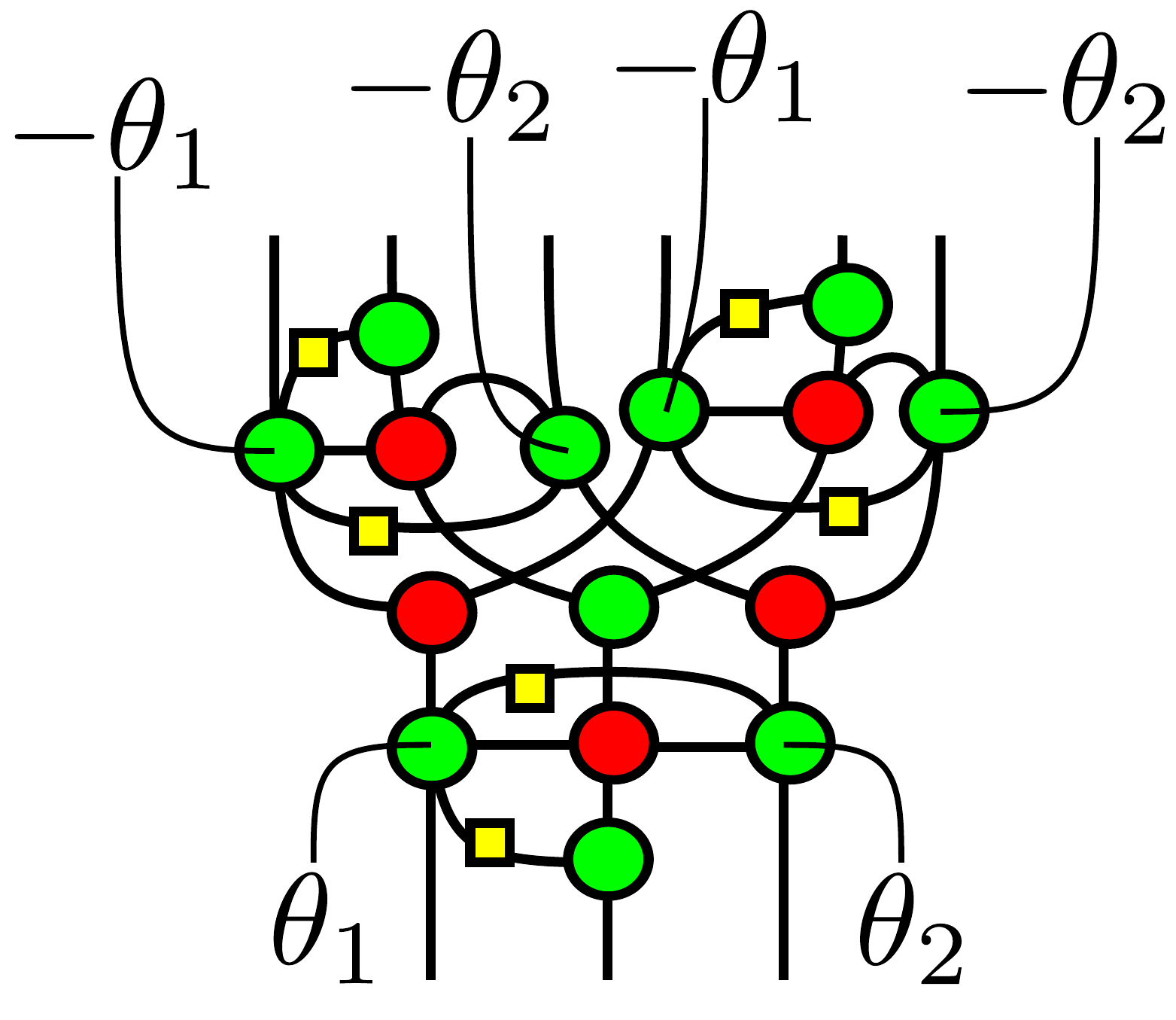} &\includegraphics[scale=0.2]{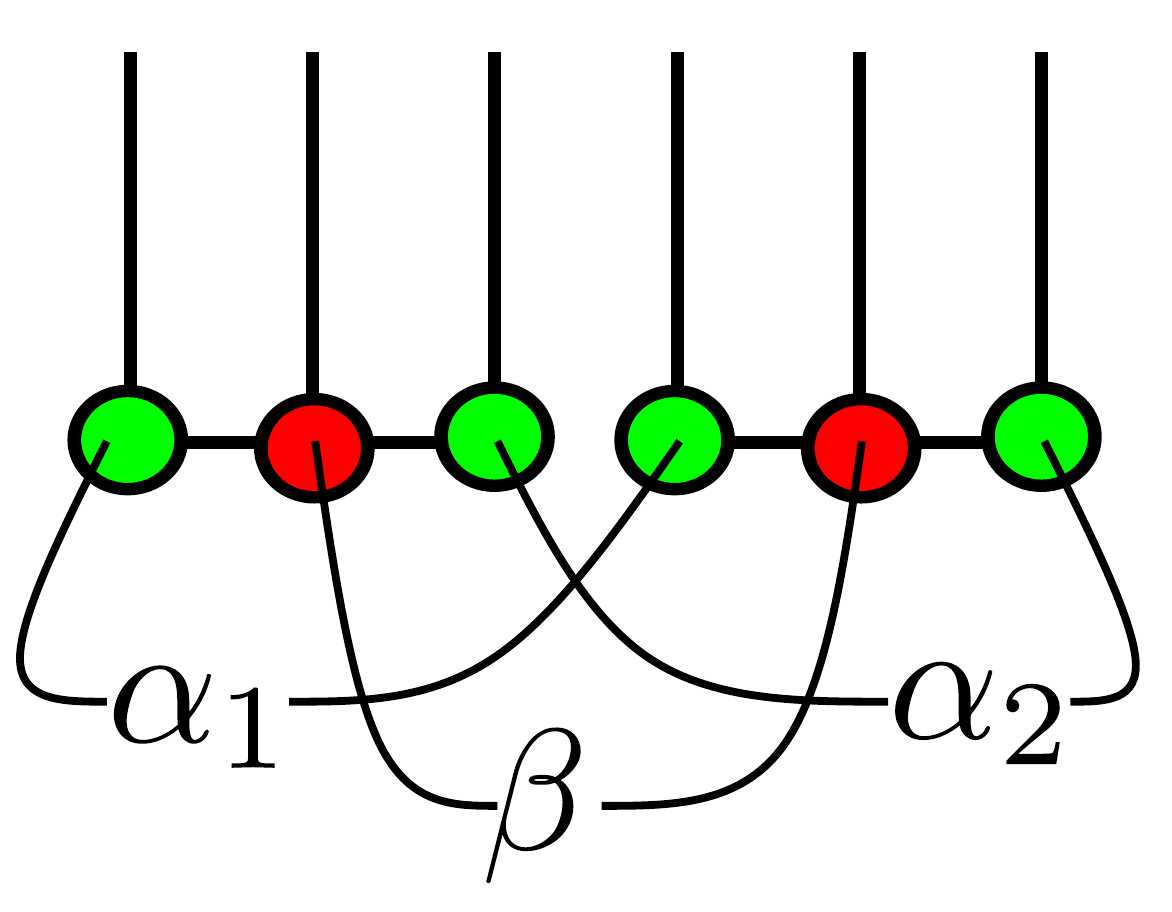}\\
&2,1-spider & Underlying basis\\
\\
\textbf{Example 7} & &\\
\\
&\includegraphics[scale=0.2]{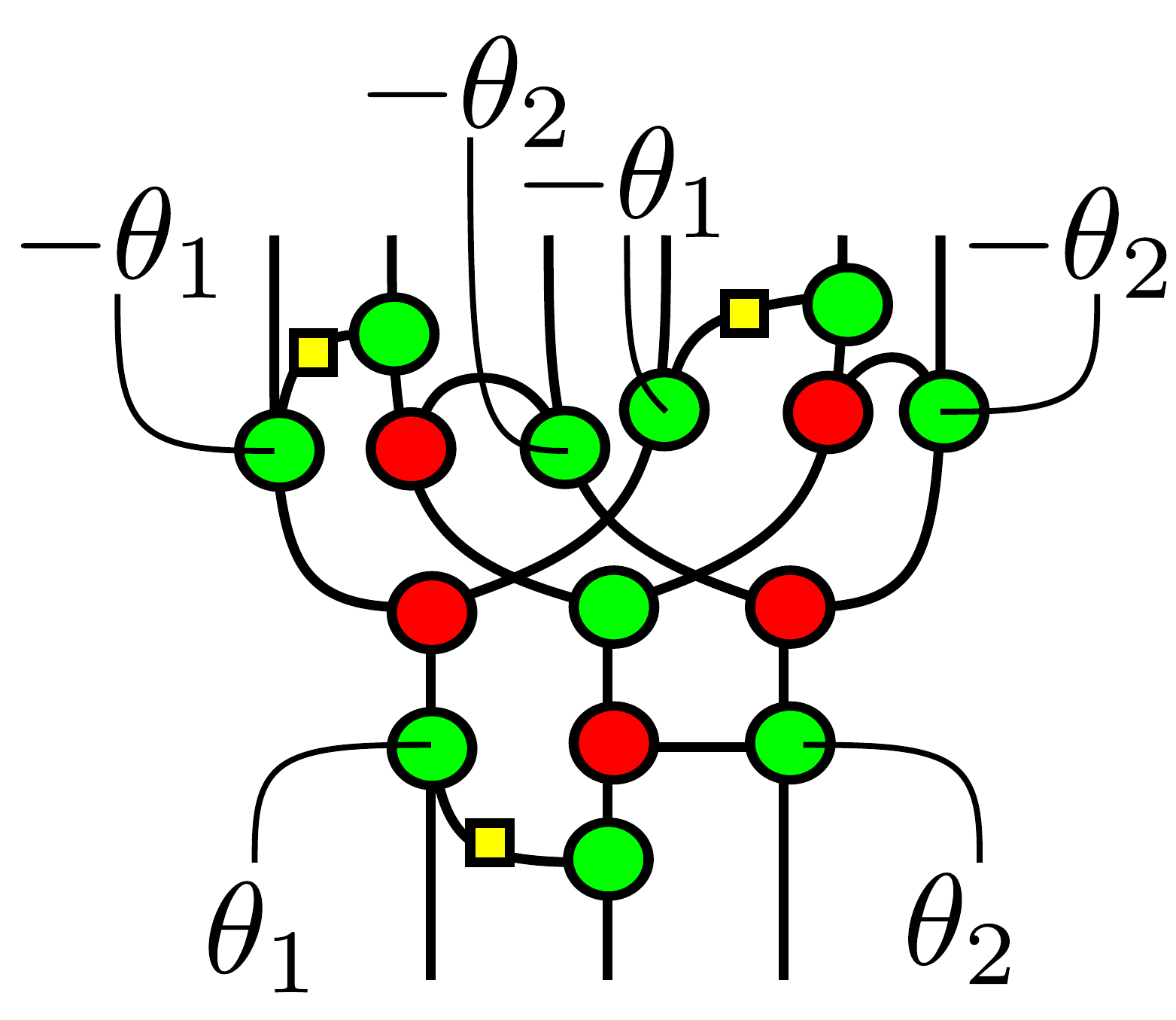} &\includegraphics[scale=0.2]{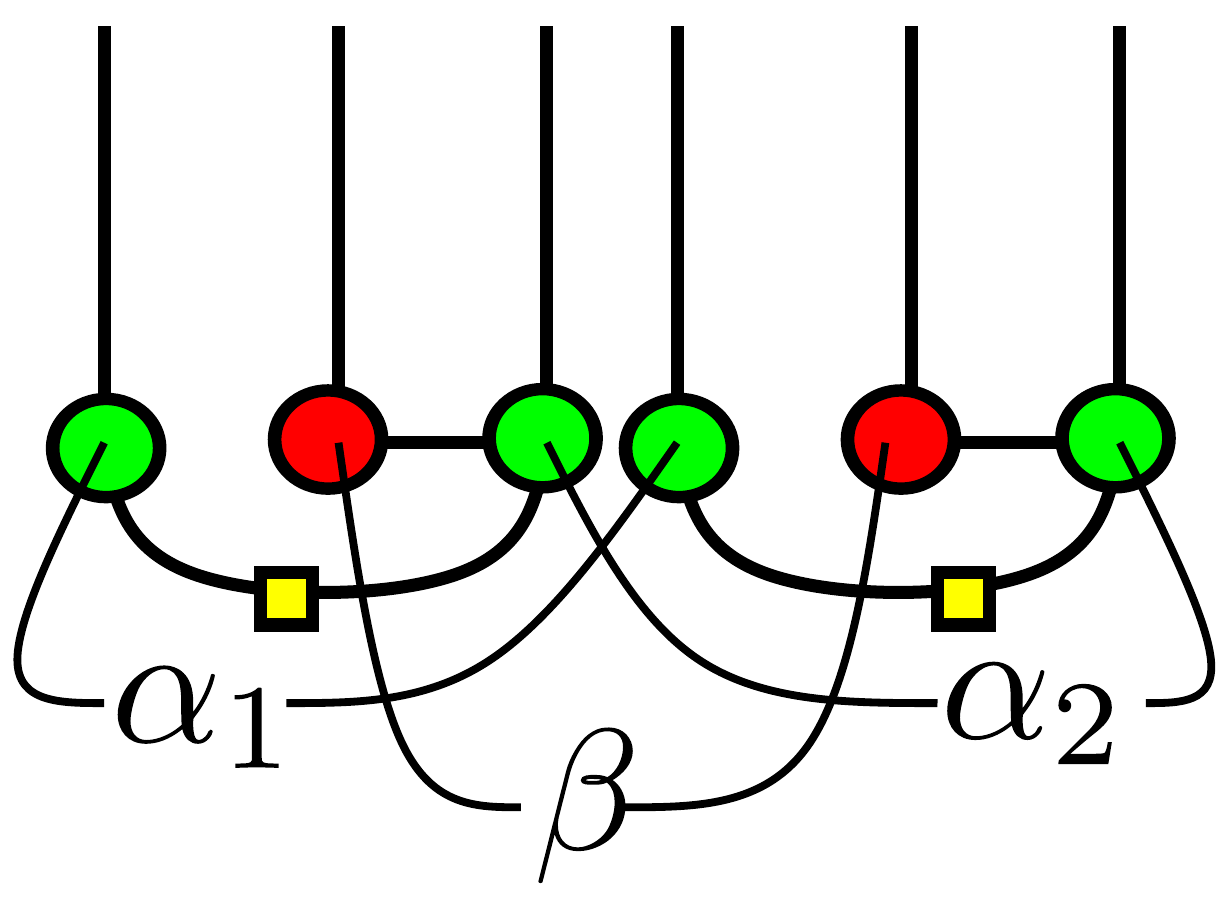}\\
&2,1-spider & Underlying basis
\end{longtable}

\newpage
\begin{longtable}{ccc}
\textbf{Example 8} & &\\
\\
&\includegraphics[scale=0.2]{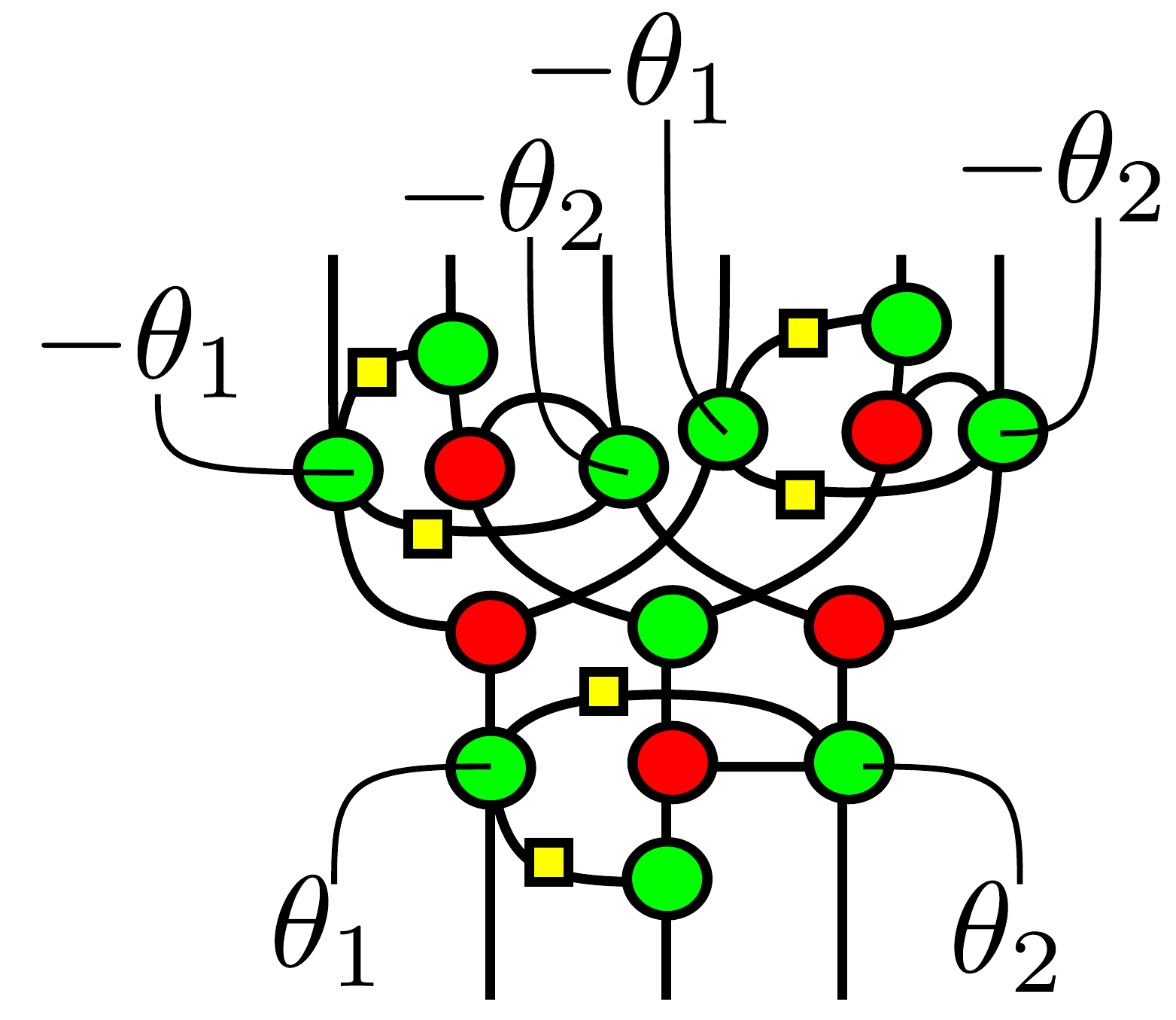} &\includegraphics[scale=0.2]{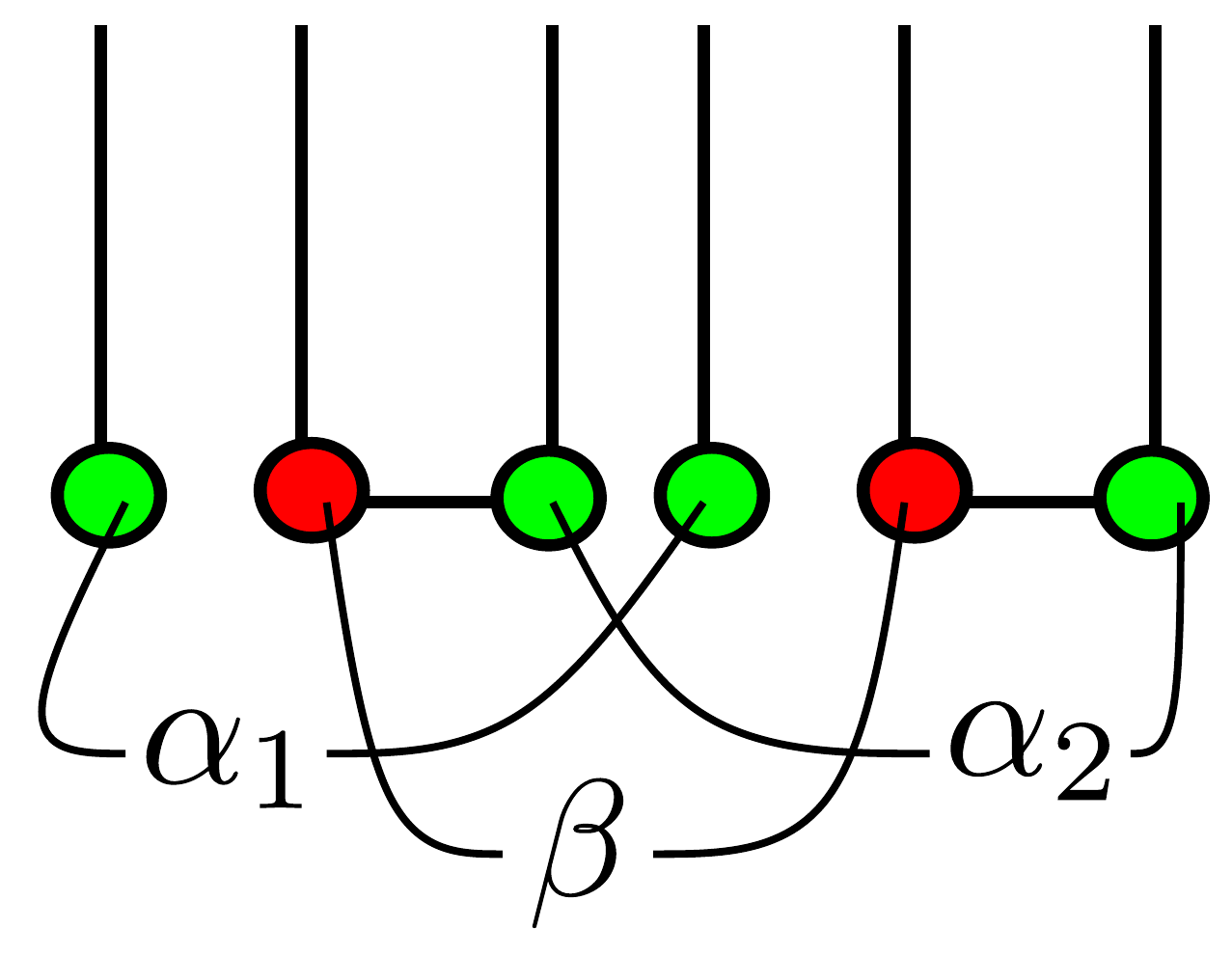}\\
&2,1-spider & Underlying basis\\
\end{longtable}

Using the examples above, we can devise a procedure to obtain a set of complementary classical structures on three qubits with one entanglement configuration from another set with a different configuration. We can describe such a procedure by a mapping on the names of classical structures on three qubits. Below, we define the mapping, which we shall denote as $\text{ent}_m$, where $m\in\{0,1,2\}$. 

\noindent Let 
$A=
\begin{pmatrix}
c_0 & c_1 & c_2\\
0 & e_{1,0} & e_{2,0}\\
e_{0,1} & 0 & e_{2,1}\\
e_{0,2} & e_{1,2} & 0
\end{pmatrix}$
be the name of a classical structure on three qubits. We define $\text{ent}_m$ on the permutations of the columns and rows of $A$. So the the indices of the entries of the names in the definition below are modulo 3.  

If $c_m\in\{0,1\}$, $c_{m+1}=Z$, $c_{m+2}\in\{0,1\}$ and $e_{m+1,m+2}=k$,
\begin{equation*}
\begin{pmatrix}
c_m & c_{m+ 1} & c_{m+ 2}\\
0 & e_{m+ 1,m} & e_{m+ 2,m}\\
e_{m,m+ 1} & 0 & e_{m+ 2,m+ 1}\\
e_{m,m+ 2} & e_{m+ 1,m+2} & 0
\end{pmatrix}
\xmapsto{\text{ent}_m}
\begin{pmatrix}
c_m & c_{m+1} & c_{m+2}\\
0 & e_{m+1,m} & e_{m+2,m}+1\\
e_{m,m+1} & 0 & e_{m+2,m+1}\\
e_{m,m+2}+1 & e_{m+1,m+2} & 0
\end{pmatrix}
\end{equation*}

If $c_m\in\{0,1\}$, $c_{m+1}=Z$, $c_{m+2}=Z$ and $e_{m+1,m+2}=i$,
\begin{equation*}
\begin{pmatrix}
c_m & c_{m+1} & c_{m+2}\\
0 & e_{m+1,m} & e_{m+2,m}\\
e_{m,m+1} & 0 & e_{m+2,m+1}\\
e_{m,m+2} & e_{m+1,m+2} & 0
\end{pmatrix}
\xmapsto{\text{ent}_m}
\begin{pmatrix}
c_m & c_{m+1} & c_{m+2}\\
0 & e_{m+1,m} & e_{m+2,m}+k\\
e_{m,m+1} & 0 & e_{m+2,m+1}\\
e_{m,m+2}+j & e_{m+1,m+2} & 0
\end{pmatrix}
\end{equation*}

If $c_m=Z$, $c_{m+1}\in\{0,1\}$, $c_{m+2}\in\{0,1\}$ and $e_{m+2,m}=k$,
\begin{equation*}
\begin{pmatrix}
c_m & c_{m+ 1} & c_{m+ 2}\\
0 & e_{m+ 1,m} & e_{m+ 2,m}\\
e_{m,m+ 1} & 0 & e_{m+ 2,m+ 1}\\
e_{m,m+ 2} & e_{m+ 1,m+2} & 0
\end{pmatrix}
\xmapsto{\text{ent}_m}
\begin{pmatrix}
c_m & c_{m+1} & c_{m+2}\\
0 & e_{m+1,m} & e_{m+2,m}\\
e_{m,m+1} & 0 & e_{m+2,m+1}+1\\
e_{m,m+2} & e_{m+1,m+2}+1 & 0
\end{pmatrix}
\end{equation*}

If $c_m=Z$, $c_{m+1}\in\{0,1\}$, $c_{m+2}=Z$ and $e_{m+2,m}=i$,
\begin{equation*}
\begin{pmatrix}
c_m & c_{m+1} & c_{m+2}\\
0 & e_{m+1,m} & e_{m+2,m}\\
e_{m,m+1} & 0 & e_{m+2,m+1}\\
e_{m,m+2} & e_{m+1,m+2} & 0
\end{pmatrix}
\xmapsto{\text{ent}_m}
\begin{pmatrix}
c_m & c_{m+1} & c_{m+2}\\
0 & e_{m+1,m} & e_{m+2,m}\\
e_{m,m+1} & 0 & e_{m+2,m+1}+k\\
e_{m,m+2} & e_{m+1,m+2}+j & 0
\end{pmatrix}
\end{equation*}

If $c_m=c_{m+1}=Z$ and $e_{m,m+1}=i$, 
\begin{equation*}
\begin{pmatrix}
Z & Z & c_{m+2}\\
0 & i & e_{m+2,m}\\
i & 0 & e_{m+2,m+1}\\
e_{m,m+2} & e_{m+1,m+2} & 0
\end{pmatrix}
\xmapsto{\text{ent}_m}
\begin{pmatrix}
0 & 0 & c_{m+2}\\
0 & 1 & e_{m+2,m}'\\
1 & 0 & e_{m+2,m+1}'\\
e_{m+1,m+2}' & e_{m,m+2}' & 0
\end{pmatrix}
\end{equation*}
where 
$e_{p,q}'=
\begin{cases}
0 & \text{ if } e_{p,q}=0\\
1 & \text{ if } e_{p,q}\in\{j,k\}\\
j & \text{ if } e_{p,q}=i \text{ and } p<q\\
k & \text{ if } e_{p,q}=i \text{ and } p>q
\end{cases}$

If $c_m\in\{0,1\}$, $c_{m+1}=Z$ and $e_{m+1,m+2}=0$, or $c_m=Z$, $c_{m+1}\{0,1\}$ and $e_{m,m+2}=0$, the classical structure in unchanged under $\text{ ent}_m$. The same is true for $c_m=c_{m+1}=Z$ and $e_{m,m+1}=0$.

Recall the example of a maximal complete set of complementary classical structures on three qubits with entanglement configuration (0,9,0) from Section \ref{sec:CCS-3Q}: 

\input{list/eg090}

Suppose, we denote this set by eg090. 

Applying $\text{ent}_1$ to members of eg090, we obtain:

\input{list/egcwgg234-090}

Applying $\text{ent}_2\circ\text{ent}_1$ to members of eg090, we obtain:

\input{list/egcwgg162-090}

Indeed, both sets are also maximal complete sets of complementary classical structures on three qubits, where the first set has entanglement configuration (2,3,4) and the second set has entanglement configuration (1,6,2). Here, we described the composition of $\cwgg$ on the legs of composite classical structures via the mapping $\text{ent}_m$. It shows that understanding composite connecting wires leads to a more efficient way of finding maximal complete sets of complementary classical structures.

\section{Equivalent Classical Structures}

In this section, we use the notation $\bullet$ to show that the spiders of a classical structure is composed with $\cwgg$ at their legs. For example, we write the classical structures in Examples 1 and 2 from the previous section as $\mathcal{X}\diamond\bullet\mathcal{Z}$ or $\mathcal{Y}\diamond\bullet\mathcal{Z}$, and $\mathcal{X}\bullet\mathcal{Z}$ or $\mathcal{Y}\bullet\mathcal{Z}$, respectively.

As we previously mentioned in Section \ref{sec:process-intro}, the 1,2-spider of a classical structure can be interpreted as encoding classical data into a quantum system, and the 2,1-spider of a classical structure can be interpreted as measuring a quantum system where the outcome is in the underlying basis of the classical structure. Consequently, we can test the compatibility between two classical structures by composing a 1,2-spider and a 2,1-spider. On one extreme, such a composition will give us the identity process, implying that the two spiders belong to the same classical structure, and on the other extreme, we obtain a disconnected process implying that the two spiders belong to a pair of complementary classical structures. So it is reasonable to infer that such a composition could provide some information on the similarity between the underlying bases of two classical structures. 

We begin with $\mathcal{Z}\bullet\mathcal{Z}$ and $\mathcal{ZZ}$. These two have underlying bases which are equivalent up to scalars. That is, the underlying basis of $\mathcal{Z}\bullet\mathcal{Z}$ is $\{\ket{00},\ket{01},\ket{10},-\ket{11}\}$ and the underlying basis of $\mathcal{ZZ}$ is $\{\ket{00},\ket{01},\ket{10},\ket{11}\}$. Composing the 1,2-spider of $\mathcal{Z}\bullet\mathcal{Z}$ and the 2,1 spider $\mathcal{ZZ}$, we obtain:
\begin{equation}\label{eq:equiv-entangled-center}
    \includegraphics[scale=0.2]{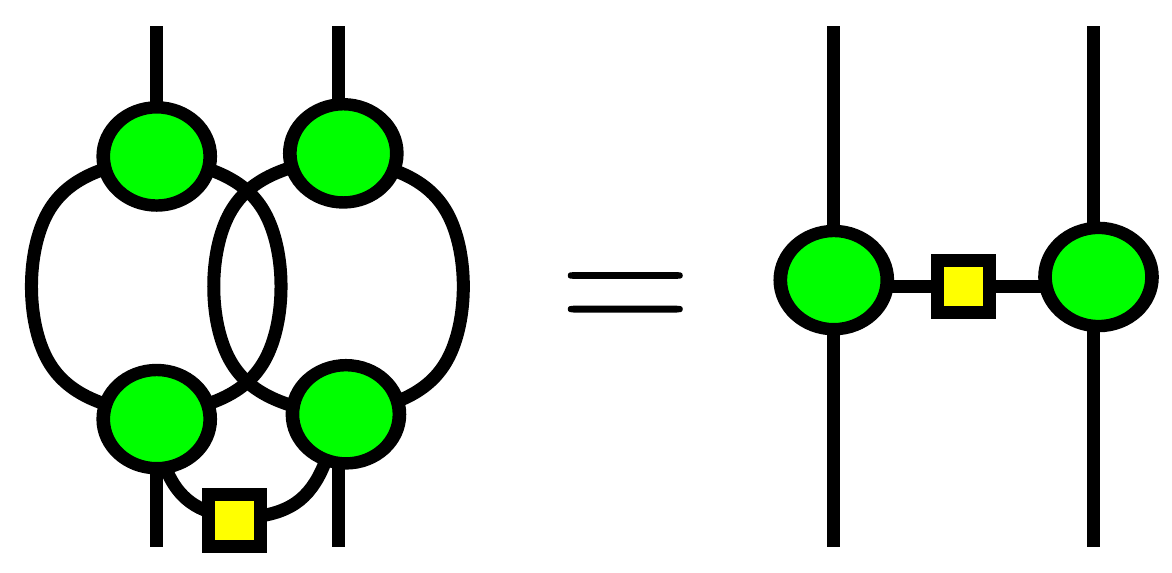}
\end{equation}
This implies that $\mathcal{Z}\bullet\mathcal{Z}$ and $\mathcal{ZZ}$ are not equivalent. 

Indeed, \citeauthor{Coecke2011b} discussed this problem in Sections 5 and 6 of reference \cite{Coecke2011b}. They proposed a solution which involves distinguishing between states and its vector representation. That is, a \textbf{state} $\psi$ is the equivalence class containing the unit vector $\ket{\psi}$ and all unit vectors $\ket{\phi}$ which satisfy $\ket{\phi}=e^{i\theta}\ket{\psi}$ for some $\theta\in[-\pi,\pi]$. The counterpart of a classical structure in such a projective space is called a state basis, which is an orthonormal basis in a finite-dimensional Hilbert space along with a state which is the 0,1-spider of the aforementioned classical structure.

\begin{longtable}{|c|c|}
\caption{Examples of classical structures with underlying bases that are equivalent up to scalars, where $\mathcal{A},\mathcal{B}\in\{\mathcal{X},\mathcal{Y}\}$.}
\label{tab:CSequiv}
\\
\hline
$\mathcal{A}\bullet\mathcal{Z}$ & $\mathcal{AZ}$\\ \hline
$\mathcal{A}\diamond\bullet\mathcal{Z}$ & $\mathcal{A}\diamond\mathcal{Z}$\\ \hline
$\mathcal{A}\bullet\mathcal{Z}\mathcal{B}\diamond$ & $\mathcal{AZ}\mathcal{B}\diamond$\\ \hline
$\mathcal{A}\diamond\bullet\mathcal{Z}\mathcal{B}\diamond$ & $\mathcal{A}\diamond\mathcal{Z}\mathcal{B}\diamond$\\ \hline
$\mathcal{A}\bullet\mathcal{Z}\diamond\mathcal{B}$ & $\mathcal{AZ}\diamond\mathcal{B}\diamond$\\ \hline
$\mathcal{A}\bullet\mathcal{Z}\diamond\mathcal{B}\diamond$ & $\mathcal{AZ}\diamond\mathcal{B}$\\ \hline
$\mathcal{A}\diamond\bullet\mathcal{Z}\diamond\mathcal{B}$ & $\mathcal{A}\diamond\mathcal{Z}\diamond\mathcal{B}\diamond$\\ \hline
$\mathcal{A}\diamond\bullet\mathcal{Z}\diamond\mathcal{B}\diamond$ & $\mathcal{A}\diamond\mathcal{Z}\diamond\mathcal{B}$ \\ \hline
\end{longtable}

However, in the present work, we are more interested in the separability of a classical structure. While $\mathcal{Z}\bullet\mathcal{Z}$ and $\mathcal{ZZ}$ are not equivalent with respect to their state bases, their underlying vector bases have members which are all separable. Other examples of such a similarity between classical structures can be found in Table \ref{tab:CSequiv}. So is there a way to characterize this invariance in separability? We do not have an answer as of yet, but we can speculate based on Eq. \ref{eq:equiv-entangled-center} and other similar results involving other pairs of classical structures with underlying vector bases containing equivalent states.

\begin{conjecture}\label{conj:equivalentCS}
Let $\mathcal{A}$ and $\mathcal{B}$ be classical structures on $N$ qubits with spiders $\black$ and $\blue$, respectively. The underlying bases of $\mathcal{A}$ and $\mathcal{B}$ contain equivalent states if they satisfy the following equation:
\begin{equation}\label{eq:conjecture-equiv}
    \includegraphics[scale=0.2]{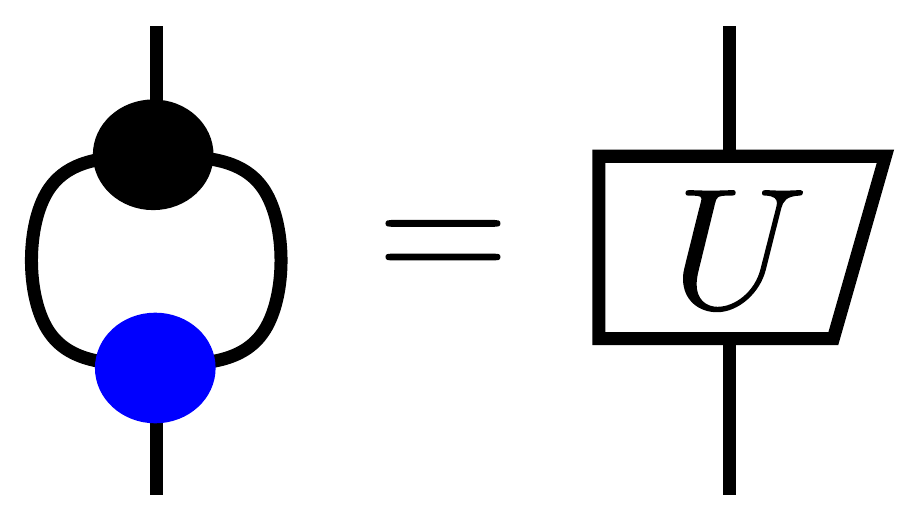}
\end{equation}
for some unitary process $U$.
\end{conjecture}

An example other than Eq. \ref{eq:equiv-entangled-center} is:
\begin{equation}
    \includegraphics[scale=0.15]{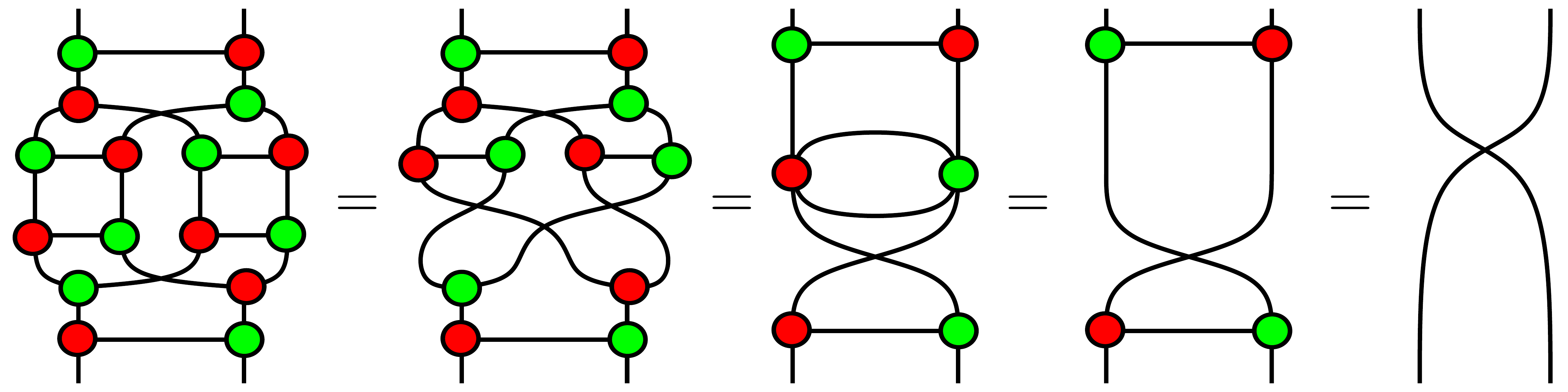}
\end{equation}
The underlying bases of $\mathcal{X}\diamond\mathcal{Z}$ and $\mathcal{Z}\diamond\mathcal{X}$ are equivalent, up to scalars, to:
\begin{equation*}
    \includegraphics[scale=0.25]{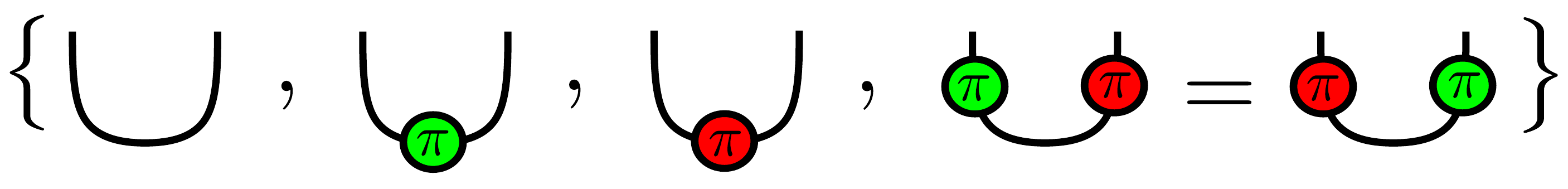}
\end{equation*}
For the pairs of classical structures in Table \ref{tab:CSequiv} in the third row and rows below, we obtain the identity by composing the 1,2-spider of the classical structure on the left to the 2,1-spider of the classical structure on the right. For example: 
\begin{equation*}
    \includegraphics[scale=0.18]{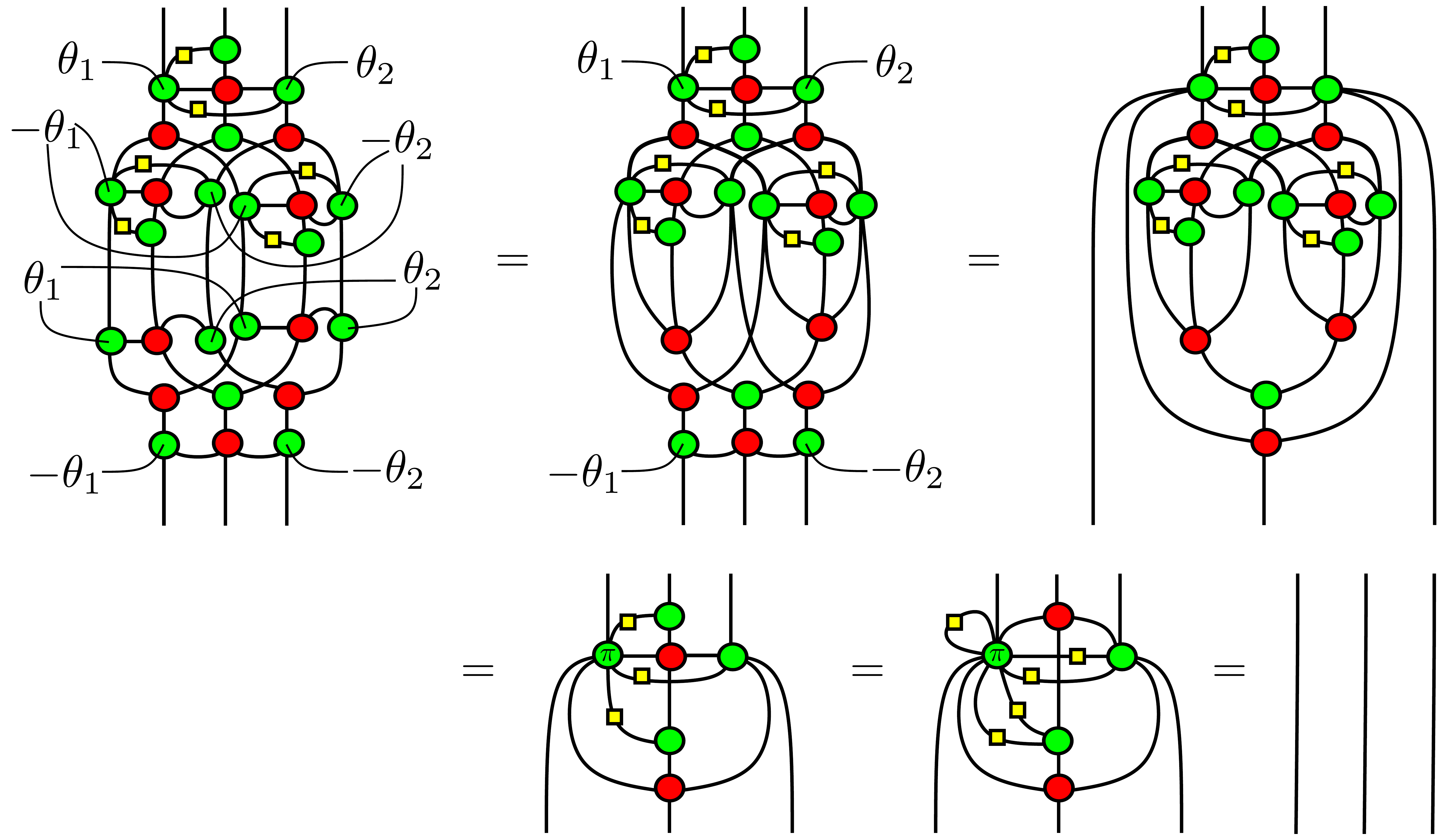}
\end{equation*}

We can also look at some examples that do not satisfy Eq. \ref{eq:conjecture-equiv}. The most obvious example would be a pair of complementary classical structures. No unitary exists which could result in the identity when the 1,2-spider of one of the classical structures is composed to the 2,1-spider of the other classical structure. However, a pair of complementary classical structures is not the only example for those classical structures which satisfy Eq. \ref{eq:conjecture-equiv}.

An example of a pair of classical structures on two qubits which do not satisfy Eq. \ref{eq:conjecture-equiv} is $\mathcal{X}\diamond\mathcal{X}$ and $\mathcal{Z}\diamond\mathcal{X}$:
\begin{equation*}
    \includegraphics[scale=0.2]{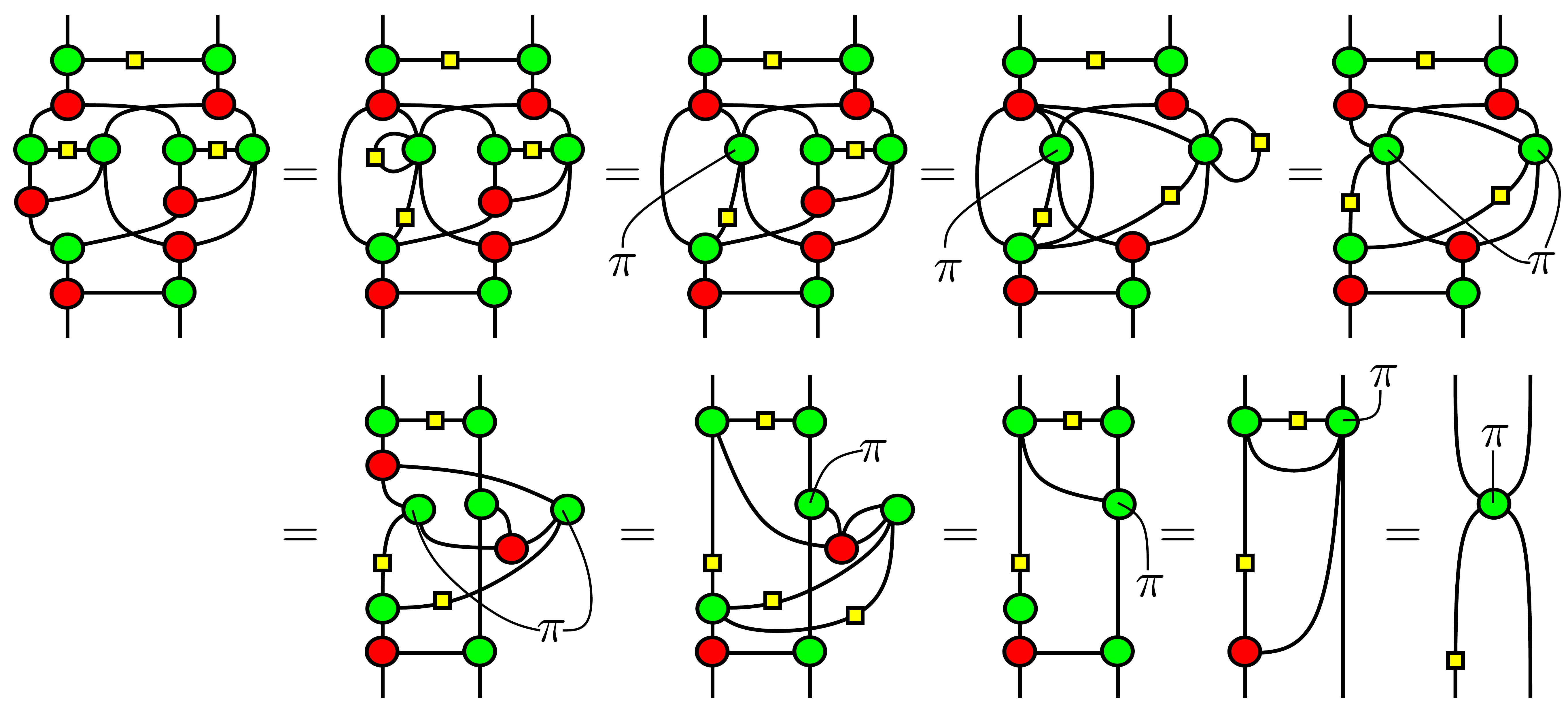}
\end{equation*}
The underlying basis of $\mathcal{X}\diamond\mathcal{X}$ is
$\egonebaseone$,
and the underlying basis of $\mathcal{Z}\diamond\mathcal{X}$ is
$\egonebasetwo$.

An example of a pair of classical structures on three qubits which does not satisfy Eq. \ref{eq:conjecture-equiv} is $\mathcal{X}\diamond\mathcal{X}\diamond\mathcal{X}\diamond$ and $\mathcal{Z}\diamond\mathcal{Z}\diamond\mathcal{Z}\diamond$:
\begin{equation*}
    \includegraphics[scale=0.18]{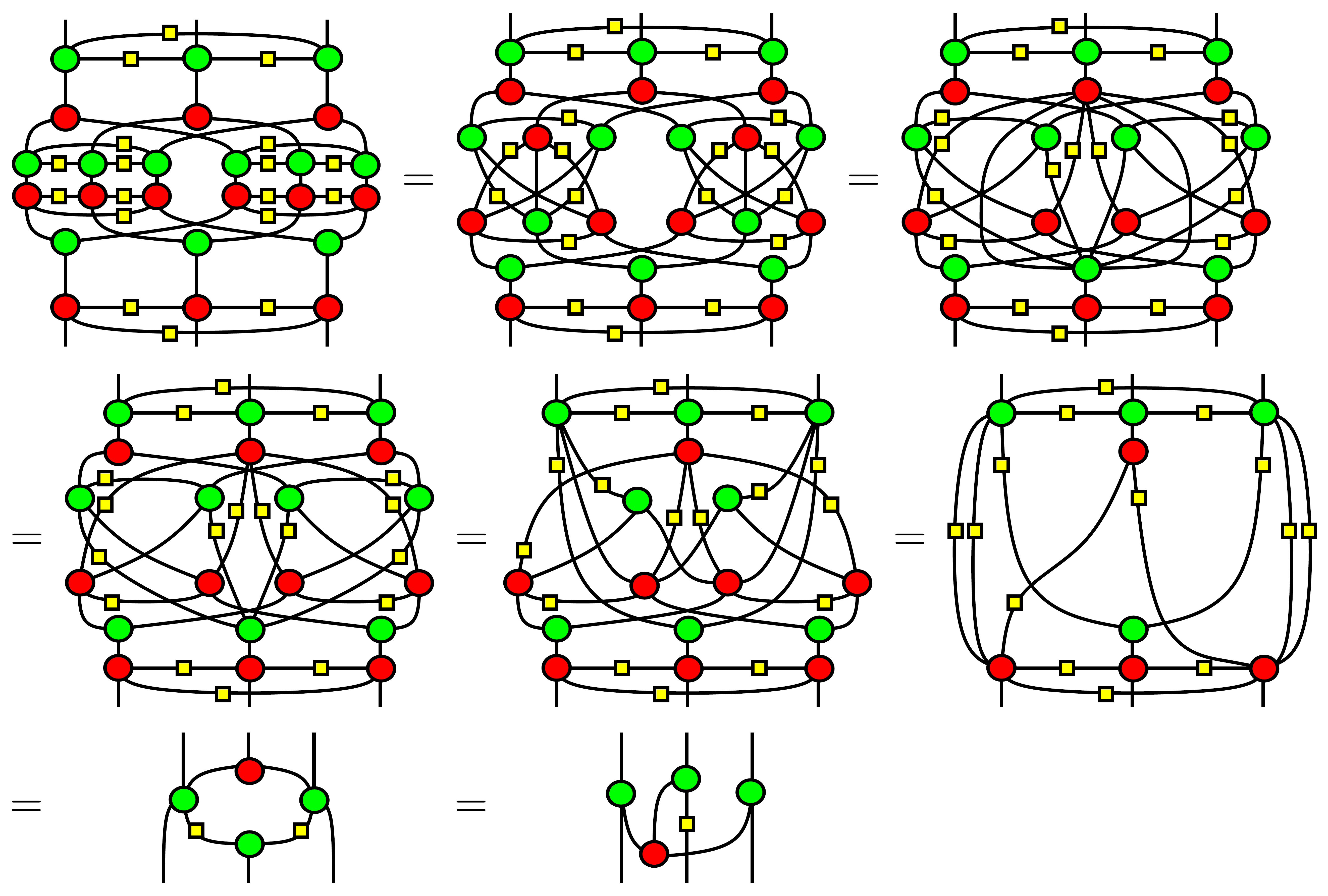}
\end{equation*}

Composing the final diagram in the equation above with its adjoint, we obtain:
\begin{equation*}
    \includegraphics[scale=0.15]{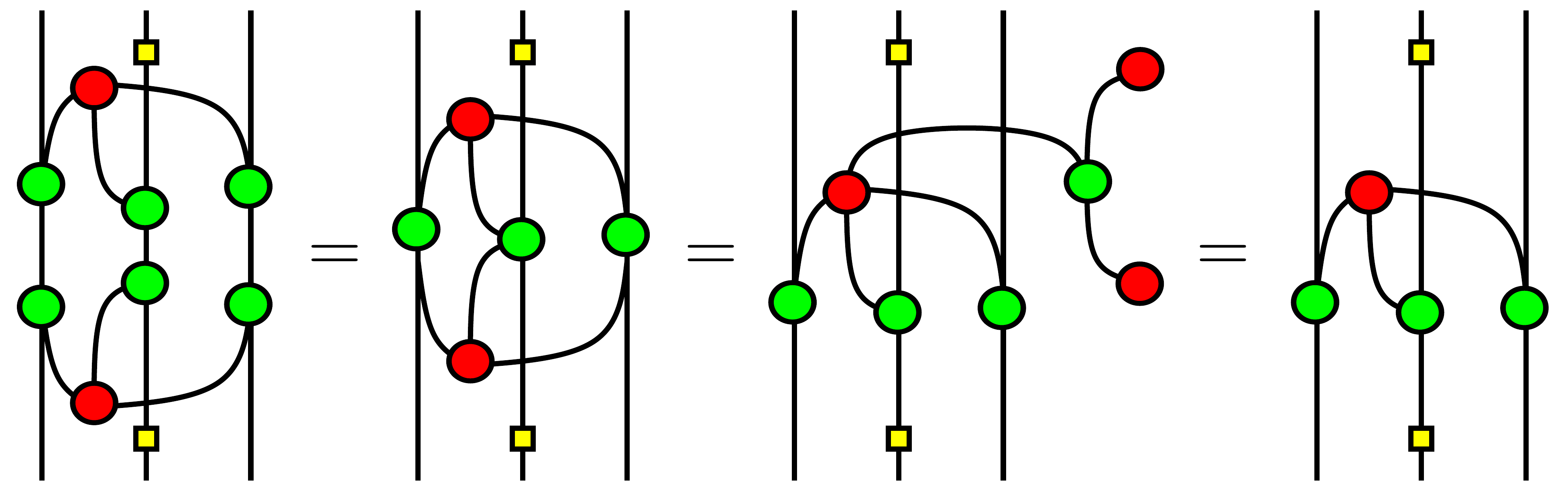}
\end{equation*}
So, the pair $\{\mathcal{X}\diamond\mathcal{X}\diamond\mathcal{X}\diamond,\mathcal{Z}\diamond\mathcal{Z}\diamond\mathcal{Z}\diamond\}$ does not satisfy Eq. \ref{eq:equiv-entangled-center}. The underlying bases of the two classical structures also do not contain equivalent states. 

Our examples show that the LHS diagram in Eq. \ref{eq:conjecture-equiv}, which we shall call the metric diagram between two classical structures, provides us with insight on the similarity and/or compatibility between the states in a pair of orthogonal bases. This warrants further study on the interaction between classical structures via their metric diagrams. We already know that the metric diagram between a classical structure and itself is the identity while the metric diagram between complementary classical structures is the RHS diagram of the complementarity equation (i.e. Eq. \ref{eq:complementarity}) where the nodes might have phases. In addition to confirming Conjecture \ref{conj:equivalentCS}, we also would like to study other types of metric diagrams, and perhaps, find a numerical metric that could represent these diagrams.

\begin{appendices}
\chapter{}
\section{Names of Generators of Entangled CD on Three Qubits}\label{sec:entCD3Q}

\begin{longtable}{ccccc}$\begin{pmatrix}0&0&0\\0&0&1\\0&0&1\\1&1&0\end{pmatrix}$&$\begin{pmatrix}0&0&1\\0&0&1\\0&0&1\\1&1&0\end{pmatrix}$&$\begin{pmatrix}0&0&Z\\0&0&1\\0&0&1\\1&1&0\end{pmatrix}$\\&&\\\fbox{$\begin{pmatrix}0&0&Z\\0&0&1\\0&0&k\\1&j&0\end{pmatrix}$}&\fbox{$\begin{pmatrix}0&0&1 + Z\\0&0&1\\0&0&k\\1&j&0\end{pmatrix}$}&\fbox{$\begin{pmatrix}0&0&Z\\0&0&1\\0&0&1 + k\\1&1 + j&0\end{pmatrix}$}\\&&\\\fbox{$\begin{pmatrix}0&0&1 + Z\\0&0&1\\0&0&1 + k\\1&1 + j&0\end{pmatrix}$}&\fbox{$\begin{pmatrix}0&1&0\\0&0&1\\0&0&1\\1&1&0\end{pmatrix}$}&\fbox{$\begin{pmatrix}0&1&1\\0&0&1\\0&0&1\\1&1&0\end{pmatrix}$}\\&&\\$\begin{pmatrix}0&1&Z\\0&0&1\\0&0&1\\1&1&0\end{pmatrix}$&\fbox{$\begin{pmatrix}0&1&Z\\0&0&1\\0&0&k\\1&j&0\end{pmatrix}$}&\fbox{$\begin{pmatrix}0&1&1 + Z\\0&0&1\\0&0&k\\1&j&0\end{pmatrix}$}\\&&\\
\end{longtable}
\end{appendices}

\addcontentsline{toc}{chapter}{Bibliography}
\bibliographystyle{plainnat}
\bibliography{library}

\begin{thebibliography}{68}
\providecommand{\natexlab}[1]{#1}
\providecommand{\url}[1]{\texttt{#1}}
\expandafter\ifx\csname urlstyle\endcsname\relax
  \providecommand{\doi}[1]{doi: #1}\else
  \providecommand{\doi}{doi: \begingroup \urlstyle{rm}\Url}\fi

\bibitem[Abramsky and Coecke(2004)]{Abramsky2004}
S.~Abramsky and B.~Coecke.
\newblock {A categorical semantics of quantum protocols}.
\newblock In \emph{Proceedings of the 19th Annual IEEE Symposium on Logic in
  Computer Science, 2004.}, pages 415--425. IEEE, 2004.
\newblock \doi{10.1109/LICS.2004.1319636}.
\newblock URL \url{http://ieeexplore.ieee.org/document/1319636/}.

\bibitem[Abramsky and Coecke(2009)]{AbramskyCQM2009}
Samson Abramsky and Bob Coecke.
\newblock {Categorical quantum mechanics}.
\newblock In Daniel~Lehmann {Kurt Engesser, Dov M. Gabbay}, editor,
  \emph{Handbook of quantum logic and quantum structures: quantum logic}, pages
  261--324. Elsevier Amsterdam, 2009.
\newblock URL \url{http://arxiv.org/abs/0808.1023}.

\bibitem[Atiyah(1988)]{Atiyah1988}
Michael Atiyah.
\newblock {Topological Quantum Field Theory}.
\newblock \emph{Publications math{\'{e}}matiques de l'I.H.{\'{E}}.S.},
  68:\penalty0 175--186, 1988.

\bibitem[Awodey(2010)]{Awodey2010}
Steve Awodey.
\newblock \emph{{Category Theory}}.
\newblock Oxford University Press, Oxford, second edi edition, 2010.
\newblock ISBN 978-0-19-958736-0.

\bibitem[Backens(2014)]{Backens2013}
Miriam Backens.
\newblock {The ZX-calculus is complete for stabilizer quantum mechanics}.
\newblock \emph{New Journal of Physics}, 16\penalty0 (9):\penalty0 93021, sep
  2014.
\newblock ISSN 1367-2630.
\newblock \doi{10.1088/1367-2630/16/9/093021}.
\newblock URL
  \url{https://iopscience.iop.org/article/10.1088/1367-2630/16/9/093021}.

\bibitem[Backens(2016)]{Backens2016}
Miriam Backens.
\newblock \emph{{Completeness and the ZX-calculus}}.
\newblock Doctoral dissertation, University of Oxford, 2016.
\newblock URL \url{https://arxiv.org/abs/1602.08954}.

\bibitem[Baez and Stay(2010)]{Baez2010PhysicsStone}
J.~Baez and M.~Stay.
\newblock {Physics, topology, logic and computation: A Rosetta Stone}.
\newblock \emph{New Structures for Physics: Lecture Notes in Physics},
  813:\penalty0 95--172, 2010.

\bibitem[Baez and Biamonte(2018)]{Baez2012}
John~C. Baez and Jacob~D. Biamonte.
\newblock \emph{{Quantum Techniques in Stochastic Mechanics}}.
\newblock WORLD SCIENTIFIC, apr 2018.
\newblock ISBN 978-981-322-693-7.
\newblock \doi{10.1142/10623}.
\newblock URL
  \url{https://www.worldscientific.com/worldscibooks/10.1142/10623}.

\bibitem[Baez and Erbele(2015)]{Baez2015CategoriesControl}
John~C. Baez and Jason Erbele.
\newblock {Categories in Control}.
\newblock \emph{Theory and Applications of Categories}, 30\penalty0
  (24):\penalty0 836--881, may 2015.
\newblock URL \url{https://arxiv.org/abs/1405.6881v3}.

\bibitem[Baez and Lauda(2011)]{Baez2009}
John~C. Baez and Aaron~D. Lauda.
\newblock {A prehistory of n-categorical physics}.
\newblock In \emph{Deep Beauty: Understanding the Quantum World Through
  Mathematical Innovation}, pages 13--128. Cambridge University Press, aug
  2011.
\newblock ISBN 9780511976971.
\newblock \doi{10.1017/CBO9780511976971.003}.
\newblock URL \url{http://arxiv.org/abs/0908.2469}.

\bibitem[Barnum et~al.(2013)Barnum, Duncan, and Wilce]{Barnum2013}
Howard Barnum, Ross Duncan, and Alexander Wilce.
\newblock {Symmetry, Compact Closure and Dagger Compactness for Categories of
  Convex Operational Models}.
\newblock \emph{Journal of Philosophical Logic}, 42\penalty0 (3):\penalty0
  501--523, 2013.
\newblock ISSN 00223611.
\newblock \doi{10.1007/s10992-013-9280-8}.

\bibitem[Bonchi et~al.(2017)Bonchi, Soboci{\'{n}}ski, and
  Zanasi]{Bonchi2017InteractingAlgebras}
Filippo Bonchi, Pawe{\l} Soboci{\'{n}}ski, and Fabio Zanasi.
\newblock {Interacting Hopf algebras}.
\newblock \emph{Journal of Pure and Applied Algebra}, 221\penalty0
  (1):\penalty0 144--184, jan 2017.
\newblock \doi{10.1016/j.jpaa.2016.06.002}.
\newblock URL
  \url{https://linkinghub.elsevier.com/retrieve/pii/S0022404916300731}.

\bibitem[Chiribella(2014)]{Chiribella2014DistinguishabilityTheories}
Giulio Chiribella.
\newblock {Distinguishability and Copiability of Programs in General Process
  Theories}.
\newblock \emph{International Journal of Software and Informatics}, 8\penalty0
  (3-4):\penalty0 209--223, 2014.
\newblock URL \url{http://arxiv.org/abs/1411.3035}.

\bibitem[Chiribella et~al.(2010)Chiribella, D'Ariano, and
  Perinotti]{Chiribella2010}
Giulio Chiribella, Giacomo~Mauro D'Ariano, and Paolo Perinotti.
\newblock {Probabilistic theories with purification}.
\newblock \emph{Physical Review A}, 81:\penalty0 062348, jun 2010.
\newblock \doi{10.1103/PhysRevA.81.062348}.
\newblock URL \url{https://link.aps.org/doi/10.1103/PhysRevA.81.062348}.

\bibitem[Chiribella et~al.(2011)Chiribella, D'Ariano, and
  Perinotti]{Chiribella2011InformationalTheory}
Giulio Chiribella, Giacomo~Mauro D'Ariano, and Paolo Perinotti.
\newblock {Informational derivation of quantum theory}.
\newblock \emph{Physical Review A}, 84\penalty0 (1):\penalty0 012311, jul 2011.
\newblock \doi{10.1103/PhysRevA.84.012311}.
\newblock URL \url{https://link.aps.org/doi/10.1103/PhysRevA.84.012311}.

\bibitem[Coecke(2006)]{Coecke2006}
Bob Coecke.
\newblock {Kindergarten Quantum Mechanics: Lecture Notes}.
\newblock \emph{AIP Conference Proceedings}, 810\penalty0 (1):\penalty0 81--98,
  2006.
\newblock \doi{10.1063/1.2158713}.
\newblock URL \url{http://aip.scitation.org/doi/abs/10.1063/1.2158713}.

\bibitem[Coecke(2007)]{Coecke2007}
Bob Coecke.
\newblock {De-linearizing Linearity: Projective Quantum Axiomatics From Strong
  Compact Closure}.
\newblock \emph{Electronic Notes in Theoretical Computer Science},
  170:\penalty0 49--72, mar 2007.
\newblock ISSN 15710661.
\newblock \doi{10.1016/j.entcs.2006.12.011}.
\newblock URL
  \url{https://linkinghub.elsevier.com/retrieve/pii/S1571066107000552}.

\bibitem[Coecke(2008)]{Coecke2008a}
Bob Coecke.
\newblock {Axiomatic Description of Mixed States From Selinger's
  CPM-construction}.
\newblock \emph{Electronic Notes in Theoretical Computer Science}, 210\penalty0
  (C):\penalty0 3--13, jul 2008.
\newblock ISSN 15710661.
\newblock \doi{10.1016/j.entcs.2008.04.014}.
\newblock URL
  \url{https://linkinghub.elsevier.com/retrieve/pii/S1571066108002296}.

\bibitem[Coecke(2011)]{Coecke2011e}
Bob Coecke.
\newblock {A universe of processes and some of its guises}.
\newblock \emph{Deep Beauty: Understanding the Quantum World through
  Mathematical Innovation}, pages 128----186, 2011.
\newblock \doi{10.1017/CBO9780511976971.004}.
\newblock URL \url{http://arxiv.org/abs/1009.3786}.

\bibitem[Coecke and Duncan(2008)]{Coecke2008InteractingObservables}
Bob Coecke and Ross Duncan.
\newblock {Interacting Quantum Observables}.
\newblock In \emph{Lecture Notes in Computer Science}, volume 5126, pages
  298--310. Springer Berlin Heidelberg, Berlin, Heidelberg, 2008.
\newblock ISBN 3540705821.
\newblock \doi{10.1007/978-3-540-70583-3_25}.
\newblock URL
  \url{https://link.springer.com/chapter/10.1007/978-3-540-70583-3_25}.

\bibitem[Coecke and Duncan(2011)]{Coecke2011b}
Bob Coecke and Ross Duncan.
\newblock {Interacting quantum observables: Categorical algebra and
  diagrammatics}.
\newblock \emph{New Journal of Physics}, 13\penalty0 (85pp):\penalty0 43016,
  apr 2011.
\newblock ISSN 13672630.
\newblock \doi{10.1088/1367-2630/13/4/043016}.
\newblock URL \url{http://www.njp.org/}.

\bibitem[Coecke and Kissinger(2010)]{Coecke2010}
Bob Coecke and Aleks Kissinger.
\newblock {The Compositional Structure of Multipartite Quantum Entanglement}.
\newblock In Samson Abramsky, Cyril Gavoille, Claude Kirchner, Friedhelm {Meyer
  auf der Heide}, and Paul~G. Spirakis, editors, \emph{Automata, Languages and
  Programming}, pages 297--308. Springer, Berlin, Heidelberg, Berlin,
  Heidelberg, 2010.
\newblock ISBN 978-3-642-14162-1.
\newblock \doi{10.1007/978-3-642-14162-1_25}.
\newblock URL \url{http://link.springer.com/10.1007/978-3-642-14162-1_25}.

\bibitem[Coecke and Kissinger(2015)]{Coecke2015CQM}
Bob Coecke and Aleks Kissinger.
\newblock {Categorical Quantum Mechanics I: Causal Quantum Processes}.
\newblock \emph{Categories for the Working Philosopher}, pages 286--328, oct
  2015.
\newblock URL \url{http://arxiv.org/abs/1510.05468}.

\bibitem[Coecke and Kissinger(2016)]{Coecke2016a}
Bob Coecke and Aleks Kissinger.
\newblock {Categorical Quantum Mechanics II: Classical-Quantum Interaction}.
\newblock \emph{International Journal of Quantum Information}, 14\penalty0 (4),
  jun 2016.
\newblock ISSN 02197499.
\newblock \doi{10.1142/S0219749916400207}.

\bibitem[Coecke and Kissinger(2017)]{Coecke2017PicturingProcesses}
Bob Coecke and Aleks Kissinger.
\newblock \emph{{Picturing Quantum Processes}}.
\newblock Cambridge University Press, Cambridge, 2017.
\newblock ISBN 9781316219317.
\newblock \doi{10.1017/9781316219317}.
\newblock URL \url{http://ebooks.cambridge.org/ref/id/CBO9781316219317}.

\bibitem[Coecke and Lal(2013)]{coecke2013causalProcesses}
Bob Coecke and Raymond Lal.
\newblock {Causal Categories: Relativistically Interacting Processes}.
\newblock \emph{Foundations of Physics}, 43\penalty0 (4):\penalty0 458--501,
  2013.
\newblock \doi{10.1007/s10701-012-9646-8}.
\newblock URL
  \url{https://link.springer.com/article/10.1007/s10701-012-9646-8}.

\bibitem[Coecke and Paquette(2008)]{Coecke2008POVMsSums}
Bob Coecke and Eric~Oliver Paquette.
\newblock {POVMs and Naimark's theorem without sums}.
\newblock \emph{Electronic Notes in Theoretical Computer Science},
  210:\penalty0 15--31, aug 2008.
\newblock URL \url{http://arxiv.org/abs/quant-ph/0608072}.

\bibitem[Coecke and Pavlovic(2008)]{Coecke2008c}
Bob Coecke and Dusko Pavlovic.
\newblock {Quantum measurements without sums}.
\newblock In Goong Chen, Louis Kauffman, and Samuel~J Lomonaco, editors,
  \emph{Mathematics of Quantum Computation and Quantum Technology}, chapter~16,
  pages 559--596. Chapman and Hall/CRC, sep 2008.
\newblock ISBN 978-1-58488-899-4.
\newblock \doi{10.1201/9781584889007.ch16}.
\newblock URL \url{http://arxiv.org/abs/quant-ph/0608035}.

\bibitem[Coecke and Perdrix(2012)]{Coecke2012EnvironmentMechanics}
Bob Coecke and Simon Perdrix.
\newblock {Environment and classical channels in categorical quantum
  mechanics}.
\newblock \emph{Logical Methods in Computer Science}, 8\penalty0 (4), nov 2012.
\newblock ISSN 18605974.
\newblock \doi{10.2168/LMCS-8(4:14)2012}.
\newblock URL \url{http://dx.doi.org/10.2168/LMCS-8(4:14)2012}.

\bibitem[Coecke et~al.(2009)Coecke, Paquette, and Pavlovic]{Coecke2010a}
Bob Coecke, {\'{E}}ric~Oliver Paquette, and Dusko Pavlovic.
\newblock {Classical and Quantum Structuralism}.
\newblock In Simon Gay and Ian Mackie, editors, \emph{Semantic Techniques in
  Quantum Computation}, pages 29--69. Cambridge University Press, Cambridge,
  2009.
\newblock ISBN 9781139193313.
\newblock \doi{10.1017/CBO9781139193313.003}.
\newblock URL
  \url{https://www.cambridge.org/core/product/identifier/CBO9781139193313A009/type/book_part}.

\bibitem[Coecke et~al.(2011)Coecke, Edwards, and Spekkens]{Coecke2011c}
Bob Coecke, Bill Edwards, and Robert~W. Spekkens.
\newblock {Phase Groups and the Origin of Non-locality for Qubits}.
\newblock \emph{Electronic Notes in Theoretical Computer Science}, 270\penalty0
  (2):\penalty0 15--36, feb 2011.
\newblock \doi{10.1016/j.entcs.2011.01.021}.
\newblock URL
  \url{https://linkinghub.elsevier.com/retrieve/pii/S1571066111000223}.

\bibitem[Coecke et~al.(2012)Coecke, Duncan, Kissinger, and Wang]{Coecke2012}
Bob Coecke, Ross Duncan, Aleks Kissinger, and Quanlong Wang.
\newblock {Strong Complementarity and Non-locality in Categorical Quantum
  Mechanics}.
\newblock In \emph{2012 27th Annual IEEE Symposium on Logic in Computer
  Science}, pages 245--254. IEEE, jun 2012.
\newblock \doi{10.1109/LICS.2012.35}.
\newblock URL \url{http://ieeexplore.ieee.org/document/6280443/}.

\bibitem[Coecke et~al.(2013{\natexlab{a}})Coecke, Heunen, and
  Kissinger]{Coecke2013}
Bob Coecke, Chris Heunen, and Aleks Kissinger.
\newblock {Compositional Quantum Logic}.
\newblock In Bob Coecke, Luke Ong, and Prakash Panangaden, editors,
  \emph{Computation, Logic, Games, and Quantum Foundations. The Many Facets of
  Samson Abramsky. Lecture Notes in Computer Science}, volume 7860, pages
  21--36. Springer, Berlin, Heidelberg, 2013{\natexlab{a}}.
\newblock ISBN 978-3-642-38163-8.
\newblock \doi{10.1007/978-3-642-38164-5_3}.
\newblock URL \url{http://link.springer.com/10.1007/978-3-642-38164-5_3}.

\bibitem[Coecke et~al.(2013{\natexlab{b}})Coecke, Pavlovic, and
  Vicary]{Coecke2013New}
Bob Coecke, Dusko Pavlovic, and Jamie Vicary.
\newblock {A new description of orthogonal bases}.
\newblock \emph{Mathematical Structures in Computer Science}, 23\penalty0
  (3):\penalty0 555--567, 2013{\natexlab{b}}.
\newblock \doi{10.1017/S0960129512000047}.
\newblock URL \url{http://arxiv.org/abs/0810.0812}.

\bibitem[Coecke et~al.(2014)Coecke, Heunen, and
  Kissinger]{Coecke2014CategoriesChannels}
Bob Coecke, Chris Heunen, and Aleks Kissinger.
\newblock {Categories of quantum and classical channels}.
\newblock \emph{Quantum Information Processing}, 15\penalty0 (12):\penalty0
  5179--5209, 2014.
\newblock \doi{10.1007/s11128-014-0837-4}.
\newblock URL \url{http://doi.org/10.1007/s11128-014-0837-4}.

\bibitem[Coecke et~al.(2016{\natexlab{a}})Coecke, Duncan, Kissinger, and
  Wang]{Coecke2015}
Bob Coecke, Ross Duncan, Aleks Kissinger, and Quanlong Wang.
\newblock {Generalised Compositional Theories and Diagrammatic Reasoning}.
\newblock In Giulio Chiribella and Robert~W. Spekkens, editors,
  \emph{Fundamental Theories of Physics}, volume 181, pages 309--366. Springer,
  Dordrecht, 2016{\natexlab{a}}.
\newblock ISBN 9789401773034.
\newblock \doi{10.1007/978-94-017-7303-4_10}.
\newblock URL
  \url{https://link.springer.com/chapter/10.1007/978-94-017-7303-4_10}.

\bibitem[Coecke et~al.(2016{\natexlab{b}})Coecke, Fritz, and
  Spekkens]{Coecke2016}
Bob Coecke, Tobias Fritz, and Robert~W. Spekkens.
\newblock {A mathematical theory of resources}.
\newblock \emph{Information and Computation}, 250:\penalty0 59--86, oct
  2016{\natexlab{b}}.
\newblock \doi{10.1016/j.ic.2016.02.008}.
\newblock URL
  \url{https://linkinghub.elsevier.com/retrieve/pii/S0890540116000353}.

\bibitem[Dixon and Duncan(2009)]{Dixon2009}
Lucas Dixon and Ross Duncan.
\newblock {Graphical reasoning in compact closed categories for quantum
  computation}.
\newblock \emph{Annals of Mathematics and Artificial Intelligence},
  56:\penalty0 23, may 2009.
\newblock \doi{10.1007/s10472-009-9141-x}.
\newblock URL \url{http://link.springer.com/10.1007/s10472-009-9141-x}.

\bibitem[Duncan(2012)]{Duncan2012}
Ross Duncan.
\newblock {A graphical approach to measurement-based quantum computing}.
\newblock \emph{arXiv preprint arXiv:1203.6242}, 2012.
\newblock URL \url{http://arxiv.org/abs/1203.6242}.

\bibitem[Durt et~al.(2010)Durt, Englert, Bengtsson, and
  {\.{Z}}yczkowski]{DurtMUB2010}
Thomas Durt, Berthold-Georg Englert, Ingemar Bengtsson, and Karol
  {\.{Z}}yczkowski.
\newblock {On Mutually Unbiased Bases}.
\newblock \emph{International Journal of Quantum Information}, 08\penalty0
  (04):\penalty0 535--640, jun 2010.
\newblock ISSN 0219-7499.
\newblock \doi{10.1142/S0219749910006502}.
\newblock URL
  \url{https://www.worldscientific.com/doi/abs/10.1142/S0219749910006502}.

\bibitem[Eilenberg and MacLane(1945)]{Eilenberg1945}
Samuel Eilenberg and Saunders MacLane.
\newblock {General Theory of Natural Equivalences}.
\newblock \emph{Transactions of the American Mathematical Society}, 58\penalty0
  (2):\penalty0 231, 1945.
\newblock \doi{10.2307/1990284}.

\bibitem[Evans et~al.(2009)Evans, Duncan, Lang, and Panangaden]{Evans2009MUB}
Julia Evans, Ross Duncan, Alex Lang, and Prakash Panangaden.
\newblock {Classifying all mutually unbiased bases in Rel}.
\newblock \emph{arxiv preprint}, sep 2009.
\newblock URL \url{http://arxiv.org/abs/0909.4453}.

\bibitem[Fong(2016)]{Fong2016}
Brendan Fong.
\newblock \emph{{The Algebra of Open and Interconnected Systems}}.
\newblock PhD thesis, University of Oxford, 2016.
\newblock URL \url{http://arxiv.org/abs/1609.05382}.

\bibitem[Harding(2009)]{Harding2009}
John Harding.
\newblock {A Link between Quantum Logic and Categorical Quantum Mechanics}.
\newblock \emph{International Journal of Theoretical Physics}, 48:\penalty0
  769--802, mar 2009.
\newblock \doi{10.1007/s10773-008-9853-4}.
\newblock URL \url{http://link.springer.com/10.1007/s10773-008-9853-4}.

\bibitem[Heunen(2012)]{Heunen2012a}
Chris Heunen.
\newblock {Complementarity in Categorical Quantum Mechanics}.
\newblock \emph{Foundations of Physics}, 42:\penalty0 856--873, jul 2012.
\newblock \doi{10.1007/s10701-011-9585-9}.
\newblock URL \url{http://link.springer.com/10.1007/s10701-011-9585-9}.

\bibitem[Heunen et~al.(2014)Heunen, Kissinger, and Selinger]{Heunen2014}
Chris Heunen, Aleks Kissinger, and Peter Selinger.
\newblock {Completely positive projections and biproducts}.
\newblock In \emph{Electronic Proceedings in Theoretical Computer Science},
  volume 171, pages 71--83, dec 2014.
\newblock \doi{10.4204/EPTCS.171.7}.
\newblock URL \url{http://arxiv.org/abs/1308.4557v2}.

\bibitem[Inc.()]{WolframComputing}
Wolfram~Research{,} Inc.
\newblock {Mathematica, {V}ersion 12.3}.
\newblock URL \url{https://www.wolfram.com/mathematica/}.
\newblock Champaign, IL, 2021.

\bibitem[Jeandel et~al.(2019)Jeandel, Perdrix, and
  Vilmart]{Jeandel2019CompletenessZX-Calculus}
Emmanuel Jeandel, Simon Perdrix, and Renaud Vilmart.
\newblock {Completeness of the ZX-Calculus}.
\newblock \emph{Logical Methods in Computer Science}, 16\penalty0 (2), mar
  2019.
\newblock \doi{10.23638/LMCS-16(2:11)2020}.
\newblock URL \url{http://dx.doi.org/10.23638/LMCS-16(2:11)2020}.

\bibitem[Joyal and Street(1991)]{Joyal1991}
Andr{\'{e}} Joyal and Ross Street.
\newblock {The geometry of tensor calculus, I}.
\newblock \emph{Advances in Mathematics}, 88\penalty0 (1):\penalty0 55--112,
  apr 1991.
\newblock \doi{10.1016/0001-8708(91)90003-P}.
\newblock URL
  \url{http://linkinghub.elsevier.com/retrieve/pii/000187089190003P}.

\bibitem[Joyal and Street(1993)]{Joyal1993BraidedCategories}
Andr{\'{e}} Joyal and Ross Street.
\newblock {Braided Tensor Categories}, 1993.
\newblock ISSN 00018708.

\bibitem[Kissinger and Zamdzhiev(2015)]{Kissinger2015}
Aleks Kissinger and Vladimir Zamdzhiev.
\newblock {Quantomatic: A Proof Assistant for Diagrammatic Reasoning}.
\newblock In Amy~P. Felty and Aart Middeldorp, editors, \emph{Lecture Notes in
  Computer Science}, volume 9195, pages 326--336. Springer, Cham, 2015.
\newblock ISBN 978-3-319-21400-9.
\newblock \doi{10.1007/978-3-319-21401-6_22}.
\newblock URL \url{http://link.springer.com/10.1007/978-3-319-21401-6_22}.

\bibitem[Kissinger et~al.(2014)Kissinger, Merry, and Soloviev]{Kissinger2014}
Aleks Kissinger, Alex Merry, and Matvey Soloviev.
\newblock {Pattern graph rewrite systems}.
\newblock \emph{Electronic Proceedings in Theoretical Computer Science},
  143:\penalty0 54--66, mar 2014.
\newblock \doi{10.4204/EPTCS.143.5}.
\newblock URL \url{http://arxiv.org/abs/1204.6695v3}.

\bibitem[Kissinger et~al.(2018)Kissinger, Merry, Frot, Coecke, Quick,
  Miller-Bakewell, Dixon, Matyev, Duncan, and Zamdzhiev]{Quantomatic}
Aleks Kissinger, Alex Merry, Ben Frot, Bob Coecke, David Quick, Hector
  Miller-Bakewell, Lucas Dixon, Soloviev Matyev, Ross Duncan, and Vladimir
  Zamdzhiev.
\newblock {Quantomatic}.
\newblock \emph{https://quantomatic.github.io/}, 2018.
\newblock URL \url{https://quantomatic.github.io/about.html}.

\bibitem[Kock(2004)]{Kock2004FrobeniusTheories}
Joachim Kock.
\newblock \emph{{Frobenius Algebras and 2D Topological Quantum Field
  Theories}}.
\newblock Cambridge University Press, 2004.
\newblock ISBN 9780898713800.

\bibitem[Lack(2004)]{Lack2004}
Stephen Lack.
\newblock {Composing Props}.
\newblock \emph{Theory and Applications of Categories}, 13\penalty0
  (9):\penalty0 147--163, 2004.
\newblock URL \url{http://www.tac.mta.ca/tac/volumes/13/9/13-09abs.html}.

\bibitem[Lambek and {J. Scott}(1988)]{Lambek1988IntroductionLogic}
Joachim Lambek and Philip {J. Scott}.
\newblock \emph{{Introduction to higher-order categorical logic}}.
\newblock Cambridge University Press, 1988.
\newblock ISBN 97805213565634.

\bibitem[{Mac Lane}(1988)]{MacLane1988ConceptsPerspective}
Saunders {Mac Lane}.
\newblock {Concepts and Categories in Perspective}.
\newblock In Peter Duren, Richard~A. Askey, and Uta~C. Merzbach, editors,
  \emph{A Century of Mathematics in America, Part I}, pages 323--365. American
  Mathematical Society, 1988.
\newblock ISBN 0-8218-0124-4.
\newblock URL \url{www.ams.org/publicoutreach/math-history/hmath1-index}.

\bibitem[Merry(2014)]{Merry2014}
Alexander Merry.
\newblock \emph{{Reasoning with !-Graphs}}.
\newblock Doctoral dissertation, University of Oxford, 2014.
\newblock URL \url{http://arxiv.org/abs/1403.7828}.

\bibitem[Musto(2017)]{Musto2016}
Benjamin Musto.
\newblock {Constructing Mutually Unbiased Bases from Quantum Latin Squares}.
\newblock \emph{Electronic Proceedings in Theoretical Computer Science},
  236:\penalty0 108--126, jan 2017.
\newblock ISSN 2075-2180.
\newblock \doi{10.4204/EPTCS.236.8}.
\newblock URL \url{http://arxiv.org/abs/1605.08919v2}.

\bibitem[Piedeleu et~al.(2015)Piedeleu, Kartsaklis, Coecke, and
  Sadrzadeh]{Piedeleu2015OpenProcessing}
Robin Piedeleu, Dimitri Kartsaklis, Bob Coecke, and Mehrnoosh Sadrzadeh.
\newblock {Open System Categorical Quantum Semantics in Natural Language
  Processing}.
\newblock \emph{Leibniz International Proceedings in Informatics, LIPIcs},
  35:\penalty0 270--289, feb 2015.
\newblock URL \url{http://arxiv.org/abs/1502.00831}.

\bibitem[Romero et~al.(2005)Romero, Bj{\"{o}}rk, Klimov, and
  S{\'{a}}nchez-Soto]{Romero2005}
J.~L. Romero, G.~Bj{\"{o}}rk, A.~B. Klimov, and L.~L. S{\'{a}}nchez-Soto.
\newblock {Structure of the sets of mutually unbiased bases for N qubits}.
\newblock \emph{Physical Review A}, 72\penalty0 (6):\penalty0 062310, 2005.
\newblock \doi{10.1103/PhysRevA.72.062310}.

\bibitem[Selby and Coecke(2017)]{Selby2017Adjoint}
John~H. Selby and Bob Coecke.
\newblock {A Diagrammatic Derivation of the Hermitian Adjoint}.
\newblock \emph{Foundations of Physics}, 47:\penalty0 1191--1207, sep 2017.
\newblock \doi{10.1007/s10701-017-0102-7}.
\newblock URL
  \url{https://link.springer.com/article/10.1007/s10701-017-0102-7}.

\bibitem[Selinger(2011)]{Selinger2011}
P.~Selinger.
\newblock {A survey of graphical languages for monoidal categories}.
\newblock \emph{Lecture Notes in Physics}, 813:\penalty0 289--355, 2011.
\newblock ISSN 00758450.
\newblock \doi{10.1007/978-3-642-12821-9_4}.

\bibitem[Selinger(2007)]{Selinger2007a}
Peter Selinger.
\newblock {Dagger Compact Closed Categories and Completely Positive Maps}.
\newblock \emph{Electronic Notes in Theoretical Computer Science},
  170:\penalty0 139--163, mar 2007.
\newblock \doi{10.1016/j.entcs.2006.12.018}.
\newblock URL
  \url{https://linkinghub.elsevier.com/retrieve/pii/S1571066107000606}.

\bibitem[Soboci{\'{n}}ski(2015)]{Sobocinski201510.Algebra}
Pawe{\l} Soboci{\'{n}}ski.
\newblock {10. Paths and Matrices | Graphical Linear Algebra}, 2015.
\newblock URL
  \url{https://graphicallinearalgebra.net/2015/05/19/paths-and-matrices-part-1/}.

\bibitem[Vidal and Dawson(2004)]{Vidal2004UniversalGates}
G.~Vidal and C.~M. Dawson.
\newblock {Universal quantum circuit for two-qubit transformations with three
  controlled-NOT gates}.
\newblock \emph{Physical Review A}, 69\penalty0 (1):\penalty0 010301, jan 2004.
\newblock ISSN 1050-2947.
\newblock \doi{10.1103/PhysRevA.69.010301}.
\newblock URL \url{https://link.aps.org/doi/10.1103/PhysRevA.69.010301}.

\bibitem[Wang(2018)]{Wang2018qutritZX}
Quanlong Wang.
\newblock {Qutrit ZX-calculus is Complete for Stabilizer Quantum Mechanics}.
\newblock \emph{Electronic Proceedings in Theoretical Computer Science},
  266\penalty0 (10):\penalty0 58--70, feb 2018.
\newblock ISSN 2075-2180.
\newblock \doi{10.4204/EPTCS.266.3}.
\newblock URL \url{http://arxiv.org/abs/1803.00696v1}.

\bibitem[Wootters and Fields(1989)]{Wootters1989}
William~K Wootters and Brian~D Fields.
\newblock {Optimal state-determination by mutually unbiased measurements}.
\newblock \emph{Annals of Physics}, 191\penalty0 (2):\penalty0 363--381, may
  1989.
\newblock \doi{10.1016/0003-4916(89)90322-9}.

\end{thebibliography}

%\input{biodata}

%\input{publication}

%\includepdf[height=\textheight,keepaspectratio]{images/STATUS-CONFIRMATION-FORM}
\end{document}